\documentclass[11pt]{article}

\usepackage{ifthen}
\newboolean{margins}
\setboolean{margins}{true} 
\newboolean{draft}
\setboolean{draft}{false}   

\ifthenelse{\boolean{margins}}{%
  \usepackage[margin=1in,top=1in,bottom=1in]{geometry} 
  \setlength{\marginparwidth}{0pt}
  \setlength{\marginparsep}{0pt}
  \raggedbottom%
  \setlength{\parskip}{\medskipamount}%
  
}{
}

\usepackage[T1]{fontenc}
\usepackage[utf8]{inputenc} 
\usepackage[normalem]{ulem}  
\usepackage[dvipsnames]{xcolor}

\usepackage{amsmath,amsthm,amssymb,mathtools,thmtools}
\usepackage{bbm}
\usepackage{physics}
\usepackage{dsfont}

\renewcommand{\inf}{\operatorname*{inf\vphantom{\mathrm{p}}}}
\renewcommand{\lim}{\operatorname*{lim\vphantom{\mathrm{p}}}}

\usepackage[pdftex]{graphicx}
\usepackage{tikz}
\usepackage{upgreek} 
\usetikzlibrary{patterns}
\usepackage{tikz-cd}
\usepackage{qcircuit}
\usepackage{caption,subcaption}
\usepackage{float}

\usepackage{tabularx,booktabs,multirow}
\usepackage{multicol}
\usepackage{rotating}
\usepackage{makecell}

\usepackage{algorithm}
\usepackage{algorithmic}
\usepackage{tcolorbox}
\tcbuselibrary{listings,breakable}
\usepackage[linewidth=1pt]{mdframed}

\usepackage[backend=biber,
              style=alphabetic,
        maxbibnames=99,
      minalphanames=3,
      maxalphanames=4,
            backref=true]{biblatex}

\usepackage{amsthm}
\usepackage{thmtools}

\declaretheorem[name=Theorem,numberwithin=section]{theorem}
\declaretheorem[sibling=theorem]{lemma}
\declaretheorem[sibling=theorem]{proposition}
\declaretheorem[sibling=theorem]{corollary}

\declaretheorem[sibling=theorem]{claim}
\declaretheorem[name=Definition,style=definition,sibling=theorem]{definition}

\usepackage[pagebackref=false]{hyperref}

\hypersetup{
  plainpages=false,
  unicode=false,
  pdftoolbar=true,
  pdfmenubar=true,
  pdffitwindow=false,
  pdfstartview={FitH},
  pdftitle={},
  pdfauthor={},
  pdfsubject={},
  pdfnewwindow=true,
  colorlinks=true,
  linkcolor=blue,
  citecolor=green,
  filecolor=magenta,
  urlcolor=cyan
}

\usepackage[nameinlink,noabbrev]{cleveref}



\crefname{theorem}{theorem}{theorems}
\Crefname{theorem}{Theorem}{Theorems}

\crefname{lemma}{lemma}{lemmas}
\Crefname{lemma}{Lemma}{Lemmas}

\crefname{proposition}{proposition}{propositions}
\Crefname{proposition}{Proposition}{Propositions}

\crefname{corollary}{corollary}{corollaries}
\Crefname{corollary}{Corollary}{Corollaries}

\crefname{conjecture}{conjecture}{conjectures}
\Crefname{conjecture}{Conjecture}{Conjectures}

\crefname{claim}{claim}{claims}
\Crefname{claim}{Claim}{Claims}

\crefname{definition}{definition}{definitions}
\Crefname{definition}{Definition}{Definitions}

\crefname{assumption}{assumption}{assumptions}
\Crefname{assumption}{Assumption}{Assumptions}

\crefname{example}{example}{examples}
\Crefname{example}{Example}{Examples}

\usepackage{enumerate}
\usepackage{thm-restate}

\usepackage{comment}

\usepackage{xparse}

\NewDocumentCommand{\progbox}{m +m +m}{%
  \begingroup
  \renewcommand{\arraystretch}{1.15}%
  \ensuremath{%
    \begin{array}[t]{@{}c@{\quad}l@{}}
      \displaystyle #1
      &
      \displaystyle #2
      \\[0.15em]
      \mathrm{s.t.}
      &
      \begin{array}[t]{@{}l@{\quad}l@{}}
        #3
      \end{array}
    \end{array}%
  }%
  \endgroup
}

\NewDocumentCommand{\progline}{m +m +m}{%
  \ensuremath{%
    \displaystyle #1\; #2
    \quad \mathrm{s.t.}\quad
    \begin{aligned}[t]
      #3
    \end{aligned}%
  }%
}

\NewDocumentCommand{\progat}{m m +m +m}{%
  \begingroup
  \ensuremath{%
    \displaystyle
    \hbox to \dimexpr#1\relax{%
      \(\displaystyle #2\; #3\quad\)\hfil
    }%
    \mathrm{s.t.}\quad
    \begin{aligned}[t]
      #4
    \end{aligned}%
  }%
  \endgroup
}

\NewDocumentEnvironment{progalign}{}{%
  \begingroup
  \renewcommand{\arraystretch}{1.15}%
  \begin{array}[t]{@{}r@{}l@{\;}l@{\quad}c@{\quad}l@{}}
}{%
  \end{array}%
  \endgroup
}

\NewDocumentCommand{\progrow}{O{=} m m +m +m}{%
  \displaystyle #2
  & \mathrel{#1}\displaystyle #3
  & \displaystyle #4
  & \mathrm{s.t.}
  & \begin{array}[t]{@{}l@{\quad}l@{}}#5\end{array}
}

\NewDocumentCommand{\progboxrow}{O{=} m m +m +m}{%
  \displaystyle #2
  & \multicolumn{4}{l@{}}{\mathrel{#1}\progbox{#3}{#4}{#5}}
}

\NewDocumentCommand{\progexprrow}{O{=} m +m}{%
  \displaystyle #2
  & \multicolumn{4}{l@{}}{\mathrel{#1}\displaystyle #3}
}

\makeatletter
\ExplSyntaxOn
\bool_new:N \l__crefitem_in_statement_bool
\tl_new:N \l__crefitem_separator_tl
\tl_new:N \l__crefitem_parent_tl
\tl_new:N \l__crefitem_statement_tl
\cs_new_protected:Npn \__crefitem_capture_parent:
  {
    \bool_set_false:N \l__crefitem_in_statement_bool
    \exp_args:NnV \clist_if_in:nnT
      {theorem,lemma,proposition,corollary,conjecture,claim,definition,assumption,example}
      \thmt@envname
      {
        \bool_set_true:N \l__crefitem_in_statement_bool
        \tl_set:Nx \l__crefitem_statement_tl {crefitem:parent:\@currentHref}
        \exp_args:NV \label \l__crefitem_statement_tl
      }
  }
\addtotheorempostheadhook{\__crefitem_capture_parent:}
\cs_new_protected:Npn \__crefitem_record_parent:n #1
  {
    \bool_if:NT \l__crefitem_in_statement_bool
      {
        \exp_args:NnV \regex_match:nnT { \A enum [iv]+ \Z } \@currentcounter
          {
            \protected@write\@auxout{}
              {
                \string\newlabel{#1@parent}
                  {{\l__crefitem_statement_tl}{}}
              }
          }
      }
  }
\AddToHookWithArguments{label}{\__crefitem_record_parent:n {#1}}
\cs_new_protected:Npn \__crefitem_print:nn #1#2
  {
    \hyperref[#1]
      {
        \Cref*{#1}
        \tl_clear:N \l__crefitem_separator_tl
        \clist_map_inline:nn {#2}
          {
            \tl_use:N \l__crefitem_separator_tl
            (
            \regex_match:nnTF { \A ( [ivxlcdmIVXLCDM]+ | [0-9]+ ) \Z } {##1}
              {##1}
              {\ref*{##1}}
            )
            \tl_set:Nn \l__crefitem_separator_tl {,}
          }
      }
  }
\NewDocumentCommand{\Crefitem}{o m}
  {
    \group_begin:
    \IfValueTF{#1}
      {\__crefitem_print:nn {#2}{#1}}
      {
        \tl_set:Nx \l__crefitem_parent_tl {\clist_item:nn {#2}{1}@parent}
        \cs_if_exist:cT {r@\l__crefitem_parent_tl}
          {
            \tl_set:Nx \l__crefitem_parent_tl
              {\getrefnumber{\l__crefitem_parent_tl}}
          }
        \exp_args:NV \__crefitem_print:nn \l__crefitem_parent_tl {#2}
      }
    \group_end:
  }
\ExplSyntaxOff
\makeatother

\newcommand{\classes}[1]{\mathsf{#1}}       

\renewcommand{\P}{\classes{P}}

\DeclareMathOperator*{\E}{\mathbbm{E}}

\newcommand{\supp}{\textsf{supp}}

\newcommand{\dom}{\textsf{dom}}

\newcommand{\bN}{\mathbb{N}}
\newcommand{\bR}{\mathbb{R}}

\newcommand{\bC}{\mathbb{C}}

\newcommand{\cL}{\mathcal{L}}
\newcommand{\cE}{\mathcal{E}}

\newcommand{\cP}{\mathcal{P}}
\newcommand{\cF}{\mathcal{F}}
\newcommand{\cG}{\mathcal{G}}
\newcommand{\cH}{\mathcal{H}}
\newcommand{\cU}{\mathcal{U}}
\newcommand{\cW}{\mathcal{W}}
\newcommand{\cR}{\mathcal{R}}
\newcommand{\cA}{\mathcal{A}}

\newcommand{\cK}{\mathcal{K}}
\newcommand{\cS}{\mathcal{S}}
\newcommand{\cX}{\mathcal{X}}
\newcommand{\cY}{\mathcal{Y}}
\newcommand{\cZ}{\mathcal{Z}}

\newcommand{\cD}{\mathcal{D}}
\newcommand{\cV}{\mathcal{V}}
\newcommand{\cO}{\mathcal{O}}
\newcommand{\cI}{\mathcal{I}}

\newcommand{\bx}{\mathbf{x}}
\newcommand{\by}{\mathbf{y}}
\newcommand{\bz}{\mathbf{z}}
\newcommand{\bg}{\mathbf{g}}

\newcommand{\bG}{\mathbf{G}}

\newcommand{\tx}{\tilde{x}}
\newcommand{\ty}{\tilde{y}}
\newcommand{\tz}{\tilde{z}}

\renewcommand{\iff}{\Leftrightarrow}
\renewcommand{\to}{\rightarrow}
\newcommand{\one}{\mathbf{1}}

\newcommand{\ind}{\mathds{1}}

\newcommand{\eps}{\epsilon}

\usepackage{scalerel, stackengine, trimclip}
\stackMath
\newcommand\fo[1]{%
\savestack{\tmpbox}{\stretchto{%
  \scaleto{%
    \scalerel*[\widthof{\ensuremath{#1}}]{\kern-.6pt\bigwedge\kern-.6pt}%
    {\rule[-\textheight/2]{1ex}{\textheight}}
  }{\textheight}%
}{0.5ex}}%
\stackon[1pt]{#1}{\tmpbox}%
}

\newcommand{\bk}[2]{\braket{#1}{#2}}

\newcommand{\kbr}[1]{\ketbra{#1}{#1}}

\newcommand{\doubletilde}[1]{\widetilde{\raisebox{0pt}[.9\height][0pt]{\ensuremath{\widetilde{#1}}}}}
\bibliography{references}

\title{Distributional Quantum Query Complexity}

\author{ Shalev Ben{-}David\\
        \small Institute for Quantum Computing\\
        \small University of Waterloo\\
        \small \texttt{shalev.b@uwaterloo.ca}
        \and
        M. H.\ Ebtehaj\\
        \small Institute for Quantum Computing\\
        \small University of Waterloo\\
        \small \texttt{mhebtehaj@uwaterloo.ca} }

\date{}

\begin{document}

\maketitle

\begin{abstract}
Quantum query complexity enjoys a variety of pleasing joint computation properties: for example, a
composition theorem asserting $Q(f\circ g)=\Theta(Q(f)Q(g))$ for all Boolean functions $f$ and $g$;
a direct sum theorem asserting that computing $k$ copies of a function (or search problem) costs
$\Omega(k)$ times as much as the cost of computing one copy; and a direct product theorem asserting
that for Boolean functions, even succeeding at the direct sum problem with exponentially small
probability still requires $\Omega(k)$ times the cost of computing one copy to bounded error.

However, all of these results are strictly for worst-case quantum query complexity. For example, if
we have a fixed distribution $\mu$ over inputs, the direct sum theorem says nothing about the
quantum query complexity of computing $k$ copies of $f$ when the input comes from the product
distribution $\mu^k$ instead of being worst-case. (Note that while a standard Yao-type minimax
theorem guarantees a hard distribution for the direct sum problem, there's no guarantee that this
hard distribution is a product distribution.) A similar problem occurs for the direct product
theorem and the composition theorem: none of these results respect distributions.

In this work, we give \emph{distributional} joint computation lower bounds for the composition,
direct sum, and direct product problems. Along the way, we introduce some new tools for handling
quantum query lower bounds, including (a) a new ``multiplicative'' variant of the $\gamma_2$ norm
(which we use in place of the multiplicative adversary method for proving the direct product
theorem), and (b) a new ``Shaltiel-free'' measure of quantum query complexity, which we show
characterizes the composition behavior of distributional quantum query complexity and satisfies
pleasing properties.
\end{abstract}

\setcounter{tocdepth}{2}
\tableofcontents


\section{Introduction}\label{sec:introduction}


Quantum query complexity is usually studied as
worst-case complexity: the quantum query complexity
of a function $f\colon\{0,1\}^n\to\{0,1\}$, denoted
$Q(f)$, is the minimum number of queries
that can be made by a quantum algorithm which computes
$f$ to error at most $1/3$ on all inputs.
However, it is often of interest to study quantum query
complexity against a known distribution over the inputs.
If $\mu$ is a distribution over $\{0,1\}^n$,
we denote by $Q^\mu(f)$ the minimum number of queries
that can be made by a quantum algorithm which computes
$f$ to error at most $1/3$ when the input is sampled from
$\mu$; this time, the error is with respect to both
the internal randomness of the algorithm and the randomness
of the sampled input string.

A variety of joint computation results are known for
quantum query complexity. One notable example is
the composition theorem:
$Q(f\circ g)=\Theta(Q(f)Q(g))$, where $f\circ g$
is the composition of $f$ with $n$ independent copies
of $g$ \cite{HLS07,LMR+11,Kim13}. Another important example
is the direct product theorem \cite{LR12}:
$Q_{1-0.9^k}(f^k)=\Omega(kQ(f))$,
where here $f^k$ denotes the task of computing $k$
copies of $f$, and $Q_{1-0.9^k}$ denotes
the quantum query complexity when the error must be
at most $1-0.9^k$ (that is, the success probability
must be at least $0.9^k$). The direct product theorem
says that solving $k$ copies of $f$ requires $k$
times as many queries as solving one copy, and this remains
true even if we are only required to solve the $k$ copies
with exponentially small probability. An easier version
of the direct product problem is known as the direct sum
theorem, which says that computing $k$ copies
of $f$ to bounded error (rather than with exponentially
small success) requires $\Omega(kQ(f))$ queries.

While these powerful theorems are known
for worst-case quantum query complexity, they are not
known in the distributional setting. This is because
the proofs of these theorems use the negative weight
adversary method and/or the multiplicative adversary method,
which work naturally in the worst-case
setting
(though see \cite{Bel15,BBH+18,Bel24,JZ25} for previous work
on distributional adversary bounds).
In this work,
we give the first distributional composition theorem,
direct sum theorem, and direct product theorem for quantum
query complexity. In such joint computation tasks,
the upper bounds are generally easy, at least up to log
factors; we focus our efforts on lower bounding the
distributional complexity of these joint computation
problems.

We view this setting as naturally motivated: it is often the
case that we know the distribution from which the input
will come, but still wish to lower bound the quantum
query complexity. For instance, this occurs in cryptographic
applications, in which a secret key may be sampled from a known distribution but a quantum lower bound is still desired. We believe that a
distributional direct product theorem is a particularly
promising tool for cryptographic applications, since it
gives a lower bound even against an algorithm that succeeds
with very low probability, and the hard distribution
is easy to sample from (it is a product distribution).

\subsection{Our results}

\subsubsection{Composition theorem}
Let $f$ be a possibly partial Boolean function on $n$
bits and let $g$ be a possibly partial Boolean function
on $m$ bits. Let $\mu$ be a distribution over the domain
of $f$, and let $\nu_0,\nu_1$ be distributions over
$g^{-1}(0)$ and $g^{-1}(1)$ respectively. The composition
of $\mu$ with $(\nu_0,\nu_1)$, denoted
$\mu\circ(\nu_0,\nu_1)$, is the distribution over inputs
to $f\circ g$ that we get by sampling an $n$-bit string $y$
from $\mu$, then for each of the $n$ bits $y_i$ of $y$,
sampling a string $x^{(i)}$ from $\nu_{y_i}$, and
concatenating these strings to form the final string
$x=x^{(1)}x^{(2)}\dots x^{(n)}$ of length $mn$.
This distribution $\mu\circ(\nu_0,\nu_1)$ is the natural
way of composing a distribution $\mu$ for $f$ with
a distribution $\nu$ for $g$ (where $\nu$ is presumed
to be an even mixture of $\nu_0$ and $\nu_1$).

One may wish for a distributional composition theorem
of the form
\[Q^{\mu\circ(\nu_0,\nu_1)}(f\circ g)=\Omega(Q^{\mu}(f)Q^{\nu}(g)).\]
However, there are a couple of issues with this statement.
First, in the distributional setting, we do not have
error reduction (repeating the algorithm to amplify
the success would require a new sample from the distribution
in question, and we only have one sampled input string).
Instead of writing $Q^{\mu}(f)$, we should really be
writing $Q^\mu_\epsilon(f)$ where we specify the error
level $\epsilon$.

A more serious problem comes from Shaltiel \cite{Sha03}.
Shaltiel noticed that in the distributional setting,
the distribution on $g$ may rarely give a very easy
input: for example, suppose with probability $0.1$
the input is trivial to solve and can immediately be seen
to be trivial to solve, while with probability $0.9$
the input is hard. Solving $g$ against this distribution
to error $1/3$ is difficult, and requires many queries.
However, if we happened to get $100$ independent samples
with correlated $g$-values (either all samples are $0$-inputs or all are $1$-inputs), then we can solve the
copies of $g$ easily: we just need to find a single
easy input. In other words, if we take $f$ to be the
``trivial function'' that maps $0^n$ to $0$
and $1^n$ to $1$ (with all other strings being outside
the domain of $f$), then the distributional
complexity of $f\circ g$ on a composed distribution
(composing the obvious uniform distribution for $f$
with a Shaltiel-type distribution for $g$)
may be lower than the distributional complexity of $g$.

Instead, our distributional composition theorem has
the following form.

\begin{theorem}[Informal, see \Cref{thm:composition-distributional-query}]
Let $f$ and $g$ be (possibly partial) Boolean functions,
let $\mu$ be a distribution over the domain of $f$,
and let $\nu_0$ and $\nu_1$ be distributions over
$g^{-1}(0)$ and $g^{-1}(1)$ respectively. Then
for any $\epsilon>0$, we have
\[Q_\epsilon^{\mu\circ(\nu_0,\nu_1)}(f\circ g)
=\Omega(\epsilon Q_{2\epsilon}^\mu(f)\cdot \mathrm{sfAdv}^{(\nu_0,\nu_1)}(g)).\]
Moreover, the new measure $\mathrm{sfAdv}^{(\nu_0,\nu_1)}(g)$
is the best possible for this theorem statement.
\end{theorem}

In this theorem, we lose a factor of $\epsilon$, and the error parameter changes from $\epsilon$ to $2\epsilon$.
The lower bound dependence on $g$ is also
in terms of a new measure, the ``Shaltiel-free
adversary measure'' $\mathrm{sfAdv}$. We show this theorem
is tight by showing that
$\mathrm{sfAdv}^{(\nu_0,\nu_1)}(g)$ is characterized
as the distributional complexity of $\mathtt{triv}_n$
composed with $g$ (with the obvious choice of distributions),
in the limit as $n\to\infty$,
where $\mathtt{triv}_n$
is the trivial function described above which maps $0^n\to 0$
and $1^n\to 1$. Since $Q_\epsilon(\mathtt{triv}_n)=1$,
it follows
that the theorem is tight when $f=\mathtt{triv}_n$ for large
$n$. In other words, the measure $\mathrm{sfAdv}(g)$ cannot
be replaced by any larger measure in the above theorem
when $f$ is arbitrary.

\subsubsection{Direct product theorem}

Our distributional direct product theorem is as follows.

\begin{theorem}[Informal, see \Cref{thm:direct-product-distributional-query}]\label{infthm:query-direct-product}
For a decision or search problem $f$, a distribution $\mu$, $0<\alpha<1$,
    and an integer $k\geq1$, we have
    \[
        Q^{\mu^{\otimes k}}_{1-\left(\frac{1+\alpha}{2}\right)^k}
        (f^{\otimes k})
        =
        \Omega\left(k(1-\alpha)\,Q^\mu_{1-\alpha}(f)\right).
    \]
\end{theorem}

To prove this theorem, we develop a new \emph{Alberti adversary method}, denoted by $\mathrm{AAdv}^{\mu}_{\alpha}(f)$.
In fact, we show that measure charecterises the direct-product behaviour in the limit.
\begin{corollary}[Informal, see \Cref{cor:aadv-asymptotic-characterization}]
    For any decision or search problem $f$, distribution $\mu$ over
    its inputs, and $\alpha\in(0,1)$, we have
    $\mathrm{AAdv}_{\alpha}^{\mu}(f)
    =\Theta(\lim_{k\to\infty}
    \mathsf{Q}_{1-\alpha^k}^{\mu^{\otimes k}}(f^{\otimes k})/k)$.
\end{corollary}

This theorem strengthens
the direct product theorem of \cite{LR12} by making it
distributional instead of worst-case. We also have a better
dependency on the error parameter, and our result works
for search problems rather than only functional problems.
In particular, we obtain the following worst-case corollary.

\begin{corollary}[Informal, see \Cref{cor:direct-product-worst-case-query}]\label{infcor:query-direct-product-worst-case}
For a possibly partial Boolean function $f$, any $0<\eps<1/2$,
    and an integer $k\geq1$, we have
    \[
        Q_{1-(1/2+\eps)^k}(f^{\otimes k})
        =
        \Omega_{\eps}\left(k Q_{1/2-\eps}(f)\right).
    \]
\end{corollary}

We are also able to use our techniques to get related results,
such as distributional XOR lemma (see \Cref{thm:distributional-xor}) and a distributional direct-sum theorem (see \Cref{cor:direct-sum-distributional-query}).

\subsection{Our techniques}

\subsubsection{Foundations of the distributional framework}

\paragraph{A Gram-matrix characterization of distributional query complexity.}

Our starting point is to characterize distributional quantum query complexity as a state-conversion problem.

Let $f$ be a function and $\mu$ a distribution. Then, a quantum algorithm starts in the coherent state $\sum_x\sqrt{\mu(x)}\ket{x}\ket{0}$ and produces the final state $\ket{\psi}=\sum_x\sqrt{\mu(x)}\ket{x}\ket{\psi_x}$. Roughly speaking, to measure the success of the algorithm, we want to consider its overlap $|\braket{w}{\psi}|$ with a witness state $\ket{w}=\sum_x\sqrt{\mu(x)}\ket{x}\ket{w_x}$ that perfectly distinguishes inputs with different outputs, meaning $\braket{w_x}{w_y}=0$ whenever $f(x)\neq f(y)$.

In particular, unlike the worst-case setting, the vectors $\ket{w_x}$ might not be normalized individually. Intuitively, this allows us to put more weight on inputs where the algorithm succeeds more often and less weight on inputs where it fails. However, we still require the joint state $\ket{w}$ to have norm one.

Using Uhlmann's characterization of fidelity, we can then write the success condition
\begin{align*}
    \cF(\Psi\circ\kbr{\mu},W\circ\kbr{\mu})\geq\sqrt{1-\eps},
\end{align*}
where $\ket{\mu} = \sum_{x} \sqrt{\mu(x)}\ket{x}$ is the unit vecotr of entrywise squareroots of the probabilites in $\mu$,
$\Psi(x,y) = \braket{\psi_x}{\psi_y}$ is the Gram matrix corresponding to the algorithm, and $W(x,y)=\braket{w_x}{w_y}$ is a valid witness matrix for the problem $f$. Building on this observation, we obtain the following characterization of distributional quantum query complexity.
\begin{theorem}[Informal, see \Cref{thm:fidelity-charecterization}]\label{infthm:gram-matrix-characterization}
    Let $\cW_f^\mu$ denote the set of valid witness matrices for $f$
    under $\mu$, and let $\cG^\mu_\eps(f)$ be the set of all Gram
    matrices $\Psi$ for which there exists
    a witness $W\in\cW_f^\mu$ satisfying
    $\cF(\Psi\circ\kbr{\mu},W\circ\kbr{\mu})\geq\sqrt{1-\eps}$.
    Then,
    \[
        Q_\eps^\mu(f)
        = \min_{\Psi\in\cG^\mu_\eps(f)} Q(\Psi).
    \]
\end{theorem}
Having this characterization, we can focus on state conversion and
work with Gram matrices directly.
In \Cref{sec:gram-matrix-formulation}, we formally define the witness
and target sets and prove this equivalence.

\paragraph{Distributional state conversion.}
To upper bound quantum query complexity, we also need to be able
to approximately generate Gram matrices with quantum algorithms, where closeness
is measured by fidelity. More generally,
we consider converting an initial Gram matrix $\Phi$ into a target
Gram matrix $\Psi$ as follows.

\begin{theorem}[Informal, see \Cref{thm:catalyst-state-conversion}]\label{infthm:state-conversion}
    Suppose we can write the difference between two Gram matrices as
    \begin{align*}
        & \Phi-\Psi=\sum_i\Delta_i\circ X_i, && \Psi\succeq\frac{r-1}{2}\sum_iX_i,
    \end{align*}
     for PSD matrices $X_i$ and $r\geq 1$.
    Then, for any integer $T\geq1$, there exists a Gram matrix
    $\widetilde{\Psi}$ with query complexity $Q(\Phi\to\widetilde{\Psi})\leq T$ and fidelity
    \[
        \cF(\widetilde{\Psi}\circ\kbr{\mu},\Psi\circ\kbr{\mu})
        \geq
        1-\frac{\sum_i\Tr[X_i\circ\kbr{\mu}]}
                {\sum_{j=0}^{T-1}r^j}.
    \]
\end{theorem}
Here $\Delta_i (x,y) = \mathds{1}[x_i \neq y_i]$ are the filter matrices which represent the pairs of
inputs that an algorithm can distinguish by querying the $i$-th bit and
$Q(\Phi \to \widetilde{\Psi})$ denotes the query complexity of state conversion from $\Phi$ to $\widetilde{\Psi}$. We note that
The case $r=1$ follows by adapting the catalyst-reuse construction of Belovs
and Yolcu \cite{BY23} to the distributional setting. For the case $r>1$, however,
we extend their paradigm by allowing the
catalyst to be rescaled at each step.
We give the formal statement and proof in
\Cref{sec:distributional-state-conversion}.

\subsubsection{The distributional adversary method}
\paragraph{distributional filtered $\gamma_2$ norm.}
We introduce a number of variant of the quantum adversary method.
To develop our first adversary measure, we start with a
distributional version of the filtered gamma 2 norm of Lee et al.~\cite{LMR+11}; i.e.,
\[
    \bar{\gamma}_2^\mu(B\mid\{\Delta_i\})
    = \max_{\Gamma=\Gamma^\dagger}
      \frac{\langle\Gamma,B\circ\kbr{\mu}\rangle}
           {\max_i\|\Gamma\circ\Delta_i\|}.
\]
Here the matrix $B$ should be viewed as governing which pairs of inputs the quantum algorithm
must distinguish from each other and $\Gamma$ is a Hermitian \emph{adversary matrix} that intuitively
encodes the importance of distinguishing each pair. We note that, maximizing over the choice of $\mu$,
we would get the usual filtered gamma 2 norm.
\par
This measure has nice properties similar to the filtered $\gamma_2$
norm, and it can be used to lower and upper bound query complexity
for state conversion. The lower
bound uses a simple hybrid argument, while the upper bound uses its
minimization form to give certificates for the state-conversion
\Cref{infthm:state-conversion}. We cover the primal-dual forms and the query complexity bounds in
\Cref{sec:dist-gamma2}.
\paragraph{distributional adversary method.}
Next, we define our (additive) distributional adversary measure as the minimum
$\bar{\gamma}_2^\mu$ cost of going from the all-ones matrix $J$ (where no information is known) to a Gram matrix $\Psi$
that computes $f$ with error $\eps$ under $\mu$; i.e.,
\[
    \mathrm{Adv}_\eps^\mu(f)
    := \min_{\Psi\in\cG_\eps^\mu(f)}\bar{\gamma}_2^\mu(J-\Psi\mid\{\Delta_i\})
    = \max_{\Gamma=\Gamma^\dagger}\min_{\Psi\in\cG_\eps^\mu(f)}
      \frac{\langle\Gamma,(J-\Psi)\circ\kbr{\mu}\rangle}
           {\max_i\|\Gamma\circ\Delta_i\|},
\]
where, as previously mentioned, $\cG^{\mu}_{\eps}(f)$ is the set of Gram matrices that compute $f$ against $\mu$ with error at most $\eps$. We emphasize that the second exact maximization equality, which follows from strong duality, allows a single adversary matrix $\Gamma$ to certify a lower bound on the query complexity of $f$.
\par
We then show that this measure characterizes distributional quantum query complexity in the bounded-error setting up to a constant change in the error.

\begin{proposition}[Informal, see \Cref{cor:adv-query-comparison}]\label{infthm:distributional-adversary}
For any search or decision problem $f$, any distribution $\mu$, and any error parameter $0<\eps<1$, we have
    \[
        \Omega\bigl(\mathrm{Adv}_{\eps}^\mu(f)\bigr)
        \leq Q_{\eps}^\mu(f)
        \leq O\left(\frac{\mathrm{Adv}_{\eps/2}^\mu(f)}{\eps}\right).
    \]
\end{proposition}
The proof of this result is a direct consequence of our definition and the state-conversion bounds
for $\bar{\gamma}_2^{\mu}$. Also, to draw an analogy,we mention that maximizing this measure over $\mu$ bounds the general adversary method, $\mathrm{Adv}^{\pm}$. We give the formal definition, primal-dual formulations, and the above bound in
\Cref{sec:distributional-adversary}.

\subsubsection{Composition via the Shaltiel-free adversary}

\paragraph{Shaltiel-free measure.}
The main object of interest for the composition theorem is our
Shaltiel-free adversary measure for the inner function, which can be defined in its dual (maximization) form as follows.
    \begin{align*}
        \mathrm{sfAdv}^{(\nu_0,\nu_1)}(g)
        &= \progline{\max_{\Gamma=\Gamma^\dagger}}
           {\frac{1}{\max_i\|\Gamma\circ\Delta_i\|}}
           {\bra{\nu_0}\Gamma\ket{\nu_1}=\|\Gamma\|=1;
            \quad\Gamma\circ\Delta_g=\Gamma}.
    \end{align*}
    Here $\ket{\nu_0},\ket{\nu_1}$ are the corresponding unit vectors to $0$ and $1$ probability
    distributions $\nu_0$ and $\nu_1$ prescribed for the $0$ and $1$ inputs, respectively.
This measure ends up tracking a certain
cost/bias tradeoff for the inner function
in a composition, which is necessary to overcome
Shaltiel's counterexample mentioned above.
In \Cref{sec:relaxed-adversary},
we introduce an intermediate measure $\mathrm{rAdv}$, which we use in \Cref{sec:shaltiel-free-adversary} to define $\mathrm{sfAdv}$ in its primal formulation and derive
its dual formulation.
\paragraph{Composition theorem.}
Our main composition result is a consequence
of the following theorem, which shows that $\mathrm{sfAdv}$ characterizes the composition behavior of the inner function.
\begin{theorem}[Informal, see \Cref{thm:composition-adversary}]\label{infthm:composition-adversary}
    Let $f$ be a function on $n$ bits, let $g$ be a Boolean function,
    and let $\mu$ be a distribution on the inputs of $f$. For
    distributions $\nu_b$ supported on $g^{-1}(b)$ and $0<\eps<1/2$,
    we have
    \[
        \mathrm{Adv}^{\mu\circ(\nu_0,\nu_1)}_\eps(f\circ g^{\otimes n})
        \geq
        \mathrm{Adv}^\mu_\eps(f)\,
        \mathrm{sfAdv}^{(\nu_0,\nu_1)}(g).
    \]
\end{theorem}

We show this theorem using a ``dual-dual'' proof; i.e.,
taking dual (lower bound type)
objects for $f$ and $g$ and combining them
into a dual object for the composition of $f$
and $g$. More specifically, we
start with an adversary matrix $\Gamma_g$ for  $\mathrm{sfAdv}^{(\nu_0,\nu_1)}(g)$, and use it to lift an adversary matrix $\Gamma_f$ for $\mathrm{Adv}^{\mu}_{\eps}(f)$ to one for the composed function. The use of the right definition of
$\mathrm{sfAdv}$ is key
to get this construction to work.
\par
Combining \Cref{infthm:composition-adversary} with \Cref{infthm:distributional-adversary} immediately
gives our desired composition theorem for distributional quantum query complexity.
The formal statements and their proofs can be found in
\Cref{sec:composition}.

\subsubsection{Direct product via the Alberti adversary method}


\paragraph{Multiplicative analogue of the $\gamma_2$ norm.}
To develop our direct product theorem, we begin by defining a multiplicative analogue of the $\bar{\gamma}_2^{\mu}$ measure
for a pair of PSD matrices $A$ and $B$ as follows.
    \begin{align*}
        \bar{\delta}_2^\mu(A,B\mid\{S_i\})
        &:= \sup_{\Gamma\succ0}
        \frac{-2\log\sqrt{
            \langle\Gamma,A\circ\kbr{\mu}\rangle
            \langle\Gamma^{-1},B\circ\kbr{\mu}\rangle}}
             {\max_i D_{\max}(\Gamma\circ S_i\Vert\Gamma)}.
    \end{align*}

Intuitively speaking, this strange-looking object measures
the supremum of a log-over-log ratio
$$  \bar{\delta}_2^\mu(A,B\mid\{S_i\}) = \sup_{c>1} \frac{-2\log \cF_{c}(A\circ \kbr{\mu},B\circ \kbr{\mu})}{\log c}.$$
The numerator captures dissimilarity through our notion  ``filtered fidelity'' , $\cF_{c}$, which measures how similar two matrices are under a per-query budget $c$.
We obtain this notion by restricting Alberti's characterization of fidelity to that budget.
The denominator, on the other hand, normalizes the dissimilarity by the log of the per-query budget $c$.
\par
To put this measure into context, we note that for Gram matrices
$\Phi$ and $\Psi$ it recovers our additive $\bar{\gamma}_2^{\mu}$ measure
$\bar{\delta}_2^\mu(\Phi,\Psi\mid\{S_i\})
= \bar{\gamma}_2^\mu(\Phi-\Psi\mid\{\Delta_i\})/2$. In addition,
similar to $\bar{\gamma}_2^{\mu}$, this measure has nice properties and yields upper and lower bounds for the query complexity of state conversion.
Most strikingly, it has a
magical exact tensorization property across $k$ copies; i.e.,
\begin{align*}
      \bar{\delta}_2^{\mu^{\otimes k}}
    \bigl(A^{\otimes k},B^{\otimes k}\mid\{S_{ij}\}\bigr)
    = k\,\bar{\delta}_2^\mu(A,B\mid\{S_j\}).
\end{align*}
This property is the heart of our direct product result and is also directly used to prove the XOR lemma.
We give the formal treatment of
 filtered fidelity and the $\bar{\delta}_2^\mu$ measure in
\Cref{sec:filtered-fidelity} and \Cref{sec:multiplicative-gamma2},
respectively, and present the distributional XOR lemma in \Cref{sec:xor-lemma}.


\paragraph{Alberti adversary measure.}
To relate the $\bar{\delta}_2^{\mu}$ measure to quantum query complexity, we define a new ``Alberti adversary measure,''
which we essentially use in place of the
multiplicative adversary method. More formally, for a success parameter $\alpha \in (0,1)$,
    \begin{align}\label{eq:alberti-adversary}
        \mathrm{AAdv}_\alpha^\mu(f)
        &:= \min_{W\in\cW_f^\mu}
            \bar{\delta}_2^\mu(J,W/\alpha\mid\{S_i\})\\
        &= \progline{\sup_{\Gamma\succ0}}
           {\frac{\log\alpha-\log\langle\Gamma,\kbr{\mu}\rangle}
                {\max_i D_{\max}(\Gamma\circ S_i\Vert\Gamma)}}
           {\Gamma\succeq I_{f^{-1}(s)}\quad\forall s\in\Omega}.\notag
    \end{align}
    Here, as mentioned above, $\cW^{\mu}_f$ denotes the set of valid witnesses for $f$ and $I_{f^{-1}(s)} = \sum_{x:s\in f(x)} \ketbra{x}$ denotes the projector onto inputs for which $s$ is a valid output. Through state-conversion bounds for $\bar{\delta}_2^{\mu}$, we can show that the Alberti
    adversary measure characterizes distributional quantum query complexity as follows.
    \begin{proposition}[Informal, see \Cref{prop:aadv-lowerbound}]\label{infthm:alberti-adversary}
        For any search or decision problem $f$, any distribution $\mu$, and success parameters $0<\alpha<\beta<1$, we have
        \[
            \mathrm{AAdv}_{\alpha}^\mu(f)
            \leq Q_{1-\alpha}^\mu(f)
            \leq O_{\alpha,\beta}\bigl(\mathrm{AAdv}_{\beta}^\mu(f)\bigr).
        \]
    \end{proposition}

    Unlike the additive $\mathrm{Adv}^{\mu}_{\eps}(f)$, the minimization in \eqref{eq:alberti-adversary} is directly over valid witness matrices, with the success parameter entering through the rescaling $W/\alpha$. This rescaling preserves the tensor-product structure
$(W/\alpha)^{\otimes k}=W^{\otimes k}/\alpha^k$. As a result, using the exact tensorization of $\bar{\delta}_2^\mu$ as well as primal and dual formulations of $\mathrm{AAdv}$ in the above, we prove an exact direct product theorem for $\mathrm{AAdv}$ as follows.
\begin{theorem}[Informal, see \Cref{prop:aadv-direct-product}]\label{infthm:aadv-direct-product}
    For a search problem $f$, a distribution $\mu$, $0<\alpha\leq1$,
    and an integer $k\geq1$, we have
    \[
        \mathrm{AAdv}_{\alpha^k}^{\mu^{\otimes k}}(f^{\otimes k})
        = k\,\mathrm{AAdv}_\alpha^\mu(f).
    \]
\end{theorem}
Finally, combining \Cref{infthm:alberti-adversary} and \Cref{infthm:aadv-direct-product}, we immediately obtain the desired direct product theorem for distributional quantum query complexity. In \Cref{sec:alberti-adversary-method}, we study $\mathrm{AAdv}$
and derive its dual formulation, and in \Cref{sec:direct-product},
we prove the direct product theorem.

\subsection{Related work}

\paragraph{Adversary method for worst-case quantum computing.}

The adversary method was initially introduced by Ambainis \cite{Amb00} in its combinatorial form as a technique for proving lower bounds on quantum query complexity. Following several independent extensions of this technique \cite{Amb06,BSS03,Zha05,LM04}, Špalek and Szegedy \cite{SS06} showed that the resulting formulations were all equivalent, unifying them into what is now known as the positive adversary method, $\mathrm{Adv}^{+}$.
\par
Høyer, Lee, and Špalek \cite{HLS07} strengthened this lower bound method to the generalized adversary method $\mathrm{Adv}^{\pm}$ by allowing negative weights. They also proved a composition theorem for $\mathrm{Adv}^{\pm}$, which was used to prove strict separation from $\mathrm{Adv}^{+}$. Subsequently, Reichardt \cite{Rei11}
transformed the minimization formulation of $\mathrm{Adv}^{\pm}$ into a quantum algorithm, proving that $\mathrm{Adv}^{\pm}$ characterizes the bounded-error quantum complexity, up to constant factors. This, together with the composition properties of $\mathrm{Adv}^{\pm}$, immediately proved a tight composition theorem for quantum query complexity. We also mention that Belovs and Lee \cite{BL20} extended the adversary composition theorem to outer search problems by adapting its proof to the relational adversary formulation introduced by Belovs \cite{Bel15}.
\par
Building on Reichardt’s work, Lee et al. \cite{LMR+11}
showed that the filtered $\gamma_2$ norm gives a semi-tight characterization of the quantum query complexity of general state-conversion problems. They also exploited this new framework to provide a matching upper bound to the lower bound of Ambainis et al. \cite{ACLT10} for $\mathrm{Adv}^{\pm}(f^{\otimes k})$, yielding a tight direct-sum theorem for quantum query complexity.
In a parallel line of work, Špalek \cite{Spa08} introduced the multiplicative adversary method, $\mathrm{MAdv}$, which lower bounds quantum query complexity and satisfies a strong direct-product theorem. Later, Lee and Roland \cite{LR12} strengthened the comparison between $\mathrm{MAdv}$ and $\mathrm{Adv}^{\pm}$ established by Ambainis et al. \cite{AMRR11} to prove a strong direct-product theorem for quantum query complexity, as well as an XOR lemma for Boolean functions.
\par
Finally, we mention that Belovs and Yolcu \cite{BY23} gave the catalyst-reuse construction, which provides a much simpler algorithm and removes a logarithmic factor in the dependence of the query complexity on the error parameter compared with prior state-conversion algorithms.

\paragraph{Comparison with prior distributional results.}
The adversary method has also been adapted to the distributional setting. Belovs \cite{Bel15}
formulated an adversary bound for search problems in the context of general unitary input oracles,
and later, Belovs et al. \cite{BBH+18} developed a related bound for distinguishing input
distributions and used it to prove a composition theorem for this bound when the inner problem is a
promise version of unstructured search, $\mathrm{pSEARCH}$. We note that our distributional
adversary measure, although related to these bounds, has a different formulation. Moreover, compared
with their work, our composition theorem applies to arbitrary Boolean inner functions and gives the best asymptotic
composition lower bound directly in terms of the outer problem's distributional quantum query
complexity and the Shaltiel-free measure of the inner function.  Our Shaltiel-free measure is inspired by the classical Shaltiel-free measure
introduced by Ben-David and Blais \cite{BB20a}, who formulated their measure to account for the
obstruction illustrated by Shaltiel's example \cite{Sha03}.
(See also \cite{BDG+20,BBGM22} for related classical composition results.)
\par
In the multiplicative setting, Jeffery and Zur \cite{JZ25} introduced the distributional
multiplicative ladder adversary method and showed that it satisfies a strong direct-product theorem.
As a result, they obtain a lower bound on the distributional quantum query complexity of the $k$-copy problem in terms of the single-copy value of their
measure. In recent work parallel to ours, Beame, Kornerup, and Whitmeyer \cite{BKW26} proved strong selective and list-decoding direct-product theorems, which also yield lower bounds on direct-product problems in terms of the single-copy distributional multiplicative adversary measure. In contrast, our method gives a direct-product lower bound directly in terms of the single-copy distributional quantum query complexity.


\section{Preliminaries}
\newcommand{\Dens}{\mathcal{D}} \newcommand{\Pos}{\mathrm{Pos}} \newcommand{\Herm}{\mathrm{Herm}}

\subsection{Notation and Hilbert spaces}\label{sec:prelim-hilbert-spaces}

For integers $0\leq m\leq n$ we write $[n] := \{1,\dots, n\}$ and $[m:n] :=
\{m,m+1,\dots,n\}$.
\paragraph{Operators.} We assume that all Hilbert spaces are finite-dimensional and complex unless otherwise stated. For
Hilbert spaces $\cH$ and $\cK$ we write $\cL(\cH,\cK)$ to denote the space of linear operators  from
$\cH$ to $\cK$ and abbreviate $\cL(\cH):=\cL(\cH,\cH)$. A linear map $V\in\cL(\cH,\cK)$ is called an
isometry if $V^\dagger V=I_{\cH}$, where $I_{\cH}$ is the identity map. If, in addition,
$\dim(\cH)=\dim(\cK)$, then $V$ is a unitary map from $\cH$ to $\cK$. We write $\cU(\cH,\cK)$ for
the set of unitary maps from $\cH$ to $\cK$ and abbreviate $ \cU(\cH):=\cU(\cH,\cH).$ Additionally,
we write $\Herm(\cH):=\{A\in\cL(\cH):A=A^\dagger\}$ for the space of Hermitian operators,
$\Pos(\cH):=\{A\in\Herm(\cH):A\succeq 0\}$ for the cone of positive semidefinite operators, and
$\Dens(\cH):=\{\rho\in\Pos(\cH):\Tr(\rho)=1\}$ for the set of density operators. For operators $A,B\in\Herm(\cH)$, we write $\supp(A)$ to denote the span of its eigenvectors with
nonzero eigenvalues, and $A\succeq B$ whenever $A-B \succeq 0$. We write $A^\top$,
$\overline{A}$, and $A^{\dagger}$ to denote the transpose, entrywise complex conjugate, and
adjoint of $A$, respectively.
\paragraph{Bra-ket notation.}
Using \emph{bra-ket} notation, we denote vectors in $\cH$ by kets, such as $\ket{u}$ and $\ket{v}$,
and their conjugate transposes by bras, such as $\bra{u}$ and $\bra{v}$. We write $\braket{u}{v}$ to
denote the Hermitian inner product of $\ket{u}$ and $\ket{v}$, $\ketbra{u}{v}$ to denote the
rank-one operator $\ket{u}\bra{v}$, and $\|u\|_{2}:=\sqrt{ \braket{u}{u}}$ to denote the Euclidean
norm of $\ket{u}$.
\paragraph{Hermitian and PSD matrices.}

The singular value decomposition (SVD) states that any operator
$A\in\cL(\cH)$ can be written as $A=UDV^\dagger$, where $U,V$
are unitary and $D$ is diagonal with nonnegative entries.
Equivalently, writing $d=\dim(\cH)$, there exist singular values
$s_i\geq0$ and orthonormal bases $\{\ket{u_i}\}_{i=1}^d$ and
$\{\ket{v_i}\}_{i=1}^d$ such that $A=\sum_{i=1}^d s_i\ketbra{u_i}{v_i}.$
If $A$ is Hermitian, its spectral decomposition is
$A=\sum_{i=1}^d\lambda_i\ketbra{u_i}$, where the eigenvalues
$\lambda_i$ are real. In particular, $A\succeq0$ if and only if all its eigenvalues are nonnegative, or equivalently, $A=UDU^\dagger$ for a unitary $U$ and a diagonal matrix $D$ with nonnegative entries. Similarly, $A\succ0$ if and only if all its eigenvalues are positive.
\par
To prove that block matrices are positive semidefinite, we often
use the Schur complement criterion. For $A\succ0$, $B=B^\dagger$,
and an operator $C$ of compatible dimensions, it states that
\[
    \begin{pmatrix}
        A & C\\
        C^\dagger & B
    \end{pmatrix}\succeq0
    \quad\Longleftrightarrow\quad
    B\succeq C^\dagger A^{-1}C.
\]

\paragraph{Indexed matrices.}
Let $\Sigma$ be a finite set. We identify $\bC^\Sigma \cong \bC^{|\Sigma|}$ with the complex Hilbert
space whose computational basis is $\{\ket{x}:x\in\Sigma\}$ and write $\cL(\Sigma)
:= \cL(\bC^{\Sigma})$ for the space of matrices whose rows and columns are indexed by $\Sigma$. With
slight abuse of notation we write $\Herm(\Sigma):=\Herm(\bC^\Sigma)$,
$\Pos(\Sigma):=\Pos(\bC^\Sigma)$, $\Dens(\Sigma):=\Dens(\bC^\Sigma)$, and $\cU(\Sigma) :=
\cU(\bC^{\Sigma})$. Equivalently, these are the Hermitian matrices, positive semidefinite matrices,
density matrices, and unitary matrices whose rows and columns are indexed by $\Sigma$. We write
$I_{\Sigma}$ and $J_{\Sigma}$ for the identity and all-one matrices in $\bC^{\Sigma\times\Sigma}$,
respectively, and omit the subscript when the index set is clear. For $A\in\cL(\Sigma)$ and $S\subseteq\Sigma$, we denote by
$A_S\in\cL(\Sigma)$ the matrix equal to $A$ on $S\times S$ and zero elsewhere, and by
$A|_S\in\cL(S)$ the principal submatrix indexed by $S$.

\paragraph{Tensor products and partial trace.}
For $A\in\cL(\cH\otimes\cK)$, we write $\Tr_{\cK}[A]\in\cL(\cH)$ for the partial trace of $A$
over $\cK$. For $\cH=\bC^\Sigma$ and an integer $k\geq1$, let $\cH_1,\dots,\cH_k$ be copies of
$\cH$, so that $\cH^{\otimes k}=\cH_1\otimes\dots\otimes\cH_k$. For an index $\ell\in[k]$, let
$\bar\ell:=[k]\setminus\{\ell\}$ and $\cH_{\bar\ell}:=\bigotimes_{r\in\bar\ell}\cH_r$. We write
$\cH^{\otimes k}=\cH_{\bar\ell}\otimes_\ell\cH_\ell$ to denote the decomposition of
$\cH^{\otimes k}$ into its $\ell$-th tensor factor and the remaining factors. For instance, for
$z=(z^1,\dots,z^{k-1})\in\Sigma^{k-1}$, we write $\ket{z}\otimes_\ell\ket{x}$ to denote the vector
$\ket{z^1}\otimes\dots\otimes\ket{z^{\ell-1}}\otimes\ket{x}\otimes
\ket{z^\ell}\otimes\dots\otimes\ket{z^{k-1}}$. We use the same notation for matrices. Given
$A\in\cL(\cH_{\bar\ell})$ and $B\in\cL(\cH_\ell)$, we use $A\otimes_\ell B$ to denote the matrix
with $B$ in its $\ell$-th tensor position and $A$ on the remaining positions. We abbreviate $\Tr_{\cH_{\bar\ell}}$ by
$\Tr_{\bar\ell}$ for the partial trace over the remaining factors.

\paragraph{Hadamard product.}
For matrices $A,B\in\cL(\Sigma)$, we write $A\circ B$ for their \emph{Hadamard product},
defined entrywise by $(A\circ B)(x,y):=A(x,y)B(x,y)$. The Schur product theorem states that
$A\circ B\succeq0$ whenever $A,B\succeq0$.

\paragraph{Distributions.}
We often use lowercase Greek letters, such as $\mu$ and $\nu$, to denote probability distributions. For a
probability distribution $\mu:\Sigma\to [0,1]$ with $\sum_{x\in\Sigma}\mu(x)=1$, we define the unit
vector $\ket{\mu} \in \mathbb{C}^{\Sigma}$ as $\ket{\mu}:=\sum_{x\in\Sigma}\sqrt{\mu(x)}\ket{x}$,
and denote its support by $\supp(\mu) = \{x\in\Sigma:\mu(x)>0\}$.

\subsection{Norms and measures of distinguishability}

Let $\cX,\cY$ be complex Hilbert spaces and let $\cL(\cX,\cY)$ be the complex Hilbert space equipped
with the \emph{Hilbert--Schmidt} inner product defined as $\langle A,B\rangle:=\Tr[A^{\dagger}B]$
for $A,B \in \cL(\cX,\cY)$. For a linear map $\cE:\cL(\cX)\to\cL(\cY)$, its adjoint
$\cE^\dagger:\cL(\cY)\to\cL(\cX)$ is uniquely defined with respect to this inner product by $\langle
\cE^\dagger(B),A\rangle=\langle B,\cE(A)\rangle$.

\subsubsection{Schatten norms}
The \emph{singular values} of $A$ are the eigenvalues of the positive semidefinite matrix
$\sqrt{A^{\dagger}A}$, denoted $s_1(A)\ge s_2(A)\ge \cdots \ge 0$. For $p\in[1,\infty)$, the
\emph{Schatten $p$-norm} of $A$ is defined by
\[
\|A\|_p := \Big(\sum_i s_i(A)^p\Big)^{1/p},
\]
\begin{itemize}
    \item The Schatten $2$-norm is also called the \emph{Hilbert--Schmidt} (or \emph{Frobenius}) norm, which is the
    inner-product norm induced by $\langle\cdot,\cdot\rangle$, namely $\|A\|_{2}=\sqrt{\langle
    A,A\rangle}=\sqrt{\Tr[A^{\dagger}A]}$.
    \item
    The Schatten $1$-norm is called the \emph{trace norm}, and satisfies
    $\|A\|_{\tr}:=\|A\|_{1}=\Tr[\sqrt{A^{\dagger}A}] =\sum_i s_i(A)$. So, if $H\in\cL(\cX)$ is
    Hermitian, then $\|H\|_{\tr}=\Tr[|H|]=\sum_i |\lambda_i(H)|$, where $|H|:=\sqrt{H^2}$. We can also characterize the trace-norm as
    \begin{align*}
        \|H\|_{\tr} &= \min_{P,Q\succeq 0} \{\Tr[P] +\Tr[Q] : H = P-Q\} \\
                    &= \progline
                       {\min_{\ket{u_x},\ket{v_x}\in\bC^d}}
                       {\frac12\sum_{x\in\Sigma}
                       \bigl(\|u_x\|^2+\|v_x\|^2\bigr)}
                       {H(x,y)=\braket{u_x}{v_y}
                       \quad\forall x,y\in\Sigma},
    \end{align*}
    where $\Sigma$ indexes an orthonormal basis of $\cX$, and $d$ is
    a sufficiently large finite dimension.
    \item The Schatten $\infty$-norm, also known as the \emph{operator} or \emph{spectral} norm, is
    defined separately as $\|A\|_{\infty}:=\max_i s_i(A)$. Throughout the paper, we use $\|A\|$ to denote the operator norm $\|A\|_{\infty}$.
\end{itemize}

With respect to the Hilbert--Schmidt inner product, the
Schatten $p$- and $q$-norms are dual whenever $1/p+1/q=1$, i.e.,
\[
\|A\|_{p}=\max_{\|M\|_{q}\le 1} |\Tr[M^{\dagger}A]|.
\]
With the convention $1/\infty=0$, $\|\cdot\|_{tr}$ is dual to $\|\cdot\|_{\infty}$. H\"older's inequality states that for $A,B\in\cL(\cX,\cY)$
\[
|\langle A,B\rangle|=|\Tr[A^{\dagger}B]|
\leq\|A\|_p\|B\|_q.
\]
The special case $p=q=2$ is the Cauchy--Schwarz inequality for the
Hilbert--Schmidt inner product.


\subsubsection{\texorpdfstring{$\gamma_2$ norm}{gamma 2 norm}}\label{sec:prelem:gamma2}
Let $\Sigma$ be a finite set. For a matrix $A\in \bC^{\Sigma\times \Sigma}$, the $\gamma_2$ norm of
$A$, denoted by $\gamma_2(A)$, measures the minimum size of vectors needed to write $A$ as an inner
product $\gamma_2(A) := \min_{U,V: A = UV^{\dagger}}\|U\|_{2\to\infty}\,\|V\|_{2\to\infty}$ where
$\|X\|_{2\to\infty} = \max_{x} \|X_{x,\cdot}\|_2$. Equivalently, we can write
\begin{align*}
    \gamma_2(A)
    &= {}
    \progbox
    {
    \min_{\substack{
    d\in\bN\\
    \ket{u_x},\ket{v_y}\in\bC^d
    }}
    }
    {
    \max\left\{
    \max_{x\in\Sigma}\|u_x\|_2^2,\,
    \max_{y\in\Sigma}\|v_y\|_2^2
    \right\}
    }
    {
    A(x,y)=\braket{u_x}{v_y}
    & \forall x,y\in\Sigma
    },
\end{align*}
where the minimum is over all vector families \(\{\ket{u_x}\}_{x\in \Sigma}\) and
\(\{\ket{v_y}\}_{y\in \Sigma}\) in an arbitrarily large common finite-dimensional Hilbert space
$\bC^{d}$. For future use, we record that the $\gamma_2$ norm satisfies the triangle inequality
$\gamma_2(X+Y)\leq\gamma_2(X)+\gamma_2(Y)$. Additionally, for a PSD matrix $P$, it can be computed as
$\gamma_2(P)=\max_x P(x,x)$ and we have bounds
\begin{align*}
    & \| P \circ A\|_{\tr} \leq \gamma_2(P) \cdot \|A\|_{\tr} && \| P \circ A\| \leq \gamma_2(P) \cdot \|A\|.
\end{align*}
More generally, let $\{Z_i\}_{i=1}^n\subseteq\bC^{\Sigma\times\Sigma}$ be a family of matrices,
called \emph{filters}. We define the \emph{filtered $\gamma_2$ norm} by
\begin{align*}
    \gamma_2\!\left(A\mid\{Z_i\}_{i=1}^n\right)
    &= {}
    \progbox
    {
    \min_{\substack{
    d\in\bN\\
    \ket{u_x^i},\ket{v_y^i}\in\bC^d
    }}
    }
    {
    \max\left\{
    \max_{x\in\Sigma}
      \sum_{i=1}^n\|u_x^i\|_2^2,
    \max_{y\in\Sigma}
      \sum_{i=1}^n\|v_y^i\|_2^2
    \right\}
    }
    {
    A(x,y)
    =
    \sum_{i=1}^n
    Z_i(x,y)\braket{u_x^i}{v_y^i}
    & \forall x,y\in\Sigma
    }
    \end{align*}
where again $\bC^d$ is an arbitrarily large finite-dimensional Hilbert space.\footnote{If no such
factorization exists, the minimum is understood to be $+\infty$. For $Z_i=\Delta_i$, every
zero-diagonal matrix admits such a factorization, as is the case in all our applications.} Note that using a
single all-one filter $J$, we recover the ordinary unfiltered norm $\gamma_2(A) =
\gamma_2(A|\{J\})$.
\par
We also note that the filtered $\gamma_2$ norm has an SDP formulation, and therefore, we can obtain a
maximization formulation for it as
\begin{align}
    \gamma_2(A|\{Z_i\}) = \max_{\Gamma\neq 0} \frac{\|\Gamma \circ A\|}{\max_{i}\|\Gamma \circ Z_i\|},
    \label{eq:filtered-gamma2-dual}
\end{align}
where $\Gamma \in \bC^{\Sigma\times \Sigma}$. If $A$ and $\{Z_i\}$ are Hermitian, we can further restrict the feasible set to Hermitian matrices. If further $Z_i$ are real, we have $\gamma_2(\overline{A}\mid\{Z_i\})=\gamma_2(A\mid\{Z_i\})$.
For a more complete list of properties of the filtered $\gamma_2$ norm, we refer the reader to
\cite{LMR+11}.

\subsubsection{Fidelity}
For PSD operators $A,B\in \Pos(\cX)$, we define the (root) fidelity by $ \mathcal{F}(A,B) :=
\bigl\|\sqrt{A}\sqrt{B}\bigr\|_{1} = \Tr\!\left[\sqrt{\sqrt{A}\,B\,\sqrt{A}}\right].$ For unit
vectors $\ket{u}$ and $\ket{v}$, fidelity reduces to $\cF(\kbr{u},\kbr{v})=|\bk{u}{v}|$. Throughout
this work, we use the non-squared fidelity $\mathcal{F}$, and we write $\mathcal{F}^2(A,B)$ for its
square. We will also use the following optimization characterization.
\begin{theorem}[Alberti's characterization of fidelity {See, e.g., \cite[Theorem~3.19 and Corollary~3.20]{Wat18}}]\label{thm:alberti}
    Let $A,B\in\Pos(\cX)$. Then,
    \begin{align*}
        \cF(A,B)
        &= \inf_{\Gamma\succ0}
        \sqrt{\langle\Gamma,A\rangle\langle\Gamma^{-1},B\rangle} = \frac{1}{2}\inf_{\Gamma\succ0}
        \left(\langle\Gamma,A\rangle+\langle\Gamma^{-1},B\rangle\right).
    \end{align*}
\end{theorem}
\par
We now mention a few standard properties of fidelity: For any PSD matrices $A_1,B_1\in
\Pos(\cX)$ and $A_2,B_2\in \Pos(\cY)$, fidelity is additive over direct sum $\cF(A_1\oplus A_2,
B_1\oplus B_2) = \cF(A_1,B_1) + \cF(A_2,B_2)$, tensorises $\mathcal{F}(A_1\otimes A_2,B_1\otimes
B_2) = \mathcal{F}(A_1,B_1)\mathcal{F}(A_2,B_2)$, and for any constants $a,c\geq 0$,
$\mathcal{F}(aA_1,cB_1) = \sqrt{ac}\mathcal{F}(A_1,B_1)$.
\par
For density operators $\rho,\sigma \in \mathcal{D}(\cX)$, the fidelity is bounded by
$\mathcal{F}(\rho,\sigma) \in [0,1]$ and can be related to trace distance by the \emph{Fuchs--van de
Graaf inequalities}
\begin{align*}
1 - \mathcal{F}(\rho,\sigma)\le
\tfrac12 \|\rho-\sigma\|_{\tr}
\le
\sqrt{1 - \mathcal{F}^2(\rho,\sigma)},
\end{align*}
with equality $\tfrac12 \|\rho-\sigma\|_{\tr} = \sqrt{1 - \mathcal{F}^2(\rho,\sigma)}$ holding
whenever $\rho$ and $\sigma$ have rank $1$. Moreover, over the space of density matrices
$\mathcal{D}(\cX)$, fidelity is jointly concave; i.e.,
$\mathcal{F}(\sum_{i=0}^{n}p_i\sigma_{i},\sum_{i=0}^{n}p_i\rho_{i}) \geq
\sum_{i}p_i\mathcal{F}(\sigma_{i},\rho_i)$ where $p:[0:n]\to [0,1]$ is a probability distribution.
\par
We will also use a generalization of the {Fuchs--van de Graaf inequalities} to a sub-normalized
setting which we prove for completeness as follows.
\begin{lemma}[Fuchs--van de Graaf with priors]\label{lem:Fuchs-van-de-Graaf-priors} Let
$\sigma_0,\sigma_1\in\Dens(\cX)$ and let $p_0,p_1\geq 0$ satisfy $p_0+p_1=1$. Then
\begin{align*}
1-2\sqrt{p_0p_1}\cF(\sigma_0,\sigma_1)
\leq
\|p_0\sigma_0-p_1\sigma_1\|_{\tr}
\leq
\sqrt{1-4p_0p_1\cF(\sigma_0,\sigma_1)^2}.
\end{align*}
\end{lemma}
\begin{proof}
    Let $\rho_0 = p_0\sigma_0 \oplus p_1\sigma_1$ and $\rho_1 = p_1\sigma_1 \oplus p_0\sigma_0$ so
    that $\rho_0,\rho_1 \in \cD(\cX\oplus \cX)$. Then,
    \begin{align*}
        \cF(\rho_0,\rho_1) &= \cF( p_0\sigma_0 \oplus p_1\sigma_1, p_1\sigma_1\oplus p_0\sigma_0) = \cF(p_0\sigma_0,p_1\sigma_1) +
        \cF(p_1\sigma_1,p_0\sigma_0) = 2\sqrt{p_0p_1}\cF(\sigma_0,\sigma_1). \\
        \frac{1}{2}\|\rho_0 - \rho_1\|_{\tr}&= \frac{1}{2} \| p_0\sigma_0 \oplus p_1\sigma_1 -  p_1\sigma_1 \oplus p_0\sigma_0\|_{\tr} = \|p_0\sigma_0-p_1\sigma_1\|_{\tr}.
    \end{align*}
    The lemma follows by applying the ordinary Fuchs-van de Graaf inequalities to the density
    operators $\rho_0$ and $\rho_1$.
\end{proof}
\subsubsection{Distances between density operators.}
For a pair $\sigma,\rho\in \cD(\cX)$ of density matrices, the trace norm induces a {trace distance}
    $\cD(\rho,\sigma) = \frac12 \|\rho-\sigma\|_{1}$ and corresponds to the classical total
    variation distance. Fidelity itself is not a metric, but it induces several useful distances
    between quantum states: the \emph{Bures angle} is  defined as $d_{\angle}(\rho,\sigma) :=
    \arccos(\mathcal{F}(\rho,\sigma))$ and the \emph{Bures distance} is defined as
    $d_{B}(\rho,\sigma):=\sqrt{2-2\cF(\rho,\sigma)}$. We also introduce the purified distance
    $\cP(\rho,\sigma):= \sqrt{1-\mathcal{F}^{2}(\rho,\sigma)} =\sin(d_{\angle}(\rho,\sigma))$.
    Specifically, the Bures distance is closely related to the classical Hellinger distance, and
    we can reinterpret the Fuchs--van de Graaf inequalities as
\begin{align*}
    \frac12d^2_{B}(\rho,\sigma) \leq \cD(\rho,\sigma) \leq \cP(\rho,\sigma) \leq d_{B}(\rho,\sigma) \leq d_{\angle}(\rho,\sigma) \leq \frac\pi2  \cP(\rho,\sigma),
\end{align*}
with the equality  $\cD(\rho,\sigma) = \cP(\rho,\sigma)$ holding whenever $\rho$ and $\sigma$ are
rank-one operators. All three fidelity-based distances defined above are metrics and therefore
satisfy the triangle inequality.

\subsubsection{Divergences and max-relative entropy}
For density operators $\rho,\sigma\in\Dens(\cX)$ satisfying $\supp(\rho)\subseteq\supp(\sigma)$, the
\emph{quantum relative entropy} is
\begin{align*}
    D(\rho\Vert\sigma):=\Tr[\rho(\log_2\rho-\log_2\sigma)].
\end{align*}
If $\supp(\rho)\nsubseteq\supp(\sigma)$, we set $D(\rho\Vert\sigma):=+\infty$. Operationally, this
quantity measures how rapidly $\rho$ can be distinguished from $\sigma$ when one is given many
independent copies of the unknown state, represented by $\rho^{\otimes n}$ or $\sigma^{\otimes n}$,
in the limit as $n\to\infty$. When $\rho$ and $\sigma$ are diagonal in the same basis, with diagonal
entries $p$ and $q$, it reduces to the classical Kullback--Leibler divergence
$D_{\mathrm{KL}}(p\Vert q)=\sum_x p(x)\log_2(p(x)/q(x))$. Quantum relative entropy is additive under
tensor products. That is,
\begin{align*}
    D(\rho_1\otimes\rho_2\Vert\sigma_1\otimes\sigma_2)
    =D(\rho_1\Vert\sigma_1)+D(\rho_2\Vert\sigma_2).
\end{align*}
\par
For PSD operators $\rho,\sigma\in\Pos(\cX)$, the \emph{max-relative entropy} measures the smallest
factor by which $\sigma$ must be rescaled in order to dominate $\rho$. It is defined as
\begin{align*}
    D_{\max}(\rho\Vert\sigma):=\log_2\inf\{\lambda\geq 0:\rho\preceq\lambda\sigma\}.
\end{align*}
When $\supp(\rho)\subseteq\supp(\sigma)$, it can equivalently be written as
\begin{align*}
    2^{D_{\max}(\rho\Vert\sigma)}
    =\|\sigma^{-1/2}\rho\sigma^{-1/2}\|_{\infty}
    =\lambda_{\max}(\sigma^{-1/2}\rho\sigma^{-1/2}),
\end{align*}
where the inverse is taken on the support of $\sigma$. Otherwise,
$D_{\max}(\rho\Vert\sigma)=+\infty$.  Max-relative entropy is also additive,
\begin{align*}
    D_{\max}(\rho_1\otimes\rho_2\Vert\sigma_1\otimes\sigma_2)
    =D_{\max}(\rho_1\Vert\sigma_1)+D_{\max}(\rho_2\Vert\sigma_2).
\end{align*}
Consequently, its exponentiated form is multiplicative.

\subsection{Quantum information theory}

\subsubsection{Quantum channels}
A linear map $\cE:\cL(\cH)\to\cL(\cK)$ is a \emph{quantum channel} if it is completely positive and
trace preserving (CPTP). Complete positivity means that $\cE\otimes\operatorname{id}_{\cR}$ maps
positive semidefinite operators to positive semidefinite operators for all finite-dimensional
auxiliary spaces $\cR$, while trace preservation means that $\Tr[\cE(X)]=\Tr[X]$ for all
$X\in\cL(\cH)$. Equivalently, $\cE$ is a quantum channel if it admits a Kraus representation
$
\cE(X)=\sum_k K_kXK_k^\dagger
$
with $K_k\in\cL(\cH,\cK)$ satisfying
$
\sum_k K_k^\dagger K_k=I_{\cH}.
$
\par
Complete positivity ensures that $\cE$ remains positive when it acts on one part of a state that may
be entangled with an auxiliary register. Important examples of quantum channels include isometric
channels $X\mapsto VXV^\dagger$ and the partial trace. Quantum channels are closed under
composition, tensor products, and convex combinations.
\par
Quantum channels cannot increase the distinguishability of quantum states. For every quantum channel
$\cE$ and density operators $\rho,\sigma$, trace distance and all three fidelity-based distances
defined above are contractive; in particular, $\cD(\cE(\rho),\cE(\sigma))\leq\cD(\rho,\sigma)$,
while fidelity is monotone: $\cF(\cE(\rho),\cE(\sigma))\geq\cF(\rho,\sigma)$. Quantum relative
entropy and max-relative entropy also satisfy the \emph{data-processing inequalities}
$D(\cE(\rho)\Vert\cE(\sigma)) \leq D(\rho\Vert\sigma)$ and $D_{\max}(\cE(\rho)\Vert\cE(\sigma)) \leq
D_{\max}(\rho\Vert\sigma)$.
\par
We also use the following norm-contraction property of quantum channels.
\begin{lemma}\label{lem:channel-norm-contraction}
    Let $\cE:\cL(\cH)\to\cL(\cK)$ be a quantum channel. Then, for every
    $X\in\Herm(\cH)$ and $Y\in\Herm(\cK)$,
    \[
        \|\cE(X)\|_{\tr}\leq\|X\|_{\tr}
        \qquad\text{and}\qquad
        \|\cE^\dagger(Y)\|\leq\|Y\|.
    \]
\end{lemma}
\begin{proof}
    Let $X = X_{+} - X_{-}$ be the decomposition of $X$ into its positive and negative parts, where  $X_+,X_-\succeq 0$. Then,
    \begin{align*}
        \|\cE(X)\|_{\tr}
        &\leq \|\cE(X_+)\|_{\tr} + \|\cE(X_-)\|_{\tr}
        = \Tr[\cE(X_+)]+\Tr[\cE(X_-)]
        = \Tr[X_+] + \Tr[X_-] = \|X\|_{\tr},
    \end{align*}
    where we used positivity and trace preservation of $\cE$. By the Kraus representation,
    $\cE^{\dagger}(Y)=\sum_r K_r^{\dagger}YK_r$ is also Hermitian. Thus, for the operator norm on the
    adjoint, we use the dual norm definition as
    \begin{align*}
        \|\cE^{\dagger}(Y)\|
        &= \max_{X\in\Herm(\cH):\|X\|_{\tr}\leq 1}
        |\langle \cE^{\dagger}(Y),X\rangle|
        = \max_{X\in\Herm(\cH):\|X\|_{\tr}\leq 1}
        |\langle Y,\cE(X)\rangle|\\
        &\leq \max_{X\in\Herm(\cH):\|X\|_{\tr}\leq 1}
        \|Y\|\cdot\|\cE(X)\|_{\tr}\leq \|Y\|. \qedhere
    \end{align*}
\end{proof}
\subsubsection{Quantum Measurements}
Let $\cH$ be a Hilbert space and let $\Omega$ be a finite set. A \emph{projective measurement}
on $\cH$ with outcome set $\Omega$ is a collection $\{\Pi_i:i\in \Omega\}\subset \cL(\cH)$ satisfying
$\Pi_s^2 = \Pi_s$, $\Pi_s \Pi_t = 0$, $\sum_{r\in \Omega} \Pi_r = I_{\cH}$ for all $s,t \in \Omega$
with $s\neq t$. When measuring a state $\rho \in \cD(\cH)$, the probability of obtaining outcome $s$
is $\Tr[\Pi_s\rho]$. Consequently, for a pure state $\ket{\psi}$ the probability of outcome $s$ can be
computed as  $\bra{\psi}\Pi_s \ket{\psi}$.
 For a subset $S\subset \Omega$, we abbreviate $\Pi_{S} := \sum_{s\in S} \Pi_s$.
 \par
 In particular, we would use the Helstrom theorem, which characterizes the optimal measurement for
 distinguishing two quantum states.
 \begin{theorem}[Helstrom's theorem \cite{Hel76}]
    \label{thm:helstrom}
    Let $\rho_0,\rho_1\in\cD(\cH)$ be given with prior probabilities
    $p_0,p_1\geq0$ satisfying $p_0+p_1=1$. Then the minimum average error
    probability over all binary projective measurements is
    \begin{align*}
    \min_{\substack{\Pi_0+\Pi_1=I_{\cH}\\
                    \Pi_0\Pi_1=0}}
    \left(
    p_0\Tr[\Pi_1\rho_0]+p_1\Tr[\Pi_0\rho_1]
    \right)
    =
    \frac{1}{2}\left(1-\|p_0\rho_0-p_1\rho_1\|_{\tr}\right).
    \end{align*}
\end{theorem}
\subsubsection{Gram matrices and purifications}
Let $n\in\bN$. By the spectral theorem, every PSD matrix $P\in\Pos(\bC^n)$ admits a decomposition
$P=W^\dagger W$, where $W\in\cL(\bC^n,\bC^r)$ and $r=\rank(P)$. Letting
$\ket{\phi_1},\dots,\ket{\phi_n}$ denote the columns of $W$, we have $P(i,j)=\bk{\phi_i}{\phi_j}$.
Then, the vectors $\ket{\phi_i}$ are all unit vectors if and only if all the diagonal entries of $P$
are one. In that case, we say $P$ is a \emph{Gram matrix} corresponding to unit vectors
$\ket{\phi_1},\dots,\ket{\phi_n}$. We can view a Gram matrix $\Psi$ as an unweighted description of
a state family $\{\ket{\phi_i}:i\in \Sigma\}$ that realizes it. Let $\mu:\Sigma \to [0,1]$ be a
distribution. Then, the weighted matrix $\Psi^{\top}\circ \kbr{\mu}$ is a density operator.
\par
Gram matrices induce a quantum channel through Hadamard product as follows.  To see this, consider matrices
$D_k:=\mathrm{Diag}(\bk{k}{\phi_1},\dots,\bk{k}{\phi_n})$, and  verify that
    $
    X\circ G^\top=\sum_{k=1}^r D_kXD_k^\dagger
    $
    and
    $
    \sum_{k=1}^rD_k^\dagger D_k=I.
    $
\begin{theorem}\label{thm:hadamard-channel} Let $G\in\bC^{n\times n}$ be a Gram matrix. Then the map
    $\cE_G:\cL(\bC^n)\to\cL(\bC^n)$ defined by $\cE_G(X):=X\circ G^\top$ is a quantum channel.
    \end{theorem}

\paragraph{Purifications and Gram-matrix representations.}
Let $\cX$ and $\cY$ be Hilbert spaces. We call a vector $\ket{\psi}\in\cX\otimes\cY$ a
\emph{purification} of $\rho\in\Pos(\cX)$ if $\Tr_{\cY}[\kbr{\psi}]=\rho$. Given a state family
$\{\ket{\psi_x}:x\in\Sigma\}\subseteq\cY$ with Gram matrix $\Psi$ and a distribution $\mu$ on
$\Sigma$, the vector $\ket{\psi^\mu}:=\sum_{x\in\Sigma}\sqrt{\mu(x)}\ket{x}\ket{\psi_x}$ is a
purification of the weighted state $\Psi^\top\circ\kbr{\mu}$, since
$\Tr_{\cY}[\kbr{\psi^\mu}]=\Psi^\top\circ\kbr{\mu}$. We can therefore view state-generation and
state-evolution problems in terms of Gram matrices: two state families with the same Gram matrix are
related by an input-independent isometry. This follows from the uniqueness of purification and its
corollary below.
\begin{theorem}[Uniqueness of Purification]
    \label{thm:uniqueness-purification}
    Let $\cX,\cY_1,\cY_2$ be Hilbert spaces with $\dim(\cY_1)\leq\dim(\cY_2)$, and for
    $\rho\in\Pos(\cX)$, let $\ket{\Psi}\in\cX\otimes\cY_1$ and $\ket{\Phi}\in\cX\otimes\cY_2$ be two
    purifications of $\rho$. Then there exists an isometry $V\in\cL(\cY_1,\cY_2)$ such that
    $
        (I_{\cX}\otimes V)\ket{\Psi}=\ket{\Phi}.
    $
    \end{theorem}
\begin{corollary}[Unitary equivalence between equal-dimensional
spaces]\label{cor:gram-unitary-equivalence} Let $\Sigma$ be a finite index set and let
$\{\ket{\phi_x}:x\in\Sigma\}\subseteq \cH_1$ and $\{\ket{\psi_x}:x\in\Sigma\}\subseteq \cH_2$ be two
families of vectors such that, for all $x,y\in \Sigma$, $\bk{\phi_x}{\phi_y}=\bk{\psi_x}{\psi_y}$.
If $\dim(\cH_1)=\dim(\cH_2)$, then there exists a unitary $U\in \cU(\cH_1,\cH_2)$ such that for all
$x\in \Sigma$,
$
    U\ket{\phi_x}=\ket{\psi_x}.
$
\end{corollary}
Finally, we use Uhlmann's theorem to relate the fidelity of weighted Gram matrices to the overlap of
their purifications.
   \begin{theorem}[Uhlmann's theorem]\label{thm:uhlmann} Let $\cX$ and $\cY$ be Hilbert spaces, and
    let $P,Q \in \Pos(\cX)$ be positive semidefinite operators both having rank at most $\dim(\cY)$.
    For any purification $\ket{u} \in \cX \otimes \cY$ of $P$, we have
    \begin{align*}
        \cF(P,Q)
        =
        \max \left\{
            |\bk{u}{v}| :
            \ket{v}\in \cX\otimes \cY,\
            \Tr_{\cY}[\kbr{v}] = Q
        \right\}.
    \end{align*}
    Equivalently, if $\ket{w}\in\cX\otimes\cY$ is any purification of $Q$, then
\begin{align*}
\cF(P,Q)
=
\max_{U\in\cU(\cY)}
\left|
\bra{u}(I_{\cX}\otimes U)\ket{w}
\right|.
\end{align*}
In particular, the value is finite and both maxima are attained.
\end{theorem}

\subsection{Search and decision problems}\label{sec:prelim-search-problems}
We define a search problem as a family of relations $f=(f_n)$, where
$f_n\subseteq\{0,1\}^n\times\Omega$ for a set $\Omega$ of possible output values. For each
$x\in\Sigma_{f_n}:=\dom(f_n)$, the nonempty set $f_n(x):=\{s\in\Omega:(x,s)\in f_n\}$ consists of
the valid outputs on input $x$.
\par
When the input length is clear, we omit $n$ and use $f$ and $f_n$ interchangeably. Throughout,
for $\Sigma \subset \{0,1\}^{n}$, $f\subseteq\Sigma\times\Omega$ denotes a total relation on $\Sigma$;
that is, $\dom(f)=\Sigma$, or
equivalently, $f(x)\neq\varnothing$ for all $x\in\Sigma$. If every $f(x)$ is a singleton, we say
$f$ is a \emph{functional problem} and denote $f:\Sigma\to\Omega$. If additionally $\Omega =
\{0,1\}$, $f$ is called a \emph{decision problem}.

\subsubsection{Agreement and disagreement matrices}\label{sec:prelim-agreement}
For $i\in[n]$, we define the input-disagreement matrix $\Delta_i\in\Herm(\Sigma)$ by
$\Delta_i(x,y):=\ind[x_i\neq y_i]$ for $x,y\in\Sigma$. We also define the sign matrix
$S_i(x,y):=(-1)^{x_i+y_i}$, so that $S_i=J-2\Delta_i$.
For a functional problem $f:\Sigma\to\Omega$, we define its disagreement matrix $\Delta_f$ and
agreement matrix $A_f$ by $\Delta_f(x,y):=\ind[f(x)\neq f(y)]$ and
$A_f(x,y):=\ind[f(x)=f(y)]$. In particular, $A_f=J-\Delta_f$.
\par
Both $A_f$ and $S_i$ are Gram matrices, so $\gamma_2(A_f)=\gamma_2(S_i)=1$.
Since $\Delta_i=(J-S_i)/2$, the triangle inequality gives
$\gamma_2(\Delta_i)\leq \frac12+\frac12=1$. Similarly, $\Delta_f=J-A_f$ gives
$\gamma_2(\Delta_f)\leq2$.
For a decision problem, the matrix $S_f$ defined by $S_f(x,y):=(-1)^{f(x)+f(y)}$
is a Gram matrix, so $\gamma_2(S_f)=1$. Since $\Delta_f=(J-S_f)/2$,
$\gamma_2(\Delta_f)\leq1$.

\subsubsection{Tensorization}\label{sec:prelim-tensorization}
Let $\Sigma\subset\{0,1\}^n$, $f\subseteq\Sigma\times\Omega$ be a search problem, and
$k\in\bN$ with $k\geq1$. We use bold letters, such as $\bx,\by,\bz$, for tuples of inputs,
writing $\bx=(x^1,\dots,x^k)\in\Sigma^k$. We define
$f^{\otimes k}\subseteq\Sigma^k\times\Omega^k$ as $f^{\otimes k}(x^{1},\dots,x^{k}) := f(x^{1})\times \dots
\times f(x^{k})$ for all $x^{1},\dots,x^{k}\in \Sigma$.  Moreover, for a distribution $\mu:\Sigma
\to [0,1]$, we define $\mu^{\otimes k} :\Sigma^{k} \to [0,1]$ as $\mu^{\otimes k}(x^1,\dots,x^k) :=
\prod_{j=1}^{k} \mu(x^{j})$.

\subsubsection{Composition}\label{sec:prelim-composition}
Suppose $f$ is functional, and let $g:\Sigma_g\to\{0,1\}$ be a decision problem with
$\Sigma_g\subseteq\{0,1\}^m$.
We use $x,y,z\in\Sigma_g$ to denote the inner function inputs, while reserving
$\tx,\ty,\tz\in\Sigma$ for the outer function inputs.
For $\bx=(x^1,\dots,x^n)\in\Sigma_g^n$, we write $\bg:=g^{\otimes n}$, so that
$\bg(\bx)=(g(x^1),\dots,g(x^n))$. We define the composed problem
$f\circ\bg:\bg^{-1}(\Sigma)\to\Omega$ by $(f\circ\bg)(\bx):=f(\bg(\bx))$.
\par
Let $\nu_b:\Sigma_g\to[0,1]$ be distributions with
$\supp(\nu_b)\subseteq g^{-1}(b)$ for $b\in\{0,1\}$.
We write $\eta:=\mu\circ(\nu_0,\nu_1)$ for the composed distribution on
$\bg^{-1}(\Sigma)$, defined by
\begin{align*}
    \eta(x^1,\dots,x^n)
    :=\mu(g(x^1),\dots,g(x^n))\prod_{i=1}^n\nu_{g(x^i)}(x^i).
\end{align*}
Intuitively, this means sampling an input $\tilde{x}$ for $f$ according to $\mu$ and then sampling
the corresponding input $x^i$ for the $i$-th $g$ according to the zero or one distribution
$\nu_{\tilde{x}_i}$.

\subsection{Quantum query complexity}
In this part, we review the basic definitions and important results of quantum query complexity
using a coherent oracle model.

\subsubsection{The coherent oracle model.}
To discuss quantum query complexity for inputs in
$\Sigma\subseteq\{0,1\}^n$, we consider the following three registers.
\begin{itemize}
  \item an input register $\cX \cong \mathbb{C}^{\Sigma}$ with computational basis $\{\ket{x} : x
  \in \Sigma\}$,
  \item an index register $\cI \cong \mathbb{C}^{n+1}$ with basis $\{\ket{i} : i \in
  \{0,1,\dots,n\}\}$,
  \item an arbitrarily large ancilla  register $\cA$.
\end{itemize}
For simplicity, we often group the index and ancilla registers and consider a  \emph{work register}
defined by $\cW = \cI \otimes \cA$.
\begin{definition}
The \emph{coherent oracle} is the unitary
$\cO:=\sum_{x\in\Sigma}\ketbra{x}\otimes\cO^x
\in\cU(\cX\otimes\cI)$, where
\begin{align*}
    \cO^{x} \ket{i} =  \begin{cases}
         \ket{0}            & i = 0        \\
         (-1)^{x_i}\ket{i} & i\in [n].
     \end{cases}
\end{align*}
Alternatively, we can write
$
    \cO\ket{x}\ket{i} =
         (-1)^{x_i}\ket{x}\ket{i} $ for $i\in [n]$ and $\cO\ket{x}\ket{0} = \ket{x}\ket{0}$.
\end{definition}
Next, we review the definition of a $T$-query quantum algorithm as follows.
\begin{definition}
Let $T\in \bN\cup\{0\}$ and $\Sigma \subset \{0,1\}^{n}$. A $T$-query quantum algorithm $Q$ (over input space
$\cX = \bC^{\Sigma}$) is specified by $T+1$ arbitrary unitaries $U_0,\dots,U_T \in \cU(\cW)$ on the
work register (independent of $x$). The algorithm acts on the initial state via the unitary $U_Q$ on
$\cX \otimes \cI \otimes \cA$ defined as
\begin{align*}
  U_Q := (I_\cX \otimes U_T)\,(\mathcal{O} \otimes I_\cA)\,(I_\cX \otimes U_{T-1})\,(\mathcal{O} \otimes I_\cA)\,\cdots\,(\mathcal{O} \otimes I_\cA)\,(I_\cX \otimes U_0).
\end{align*}
For $T=0$, we set $U_Q:=I_\cX\otimes U_0$.
Specifically, for a fixed input $x\in \Sigma$, the total unitary acts as $U_Q
\ket{x}_{\cX}{\ket{0}_{\cW}} = \ket{x}_{\cX}\ket{\psi_{x}}_{\cW}$ where
$$\ket{\psi_x}_{\cW} = U_Q^{x} \ket{0}_{\cW} = U_{T}(\cO^{x}\otimes I_{\cA}) U_{T-1} \dots
(\cO^{x}\otimes I_{\cA})  U_0\ket{0}_{\cW}.$$ We call $\ket{\psi_{x}}$ the final work state of
the algorithm on input $x$, and define the Gram matrix generated by the algorithm to be $\Psi_{Q}\in
\Pos(\Sigma)$ where $\Psi_{Q}(x,y) := \braket{\psi_x}{\psi_y}$.
\end{definition}

\subsubsection{Query complexity of  Gram matrices.}
Equivalently, the above procedure implements a $T$-query \emph{state generation} algorithm that
starts with fixed $\ket{0}$ states, and applies input-specific queries $\cO^{x}$ to generate
$\ket{\psi_x}_{\cW}$. More generally, one can consider the task of \emph{state evolution}: given two
families of input-specific states $\{\ket{\phi_x}:x\in \Sigma\}$ and $\{\ket{\psi_x}:x\in \Sigma\}$,
we ask how many queries are needed for a quantum algorithm $Q$ to implement $\ket{\phi_x} \to
\ket{\psi_x} = U_Q^{x}\ket{\phi_x}$.
\par
Throughout the paper, we fix a sufficiently large Hilbert space $\cV$, called the realization space,
and assume that all initial and target state families are realized in $\cV$. The work register $\cW$
of an algorithm is always understood to contain $\cV$ together with any additional registers required
by the algorithm.
\par
By \Cref{cor:gram-unitary-equivalence}, the query  complexity of implementing such a state evolution
only depends on the Gram matrices of the initial and final state families which are $\Phi =
(\braket{\phi_x}{\phi_y})_{x,y\in \Sigma}$ and $\Psi = (\braket{\psi_x}{\psi_y})_{x,y\in \Sigma}$.
As a result, we can characterize the complexity of this task as follows.

\begin{definition}[Query complexity of a Gram matrix transformation]
Let $\Sigma \subset \{0,1\}^{n}$ and let $\Phi,\Psi\in\Pos(\Sigma)$ be two
Gram matrices realized by state families $\{\ket{\phi_x}:x\in\Sigma\}\subseteq\cV$ and
$\{\ket{\psi_x}:x\in\Sigma\}\subseteq\cV$, respectively.
We define the query complexity of transforming $\Phi$ to $\Psi$, denoted by $\mathsf{Q}(\Phi \to \Psi)$, as
the minimum number of queries $T$ such that there exists a $T$-query quantum algorithm $Q$
with input register $\cX = \bC^{\Sigma}$ and a work space  $\cW$ containing $\cV$ such that
$\ket{\psi_x}_{\cW}=U_Q^x\ket{\phi_x}_{\cW}$ for all $x\in\Sigma$. For the all-one Gram matrix $J$,
we abbreviate $\mathsf{Q}(\Psi):= \mathsf{Q}(J\to \Psi)$ to define the state-generation complexity of $\Psi$.
In particular, the value is finite and the minimum is attained.
\end{definition}

Lee et al.~\cite{LMR+11} showed that the filtered $\gamma_2$ norm gives a semi-tight characterization of the quantum query complexity of state conversion. To state the consequence we first need the following definition.
\par
Let $\Phi,\Psi\in\Pos(\Sigma)$ be Gram matrices. The \emph{Hadamard product fidelity} between $\Phi$ and $\Psi$ is defined by $\cF_{\cD}(\Phi,\Psi):=\min_{\mu}\cF(\Phi\circ\kbr{\mu},\Psi\circ\kbr{\mu})$, where the minimum is over all probability distributions $\mu:\Sigma \to [0,1]$.
In particular, the value is finite and the minimum is attained.
\begin{theorem}[Semi-tight characterization for Gram matrices and state evolution \cite{LMR+11}]
Let $\Sigma\subseteq\{0,1\}^n$ and let $\Phi,\Psi\in\Pos(\Sigma)$ be two Gram matrices. Let $\eps=\eps(n)\in
(0,1)$ be an error parameter and let $\gamma:=\gamma_2(\Phi-\Psi|\{\Delta_i\})$. Then, the following holds.
\begin{enumerate}
    \item $\mathsf{Q}(\Phi \to \Psi) = \Omega(\gamma)$.
    \item There exists a Gram matrix $\Psi'$ such that $\cF_{\cD}(\Psi',\Psi)\geq \sqrt{1-\eps}$ and $\mathsf{Q}(\Phi \to \Psi') = O(\gamma\log (1/\eps)/\eps^2)$.
\end{enumerate}
\end{theorem}

\subsubsection{(Worst-case) quantum query complexity}
Let $f\subseteq \Sigma \times \Omega$ be a search problem. We recall that in the worst-case setting,
 the goal is to succeed on every input with high probability. That is, starting with input
 $\ket{x}_{\cX}$, we want to output a value $s\in \Omega$ such that with high probability $s\in
 f(x)$.

\begin{definition}[Worst-case query complexity of a search problem]
    Let $\Sigma \subset \{0,1\}^{n}$ and $f \subseteq \Sigma \times \Omega$ be a  search problem.
    Let $\eps=\eps(n)\in(0,1)$.  The worst-case $\eps$-error query complexity of $f$, denoted
    $\mathsf{Q}_{\eps}(f)$, is the minimum number of queries $t$ such that there exists a $t$-query quantum
    algorithm $Q$ and a projective measurement $\{\Pi_{s}:s\in \Omega\} \subset \cL(\cW)$
    such that, for all $x\in \Sigma$,
$
    \Pr_{Q}[Q(x) \in f(x)] = \sum_{s\in f(x)} \bra{\psi_{x}}\Pi_s \ket{\psi_x} \geq 1- \eps.
$
In particular, the value is finite and the minimum is attained.
\end{definition}

\paragraph{Negative-weight adversary method.}
We now restrict our attention to functional problems. The negative-weight adversary method gives a tight characterization of their bounded-error quantum query complexity~\cite{LMR+11}.
\begin{definition}[Negative-weight adversary method]
Let $\Sigma\subseteq\{0,1\}^n$ and let $f:\Sigma\to\Omega$ be a functional problem. The
\emph{negative-weight adversary bound} of $f$ is defined as
\begin{align*}
\mathrm{Adv}^{\pm}(f)
&:=\progline{\max_{\Gamma=\Gamma^\dagger}}{\|\Gamma\|}
{\Gamma\circ A_f=0; && \\
 \|\Gamma\circ\Delta_i\|\leq 1 && \forall i\in[n].}
\end{align*}
In particular, the value is finite and the maximum is attained.
\end{definition}
Equivalently, since $\Gamma\circ A_f=0$ if and only if $\Gamma=\Gamma\circ\Delta_f$, the dual
characterization of the filtered $\gamma_2$ norm gives
\begin{align*}
\mathrm{Adv}^{\pm}(f)
&=\gamma_2\bigl(\Delta_f\mid\{\Delta_i\circ\Delta_f\}\bigr).
\end{align*}
Every feasible $\Gamma$ is called an \emph{adversary matrix}. Moreover, the adversary bound is
equivalent, up to a constant factor, to the filtered $\gamma_2$ norm with the original query filters:
\begin{align*}
\mathrm{Adv}^{\pm}(f)
&=\Theta\bigl(\gamma_2(J-A_f\mid\{\Delta_i\})\bigr).
\end{align*}
When $f$ is a decision problem, the two quantities are equal~\cite{LMR+11}.
\par
\begin{theorem}[Tightness of the negative-weight adversary method (see, e.g. \cite{LMR+11})]
Let $\Sigma\subseteq\{0,1\}^n$ and let $f:\Sigma\to\Omega$ be a functional problem. For every fixed
constant $\eps\in(0,1/2)$,
\begin{align*}
\mathsf{Q}_{\eps}(f)=\Theta\bigl(\mathrm{Adv}^{\pm}(f)\bigr).
\end{align*}
\end{theorem}

\subsection{Analysis and optimization}
In this part, we present basic analysis tools that will be used throughout the paper.

\paragraph{Semicontinuity and continuity.}
Let $\cX,\cY$ be complex Hilbert spaces.
A function $f:\cX\to[-\infty,+\infty]$ is lower semicontinuous if, for any convergent
sequence $X_k\to X$ in $\cX$, it satisfies $f(X)\leq\liminf_k f(X_k)$.
Similarly, $f$ is upper semicontinuous if
$f(X)\geq\limsup_k f(X_k)$.
\par
Let $F:\cX\times\cY\to\bR$, where $\cY$ is nonempty. If
$F(\cdot,Y)$ is lower (upper) semicontinuous for every fixed
$Y\in\cY$, then $X\mapsto\sup_{Y\in\cY}F(X,Y)$
($X\mapsto\inf_{Y\in\cY}F(X,Y)$) is lower (upper) semicontinuous.
Similarly, convexity (concavity) is preserved under pointwise
supremum (infimum).
For any nonempty compact set $\cK\subset\cX$,
the maximum (minimum) of $f$ on $\cK$ is attained if $f$ is upper
(lower) semicontinuous.
\par
A function is both lower and upper semicontinuous if and only if it
is continuous; that is, for any convergent sequence $X_k\to X$ in
$\cX$, it satisfies $f(X_k)\to f(X)$. If a function $F:I\times\cY\to\bR$, where $I\subseteq\bR$ is an open interval,
is affine in its first argument (that is, $F(x,Y)=a_Yx+b_Y$), then the supremum (infimum)
$\sup_{Y\in\cY}F(x,Y)$ ($\inf_{Y\in\cY}F(x,Y)$) is continuous on $I$ whenever it is finite-valued.

A function $F:\cX\times\cY\to\bR$ is jointly continuous if, for any
convergent sequence $(X_k,Y_k)\to(X,Y)$ in $\cX\times\cY$,
it satisfies $F(X_k,Y_k)\to F(X,Y)$. Let $F:\cX\times\cY\to\bR$ be jointly continuous, and let
$\cK\subset\cY$ be nonempty and compact. Then
$X\mapsto\max_{Y\in\cK}F(X,Y)$
($X\mapsto\min_{Y\in\cK}F(X,Y)$) is continuous.


\paragraph{Sion's minimax theorem.}
We will use the following
specific variant of Sion's minimax theorem~\cite{Sio58} throuhgout the paper.

\begin{theorem}[Sion's minimax theorem, special case]\label{thm:sion-minimax}
    Let $\cL$ and $\cK$ be nonempty convex subsets of Hilbert spaces, with $\cL$ compact, and let
    $F:\cL\times \cK \to \bR$ be a continuous function where $F(\Gamma,\cdot)$ is convex and
    $F(\cdot,K)$ is concave. Then,
    \begin{align*}
        \max_{\Gamma \in \cL} \inf_{K\in \cK} F(\Gamma,K) = \inf_{K\in \cK} \max_{\Gamma \in \cL} F(\Gamma,K).
    \end{align*}
    Moreover, if $\cK$ is compact, the infima can be replaced by minima.
\end{theorem}

\paragraph{Semidefinite programming and Slater's theorem.}
Let $\cX$ and $\cY$ be two finite-dimensional complex Hilbert spaces.
Let $A\in\Herm(\cX)$ and $B\in\Herm(\cY)$ be Hermitian matrices,
and let $\cE:\Herm(\cX)\to\Herm(\cY)$ be a linear map.
The corresponding primal and dual SDPs are defined as
\[
\begin{progalign}
    \progrow[:=]
    {p}{\inf_{X\succeq0}}{\Tr[AX]}
    {\cE(X)\succeq B}
    \\[0.3em]
    \progrow[:=]
    {d}{\sup_{Y\succeq0}}{\Tr[BY]}
    {A\succeq\cE^*(Y)}.
\end{progalign}
\]
Here, $\cE^*$ denotes the adjoint of $\cE$ with respect to the
trace inner product.

\begin{theorem}[Slater's theorem (see, e.g.,
{\cite[Section~5.9.1]{BV04}})]\label{thm:slater}
If $p$ is finite and there exists a strictly feasible solution
$X_0\succ0$ satisfying $\cE(X_0)\succ B$, then $p=d$ and
the supremum defining $d$ is attained.
Similarly, if $d$ is finite and there exists a strictly feasible
solution $Y_0\succ0$ satisfying $A\succ\cE^*(Y_0)$, then
$p=d$ and the infimum defining $p$ is attained.
\end{theorem}

\paragraph{Support restriction lemma.}
The following lemma allows us to introduce or remove constraints from specific optimization problems.
\begin{lemma}\label{lem:adversary-support-replacement}
    Let $\Sigma$ be a finite set, $\Delta\in\{0,1\}^{\Sigma\times\Sigma}$ be a symmetric matrix,
    and let $\cL\subseteq\Herm(\Sigma)$ be nonempty. Also, let $F:\cL\to\bR$ be a function
    satisfying $\Gamma\circ\Delta\in\cL$ and $F(\Gamma\circ\Delta)\geq F(\Gamma)$ for all
    $\Gamma\in\cL$. Then,
    \begin{align*}
        \sup_{\Gamma\in \cL}F(\Gamma)
        &=\sup_{\substack{\Gamma\in\cL:\\ \Gamma \circ \Delta = \Gamma}}F(\Gamma).
    \end{align*}
    Moreover, if either supremum is attained, then so is the other.
\end{lemma}
\begin{proof}
    For every $\Gamma\in\cL$, let $\widetilde{\Gamma}:=\Gamma\circ\Delta$. By assumption,
    $\widetilde{\Gamma}\in\cL$ and $F(\widetilde{\Gamma})\geq F(\Gamma)$. Since
    $\Delta\circ\Delta=\Delta$, we have $\widetilde{\Gamma}\circ\Delta=\widetilde{\Gamma}$. This gives one direction of the equality;
    the reverse direction follows immediately as $\{\Gamma \in \cL: \Gamma \circ \Delta = \Gamma\} \subseteq\cL$.
\end{proof}

\section{Distributional quantum query complexity}
In the distributional setting, we are interested in the average success probability of an algorithm
under an input distribution $\mu$. To begin, we fix some notation that will be used throughout the
paper and recall the formal definition of distributional quantum query complexity.

\paragraph{Notation.} Let $\Sigma \subset \{0,1\}^{n}$ and $\cX = \bC^{\Sigma}$ be the input
register. We generally recognize a quantum algorithm by the Gram matrix $\Psi \in \Pos(\cX)$ that it
generates. We use $U_{\Psi} \in \cU(\cX\otimes \cW)$ to denote a corresponding unitary where $\cW$
denotes the work register. We can think of the coherent oracle model as a quantum channel
$\cE^{\Psi}(\rho) = U_\Psi (\rho_{\cX} \otimes \kbr{0}_{\cW})U_\Psi^{\dagger}$ where $\rho \in
\cD(\cX)$ denotes the initial state of the input register such as $\rho = \ketbra{x}$ or $\rho =
\ketbra{\mu}$ for a distribution $\mu$. Tracing out the work and input registers, respectively, we
get channels  $ \cE^{\Psi}_{\cX}(\rho) = \Tr_{\cW} [\cE^{\Psi}(\rho)]= \rho \circ {\Psi}^{\top}$ and
$\cE^{\Psi}_{\cW}(\rho) = \Tr_{\cX} [\cE^{\Psi}(\rho)]$.
\par

Throughout, the sets $\Sigma = \Sigma_n \subset \{0,1\}^{n}$ and $\Omega = \Omega_n$, problems $f =
f_n \subset \Sigma \times \Omega$, and distributions $\mu = \mu_n : \Sigma \to [0,1]$ are implicitly
indexed by the input length. For error parameters, however, we write $\eps=\eps(n)\in[0,1]$ to make
the possible dependence on $n$ explicit. We use the same convention for bias, success, and fidelity
parameters unless they are explicitly fixed constants.
\par

 \begin{definition}[Distributional Query Complexity]\label{def:distributional-query-complexity} Let
 $\Sigma \subset \{0,1\}^{n}$, $f \subseteq \Sigma \times \Omega$ be a  search problem, and
 $\mu:\Sigma \to [0,1]$ be a distribution. Also, let $\eps=\eps(n)\in[0,1)$ be an error parameter.
 We define $\mathsf{Q}^{\mu}_{\eps}(f)$ as the minimum query cost $t$ for which there exists a
 quantum algorithm with final states $\{\ket{\psi_x}:x\in \Sigma\}$ and projective measurement
 $\{\Pi_s:s\in \Omega\}$ satisfying
    \begin{align*}
        \E_{x\sim \mu} \left[ \sum_{s \in f(x)} \bra{\psi_x} \Pi_s \ket{\psi_x} \right] \geq 1-\eps.
    \end{align*}
\end{definition}
In the following, we first give a fidelity-based characterization for general search problems. In
\Cref{sec:functional-decision-problems}, we then give a simpler characterization for decision
problems and trace-norm bounds for functional and decision problems. Finally, in
\Cref{sec:distributional-state-conversion}, we construct an algorithm for distributional state
conversion.

\subsection{Gram matrix formulation of query complexity}\label{sec:gram-matrix-formulation} While
\Cref{def:distributional-query-complexity} is operationally precise, it carries a lot of moving
parts regarding the exact implementation of the work space, unitaries, and the measurements that
complicate the theoretical analysis. We therefore derive equivalent characterizations in terms of
the query complexity of Gram matrices that approximate a valid witness for the problem. Related
Gram-matrix and output-witness formulations appear in \cite[Claim~3.8]{LR12} and
\cite[Section~6.1]{Bel15}, respectively.
\subsubsection{Witness matrices and valid Gram matrices }

We begin by defining \emph{witness sets} for a given search problem and distribution; intuitively, a
valid witness encodes one of possibly many exact solutions to the problem $f$, where the solution on
input $x$  can only be supported on valid outputs in $f(x)$.

\begin{definition}[The set of $\mu$-witnesses for $f$]\label{def:witness-set} Let $\Sigma \subset
    \{0,1\}^{n}$, $f \subseteq \Sigma \times \Omega$ be a search problem, and $\mu:\Sigma\to [0,1]$
    be a distribution. We define the set of $\mu$-witnesses for $f$, denoted by
    $\mathcal{W}^{\mu}_f$, as the set of PSD matrices $W \in \Pos(\Sigma)$ that admit a
    decomposition
$
W=\sum_{s\in \Omega} W^{(s)}
$
such that the following conditions hold.
\begin{enumerate}
    \item \textbf{PSDness.} $W^{(s)} \succeq 0$ for all $s\in\Omega$.
    \item \textbf{Correctness.} For all $x\in \Sigma$ and $s\in\Omega$, if $s\notin f(x)$ then
    $W^{(s)}[x,x]=0$.
    \item \textbf{Normalization under $\mu$.} $\Tr[W\circ \kbr{\mu}] = 1$.
\end{enumerate}
In particular, for a functional problem $f:\Sigma \to \Omega$,  $\cW^{\mu}_{f} = \{W \in
\Pos(\Sigma): W\circ \Delta_f = 0,\Tr[W\circ \kbr{\mu}] = 1\}$.
\end{definition}

We emphasize that, unlike in the worst-case setting, a $\mu$-witness $W \in \mathcal{W}^{\mu}_{f}$
need not be a Gram matrix, as our goal is to get close to the density matrix $W\circ \kbr{\mu}$.
Therefore, we require normalization only on average under $\mu$.
\par
As previously discussed, quantum algorithms correspond to Gram matrices. We therefore define the set
of target Gram matrices that are sufficiently close to a valid witness under $\mu$.

\begin{definition}\label{def:valid-grams} Let $\Sigma \subset \{0,1\}^{n}$, $f \subseteq \Sigma
    \times \Omega$ be a search problem, and $\mu:\Sigma\to [0,1]$ be a distribution. Also, let
    $\eps=\eps(n)\in[0,1)$ be an error parameter. We define the set of $(\eps,\mu)$-targets for $f$
    as
    \begin{align*}
        \cG^{\mu}_{\eps}(f) := \{\Psi \in \Pos(\Sigma):\Psi \circ I = I;\,\max_{W\in \cW^{\mu}_f} \cF(W\circ \kbr{\mu}, \Psi \circ \kbr{\mu}) \geq \sqrt{1-\eps}\}.
    \end{align*}

\end{definition}
We also note that both sets have useful analytic properties that we will use repeatedly when
applying minimax theorems or establishing attainment of optima.
\begin{lemma}\label{lem:compactness-gmueps} The set $\cW^{\mu}_{f}$ is closed and convex, and its
restriction to the support of $\mu$, defined by
\begin{align*}
    \widehat{\cW}^{\mu}_{f}
    :=\cW^{\mu}_{f}\cap
    \{W:W\circ J_{\supp(\mu)}=W\}
\end{align*}
is convex and compact. Moreover, $\cG^{\mu}_{\eps}(f)$ is convex and compact.
\end{lemma}
We postpone the proof to \Cref{app:compactness}.
\subsubsection{Fidelity characterization}
We now prove an exact characterization of distributional query complexity in terms of the set
$\cG^{\mu}_{\eps}(f)$, allowing us to work directly with Gram matrices in the following sections.
\begin{theorem}[Fidelity characterization]\label{thm:fidelity-charecterization} Let $\Sigma \subset
\{0,1\}^{n}$, $f\subseteq\Sigma\times\Omega$ be a search problem, and $\mu:\Sigma\to [0,1]$ be a
distribution. Also, let $\eps=\eps(n)\in[0,1)$ be an error parameter. Then,
\[
\mathsf{Q}^\mu_{\eps}(f) = \min_{\Psi \in \cG^{\mu}_{\eps}(f)} \mathsf{Q}(\Psi).
\]
\end{theorem}

\begin{proof}
We prove each direction as follows.

\paragraph{Standard query $\Rightarrow$ witness.}
Let $t := \mathsf{Q}^{\mu}_{\eps}(f)$, and let $\Psi$ be the Gram matrix of a $t$-query quantum
algorithm that computes $f$ over $\mu$ with error at most $\eps$. Let $\ket{\psi_x}$ be the final
state of the algorithm on input $x$ (i.e., $\Psi(x,y) = \braket{\psi_x}{\psi_y}$) and let
$\{\Pi_{s}: s\in\Omega\}$ be the projective measurement used in the algorithm to generate outputs
$s\in\Omega$. Then the success probability of the algorithm can be computed as
\[
\delta := \mathbb{E}_{x\sim \mu}\big[\|\Pi_{f(x)}\ket{\psi_x}\|^2\big] \geq 1-\eps,
\]
where $f(x)\subseteq \Omega$ denotes the set of valid outputs on input $x$ and
$\Pi_{f(x)}:=\sum_{s\in f(x)}\Pi_s$. For each $s\in\Omega$ and input $x$, we define
\[
\ket{w_{x,s}} :=
\begin{cases}
\frac{1}{\sqrt{\delta}}\Pi_s\ket{\psi_x} & s\in f(x),\\
0 & \text{otherwise},
\end{cases}
\]
Then, letting $W^{(s)}(x,y):=\braket{w_{x,s}}{w_{y,s}}$ and
$
W:=\sum_{s\in\Omega} W^{(s)}$, we have $W^{(s)}\succeq 0$ for all $s$, and $W^{(s)}(x,x)=0$ whenever $s\notin f(x)$.
Moreover,
\begin{align*}
\Tr[ W\circ \kbr{\mu}]
&= \sum_x \mu(x)\,W(x,x)
= \sum_x \mu(x)\sum_{s\in f(x)} \frac{1}{\delta}\|\Pi_s\ket{\psi_x}\|^2
= \frac{1}{\delta}\mathbb{E}_{x\sim \mu}\big[\|\Pi_{f(x)}\ket{\psi_x}\|^2\big]
= 1.
\end{align*}
So, ${W}\in\cW^{\mu}_f$.
\par
Now consider the coherent states
$
\ket{\psi}:=\sum_x \sqrt{\mu(x)}\ket{x}\ket{\psi_x}
$ and
$
\ket{\phi}:=\sum_x \sqrt{\mu(x)}\ket{x}\cdot \frac{\Pi_{f(x)}\ket{\psi_x}}{\sqrt{\delta}}$,
purifying $(\Psi\circ \kbr{\mu})^\top$ and  $({W}\circ \kbr{\mu})^{\top}$, respectively. By
Uhlmann's theorem,
\begin{align*}
\cF(\Psi\circ \kbr{\mu},{W}\circ \kbr{\mu})& = \cF((\Psi\circ \kbr{\mu})^{\top},({W}\circ \kbr{\mu})^{\top})
 \geq |\braket{\phi}{\psi}|
=
\left|\frac{1}{\sqrt{\delta}}\sum_x \mu(x)\,\bra{\psi_x}\Pi_{f(x)}\ket{\psi_x}\right|
 \\
 &=
\frac{1}{\sqrt{\delta}}\,
\mathbb{E}_{x\sim \mu}\big[\|\Pi_{f(x)}\ket{\psi_x}\|^2\big]
=
\sqrt{\delta} \geq \sqrt{1-\eps}.
\end{align*}
Thus,  $\Psi \in \cG^{\mu}_{\eps}(f)$ and $\min_{\Phi \in \cG^{\mu}_{\eps}(f)} \mathsf{Q}(\Phi) \leq
\mathsf{Q}(\Psi) \leq t = \mathsf{Q}^{\mu}_{\eps}(f)$.
\paragraph{Witness $\Rightarrow$ standard query.}
Let $\Psi$ be the Gram matrix of an algorithm attaining $t=\mathsf{Q}(\Psi) = \min_{\Phi \in
\cG^{\mu}_{\eps}(f)} \mathsf{Q}(\Phi)$ with input and work registers $\cX$ and $\cW$, respectively.
Therefore, there exists a witness ${W}\in\cW^{\mu}_f$ such that $\Tr[{W}\circ \kbr{\mu}]=1$ and
$
\cF(\Psi\circ \kbr{\mu},{W}\circ \kbr{\mu})\ \ge\ \sqrt{1-\eps}.
$
Also, let
$
\ket{\psi}:=\sum_{x\in \Sigma} \sqrt{\mu(x)}\ket{x}\ket{\psi_x}
$
be a purification of $(\Psi\circ \kbr{\mu})^{\top}  = \Psi^{\top}\circ \kbr{\mu} =
\cE^{\Psi}_{\cX}(\kbr{\mu})$ generated by running the algorithm $\Psi$ with input register
initialized as the superposition $\sum_{x \in \Sigma} \sqrt{\mu(x)}\ket{x}$.

Since $W\in\cW^{\mu}_f$, we can write $W=\sum_{s\in \Omega} W^{(s)}$ with $W^{(s)}\succeq 0$. Taking
$\cY$ to be a sufficiently large space, we can find vectors $\{\ket{\phi_{x,s}}: x\in \Sigma,s\in
\Omega\}\subset \cY$ such that
$
W^{(s)}(x,y)=\braket{\phi_{x,s}}{\phi_{y,s}}
$ for all $x,y\in \Sigma$ and $s\in \Omega$.
Then,  the correctness condition $s\notin f(x)\Rightarrow W^{(s)}(x,x)=0$ implies
$\ket{\phi_{x,s}}=0$ whenever $s\notin f(x)$. Next, we define the output register $\cS :=
\bC^{\Omega}$ and a purification on registers $\cX\otimes \cS\otimes \cY$ by
\[
\ket{\phi}:=\sum_{x\in \Sigma}\sum_{s\in \Omega}\sqrt{\mu(x)}\ket{x}_{\cX}\ket{s}_{\cS}\ket{\phi_{x,s}}_{\cY}.
\]
As a result, $\ket{\phi}$ is a purification of ${W}^{\top}\circ \kbr{\mu} =
\Tr_{\cS,\cY}[\kbr{\phi}]$, and $\|\phi\|^2 = \Tr[\kbr{\phi}]=\Tr[W^{\top}\circ \kbr{\mu}] =
\Tr[W\circ \kbr{\mu}] =1$.
\par
Note that by appending fixed ancillas to the two purifying registers, we may assume without loss of
generality that $\dim(\cW)=\dim(\cS\otimes\cY)$. As a result, by Uhlmann's theorem
(\Cref{thm:uhlmann}) in the ``max-over-unitaries'' form, there exists a unitary $U\in \cU(\cW,
\cS\otimes \cY)$ such that
\begin{align*}
\big|\bra{\phi}(I_{\cX}\otimes U)\ket{\psi}\big| =  \cF((\Psi\circ \kbr{\mu})^{\top},({W}\circ \kbr{\mu})^{\top}) =
\cF(\Psi\circ \kbr{\mu},{W}\circ \kbr{\mu})
 \ge \sqrt{1-\eps}.
\end{align*}
\underline{\textit{The algorithm:}}
\begin{enumerate}
    \item Run the algorithm corresponding to $\Psi$ to prepare $\ket{\psi} \in \cX\otimes \cW$.
    \item Apply $U$ to the work registers to obtain $(I_{\cX}\otimes U)\ket{\psi} \in \cX\otimes \cS
    \otimes \cY$.
    \item Measure the register $\cS$ in the computational basis to output an outcome $s$.
\end{enumerate}
\noindent
\underline{\textit{Correctness analysis:}} Let $\Pi_f$ be the projector onto correct input-output
pairs:
\[
\Pi_f
:=\sum_{x\in \Sigma} \kbr{x}\ \otimes\ \left(\sum_{s\in f(x)} \kbr{s}\right)\ \otimes\ I_{\cY}.
\]
As a result, $\|\Pi_f(I_{\cX}\otimes U)\ket{\psi}\|^2 = \E_{x\sim\mu}[\sum_{s\in
f(x)}\bra{\psi_x}U^\dagger(\ketbra{s}\otimes I_{\cY})U\ket{\psi_x}]$ is the success probability of
the algorithm. For the target state $\ket{\phi}$ we have $\Pi_f\ket{\phi}=\ket{\phi}$, since
$\ket{\phi_{x,s}}=0$ when $s\notin f(x)$. Hence,
\begin{align*}
 \|\Pi_f(I_{\cX}\otimes U)\ket{\psi}\|^2
 &= \|\phi\|^2 \cdot \|\Pi_f(I_{\cX}\otimes U)\ket{\psi}\|^2 \ge\ |\bra{\phi}\Pi_f(I_{\cX}\otimes U)\ket{\psi}|^2
= |\bra{\phi}(I_{\cX}\otimes U)\ket{\psi}|^2
\\ &\ge ( \sqrt{1-\eps})^2
= 1-\eps
\end{align*}
where we used Cauchy--Schwarz in the first inequality. Thus, the constructed $t$-query algorithm
succeeds with probability at least $1-\eps$ under $\mu$. Therefore, $\mathsf{Q}^{\mu}_{\eps}(f)\leq
t$, completing the reverse-direction proof.
\end{proof}

\subsection{Functional and decision problems}\label{sec:functional-decision-problems}

In this part, we consider decision and functional problems where each input $x$ has a single
corresponding output, so $|f(x)|=1$. As a result, an exact algorithm must produce orthogonal states
on different inputs $x,y$ with different outputs $f(x)\neq f(y)$.
\subsubsection{A simpler characterization for decision problems}
For decision problems, we can replace the optimization over witness matrices in
\Cref{thm:fidelity-charecterization} with a single sign matrix, at the cost of a factor of two in
query complexity.

\begin{lemma}\label{lem:decision-sign-characterization} Let $\Sigma\subseteq\{0,1\}^n$,
    $g:\Sigma\to\{0,1\}$ be a decision problem, and $\mu:\Sigma\to[0,1]$ be a distribution. Then,
    for every bias $\beta=\beta(n)\in(0,1)$,
    \[
        \mathsf{Q}^{\mu}_{(1-\beta)/2}(g)
        \leq
        \min_{\substack{
            \Psi\succeq0,\ \Psi\circ I=I\\
            \cF(\Psi\circ\kbr{\mu},S_g\circ\kbr{\mu})\geq\beta
        }}
        \mathsf{Q}(\Psi)
        \leq
        2\mathsf{Q}^{\mu}_{(1-\beta)/2}(g),
    \]
    where $S_g(x,y):=(-1)^{g(x)+g(y)}$ is the sign matrix of $g$.
\end{lemma}
\begin{proof}
    Let $\ket{\sigma}:=\sum_{x\in\Sigma}(-1)^{g(x)}\sqrt{\mu(x)}\ket{x}$. Then,
    $S_g\circ\kbr{\mu}=\ketbra{\sigma}$. For any Gram matrix $\Psi$ with decomposition
    $\Psi(x,y)=\braket{\psi_x}{\psi_y}$, we have
    \begin{align*}
        \cF(\Psi\circ\kbr{\mu},S_g\circ\kbr{\mu})
        &=\sqrt{\bra{\sigma}(\Psi\circ\kbr{\mu})\ket{\sigma}}
        =\left(
            \sum_{x,y\in\Sigma}\mu(x)\mu(y)(-1)^{g(x)+g(y)}
            \braket{\psi_x}{\psi_y}
            \right)^{1/2} \\
        &=\left\|\sum_{x\in\Sigma}\mu(x)(-1)^{g(x)}\ket{\psi_x}\right\|.
    \end{align*}

    \paragraph{From decision algorithms to Gram matrices.}
    Let $T:=\mathsf{Q}^{\mu}_{(1-\beta)/2}(g)$, and let $U_Q^x$ and $\{\Pi_0,\Pi_1\}$ describe an
    optimal algorithm. Define $\ket{\phi_x}:=(U_Q^x)^\dagger(\Pi_0-\Pi_1)U_Q^x\ket{0}$, and let
    $\Psi$ be their Gram matrix. These states are normalized and can be generated with $2T$ queries
    by running the algorithm, applying the output sign, and reversing the algorithm. Moreover,
    \begin{align*}
        (-1)^{g(x)}\braket{0}{\phi_x}
        &=(-1)^{g(x)}\bra{0}(U_Q^x)^\dagger(\Pi_0-\Pi_1)U_Q^x\ket{0}\\
        &=2\Pr_Q[Q(x)=g(x)]-1.
    \end{align*}
    Therefore, we can lower bound the fidelity using Cauchy--Schwarz as follows:
    \begin{align*}
        \cF(\Psi\circ\kbr{\mu},S_g\circ\kbr{\mu})
        &=\left\|\sum_{x\in\Sigma}\mu(x)(-1)^{g(x)}\ket{\phi_x}\right\|\cdot\|\ket{0}\|\\
        &\geq\left|\sum_{x\in\Sigma}\mu(x)(-1)^{g(x)}\braket{0}{\phi_x}\right|
        =2\Pr_{x\sim\mu,Q}[Q(x)=g(x)]-1\geq\beta.
    \end{align*}
    As a result, $\Psi$ is feasible for the minimization above and $\mathsf{Q}(\Psi)\leq2T$,
    completing this direction of the proof.

    \paragraph{From Gram matrices to decision algorithms.}
    Let $\Psi$ be feasible for the minimum in the statement, and let $\{\ket{\psi_x}\}$ be the final
    states of an algorithm generating $\Psi$ with $T=\mathsf{Q}(\Psi)$ queries. Define
    \begin{align*}
        F&:=\cF(\Psi\circ\kbr{\mu},S_g\circ\kbr{\mu})\geq\beta>0,\\
        \ket{\eta}&:=
        \frac{\sum_{x\in\Sigma}\mu(x)(-1)^{g(x)}\ket{\psi_x}}{F}.
    \end{align*}
    By the fidelity identity above, $\|\ket{\eta}\|_2=1$. This state is independent of the input and
    can be prepared without queries. We can therefore prepare
    $(\ket{0}\ket{\eta}+\ket{1}\ket{\psi_x})/\sqrt{2}$ using $T$ queries. Then, we measure the first
    qubit in the Hadamard basis, outputting $0$ on $\ket{+}$ and $1$ on $\ket{-}$.
    On input $x$, the probability of obtaining $\ket{+}$ is then
    \begin{align*}
        \left\|(\bra{+}\otimes I)
        \frac{\ket{0}\ket{\eta}+\ket{1}\ket{\psi_x}}{\sqrt{2}}\right\|_2^2
        &=\frac14\bigl\|\ket{\eta}+\ket{\psi_x}\bigr\|_2^2=\frac14\left(
        \braket{\eta}{\eta}+\braket{\eta}{\psi_x}
        +\braket{\psi_x}{\eta}+\braket{\psi_x}{\psi_x}\right)\\
        &=\frac12\bigl(1+\operatorname{Re}\braket{\eta}{\psi_x}\bigr).
    \end{align*}
    Similarly, the probability of obtaining $\ket{-}$ is
    $
        \frac12\bigl(1-\operatorname{Re}\braket{\eta}{\psi_x}\bigr).
    $
    The error outcome is $\ket{-}$ when $g(x)=0$ and $\ket{+}$ when $g(x)=1$. Thus, the probability
    of error on input $x$ is
    \[
        \frac12\bigl(1-(-1)^{g(x)}
        \operatorname{Re}\braket{\eta}{\psi_x}\bigr).
    \]

    Therefore, the error probability under $\mu$ is
    \begin{align*}
        &\frac12\sum_{x\in\Sigma}\mu(x)
        \bigl(1-(-1)^{g(x)}\operatorname{Re}\braket{\eta}{\psi_x}\bigr)
        =\frac12\left(1-\operatorname{Re}\left[\bra{\eta}F\ket{\eta}\right]\right)
        =\frac12\bigl(1-F\braket{\eta}{\eta}\bigr)
        =\frac{1-F}{2}\leq\frac{1-\beta}{2}. \qedhere
    \end{align*}
\end{proof}

\subsubsection{Trace-norm bounds}
In the following lemma, we bound the overlap under $\mu$ that can remain between such inputs when
allowing distributional error, in terms of the trace norm. This bound will later allow us to define
a relaxed measure directly based on the trace-norm bound that remains closely related to
distributional query complexity and will be used later in the composition proof.

\begin{lemma}\label{lem:trace-norm-upper-bound-for-boolean-functions} Let $\Sigma\subset\{0,1\}^n$,
    $f:\Sigma\to\Omega$ be a functional problem, and $\mu:\Sigma\to[0,1]$ be a distribution. Also,
    let $\eps=\eps(n)\in[0,1/2)$ be an error parameter. Then, for every $\Psi\in\cG^\mu_\eps(f)$,
    \[
        \|\Psi\circ\Delta_f\circ\kbr{\mu}\|_{\tr}\leq 4\sqrt{\eps}.
    \]
    Moreover, if $f:\Sigma\to\{0,1\}$ is a decision problem, then
    \[
        \|\Psi\circ\Delta_f\circ\kbr{\mu}\|_{\tr}\leq 2\sqrt{\eps(1-\eps)}.
    \]
    \end{lemma}
    \begin{proof}
    By \Cref{def:valid-grams}, there exists a witness $W\in\cW_f^\mu$ such that
    $\cF(W\circ\kbr{\mu},\Psi\circ\kbr{\mu})\geq\sqrt{1-\eps}$. The Fuchs--van de Graaf inequalities
    give
    \[
        \|(W-\Psi)\circ\kbr{\mu}\|_{\tr}
        \leq 2\sqrt{1-\cF^2(W\circ\kbr{\mu},\Psi\circ\kbr{\mu})}
        \leq 2\sqrt{\eps}.
    \]
    Let $X:=(\Psi-W)\circ\kbr{\mu}$. Since $W\circ\Delta_f=0$, we have
    $\|\Psi\circ\Delta_f\circ\kbr{\mu}\|_{\tr}=\|X\circ\Delta_f\|_{\tr}$. Consequently,
    \[
        \|X\circ\Delta_f\|_{\tr}
        \leq\|X\|_{\tr}+\|X\circ A_f\|_{\tr}
        \leq \|X\|_{\tr}+\gamma_2(A_f)\|X\|_{\tr}
        \leq 2\|X\|_{\tr}
        \leq 4\sqrt{\eps}.
    \]
    Now suppose that $f$ is a decision problem. Write $\mu=p_0\mu_0+p_1\mu_1$, where $p_0+p_1=1$,
    and $\mu_0$ and $\mu_1$ are the conditional distributions of $\mu$ over the $0$-inputs
    $f^{-1}(0)$ and the $1$-inputs $f^{-1}(1)$, respectively. Without loss of generality, we assume
    that $p_0,p_1>0$, as otherwise $\Psi\circ\Delta_f\circ\kbr{\mu}=0$ and the bound is immediate.
    \par
    We define $\sigma_b^\Psi:=\cE^\Psi_{\cW}(\kbr{\mu_b})$ to be the work-register state after
    running the algorithm $\Psi$ on $\ket{\mu_b}$ and tracing out the input register, for
    $b\in\{0,1\}$. These states capture the information available to the algorithm for
    distinguishing inputs drawn from $\mu_0$ and $\mu_1$ through its final measurement. By the
    reverse direction of \Cref{thm:fidelity-charecterization}, the algorithm $\Psi$ computes $f$
    over $\mu$ with error at most $\eps$. By Helstrom's theorem (\Cref{thm:helstrom}), the minimum
    average error of any final measurement, with prior probabilities $p_0$ and $p_1$, is
    $\frac12(1-\|p_0\sigma_0^\Psi-p_1\sigma_1^\Psi\|_{\tr})$. Therefore,
    $\|p_0\sigma_0^\Psi-p_1\sigma_1^\Psi\|_{\tr}\geq1-2\eps$. We claim that
    \[
        \|\Psi\circ\Delta_f\circ\kbr{\mu}\|_{\tr}
        =2\sqrt{p_0p_1}\cF(\sigma_0^\Psi,\sigma_1^\Psi).
    \]
    Assuming the above, the generalized Fuchs--van de Graaf inequality
    (\Cref{lem:Fuchs-van-de-Graaf-priors}) yields
    \[
        \|\Psi\circ\Delta_f\circ\kbr{\mu}\|_{\tr}
        =2\sqrt{p_0p_1}\cF(\sigma_0^\Psi,\sigma_1^\Psi)
        \leq\sqrt{1-\|p_0\sigma_0^\Psi-p_1\sigma_1^\Psi\|_{\tr}^2}
        \leq 2\sqrt{\eps(1-\eps)},
    \]
   giving the desired bound.
    \par
    It only remains to prove the claim. Let $\{\ket{\psi_x}:x\in \Sigma\}$ be a Gram decomposition
of $\Psi$; i.e., $\Psi(x,y) = \braket{\psi_x}{\psi_y}$ for all $x,y\in \Sigma$. For $b\in \{0,1\}$,
let $\Sigma_{b} = f^{-1}(b)$, and define  $B_{b} \in \cL(\bC^{\Sigma_b}, \cW)$ so that $B_{b}\ket{x}
= \sqrt{\mu(x)}\ket{\psi_x}$ for all $x\in f^{-1}(b)$. Then,
    \begin{align*}
        \sigma_b^\Psi
        &= \cE^\Psi_{\cW}(\kbr{\mu_b})
        = \sum_{x\in\Sigma_b}\mu_b(x)\ketbra{\psi_x}
        = \frac{1}{p_b}B_bB_b^{\dagger}.
    \end{align*}
On the other hand, we have $(\Psi\circ\kbr{\mu})(x,y) = \braket{\psi_x}{\psi_y}\sqrt{\mu(x)\mu(y)} =
B^{\dagger}_{f(x)}B_{f(y)}^{\vphantom{\dagger}}(x,y)$. Therefore, $\Psi\circ \Delta_f \circ
\kbr{\mu}$ has the block form
\[
\Psi\circ \Delta_f \circ \kbr{\mu}
=
\begin{pmatrix}
0 & B_0^{\dagger}B_1\\
B_1^{\dagger}B_0 & 0
\end{pmatrix}.
\]
Hence,
\begin{align*}
\| \Psi\circ \Delta_f \circ \kbr{\mu} \|_{\tr}
&= 2\|B_0^{\dagger}B_1\|_{\tr}
 = 2\|\sqrt{B_0B_0^{\dagger}} \sqrt{B_1B_1^{\dagger}}\|_{\tr}
 = 2\sqrt{p_0p_1}\|\sqrt{\sigma_0^\Psi} \sqrt{\sigma_1^\Psi}\|_{\tr} \\
&= 2\sqrt{p_0p_1}\cF(\sigma_0^\Psi,\sigma_1^\Psi). \qedhere
\end{align*}
    \end{proof}

\subsection{Distributional state conversion}\label{sec:distributional-state-conversion} In this
section, we construct algorithms for approximately converting an initial Gram matrix $\Phi$ into a
target Gram matrix $\Psi$, where the approximation is measured by fidelity under $\mu$. Later, we
use these algorithms to obtain upper bounds on distributional query complexity in terms of our
adversary measures.
\par
In \Cref{thm:catalyst-state-conversion}, we provide an upper bound on the number of queries needed
to achieve a given fidelity. We use this result throughout the paper to construct quantum query
algorithms for a single instance of a problem. In contrast,
\Cref{prop:multiplicative-state-conversion} guarantees the existence of an algorithm with a
particular exponential tradeoff between query cost and fidelity (without allowing us to choose
either quantity). This gives limited control for a single instance, however, we can combine such
algorithms across $k$ many  instances to obtain asymptotic direct-product bounds for exponentially
small success that have vanishing error overhead as $k\to \infty$.

\subsubsection{State conversion with prescribed error}

\begin{theorem}[Distributional state conversion via catalyst
    reuse]\label{thm:catalyst-state-conversion} Let $\Sigma\subseteq\{0,1\}^{n}$,
    $\mu:\Sigma\to[0,1]$ be a distribution, and $\Phi,\Psi\in\Pos(\Sigma)$ be Gram matrices with
    $\Phi\circ I=\Psi\circ I=I$. Suppose that there exist $X_1,\ldots,X_n\in\Pos(\Sigma)$ and
    $r\geq1$ such that
    \begin{align*}
        \Phi-\Psi=\sum_{i=1}^{n}\Delta_i\circ X_i,
        \qquad
        \Psi\succeq\frac{r-1}{2}\sum_{i=1}^{n}X_i.
    \end{align*}
    Then, for every integer $T\geq1$, there exists a Gram matrix $\widetilde{\Psi}\in\Pos(\Sigma)$
    with $\widetilde{\Psi}\circ I=I$ such that $\mathsf{Q}(\Phi\to\widetilde{\Psi})\leq T$ and
    for every $W\succeq0$ with $\Tr(W\circ\kbr{\mu})=1$, it satisfies
    \begin{align*}
        \cF(\widetilde{\Psi}\circ\kbr{\mu},W\circ\kbr{\mu})
        &\geq(1-\eps_T)\alpha-\sqrt{\eps_T(1-\alpha)},
    \end{align*}
    where $\alpha:=\cF(\Psi\circ\kbr{\mu},W\circ\kbr{\mu})$ and $\eps_T:=
        {\sum_{i=1}^{n}\Tr\bigl(X_i\circ\kbr{\mu}\bigr)}/ {\sum_{j=0}^{T-1}r^j}.
    $
\end{theorem}
The hypothesis $\Phi - \Psi = \sum_{i=1}^{n}\Delta_i\circ X_i$ is standard in
adversary formulations of state conversion, where each $X_i$ encodes the contribution associated
with querying index $i$ in converting $\Phi$ to $\Psi$. For $r=1$, this naturally translates to
paying cost $T$ for infidelity $\sum_{i}\Tr[X_i\circ \kbr{\mu}]/ T$.
\par
For $r>1$, however, our extra condition $\Psi\succeq\frac{r-1}{2}\sum_iX_i$ states that the target
$\Psi$ already contains a PSD component proportional to $\sum_iX_i$. Roughly speaking, this allows
us to aim for the ``smaller'' target $\Psi - (r-1)/2 \sum_{i} X_i$, which remains PSD, and restore
the remaining component $(r-1)/2\sum_iX_i$ without additional queries. As a result, for the same
query cost $T$, we improve the infidelity bound by replacing the denominator $T$ with
$\sum_{j=0}^{T-1}r^j=(r^T-1)/(r-1)$.
\paragraph{Proof overview.}
For the proof, we generalize the catalyst-reuse argument of Belovs and Yolcu \cite{BY23}. More
specifically, our algorithm will have two main ingredients, which we informally describe in the
following for the case $r=1$.
\begin{itemize}
    \item \textbf{Catalyst synthesis.} From the given decomposition of $\Phi-\Psi$, we construct
    \emph{catalyst} vectors $\ket{\kappa_x}$ that enable the conversion of
    $\ket{\phi_x}\oplus\ket{\kappa_x}$ to $\ket{\psi_x}\oplus\ket{\kappa_x}$ using a single oracle
    query, while keeping the catalyst unharmed.\footnote{For $r>1$, the catalyst is rescaled, and
    the conversion uses an auxiliary target $\ket{a_x}$ instead of $\ket{\psi_x}$.}
    \item \textbf{Catalyst reuse.} But, the catalyst depends on the input, we cannot prepare it for
    free and replacing it with $\ket{0}$ might incur a large error. Therefore, we split the initial
    state into smaller pieces and use a single smaller catalyst again and again to convert these
    pieces one at a time. Omitting the catalyst then yields an algorithm whose error decreases as
    the number of queries increases.
\end{itemize}
We present these two ingredients separately in \Cref{lem:catalyst-synthesis,lem:catalyst-reuse} and
then combine them to prove \Cref{thm:catalyst-state-conversion}.
\paragraph{Catalyst synthesis.}

\begin{lemma}[Catalyst synthesis]\label{lem:catalyst-synthesis} Let $\Sigma\subseteq\{0,1\}^n$, and
    let $\Phi,A\succeq0$ be Gram matrices of vectors $\ket{\phi_x},\ket{a_x}\in\cV$, where $\cV$ is
    a finite-dimensional Hilbert space, with $\Phi(x,y)=\braket{\phi_x}{\phi_y}$ and
    $A(x,y)=\braket{a_x}{a_y}$. Assume that, for some $r\geq1$, there exist $X_1,\dots,X_n\succeq0$
    such that
    \begin{align*}
        \Phi-A=\sum_{i=1}^{n}(rJ-S_i)\circ X_i.
    \end{align*}
    Then there exist a catalyst space $\cK=\cI\otimes\cZ$, consisting of the query-index register
    $\cI$ and an auxiliary workspace $\cZ$ with $\dim(\cZ)\leq|\Sigma|$, catalyst vectors
    $\{\ket{\kappa_x}\}_{x\in\Sigma}\subseteq\cK$, and an input-independent unitary
    $U\in\cU(\cV\oplus\cK)$ such that
    \begin{align*}
        U(\ket{\phi_x}\oplus\cO^x\ket{\kappa_x})
        &=\ket{a_x}\oplus\sqrt{r}\,\ket{\kappa_x}, &&
        \braket{\kappa_x}{\kappa_y}
        =\sum_{i=1}^{n}X_i(x,y)
    \end{align*}
    for all $x,y\in\Sigma$. Here, $\cO^x$ acts on the query-index register $\cI$.
\end{lemma}
\begin{proof}
    Since each $X_i\succeq0$, we can choose vectors $\{\ket{\kappa_x^i}:x\in\Sigma\}\subseteq\cZ$ in
    a common space $\cZ$ with $\dim(\cZ)\leq|\Sigma|$ satisfying
    $\braket{\kappa_x^i}{\kappa_y^i}=X_i(x,y)$. Let $\cK:=\cI\otimes\cZ$, and for $x\in\Sigma$,
    define $\ket{\kappa_x}:=\sum_{i=1}^{n}\ket{i}_{\cI}\ket{\kappa_x^i}$. Then
    $\braket{\kappa_x}{\kappa_y}=\sum_{i=1}^{n}X_i(x,y)$. For the catalyst unitary, we claim that
    the two families of vectors
    \begin{align*}
        \left\{\ket{\phi_x}\oplus\cO^x\ket{\kappa_x}:x\in\Sigma\right\}
        \quad\text{and}\quad
        \left\{\ket{a_x}\oplus\sqrt r\,\ket{\kappa_x}:x\in\Sigma\right\}
    \end{align*}
    have the same Gram matrix. Indeed, for all $x,y\in\Sigma$,
    \begin{align*}
        \bra{\kappa_x}\left(rI-(\cO^x)^\dagger\cO^y\right)\ket{\kappa_y}
        &=\sum_{i=1}^{n}\bigl(r-S_i(x,y)\bigr)X_i(x,y)\\
        &=(\Phi-A)(x,y)
        =\braket{\phi_x}{\phi_y}-\braket{a_x}{a_y}.
    \end{align*}
    Thus, by \Cref{cor:gram-unitary-equivalence}, there exists an input-independent unitary
    $U\in\cU(\cV\oplus\cK)$ mapping the first family to the second, as required.
\end{proof}

\paragraph{Catalyst-reuse Lemma.}
\begin{restatable}[Catalyst-reuse]{lemma}{catalyst-reuse}\label{lem:catalyst-reuse} Let
    $\Sigma\subseteq\{0,1\}^{n}$, and for every $x\in\Sigma$, let $\ket{a_x},\ket{b_x}\in\cV$ be
    initial and target vectors. Let $\cK$ be a catalyst space containing the query-index register
    $\cI$, on which the oracle acts. Suppose that, for some $r>0$, there are catalyst vectors
    $\{\ket{\kappa_x}\}_{x\in\Sigma}$ and an input-independent unitary $U\in\cU(\cV\oplus\cK)$ such
    that, for every $x\in\Sigma$,
    \begin{align*}
        U\left(\ket{a_x}\oplus\cO^x\ket{\kappa_x}\right)
        =\ket{b_x}\oplus\sqrt{r}\,\ket{\kappa_x}.
    \end{align*}
    Then, for every integer $T\geq1$, there exists a $T$-query quantum algorithm $Q_T$ such
    that\footnote{Formally, $Q_T$ uses a clock register $\mathcal T\simeq\bC^T$ initialized and
    returned to $\ket{1}_{\mathcal T}$. We identify $\cV$ with $\ket{1}_{\mathcal T}\otimes\cV$ and
    suppress this fixed clock register from the input and output vectors.}
    \begin{align*}
        U_{Q_T}^{x}\left(
            \ket{a_x}\oplus\frac{\ket{\kappa_x}}{\sqrt{s}}
        \right)
        =\ket{b_x}\oplus\frac{r^{T/2}\ket{\kappa_x}}{\sqrt{s}},
    \end{align*}
    where $s:=\sum_{t=0}^{T-1}r^t$, and $U_{Q_T}^x$ denotes the unitary implemented by $Q_T$ using
    the oracle $\cO^x$.\footnote{Note that the two sides need not be unit vectors.}
\end{restatable}
To give a rough sketch of the proof, let us again consider the case $r=1$. If the catalyst vector
$\ket{\kappa_x}$ is provided, one can unitarily map $\ket{\phi_x}\oplus\cO^x\ket{\kappa_x}$ to
$\ket{\psi_x}\oplus\ket{\kappa_x}$ in a single query. However, $\ket{\kappa_x}$ depends on $x$ and
therefore cannot be prepared without knowing the input. The key idea is to split $\ket{\phi_x}$ into
$T$ pieces $\ket{\phi_x}\ket{t}/\sqrt{T}$, while using a single shared, scaled catalyst
$\ket{\kappa_x}/\sqrt{T}$ to transform the pieces one at a time. Omitting the catalyst, we then
obtain a $T$-query algorithm with squared error of order $\|\kappa_x\|^2/T$. The proof of
\Cref{lem:catalyst-reuse} follows a similar construction with geometric weights and is postponed to
\Cref{app:catalyst-reuse}.
\paragraph{Combining the ingredients.}
We now provide the proof for the general state conversion.
\begin{proof}[Proof of \Cref{thm:catalyst-state-conversion}]
    First, without loss of generality, we assume that $\Sigma=\supp(\mu)$. Indeed, let
    $S:=\supp(\mu)$ and replace all matrices by their principal submatrices indexed by $S$. Any
    algorithm constructed on $S$ can be run on the inputs $\Sigma$, starting from the original
    initial states, with the same number of queries. Its output Gram matrix agrees with the one
    constructed on $S$, and all fidelity expressions and $\eps_T$ are unchanged since $\mu$ is
    supported on $S$. Thus, the conclusion on $S$ gives the required conclusion on $\Sigma$.
    \par
    Choose initial and target states $\{\ket{\phi_x}:x\in\Sigma\}$ and $\{\ket{\psi_x}:x\in\Sigma\}$
    in $\cV$ with Gram matrices $\Phi$ and $\Psi$, respectively. We choose $\dim\cV\geq|\Sigma|$ so
    that $\cV$ can accommodate a purification of any density matrix on $\cX$. Throughout the proof,
    whenever $\cO^x$ acts on a space containing the query-index register $\cI$, we use the same
    notation for its extension by the identity on all registers other than $\cI$.
    \par
    By the second hypothesis, $A:=\Psi-(r-1)\sum_iX_i/2\succeq0$. Choose vectors
    $\{\ket{a_x}:x\in\Sigma\}\subseteq\cV$ with Gram matrix $A$. Then
    \begin{align*}
        \Phi-A
        &=\sum_{i=1}^{n}\Delta_i\circ X_i+\frac{r-1}{2}\sum_{i=1}^{n}X_i
        =\sum_{i=1}^{n}(rJ-S_i)\circ\frac{X_i}{2}.
    \end{align*}
    Applying \Cref{lem:catalyst-synthesis} with target $A$ and witnesses $X_i/2$, we obtain a
    catalyst space $\cK$, vectors $\{\ket{\kappa_x}:x\in\Sigma\}$, and an input-independent unitary
    $U\in\cU(\cV\oplus\cK)$ such that
    \begin{equation}\label{eq:catalyst-unitary}
        U\left(\ket{\phi_x}\oplus\cO^x\ket{\kappa_x}\right)
        =\ket{a_x}\oplus\sqrt r\,\ket{\kappa_x}
        \qquad(x\in\Sigma),
    \end{equation}
    and $\braket{\kappa_x}{\kappa_y} =\frac12\sum_{i=1}^{n}X_i(x,y)$.
    \par
    Write $s:=\sum_{t=0}^{T-1}r^t$. Applying \Cref{lem:catalyst-reuse}, we obtain a $T$-query
    algorithm $R_T$ satisfying
    \begin{align*}
        U_{R_T}^x\left(\ket{\phi_x}\oplus\frac{\ket{\kappa_x}}{\sqrt{s}}\right)
        =\ket{a_x}\oplus\frac{r^{T/2}\ket{\kappa_x}}{\sqrt{s}}
        \qquad(x\in\Sigma).
    \end{align*}
    Moreover, for $x,y\in\Sigma$, we have
    \begin{align*}
        \braket{a_x}{a_y}+\frac{r^T}{s}\braket{\kappa_x}{\kappa_y}
        &=\braket{\psi_x}{\psi_y}-(r-1)\braket{\kappa_x}{\kappa_y}
          +\frac{r^T}{s}\braket{\kappa_x}{\kappa_y}\\
        &=\braket{\psi_x}{\psi_y}
          +\frac{r^T-(r-1)\sum_{j=0}^{T-1}r^j}{\sum_{j=0}^{T-1}r^j}
          \braket{\kappa_x}{\kappa_y}\\
        &=\braket{\psi_x}{\psi_y}+\frac{\braket{\kappa_x}{\kappa_y}}{s},
    \end{align*}
    where the first line is by definition of $A(x,y) = \braket{a_x}{a_y}$. As a consequence, the two
    families of vectors
    \begin{align*}
        \left\{\ket{a_x}\oplus\frac{r^{T/2}\ket{\kappa_x}}{\sqrt{s}}:x\in\Sigma\right\}
        \quad\text{and}\quad
        \left\{\ket{\psi_x}\oplus\frac{\ket{\kappa_x}}{\sqrt{s}}:x\in\Sigma\right\}
    \end{align*}
    also have the same Gram matrix. Thus, by \Cref{cor:gram-unitary-equivalence}, composing with
    another input-independent unitary gives a $T$-query algorithm $Q_T$ satisfying
    \begin{equation}\label{eq:catalyst-conversion}
        U_{Q_T}^x\left(\ket{\phi_x}\oplus\frac{\ket{\kappa_x}}{\sqrt{s}}\right)
        =\ket{\psi_x}\oplus\frac{\ket{\kappa_x}}{\sqrt{s}}
        \qquad(x\in\Sigma).
    \end{equation}
    \par
    Next, we forget about the catalyst and run the algorithm on $\ket*{\phi_x}\oplus0$. More
    specifically, on input $x\in\Sigma$, our algorithm computes the final state
    $\ket*{\widetilde{\psi}_x}:=U_{Q_T}^x(\ket*{\phi_x}\oplus0)$. These are unit vectors. Thus,
    their Gram matrix $\widetilde{\Psi}$ satisfies $\widetilde{\Psi}\circ I=I$ and
    $\mathsf{Q}(\Phi\to\widetilde{\Psi})\leq T$. Consider the states
    \begin{align*}
        \ket*{\widetilde{\psi}}
        :=\sum_{x\in\Sigma}\sqrt{\mu(x)}\ket*{x}\ket*{\widetilde{\psi}_x},
        \qquad
        \ket*{\psi}
        :=\sum_{x\in\Sigma}\sqrt{\mu(x)}\ket*{x}(\ket*{\psi_x}\oplus 0),
    \end{align*}
    which are purifications of $(\widetilde{\Psi}\circ\kbr{\mu})^{\top}$ and
    $(\Psi\circ\kbr{\mu})^{\top}$, respectively. By \eqref{eq:catalyst-conversion} and linearity,
    \begin{align*}
        \ket*{\widetilde{\psi}}-\ket*{\psi}
        &=\sum_{x\in\Sigma}\sqrt{\mu(x)}\ket*{x}
          \left[U_{Q_T}^x(\ket*{\phi_x}\oplus0)-(\ket*{\psi_x}\oplus0)\right]\\
        &=\sum_{x\in\Sigma}\sqrt{\mu(x)}\ket*{x}
          \left[\begin{aligned}
              &U_{Q_T}^x(\ket*{\phi_x}\oplus0)
              -U_{Q_T}^x\bigl(\ket*{\phi_x}\oplus\ket*{\kappa_x}/\sqrt{s}\bigr)\\
              &\quad+U_{Q_T}^x\bigl(\ket*{\phi_x}\oplus\ket*{\kappa_x}/\sqrt{s}\bigr)
              -(\ket*{\psi_x}\oplus0)
          \end{aligned}\right]\\
        &=\sum_{x\in\Sigma}\sqrt{\mu(x)}\ket*{x}
          \left[\begin{aligned}
              &U_{Q_T}^x(\ket*{\phi_x}\oplus0)
              -U_{Q_T}^x\bigl(\ket*{\phi_x}\oplus\ket*{\kappa_x}/\sqrt{s}\bigr)\\
              &\quad+\bigl(\ket*{\psi_x}\oplus\ket*{\kappa_x}/\sqrt{s}\bigr)
              -(\ket*{\psi_x}\oplus0)
          \end{aligned}\right]\\
        &=\sum_{x\in\Sigma}\sqrt{\mu(x)}\ket*{x}
          (I-U_{Q_T}^x)\bigl(0\oplus\ket*{\kappa_x}/\sqrt{s}\bigr).
    \end{align*}
    Fix any $W\succeq0$ with $\Tr(W\circ\kbr{\mu})=1$. By Uhlmann's theorem, there exists a
    purification $\ket*{\omega}\in\cX\otimes\cV$ of $(W\circ\kbr{\mu})^\top$ such that
    \begin{align*}
        \alpha
        =\cF(\Psi\circ\kbr{\mu},W\circ\kbr{\mu})
        =\braket{\psi}{\omega}\in[0,1].
    \end{align*}
    Write $\ket*{\omega}=\alpha\ket*{\psi}+\ket*{e}$, where $\braket{\psi}{e}=0$ and
    $\|e\|^2=1-\alpha^2$. Then,
    \begin{align*}
        \cF(\widetilde{\Psi}\circ\kbr{\mu},W\circ\kbr{\mu})
        &\geq\Re\braket*{\omega}{\widetilde{\psi}} =\alpha\Re\braket*{\psi}{\widetilde{\psi}}
          +\Re\braket*{e}{\widetilde{\psi}}.
    \end{align*}
    We bound the two terms separately.
    \begin{itemize}
        \item For the first term, since $\ket*{\psi}$ and $\ket*{\widetilde{\psi}}$ are unit
        vectors, we have
        \begin{align*}
            \Re\braket*{\psi}{\widetilde{\psi}}
            &=1-\tfrac12\|\widetilde{\psi}-\psi\|^2\\
            &=1-\tfrac12\sum_{x\in\Sigma}\mu(x)
              \left\|(I-U_{Q_T}^x)\bigl(0\oplus\ket*{\kappa_x}/\sqrt{s}\bigr)\right\|^2\\
            &\geq1-\tfrac12\sum_{x\in\Sigma}\mu(x)
              \left(\left\|0\oplus\ket*{\kappa_x}/\sqrt{s}\right\|
              +\left\|U_{Q_T}^x\bigl(0\oplus\ket*{\kappa_x}/\sqrt{s}\bigr)\right\|\right)^2\\
            &=1-2\sum_{x\in\Sigma}\mu(x)
              \left\|\kappa_x/\sqrt{s}\right\|^2\\
            &=1-\sum_{i=1}^{n}\Tr\bigl(X_i\circ\kbr{\mu}\bigr)/s\\
            &=1-\eps_T.
        \end{align*}
        Here we used the expression for $\ket*{\widetilde{\psi}}-\ket*{\psi}$ above, followed by the
        triangle inequality and unitarity. The last two equalities use
        $\|\kappa_x\|^2=\tfrac12\sum_iX_i(x,x)$ and the definition of $\eps_T$.

        \item For the second term, since both $\ket*{\omega}$ and $\ket*{\psi}$ lie in
        $\cX\otimes\cV$, so does $\ket*{e}$, and therefore it has no component in the catalyst
        space. That is, $\bra{e}\ket*{x}(0\oplus\ket*{\kappa_x})=0$. Using this and
        $\braket{e}{\psi}=0$, we have
        \begin{align*}
            \Re\braket*{e}{\widetilde{\psi}}
            &=\Re\left[\bra{e}\bigl(\ket*{\widetilde{\psi}}-\ket*{\psi}\bigr)\right]\\
            &=\Re\left[\bra{e}\sum_{x\in\Sigma}\sqrt{\mu(x)}\ket*{x}
              (I-U_{Q_T}^x)\bigl(0\oplus\ket*{\kappa_x}/\sqrt{s}\bigr)\right]\\
            &=0-\Re\left[\bra{e}\sum_{x\in\Sigma}\sqrt{\mu(x)}\ket*{x}
              U_{Q_T}^x\bigl(0\oplus\ket*{\kappa_x}/\sqrt{s}\bigr)\right]\\
            &\geq-\|e\|\left\|\sum_{x\in\Sigma}\sqrt{\mu(x)}\ket*{x}
              U_{Q_T}^x\bigl(0\oplus\ket*{\kappa_x}/\sqrt{s}\bigr)\right\|\\
            &=-\|e\|\left(\sum_{x\in\Sigma}\mu(x)
              \left\|U_{Q_T}^x\bigl(0\oplus\ket*{\kappa_x}/\sqrt{s}\bigr)\right\|^2\right)^{1/2}\\
            &=-\|e\|\left(\sum_{x\in\Sigma}\mu(x)
              \left\|\kappa_x/\sqrt{s}\right\|^2\right)^{1/2}\\
            &=-\sqrt{(1-\alpha^2)\eps_T/2}.
        \end{align*}
        The inequality is Cauchy--Schwarz, and the last two equalities use unitarity,
        $\|e\|^2=1-\alpha^2$, and the definition of $\eps_T$.
    \end{itemize}
    Combining the two bounds, we get
    \begin{align*}
        \cF(\widetilde{\Psi}\circ\kbr{\mu},W\circ\kbr{\mu})
        &\geq \alpha(1-\eps_T)-\sqrt{(1-\alpha^2)\eps_T/2}\\
        &\geq \alpha(1-\eps_T)-\sqrt{(1-\alpha)\eps_T},
    \end{align*}
    where the last inequality uses $1+\alpha\leq2$. This completes the proof.
\end{proof}

\subsubsection{State conversion with exponential query-fidelity tradeoff}

\begin{proposition}[State conversion with exponential query-fidelity
    tradeoff]\label{prop:multiplicative-state-conversion}\label{cor:geometric-convergence-reuse} Let
    $\Sigma\subseteq\{0,1\}^n$, $\mu:\Sigma\to[0,1]$ be a distribution, and
    $\Phi,A,W\in\Pos(\Sigma)$ satisfy $\Phi\circ I=I$ and $\Tr(W\circ\kbr{\mu})=1$. Assume that, for
    some $r>1$, there exist $X_1,\ldots,X_n \succeq 0 $ such that
    \begin{align*}
        \Phi-A=\sum_{i=1}^{n}(rJ-S_i)\circ X_i.
    \end{align*}
    Then there exist an integer $t\geq0$ and a Gram matrix $\Psi\in\Pos(\Sigma)$  such that
    $\mathsf{Q}(\Phi\to\Psi)\leq t$ and
    \begin{align*}
        \cF\left(\Psi\circ\kbr{\mu},W\circ\kbr{\mu}\right)
        &\geq\frac{r^{t/4}}{4}
        \cF\left(A\circ\kbr{\mu},W\circ\kbr{\mu}\right).
    \end{align*}
\end{proposition}

Unlike in \Cref{thm:catalyst-state-conversion}, for this proposition, we cannot afford replacing the
catalyst $\ket{\kappa_x}$ with $\ket{0}$ as it could incur an additive error in terms of fidelity.
However, as shown in the following lemma, the catalyst relation allows us to recover the target as a
geometrically convergent sum.

\begin{lemma}[Geometric convergence]\label{lem:geometric-convergence} Let $\cV$ and $\cK$ be complex
    Hilbert spaces, and let $r>1$. Let $\ket{a},\ket{b}\in\cV$ be initial and target vectors, let
    $\ket{\kappa}\in\cK$ be a catalyst vector, and suppose that a unitary $V$ on $\cV\oplus\cK$
    satisfies
    \begin{align*}
        V(\ket{a}\oplus\ket{\kappa})=\ket{b}\oplus\sqrt{r}\,\ket{\kappa}.
    \end{align*}
    Then, starting without the catalyst, we can recover the target $\ket{b}$ in the limit as $
    \ket{b}=\sum_{t=0}^{\infty}r^{-t/2}\ket{b_t}$ where
    $\ket{b_0}\oplus\ket{h_0}:=V(\ket{a}\oplus0)$ and, recursively, for every integer $t\geq0$, $
    \ket{b_{t+1}}\oplus\ket{h_{t+1}}:=V(0\oplus\ket{h_t})$. For $r>1$, the series converges
    geometrically in norm; i.e., for every integer $T\geq0$,
    \begin{align*}
        & \left\|\ket{b}-\sum_{t=0}^{T}\frac{\ket{b_t}}{\sqrt{r^t}}\right\|
        \leq \frac{\|\ket{\kappa}\|}{\sqrt{r^T}}.
    \end{align*}
\end{lemma}
We defer the proof of \Cref{lem:geometric-convergence} to \Cref{app:geometric-convergence}, and
continue with the proof of \Cref{prop:multiplicative-state-conversion}.
\begin{proof}[Proof of \Cref{prop:multiplicative-state-conversion}]
    By \Cref{lem:catalyst-synthesis}, there exists a catalyst space $\cK = \cI \otimes \cZ$ and
    catalyst vectors $\{\ket{\kappa_x}\}_{x\in \Sigma} \subset \cK$ and an input-independent unitary
    $U$ such that $U(\ket{\phi_x}\oplus \cO^{x}\ket{\kappa_x}) = \ket{a_x}\oplus
    \sqrt{r}\ket{\kappa_x}$ where $\braket{\kappa_x}{\kappa_y} = \sum_{i=1}^{n} X_i(x,y)$. Let $\ell
    \geq 1$ be an integer to be determined later, and let $s= \sum_{i=0}^{\ell-1} r^{i}$. We apply
    the catalyst-reuse \Cref{lem:catalyst-reuse} to obtain an $\ell$-query unitary such that given
    oracle access to $\cO^{x}$ it computes
    \begin{align*}
        V^{x}(\ket{\phi_x}\oplus \ket{\kappa_x}/\sqrt{s}) = \ket{a_x} \oplus r^{\ell/2}\ket{\kappa_x}/\sqrt{s}.
    \end{align*}
    Applying \Cref{lem:geometric-convergence} to $V^x$ with parameter $r^\ell$ and catalyst
    $\ket{\kappa'_x}:=\ket{\kappa_x}/\sqrt{s}$, we obtain vectors $\{\ket*{b_x^{(j)}}\}_{j\geq0}$
    such that, for every integer $T\geq0$,
    \begin{align*}
        \ket{a_x}
        &=\sum_{j=0}^{\infty}\frac{\ket*{b_x^{(j)}}}{\sqrt{r^{\ell j}}},
        &&
        \left\|\ket{a_x}-\sum_{j=0}^{T}
            \frac{\ket*{b_x^{(j)}}}{\sqrt{r^{\ell j}}}\right\|
        \leq\frac{\|\ket{\kappa'_x}\|}{\sqrt{r^{\ell T}}}
        =\frac{\|\ket{\kappa_x}\|}{\sqrt{s\,r^{\ell T}}}.
    \end{align*}

    By Uhlmann's theorem, we can choose a PSD decomposition $W(x,y)=\braket{w_x}{w_y}$ such that
    \begin{align*}
        \cF(A\circ\kbr{\mu},W\circ\kbr{\mu})
        &=\left|\sum_{x\in\Sigma}\mu(x)\braket{w_x}{a_x}\right|
        =\left|\E_{x\sim\mu}\braket{w_x}{a_x}\right| \\
        &\leq \sum_{t=0}^{\infty}r^{-\ell t/2}\left|\E_{x\sim\mu}\braket{w_x}{b_x^{(t)}}\right| \\
        &=\sum_{t=0}^{\infty}r^{-\ell t/4}\left(r^{-\ell t/4}\left|\E_{x\sim\mu}\braket{w_x}{b_x^{(t)}}\right|\right).
    \end{align*}
    Taking $\ell\geq1$ to be the smallest integer such that $r^\ell\geq16$ and $r^{\ell-1}<16$, we
    get $\sum_{t=0}^{\infty}r^{-\ell t/4}=(1-r^{-\ell/4})^{-1}\leq2$. Therefore, there exists an
    integer $t\geq0$ such that
    \begin{align*}
        \left|\E_{x\sim\mu}\braket{w_x}{b_x^{(t)}}\right|
        &\geq \frac{r^{\ell t/4}}{2}\cF(A\circ\kbr{\mu},W\circ\kbr{\mu}).
    \end{align*}
    By the construction in \Cref{lem:geometric-convergence}, applying $V^x$ a total of $t+1$ times
produces a normalized state $\ket*{b_x^{(t)}}\oplus\ket{g_x}$, where $\ket{g_x}$ collects the
remaining ``garbage'' components. In the catalyst-reuse circuit, the first oracle call acts only on
the zero catalyst component and therefore can be omitted. Thus, the output Gram matrix $\Psi$
satisfies $Q(\Phi\to\Psi)\leq T$, where $T:=\ell(t+1)-1 =\ell  t + \ell -1$. Since the vectors
$\ket{w_x}\oplus0$ have Gram matrix $W$, Uhlmann's theorem gives
\begin{align*}
    \cF(\Psi\circ\kbr{\mu},W\circ\kbr{\mu})
    &\geq \left|\E_{x\sim\mu}\braket{w_x\oplus 0}{b_x^{(t)}\oplus g_x}\right|\\
    &\geq \frac{r^{\ell t/4}}{2}
    \cF(A\circ\kbr{\mu},W\circ\kbr{\mu})\\
    &\geq \frac{r^{T/4}}{4}
    \cF(A\circ\kbr{\mu},W\circ\kbr{\mu}),
\end{align*}
where the last inequality follows from $r^{\ell-1}<16$.
\end{proof}

\section{\texorpdfstring{A distributional analogue of the $\gamma_2$ norm}{A distributional analogue
of the gamma 2 norm}}\label{sec:dist-gamma2} As mentioned in \Cref{sec:prelem:gamma2}, the filtered
$\gamma_2$ norm measures a worst-case factorization cost of a matrix, and it can be used to tightly
bound the adversary method as well as constant-error quantum query complexity. In
\Cref{sec:dist-gamma2-properties}, we develop a distributional variant of the $\gamma_2$ norm along
with its basic properties. Then, in \Cref{sec:dist-gamma2-query}, we specialize it to bound
distributional quantum query complexity.

\subsection{Definition and general properties}\label{sec:dist-gamma2-properties} We begin by
defining a distributional analogue of the filtered $\gamma_2$ norm. In essence, for a matrix $B$, we
replace the spectral norm numerator $\|B\circ\Gamma\|$ in the definition of $\gamma_2$ with the
$\mu$-weighted inner product $|\langle\Gamma,B\circ\kbr{\mu}\rangle|$.
\begin{definition}
Let $n\in \bN$ and let $\Sigma$ be a finite set. For a distribution $\mu:\Sigma \to [0,1]$, the
\emph{filtered distributional $\gamma_2$ norm} of a matrix $B\in \bC^{\Sigma\times \Sigma}$ with
non-zero filters $ \{Z_i\}_{i\in [n]} \in \bC^{\Sigma\times \Sigma}$ is defined as
    \begin{align*}
        \bar{\gamma}^{\mu}_2(B|\{Z_i\}) :=
        \max_{\substack{\Gamma\in\cL(\Sigma):\\
        \max_{i\in[n]}\|\Gamma\circ Z_i\|\leq1}}
        |\langle\Gamma,B\circ\kbr{\mu}\rangle|.
    \end{align*}
    The value is understood to be $+\infty$ if the optimization is unbounded. Otherwise,  the
    maximum is attained.\footnote{To see this, assume $\bar{\gamma}_2^{\mu}(B|\{Z_i\})<\infty$.
    Then, having $Z_i(x,y)=0$ for all $i$ forces $B(x,y)\sqrt{\mu(x)\mu(y)}=0$, since otherwise
    $\Gamma_t:=t\ketbra{x}{y}$ would make the optimization unbounded as $t\to\infty$. We may
    therefore set $\Gamma(x,y)=0$ at these entries without changing the objective or any filtered
    matrix and write
    $\bar{\gamma}_2^{\mu}(B|\{Z_i\})=\max_{\Gamma\in\cL}|\langle\Gamma,B\circ\kbr{\mu}\rangle|$,
    where $\cL=\{\Gamma\in\cL(\Sigma):\|\Gamma\circ Z_i\|\leq1~\forall i\in[n],\
    \Gamma(x,y)=0~\text{whenever}~Z_i(x,y)=0~\text{for all}~i\in[n]\}$. Then, $\cL$ is bounded and
    closed, hence compact, and the objective is continuous. So the maximum is attained.
    }
\end{definition}

In the following, we focus on Hermitian matrices $B$ and real symmetric filters $\{Z_i\}$. We
provide equivalent primal and dual formulations in \Cref{prop:equiv-avg-gamma2} and collect a few
useful properties in \Cref{prop:avggamma2-properties}.

\subsubsection{Equivalent primal and dual formulations}
Any feasible solution to the maximization form in the definition gives a natural certificate for
lower-bounding $\bar{\gamma}_2^\mu$. It is also useful to obtain a minimization form that provides
certificates for upper bounds and tightness results. In the following, we use SDP strong duality to
derive equivalent primal and dual formulations of this measure.

\begin{proposition}\label{prop:equiv-avg-gamma2} Let $n\in \bN$ and $\Sigma$ be a finite set. Let $B
\in \Herm(\Sigma)$ and let $Z_1,\dots,Z_n\in\bR^{\Sigma\times\Sigma}$ be symmetric non-zero
matrices. Then, for any distribution $\mu:\Sigma \to [0,1]$, the following holds.
\begin{align*}
&\bar{\gamma}_2^{\mu}(B\mid\{Z_i\})\\[0.3em]
= &\progat{6.5cm}
{\max_{\Gamma=\Gamma^\dagger}}
{\langle \Gamma,B\circ\kbr{\mu}\rangle}
{\makebox[5.2cm][l]{$\displaystyle -I\preceq \Gamma\circ Z_i\preceq I$}
& \forall i\in[n]}
\\[0.6em]
= &\progat{6.5cm}
{\min_{H_i=H_i{}^\dagger}}
{\sum_{i\in[n]}\|H_i\circ\kbr{\mu}\|_{\tr}}
{\makebox[5.2cm][l]{$\displaystyle B(x,y)=\sum_{i\in[n]}Z_i(x,y)H_i(x,y)$}
& \forall x,y\in\supp(\mu)}
\\[0.6em]
= &\progat{6.5cm}
{\min_{d}\min}
{\frac12\sum_{i\in[n]}\E_{x\sim\mu}
\big[\|u_x^i\|^2+\|v_x^i\|^2\big]}
{
\makebox[5.2cm][l]{$\displaystyle B(x,y)=\sum_{i\in[n]} Z_i(x,y)\braket{u_x^i}{v_y^i}$}
& \forall x,y\in\supp(\mu) \\
\makebox[5.2cm][l]{$\displaystyle \ket{u_x^i},\ket{v_x^i}\in\bC^{d}$}
& \forall i\in[n],\ x\in\supp(\mu)
}
\end{align*}
Moreover, if $B$ and $\{Z_i\}_{i=1}^{n}$ are real-valued, all matrix and vector variables in the
above characterizations may be taken real-valued.
\end{proposition}
\begin{proof}
    Let $\gamma := \bar{\gamma}_{2}^{\mu}(B|\{Z_i\})$. If $\gamma=+\infty$, then some
    $x,y\in\supp(\mu)$ satisfy $B(x,y)\neq0$ and $Z_i(x,y)=0$ for every $i$. In this case, the
    maximization is unbounded and the minimization constraints cannot hold at $(x,y)$, so all values
    are understood to be $+\infty$. For the rest of the proof, we may therefore assume that
    $\gamma<+\infty$.
    \par
    Let $A:=B\circ\kbr{\mu}$, which is Hermitian. For an arbitrary $\Gamma$, choose $c\in\bC$ with
    $|c|=1$ such that $\langle c\Gamma,A\rangle=|\langle\Gamma,A\rangle|$, and define
    $\Gamma':=(c\Gamma+(c\Gamma)^\dagger)/2$. Since $A$ is Hermitian,
    $\langle\Gamma',A\rangle=\operatorname{Re}\langle c\Gamma,A\rangle =|\langle\Gamma,A\rangle|$.
    Moreover, since each $Z_i$ is Hermitian,
    \begin{align*}
        \|\Gamma'\circ Z_i\|
        &=\left\|\frac{c\Gamma\circ Z_i+(c\Gamma\circ Z_i)^\dagger}{2}\right\|
        \leq\|c\Gamma\circ Z_i\|
        =\|\Gamma\circ Z_i\|.
    \end{align*}
    Thus, we may restrict the maximum to Hermitian $\Gamma$ and omit the absolute value, giving the
    first equivalent formulation
    \begin{align*}
        \gamma &= \progat{6cm}
        {\max_{\Gamma=\Gamma^\dagger}}
        {\langle\Gamma,B\circ\kbr{\mu}\rangle}
        {-I\preceq\Gamma\circ Z_i\preceq I & \quad\forall i\in[n].}
    \end{align*}
    \par
    This is an SDP, and taking $\Gamma=0$ yields an interior feasible solution $-I\prec \Gamma \circ
    Z_i \prec I$. Therefore, Slater's condition is satisfied, so strong duality holds and the dual
    minimum is attained (\Cref{thm:slater}). To obtain the minimization form, we introduce
    Lagrangian variables $P_1,\dots,P_n\succeq 0$ and $Q_1,\dots,Q_n\succeq 0$. Thus, we can write
    \begin{align*}
        \gamma
        &=
        \sup_{\Gamma=\Gamma^\dagger}
        \inf_{P_i,Q_i\succeq 0}
        \left[
        \langle \Gamma,B\circ\kbr{\mu}\rangle
        +
        \sum_i
        \Big(
        \langle P_i,I-\Gamma\circ Z_i\rangle
        +
        \langle Q_i,I+\Gamma\circ Z_i\rangle
        \Big)
        \right]
        \\
        &=
        \sup_{\Gamma=\Gamma^\dagger}
        \inf_{P_i,Q_i\succeq 0}
        \left[
        \langle \Gamma,B\circ\kbr{\mu}\rangle
        +
        \sum_i\langle P_i+Q_i,I\rangle
        -
        \sum_i\langle\Gamma,(P_i-Q_i)\circ Z_i\rangle
        \right]
        \\
        &=
        \inf_{P_i,Q_i\succeq 0}
        \sup_{\Gamma=\Gamma^\dagger}
        \left[
        \sum_i \Tr[P_i+Q_i]
        +
        \left\langle
        \Gamma,
        B\circ\kbr{\mu}
        -
        \sum_i(P_i-Q_i)\circ Z_i
        \right\rangle
        \right]
        \\
        &=
        \progat{6cm}
        {\min_{P_i,Q_i\succeq0}}
        {\sum_{i\in[n]}\Tr[P_i+Q_i]}
        {B\circ\kbr{\mu}=\sum_{i\in[n]}(P_i-Q_i)\circ Z_i}.
    \end{align*}
    \par
    Next, we perform a variable change to remove the $\mu$-weights. Let $S:=\supp(\mu)$. Since
    $B\circ\kbr{\mu}$ is supported on $S\times S$, restricting $P_i,Q_i$ to $S$ and extending them
    by zero preserves the constraint and cannot increase the trace objective. We may therefore
    assume that $P_i,Q_i$ are supported on $S\times S$ and are zero elsewhere. We can then apply a
    PSD-preserving bijection $P_i\leftrightarrow P'_i$ and $Q_i\leftrightarrow Q'_i$ within this
    class, where $P'_i\circ\kbr{\mu}=P_i$ and $Q'_i\circ\kbr{\mu}=Q_i$, giving
    \begin{align*}
        \gamma &= \progat{6cm}
        {\min_{P'_i,Q'_i\succeq0}}
        {\sum_{i\in[n]}\Tr[(P'_i+Q'_i)\circ\kbr{\mu}]}
        {B(x,y)=\sum_{i\in[n]}Z_i(x,y)(P'_i-Q'_i)(x,y)
        & \quad\forall x,y\in S.}
    \end{align*}
    \par
    Next, we combine $P'_i$ and $Q'_i$ into $H_i:=P'_i-Q'_i$. We apply the trace-norm
    characterization $\|H_i\circ\kbr{\mu}\|_{\tr}
    =\min_{X_i,Y_i\succeq0}\{\Tr[X_i+Y_i]:H_i\circ\kbr{\mu}=X_i-Y_i\}$. A minimizing decomposition
    can be restricted to $S\times S$ and rescaled to PSD matrices $X'_i,Y'_i$ satisfying
    $H_i=X'_i-Y'_i$. Therefore, the optimal value is preserved, and we obtain the second equivalent
    formulation as
    \begin{align*}
        \gamma &= \progat{6cm}
        {\min_{H_i=H_i{}^\dagger}}
        {\sum_{i\in[n]}\|H_i\circ\kbr{\mu}\|_{\tr}}
        {B(x,y)=\sum_{i\in[n]}Z_i(x,y)H_i(x,y)
        & \quad\forall x,y\in S.}
    \end{align*}
    We now show that removing the Hermiticity condition $H_i=H_i^{\dagger}$ does not decrease the
    optimal value. So, let $\{K_i\}$ be an arbitrary family satisfying
    $\bra{x}B\ket{y}=\sum_iZ_i(x,y)\bra{x}K_i\ket{y}$ for every $x,y\in S$, and let
    $H_i=(K_i+K_i^{\dagger})/2$. By Hermiticity of $B$ and symmetry of the filters, the matrices
    $H_i$ satisfy the same constraints. Moreover, by the triangle inequality for the trace norm,
    $
        \|H_i\circ\kbr{\mu}\|_{\tr}
        =\frac{1}{2}\|K_i\circ\kbr{\mu}+K_i^{\dagger}\circ\kbr{\mu}\|_{\tr}
        \leq\|K_i\circ\kbr{\mu}\|_{\tr}$. Therefore, we conclude
        \begin{align*}
             \gamma
        &=
        \progat{6cm}
        {\min_{H_i}}
        {\sum_i\|H_i\circ\kbr{\mu}\|_{\tr}}
        {
        \bra{x}B\ket{y}=\sum_iZ_i(x,y)\bra{x}H_i\ket{y}
        &\quad\forall x,y\in S
        }.
        \end{align*}
        To get the final vector form, we apply the trace-norm factorization
    \begin{align*}
        \|H_i\circ\kbr{\mu}\|_{\tr}
        =\progbox
        {\min_d\min_{\ket{a_x^i},\ket{b_x^i}\in\bC^d}}
        {\frac12\sum_{x\in S}\big[\|a_x^i\|^2+\|b_x^i\|^2\big]}
        {\sqrt{\mu(x)\mu(y)}H_i(x,y)=\braket{a_x^i}{b_y^i}
        & x,y\in S}.
    \end{align*}
    Taking $\ket{u_x^i}:=\ket{a_x^i}/\sqrt{\mu(x)}$ and $\ket{v_x^i}:=\ket{b_x^i}/\sqrt{\mu(x)}$,
    summing over $i$, and substituting the constraints on $H_i$ gives the final formulation.
    \par
    Finally, for the moreover part, assume $B$ and the filters are real-valued. Then, in the first
    formulation replacing $\Gamma$ by $\Gamma':=(\Gamma+\overline{\Gamma})/2$ preserves feasibility
    and the objective. Similarly, in the second formulation replacing each $H_i$ by
    $H'_i:=(H_i+\overline{H_i})/2$ preserves the constraints and does not increase the trace-norm
    objective. For the third formulation, a real singular-value decomposition of each
    $H'_i\circ\kbr{\mu}$ then gives real factorization vectors. Thus, all matrix and vector
    variables may be taken real-valued as required.
\end{proof}

\subsubsection{Basic properties}
Throughout the rest of the paper, we will make use of a few basic properties of the $\bar{\gamma}_2$
norm.

\paragraph{Analogy with the worst-case $\gamma_2$ norm.}
First, we claim that, maximizing the average $\bar{\gamma}_{2}^{\mu}$ over the worst distribution
$\mu$, we recover the ordinary filtered $\gamma_2$ norm. So, our measure is in fact the
distributional analogue of the $\gamma_2$ norm.

\begin{proposition}\label{prop:average-to-worst-gamma2} Let $n\in \bN$ and let $\Sigma$ be a finite
 set. Let $A\in \Herm(\Sigma)$ and let $Z_i\in \bR^{\Sigma\times\Sigma}$ be symmetric for each $i\in
 [n]$. Then,
$$\max_{\mu} \bar{\gamma}_2^{\mu}(A|\{Z_i\}) = \gamma_2(A|\{Z_i\})$$ where the maximum is over all
distributions $\mu:\Sigma \to [0,1]$.
\end{proposition}
To give a rough sketch of the proof, consider taking the maximum over all unit vectors. That is,
$\max_{\|u\|=1} |\langle \Gamma, A\circ \kbr{u}\rangle| = \max_{\|u\|=1}
|\bra{u}(\Gamma\circ\overline{A})\ket{u}| = \|\Gamma\circ \overline{A}\|$. To restrict to vectors
$\ket{\mu}$, we can show that the phases of $\ket{u}$ can be absorbed into $\Gamma$ while preserving
the constraints $\|\Gamma\circ Z_i\|\leq 1$.  We refer the reader  to
\Cref{app:average-to-worst-gamma2} for the full details of the proof.

\paragraph{A list of properties.}
We also present a few properties of $\bar{\gamma}_2^{\mu}$ that naturally parallel those of the
filtered $\gamma_2$ norm (see, e.g., \cite[Lemma~A.2]{LMR+11}).

\begin{proposition}\label{prop:avggamma2-properties} Let $n\in \bN$ and let $\Sigma$ be a finite
    set. Let $A,B\in \Herm(\Sigma)$, where $B\neq 0$, and let $Z_i\in\bR^{\Sigma\times\Sigma}$ be
    symmetric for each $i\in[n]$. For any distribution $\mu:\Sigma\to[0,1]$, the following hold.
    \begin{enumerate}[i.]
        \item \textbf{Triangle inequality.} $\bar{\gamma}_2^{\mu}(A+B|\{ Z_i\}) \leq
        \bar{\gamma}_2^{\mu}(A|\{ Z_i\}) + \bar{\gamma}_2^{\mu}(B|\{ Z_i\}).$
        \item \textbf{Scalability.} For $\lambda > 0$, $\bar{\gamma}_2^{\mu}(\lambda A|\{ Z_i\}) =
        |\lambda| \bar{\gamma}_2^{\mu}(A|\{ Z_i\})$
        \item \textbf{Hadamard product bounds.}\label{item:avggamma2-hadamard-bounds} We have
        $\bar{\gamma}_2^{\mu}(A\circ B|\{Z_i\}) \leq \bar{\gamma}_2^{\mu}(A|\{Z_i\}) \cdot
        \gamma_2(B)$ and $\bar{\gamma}_2^{\mu}(A|\{Z_i\}) \leq \bar{\gamma}_2^{\mu}(A|\{Z_i\circ
        B\}) \cdot \gamma_2(B)$. If $B$ is real, then $\bar{\gamma}_2^{\mu}(A\circ B|\{Z_i\circ
        B\})\leq \bar{\gamma}_2^{\mu}(A|\{Z_i\})$.
        \item \textbf{Trace norm lower bound.}\label{item:avggamma2-tr-lower}
        $\bar{\gamma}_2^{\mu}(A|\{Z_i\}) \cdot {\max_i\gamma_2(Z_i)} \geq
        {\|A\circ\kbr{\mu}\|_{\tr}}$.
        \item \textbf{Trace norm upper bound.}\label{item:avggamma2-tr-upper}
        $\bar{\gamma}_2^{\mu}(A\circ B|\{Z_i\}) \leq \|A\circ\kbr{\mu}\|_{\tr} \cdot
        \gamma_2(B|\{Z_i\})$.
        \item \textbf{Support restriction.}\label{item:avggamma2-support} Let $\supp(\mu)\subseteq
        S\subseteq\Sigma$. Then,
        \begin{align*}
            \bar{\gamma}_2^\mu(A\mid\{Z_i\})
            &=\bar{\gamma}_2^\mu(A\circ J_S\mid\{Z_i\})
            =\bar{\gamma}_2^{\mu|_S}(A|_S\mid\{Z_i|_S\}),
        \end{align*}
        where $X|_S$ denotes the principal submatrix of $X$ indexed by $S$.
    \end{enumerate}
\end{proposition}
All proofs are straightforward, and, for the sake of brevity, are postponed to
\Cref{app:avggamma2-properties}.

\subsection{Bounding query complexity}\label{sec:dist-gamma2-query} To talk about query complexity,
we specialize the preceding formulations to the query filters $\{\Delta_i\}$.
\begin{lemma}[Equivalent forms for query filters]\label{lem:equiv-avg-gamma2-query} Let
$\Sigma\subseteq\{0,1\}^n$ and let $B\in\Herm(\Sigma)$ satisfy $B\circ I=0$. Let
$\cL:=\{\Gamma\in\Herm(\Sigma):\Gamma\circ I=0,\ \|\Gamma\circ\Delta_i\|\leq1 \text{ for all
}i\in[n]\}$. Then, $\cL$ is compact and convex, and for every distribution $\mu:\Sigma\to[0,1]$, we
have
\begin{align*}
    \bar{\gamma}_2^\mu(B\mid\{\Delta_i\})
    =\max_{\Gamma\in\cL}\ \langle\Gamma,B\circ\kbr{\mu}\rangle.
\end{align*}
Moreover, the map $(B,\mu)\mapsto\bar{\gamma}_2^\mu(B|\{\Delta_i\})$ is jointly continuous on the
Hermitian matrices $B$ with $B\circ I=0$ and probability distributions $\mu$ on $\Sigma$.
Consequently, we can equivalently write
\begin{align*}
    \bar{\gamma}_2^\mu(B\mid\{\Delta_i\})
    =\progline{\max_{\Gamma\neq0}}
    {\frac{\langle\Gamma,B\circ\kbr{\mu}\rangle}{\max_i\|\Gamma\circ\Delta_i\|}}
    {\Gamma=\Gamma^\dagger;\quad\Gamma\circ I=0.}
\end{align*}
In particular, the value is finite and both maxima are attained.\footnote{ Implicitly, we are
assuming $|\Sigma|\geq2$ throughout the paper. }
\end{lemma}
\begin{proof}
Since $B\circ I=0$ and $\Delta_i\circ I=0$ for every $i\in[n]$, replacing $\Gamma$ by
$\Gamma':=\Gamma\circ(J-I)$ preserves the objective and every constraint in
\Cref{prop:equiv-avg-gamma2} as $\Gamma'\circ\Delta_i=\Gamma\circ\Delta_i$. Therefore, by
\Cref{lem:adversary-support-replacement}, we may assume $\Gamma\circ I=0$. Convexity and closedness
of the set $\cL$ are easy to verify. The set $\cL$ is also bounded. To see this, let $\Gamma\in\cL$
and $x,y\in\Sigma$. If $x=y$, then $\Gamma(x,y)=0$. Otherwise, there exists $i\in[n]$ such that
$\Delta_i(x,y)=1$, so $|\Gamma(x,y)|\leq\|\Gamma\circ\Delta_i\|\leq1$. Thus, $\cL$ is compact. Since
the objective is continuous on $\cL$, its maximum is finite and attained.
\par
Moreover, the objective is jointly continuous in $\Gamma$, $B$, and $\mu$. Since $\cL$ is compact
and independent of $B$ and $\mu$, its maximum over $\cL$, $\bar{\gamma}_2^\mu(B\mid\{\Delta_i\})$,
is also jointly continuous in $B$ and $\mu$.
\par
Finally, for every nonzero zero-diagonal matrix $\Gamma$, we have
$M:=\max_{i\in[n]}\|\Gamma\circ\Delta_i\|>0$. Indeed, if $\Gamma(x,y)\neq0$, then $x\neq y$, so
there exists an index $i\in[n]$ such that $\Delta_i(x,y)=1$. Consequently, rescaling to
$\Gamma':=\Gamma/M$ gives $\max_i\|\Gamma'\circ\Delta_i\|=1$, yielding the quotient formulation by
homogeneity.
\end{proof}

Next, we use $\bar{\gamma}_2^{\mu}$ to obtain lower bounds for generating a Gram matrix in
\Cref{prop:hybrid-gamma2} and upper bounds for state conversion in
\Cref{prop:distalg-from-avggamma2}.

\subsubsection{Lower bounding query complexity}
To prove the lower bound, we first present a standard one-query decomposition for quantum algorithms
\cite[Lemma~4.7]{LMR+11}.

\begin{lemma}[One-query decomposition]\label{lem:one-query-decomposition} Let $\Sigma\subset
\{0,1\}^{n}$ and let $Q$ be a $T$-query quantum algorithm. For $t\in[0:T]$, let $\Psi^{t}$ be the
Gram matrix of the algorithm after $t$ queries; i.e., $\Psi^{t}(x,y)=\braket{\psi_x^t}{\psi_y^t}$,
where $\ket{\psi_x^0}=U_0\ket{0}$ and $\ket{\psi_x^t}=U_t\cO^x\cdots
U_1\cO^xU_0\ket{0}=U_t\cO^x\ket{\psi_x^{t-1}}$ is the state of the algorithm on input $x$ after the
$t$-th query.
 Then, for $t\in[T]$, there exists a PSD decomposition $\Psi^{t-1}=\sum_{i=0}^{n}Z_i^t$, with
$Z_i^t\succeq0$, such that $\Psi^{t-1}-\Psi^t=2\sum_{i=1}^{n}\Delta_i\circ Z_i^t$.
\end{lemma}
\begin{proof}
    Throughout the proof, $\cO^x$ is understood to act trivially on the auxiliary register. For
    simplicity, we define $x_0:=0$. Then, we can simply write $\cO^x=\sum_{i=0}^{n}(-1)^{x_i}P_i$,
    where $P_i$ is the projector onto the subspace where the index register equals $i$. Therefore,
    $\sum_{i=0}^{n}P_i=I$ and $P_i\cO^x=(-1)^{x_i}P_i$.
    \par
    Let $Z^t_i(x,y):=\bra{\psi_x^{t-1}}P_i\ket{\psi_y^{t-1}}$ for $x,y\in\Sigma$. Then,
    $Z_i^t\succeq0$, and
    \begin{align*}
        \sum_{i=0}^{n}Z_i^t(x,y)
        &=\bra{\psi_x^{t-1}}\Big(\sum_{i=0}^{n}P_i\Big)\ket{\psi_y^{t-1}}
        =\Psi^{t-1}(x,y).
    \end{align*}
    \par
    Next, for $x,y\in\Sigma$, we compute the entries of $\Psi^t$ as follows.
    \begin{align*}
        \Psi^t(x,y)
        &=\braket{\psi_x^t}{\psi_y^t}
        =\bra{\psi_x^{t-1}}(\cO^x)^{\dagger}U_t^{\dagger}U_t\cO^y\ket{\psi_y^{t-1}}
        =\sum_{i=0}^{n}\bra{\psi_x^{t-1}}(\cO^x)^{\dagger}P_i^2\cO^y\ket{\psi_y^{t-1}}\\
        &=\sum_{i=0}^{n}(-1)^{x_i\oplus y_i}Z_i^t(x,y).
    \end{align*}
    Therefore,
    \begin{align*}
        (\Psi^{t-1} - \Psi^{t})(x,y) &=  \sum_{i\in [0:n]} (1-(-1)^{x_i \oplus y_i}) Z_i^{t}(x,y) =
        0 + \sum_{i\in [n]} 2\Delta_i(x,y) Z_i^{t}(x,y)
    \end{align*}
    where we used the fact that for $i=0$, $x_0 = y_0 =0$ and for $i\in [n]$, $\Delta_i(x,y) =
    \mathds{1}[x_i\neq y_i]$. Therefore, $\Psi^{t-1}-\Psi^t=2\sum_{i=1}^{n}\Delta_i\circ Z_i^t$,
    completing the proof.
\end{proof}
We can now apply a hybrid argument to obtain our lower bound.
\begin{proposition}\label{prop:hybrid-gamma2} Let $\Sigma \subset \{0,1\}^{n}$ and let $\mu:\Sigma
\to [0,1]$ be a distribution. Also, let $\Psi \in \Pos(\Sigma)$ be a Gram matrix with $\Psi \circ I
= I$. Then, $\mathsf{Q}(\Psi) \geq \bar{\gamma}^{\mu}_2(J-\Psi|\{\Delta_i\})/2$.
\end{proposition}
\begin{proof}
Let $Q$ be a $T$-query algorithm computing $\Psi$, where $T=\mathsf{Q}(\Psi)$, and let $\Psi^t$ be
its Gram matrix after $t$ queries, as in \Cref{lem:one-query-decomposition}. Then, $\Psi^0=J$ and
$\Psi^T=\Psi$. Let $\Gamma \in \bC^{\Sigma \times \Sigma}$ be an arbitrary matrix. We define
$d_{t}:= \langle \Gamma, \Psi^{t} \circ \kbr{\mu}\rangle$ for $t\in [0:T]$, and keep track of
sequential differences as follows.
\begin{align*}
    | d_{t-1} - d_{t}| &= |\langle \Gamma, (\Psi^{t-1}-\Psi^{t}) \circ \kbr{\mu}\rangle| \\
     &= 2|\langle \Gamma, \sum_{i\in [n]} \Delta_i \circ Z_i^{t} \circ \kbr{\mu}\rangle| && \text{(By \Cref{lem:one-query-decomposition})}\\
     & \leq 2\sum_{i\in [n]} |\langle \Gamma, \Delta_i \circ Z_i^{t} \circ \kbr{\mu}\rangle|  \\
     & \leq 2\sum_{i\in [n]} \bigl(\|\Gamma\circ\Delta_i\|\,\|Z_i^{t}\circ\kbr{\mu}\|_{\tr}\bigr) \\
     & \leq 2\left(\max_{i\in[n]}\|\Gamma\circ\Delta_i\|\right)
     \Tr[\textstyle\sum_{i\in[0:n]}Z_i^{t}\circ\kbr{\mu}] && \text{(Since $Z_i^{t}\succeq 0$)} \\
     &= 2\left(\max_{i\in[n]}\|\Gamma\circ\Delta_i\|\right)
     \Tr[\Psi^{t-1}\circ\kbr{\mu}] \\
     &= 2\max_{i\in[n]}\|\Gamma\circ\Delta_i\|.
\end{align*}
Therefore, by the triangle inequality,
\begin{align*}
    |\langle \Gamma,(J-\Psi)\circ \kbr{\mu}\rangle|
    = |d_0 - d_T| \leq \sum_{t=1}^{T} |d_{t-1} - d_{t}| \leq T\cdot 2\max_{i}\|\Gamma \circ \Delta_i\|.
\end{align*}
Since $\Gamma$ was arbitrary, by \Cref{lem:equiv-avg-gamma2-query} we have
\begin{align*}
T &\geq \max_{\substack{\Gamma\neq0,\ \Gamma=\Gamma^{\dagger}\\\Gamma\circ I=0}}
\frac{\langle\Gamma,(J-\Psi)\circ\kbr{\mu}\rangle}{2\max_i\|\Gamma\circ\Delta_i\|}
=\frac{1}{2}\bar{\gamma}_2^{\mu}(J-\Psi|\{\Delta_i\}). \qedhere
\end{align*}
\end{proof}
\subsubsection{Upper bounding query complexity}
For the upper bound, we use the witnesses $\ket{u_x^i}$ and $\ket{v_x^i}$ from the minimization
formulation in \Cref{prop:equiv-avg-gamma2} to apply the distributional state-conversion algorithm
of \Cref{thm:catalyst-state-conversion}.

\begin{proposition}\label{prop:distalg-from-avggamma2} Let $\Sigma\subset \{0,1\}^{n}$ and let
    $\mu:\Sigma \to [0,1]$ be a distribution. Also, let $\Phi,\Psi\in \Pos(\Sigma)$ be Gram matrices
    where $\Phi\circ I=\Psi\circ I=I$. Then, for any fidelity parameter $\alpha=\alpha(n)\in(0,1)$,
    there exists a Gram matrix $\widetilde{\Psi}\in\Pos(\Sigma)$ with $\widetilde{\Psi}\circ I=I$
    such that $\cF(\widetilde{\Psi}\circ\kbr{\mu},\Psi\circ\kbr{\mu})\geq \alpha$ and
    \begin{align*}
        \mathsf{Q}(\Phi\to\widetilde{\Psi})
        \leq
            \frac{2}{1-\alpha}\cdot \bar{\gamma}_2^{\mu}(\Phi-\Psi|\{\Delta_i\}).
    \end{align*}
\end{proposition}
\begin{proof}
    Without loss of generality, we assume that $\Sigma=\supp(\mu)$. Indeed, replacing each matrix by
its principal submatrix indexed by $S:=\supp(\mu)$ leaves $\bar{\gamma}_2^\mu$ unchanged (see
\Crefitem{item:avggamma2-support}). Also, any state-conversion algorithm on $S$ can be run on all
inputs in $\Sigma$ with the same query cost, and its output Gram matrix agrees with the one
constructed on $S$, so the fidelity under Hadamard product with $\kbr{\mu}$ is unchanged.
\par

    Let $\gamma:=\bar{\gamma}_2^\mu(\Phi-\Psi\mid\{\Delta_i\})$. If $\gamma\leq1-\alpha$, then
    taking $\widetilde{\Psi}=\Phi$ requires zero queries and satisfies the fidelity requirement,
    since
    \begin{align*}
        \cF(\Phi\circ\kbr{\mu},\Psi\circ\kbr{\mu})
        &\geq 1-\frac12\|(\Phi-\Psi)\circ\kbr{\mu}\|_{\tr}
        \geq 1-\left(\max_i\gamma_2(\Delta_i)\right)
        \bar{\gamma}_2^\mu(\Phi-\Psi\mid\{\Delta_i\})
        \geq 1-\gamma
        \geq \alpha,
    \end{align*}
    where we used the Fuchs--van de Graaf inequalities, \Crefitem{item:avggamma2-tr-lower}, and
    $\gamma_2(\Delta_i)\leq1$.
  \par
  Therefore, assume that $\gamma>1-\alpha$. Choose witnesses $\ket{u_x^i}$ and $\ket{v_x^i}$
    attaining the minimum in \Cref{prop:equiv-avg-gamma2}, so that, for all $x,y\in \Sigma$,
    \begin{align*}
        \Phi(x,y)-\Psi(x,y)=\sum_{i=1}^{n}\Delta_i(x,y)\braket{u_x^i}{v_y^i}.
    \end{align*}
    To put the witnesses in the form required by \Cref{thm:catalyst-state-conversion}, we combine
    the two families of vectors by defining
    \begin{align*}
        \ket{\kappa_x^i}:=
        \frac{\ket{u_x^i}+\ket{v_x^i}}{2}
        \oplus(-1)^{x_i}\frac{\ket{u_x^i}-\ket{v_x^i}}{2}.
    \end{align*}
    Let $X_i(x,y):=\braket{\kappa_x^i}{\kappa_y^i}$. Therefore, $X_i\succeq0$. Moreover, whenever
    $\Delta_i(x,y)=1$, we have $(-1)^{x_i+y_i}=-1$. Thus, for all $x,y\in \Sigma$,
    \begin{align*}
        \sum_{i=1}^{n}\Delta_i(x,y)X_i(x,y)
        &=\frac12\sum_{i=1}^{n}\Delta_i(x,y)
        \bigl(\braket{u_x^i}{v_y^i}+\braket{v_x^i}{u_y^i}\bigr)
        =\Phi(x,y)-\Psi(x,y),
    \end{align*}
    where the last equality follows from Hermiticity of $\Phi-\Psi$. In addition,
    \begin{align*}
        \sum_{i=1}^{n}\Tr\bigl(X_i\circ\kbr{\mu}\bigr)
        &=\sum_{i=1}^{n}\E_{x\sim\mu}\|\kappa_x^i\|^2
        =\frac12\sum_{i=1}^{n}\E_{x\sim\mu}
        \bigl[\|u_x^i\|^2+\|v_x^i\|^2\bigr]
        =\gamma.
    \end{align*}
    Therefore, applying \Cref{thm:catalyst-state-conversion} with $r=1$, $W = \Psi$, and
    $T:=\left\lceil\gamma/(1-\alpha)\right\rceil\leq 2\gamma/(1-\alpha)$ yields a Gram matrix
    $\widetilde{\Psi}$ with $\widetilde{\Psi}\circ I=I$, where
    $\mathsf{Q}(\Phi\to\widetilde{\Psi})\leq T$ and
    \begin{align*}
        \cF(\widetilde{\Psi}\circ\kbr{\mu},\Psi\circ\kbr{\mu})
        &\geq1-\frac{\sum_{i=1}^{n}\Tr\bigl(X_i\circ\kbr{\mu}\bigr)}{T}
        =1-\frac{\gamma}{T}\geq\alpha. \qedhere
    \end{align*}
\end{proof}


\section{Distributional adversary method}\label{sec:distributional-adversary}

In this section, we develop our notion of a distributional adversary method in
\Cref{sec:adv-definition} and relate it to quantum query complexity in \Cref{sec:adv-query}. In
\Cref{sec:adv-direct-sum}, we also show a direct-sum theorem for the measure itself.

\subsection{Definition and general properties}\label{sec:adv-definition} We begin by defining our
notion of a distributional adversary method, based on the Gram-matrix characterization of
distributional query complexity in \Cref{thm:fidelity-charecterization}. Specifically, we consider
the minimum cost of a transition from $J$ to a Gram matrix $\Psi$ computing $f$ over $\mu$, in terms
of the $\bar{\gamma}_2^\mu$ norm.

We note that Belovs \cite[Theorem~33 and Corollary~34]{Bel15} also gives a related distributional
adversary formulation for search problems.

\begin{definition}\label{def:adv} Let $\Sigma\subseteq\{0,1\}^n$, $f\subseteq\Sigma\times\Omega$ be
    a search problem, and $\mu:\Sigma\to[0,1]$ be a distribution. Also, let $\eps=\eps(n)\in[0,1)$
    be an error parameter. Then, the adversary measure for $f$ with respect to distribution $\mu$
    and error $\eps$ is defined as
    \begin{align*}
         \mathrm{Adv}^{\mu}_{\eps}(f) := \min_{\Psi \in \cG^{\mu}_{\eps}(f)}
         \bar{\gamma}_{2}^{\mu}(J-\Psi|\{\Delta_i\}),
    \end{align*}
    where $\cG^{\mu}_{\eps}(f)$ is the set of $(\eps,\mu)$-targets for $f$ as in
    \Cref{def:valid-grams}. In particular, the value is finite and the minimum is
    attained.\footnote{$\cG^{\mu}_{\eps}(f)$ is compact by \Cref{lem:compactness-gmueps} and
    $\bar{\gamma}_2^{\cdot}(\cdot|\{\Delta_i\})$ is jointly continuous as in
    \Cref{lem:equiv-avg-gamma2-query}.}
\end{definition}

\subsubsection{Equivalent primal and dual formulations}
We now show that a single adversary matrix can certify a common lower bound for all
$(\eps,\mu)$-targets for $f$. This will be especially useful when lifting a hardness witness for a
single instance to multiple copies of the problem.

\begin{lemma}\label{lem:adv-universal-hardness-cert} Let $\Sigma\subseteq\{0,1\}^n$,
$f\subseteq\Sigma\times\Omega$ be a search problem, and $\mu:\Sigma\to[0,1]$ be a distribution. For
$\eps=\eps(n)\in[0,1)$, we have
\begin{align*}
     \mathrm{Adv}^{\mu}_{\eps}(f)
     =\max_{\substack{\Gamma=\Gamma^\dagger\\
     \max_i\|\Gamma\circ\Delta_i\|\leq 1}}
     \min_{\Psi \in \cG^{\mu}_{\eps}(f)}
     \langle \Gamma,(J-\Psi) \circ \kbr{\mu}\rangle.
\end{align*}
Moreover, we may assume without loss of generality that $\Gamma\circ I=0$. In particular, the value
is finite and both the maximum and the inner minimum are attained.
\end{lemma}
\begin{proof}
Let $\cL:=\{\Gamma\in\Herm(\Sigma):\Gamma\circ I=0,\ \|\Gamma\circ\Delta_i\|\leq1 \text{ for all
}i\in[n]\}$. Then, we can write
\[
    \mathrm{Adv}^{\mu}_{\eps}(f)
    =\min_{\Psi\in\cG^{\mu}_{\eps}(f)}\max_{\Gamma\in\cL}
    \langle\Gamma,(J-\Psi)\circ\kbr{\mu}\rangle,
\]
where we expanded $\bar{\gamma}_2^\mu$ according to \Cref{lem:equiv-avg-gamma2-query}. The sets
$\cL$ and $\cG^{\mu}_{\eps}(f)$ are compact and convex by \Cref{lem:equiv-avg-gamma2-query} and
\Cref{lem:compactness-gmueps}, respectively. The objective is continuous, linear in $\Gamma$, and
affine in $\Psi$. As a result, by Sion's minimax theorem (\Cref{thm:sion-minimax}),
\[
    \mathrm{Adv}^{\mu}_{\eps}(f)
    =\max_{\Gamma\in\cL}\min_{\Psi\in\cG^{\mu}_{\eps}(f)}
    \langle\Gamma,(J-\Psi)\circ\kbr{\mu}\rangle.
\]
Finally, replacing any feasible $\Gamma$ by $\Gamma':=\Gamma-\Gamma\circ I$ preserves the objective,
since $(J-\Psi)\circ I=0$, and preserves the constraints, since
$\|\Gamma'\circ\Delta_i\|=\|(\Gamma-\Gamma\circ I)\circ\Delta_i\| =\|\Gamma\circ\Delta_i\|$ for all
$i\in[n]$. Thus, by \Cref{lem:adversary-support-replacement}, we may remove the condition
$\Gamma\circ I=0$, completing the proof.
\end{proof}

\subsubsection{Analogy with the worst-case adversary bound}
In this part, we present the relation between $\mathrm{Adv}^{\mu}$ and the worst-case
negative-weight adversary bound $\mathrm{Adv}^{\pm}$, showing that our measure is in fact a
distributional analogue.
\begin{proposition}\label{prop:recovering-worst-case-adversary} Let $\Sigma\subset\{0,1\}^n$,
$f:\Sigma\to\Omega$ be a functional problem, and $\eps=\eps(n)\in[0,1/2)$ be an error parameter.
Then,
\[
    (1-4\sqrt{\eps})\mathrm{Adv}^{\pm}(f)
    \leq \max_{\mu}\mathrm{Adv}^{\mu}_{\eps}(f)
    \leq O\bigl(\mathrm{Adv}^{\pm}(f)\bigr),
\]
where the maximum is over all distributions $\mu:\Sigma\to[0,1]$.
Moreover, if $f:\Sigma\to\{0,1\}$ is a decision problem, we can replace $(1-4\sqrt{\eps})$ by the
tighter coefficient $(1-2\sqrt{\eps(1-\eps)})$ and $O\bigl(\mathrm{Adv}^{\pm}(f)\bigr)$ by
$\mathrm{Adv}^{\pm}(f)$.
\end{proposition}

We postpone the proof to \Cref{app:recovering-worst-case-adversary}.

\subsection{Bounding query complexity}\label{sec:adv-query} Similar to the general adversary bound
\cite{LMR+11}, we show that $\mathrm{Adv}^{\mu}$ tightly characterizes $\mathsf{Q}^{\mu}$ up to a
constant-factor change in the error parameter. We begin with a one-line proof of the lower bound.
\begin{proposition}\label{adv-lowerbound} Let $\Sigma \subset \{0,1\}^{n}$, $f\subseteq \Sigma
    \times \Omega$ be a search problem, and $\mu:\Sigma\to [0,1]$ be a distribution. Also, let
    $\eps=\eps(n)\in[0,1)$ be an error parameter. Then, $\mathsf{Q}_{\eps}^{\mu}(f) \geq
    \mathrm{Adv}^{\mu}_{\eps}(f)/2.$
\end{proposition}
\begin{proof}
    \[
        \mathsf{Q}_{\eps}^{\mu}(f)
        = \min_{\Psi \in \cG^{\mu}_{\eps}(f)} \mathsf{Q}(\Psi)
        \geq \frac{1}{2}\min_{\Psi \in \cG^{\mu}_{\eps}(f)}
        \bar{\gamma}_2^{\mu}(J-\Psi|\{\Delta_i\})
        = \frac{1}{2}\mathrm{Adv}^{\mu}_{\eps}(f),
    \]
    where the first equality and the inequality follow from \Cref{thm:fidelity-charecterization} and
    \Cref{prop:hybrid-gamma2}, respectively.
\end{proof}

Next, we prove a general upper bound where the error from $\mathrm{Adv}^{\mu}$ is combined with a
state-conversion error.
\begin{proposition}\label{prop:adv-upperbound} Let $\Sigma \subset \{0,1\}^{n}$, $f\subseteq\Sigma
     \times \Omega$ be a search problem, and $\mu:\Sigma\to [0,1]$ be a distribution. Also, let
     $\eps=\eps(n)\in[0,1)$ and $\eta=\eta(n)\in(0,1)$ be error parameters satisfying
     $\sqrt{\eps}+\sqrt{\eta}<1$. Then,
    \begin{align*}
         \mathsf{Q}^{\mu}_{(\sqrt{\eps}+\sqrt{\eta})^2}(f)
         \leq \frac{4}{\eta} \cdot \mathrm{Adv}^{\mu}_{\eps}(f)
    \end{align*}
\end{proposition}

\begin{proof}[Proof of \Cref{prop:adv-upperbound}]
    Let $\Psi$ be a Gram matrix attaining the minimum $\mathrm{Adv}^{\mu}_{\eps}(f) =
    \bar{\gamma}_{2}^{\mu}(J-\Psi|\{\Delta_i\})$, and let $W\in \cW^{\mu}_{f}$ be the corresponding
    witness that satisfies $\cF(\Psi \circ \kbr{\mu},W\circ \kbr{\mu}) \geq \sqrt{1-\eps}$. By
    \Cref{prop:distalg-from-avggamma2}, there exists a quantum algorithm that converts Gram matrix
    $J$ to a Gram matrix $\widetilde{\Psi}$ satisfying
    $\cF(\Psi\circ\kbr{\mu},\widetilde{\Psi}\circ\kbr{\mu}) \geq \sqrt{1-\eta}$ using at most
    \[
       \frac{2}{1-\sqrt{1-\eta}}\bar{\gamma}_2^{\mu}(J-\Psi|\{\Delta_i\})
        \leq\frac{4}{\eta}\mathrm{Adv}^{\mu}_{\eps}(f)
    \]
    queries, where we used $1-\sqrt{1-\eta}\geq\eta/2$. As a result, by the triangle inequality for
    purified distance, we have
    \begin{align*}
        \sqrt{1-\cF^{2}(\widetilde{\Psi}\circ\kbr{\mu},W\circ\kbr{\mu})}
        &\leq \sqrt{1-\cF^{2}(\widetilde{\Psi}\circ\kbr{\mu},\Psi\circ\kbr{\mu})}
        +\sqrt{1-\cF^{2}(\Psi\circ\kbr{\mu},W\circ\kbr{\mu})} \\
        &\leq \sqrt{\eta}+\sqrt{\eps}.
    \end{align*}
    Therefore, by the reverse direction of \Cref{thm:fidelity-charecterization}, the algorithm
    computes $f$ over distribution $\mu$ with error at most $(\sqrt{\eps}+\sqrt{\eta})^2$.
\end{proof}

As a direct corollary, we obtain bounds in terms of a prescribed error or bias.
\begin{corollary}\label{cor:adv-query-comparison} Let $\Sigma \subset \{0,1\}^{n}$,
    $f\subseteq\Sigma\times\Omega$ be a search problem, and $\mu:\Sigma\to[0,1]$ be a distribution.
    Also, let $\eps=\eps(n)\in(0,1)$ be an error parameter. Then,
    \begin{align*}
        \Omega\left(\mathrm{Adv}^{\mu}_{\eps}(f)\right)
        &\leq \mathsf{Q}^{\mu}_{\eps}(f)
        \leq O\left(\frac{\mathrm{Adv}^{\mu}_{\eps/2}(f)}{\eps}\right).
    \end{align*}
    Moreover, for a decision problem $f$, we also obtain the following bound in terms of a bias
    parameter $\delta=\delta(n)\in(0,1/2)$.\footnote{For $\delta>1/4$, the target error parameter
    $1/2-2\delta$ becomes negative. Since no valid target can have negative error, the second upper
    bound is vacuous in this case. The lower bound remains valid throughout the stated range.}
    \begin{align*}
        \Omega\left(\mathrm{Adv}^{\mu}_{1/2-\delta}(f)\right)
        &\leq \mathsf{Q}^{\mu}_{1/2-\delta}(f)
        \leq O\left(\frac{\mathrm{Adv}^{\mu}_{1/2-2\delta}(f)}{\delta^2}\right).
    \end{align*}
\end{corollary}
\begin{proof}
    The lower-bound inequalities follow from \Cref{adv-lowerbound}. For the first upper bound, apply
    \Cref{prop:adv-upperbound} with error parameter $\eps/2$ and $\eta=\eps/18$. We have
    \begin{align*}
        \left(\sqrt{\eps/2}+\sqrt{\eps/18}\right)^2
        = \frac{8}{9}\eps \leq \eps.
    \end{align*}
    The query complexity is at most $72\mathrm{Adv}^{\mu}_{\eps/2}(f)/\eps$, which gives the first
    statement.
    \par
    For the second upper bound, assume $\delta\leq1/4$ and apply \Cref{prop:adv-upperbound} with
    error parameter $1/2-2\delta$ and $\eta=\delta^2/4$. We have
    \begin{align*}
        \left(\sqrt{\frac{1}{2}-2\delta}+\delta/2\right)^2
        &= \frac{1}{2}-2\delta+\delta\sqrt{\frac{1}{2}-2\delta}+\frac{\delta^2}{4}
        \leq \frac{1}{2}-\delta,
    \end{align*}
    where we used $\sqrt{1/2-2\delta}+\delta/4\leq1$. The query complexity is at most
    $16\mathrm{Adv}^{\mu}_{1/2-2\delta}(f)/\delta^2$, which gives the second statement.
\end{proof}

\subsection{Direct sum for distributional adversary}\label{sec:adv-direct-sum} In this part, we
state a direct-sum proposition for the distributional adversary measure. The lower bound preserves
the error exactly, as in the worst-case setting \cite{ACLT10,LMR+11}. For the upper bound, however,
our construction uses a one-copy target with a smaller error than the $k$-copy target. The reason
for this distinction is that a single distribution can genuinely put weights on extremely hard
inputs, which naturally accumulates the error across the $k$ copies.

\begin{proposition}\label{prop:direct-sum-distributional-adversary} Let $\Sigma \subset
    \{0,1\}^{n}$, $f\subseteq \Sigma \times \Omega$ be a search problem, $\mu:\Sigma \to [0,1]$ be a
    distribution, $\eps=\eps(n)\in(0,1)$, and $k\in \bN$ with $k\geq 1$. Then,
    \begin{align*}
        k\cdot \mathrm{Adv}^{\mu}_{\eps}(f)\leq \mathrm{Adv}^{\mu^{\otimes k}}_{\eps}(f^{\otimes k})\leq
         k\cdot \mathrm{Adv}^{\mu}_{1-(1-\eps)^{1/k}}(f)
    \end{align*}
 \end{proposition}
We include this proposition mainly as an additional structural property of the distributional
adversary measure. For $0<\eps<\eps'<1/2$, combining \Cref{adv-lowerbound,prop:adv-upperbound} with
the proposition gives
\begin{align*}
    \mathsf{Q}^{\mu^{\otimes k}}_{\eps}(f^{\otimes k})
    =\Omega\!\left(
        (\sqrt{\eps'}-\sqrt{\eps})^2
        \cdot k\cdot\mathsf{Q}^{\mu}_{\eps'}(f)
    \right).
\end{align*}
Later, in \Cref{cor:direct-sum-distributional-query}, we use the direct-product theorem to remove
the error slack and obtain
\begin{align*}
    \mathsf{Q}^{\mu^{\otimes k}}_{\eps}(f^{\otimes k})
    =\Omega\!\left(\eps k\cdot\mathsf{Q}^{\mu}_{\eps}(f)\right)
\end{align*}
for $0<\eps<1/2$. Since that argument does not use the adversary direct-sum proposition, we defer
the detailed proof of \Cref{prop:direct-sum-distributional-adversary} to \Cref{app:adv-direct-sum}
and outline the main ideas below.

\paragraph{Proof overview.}
 The main idea in both directions of the proof is to design a lifting map. For example, consider the
 lower bound. We construct a map $\cE^{\uparrow}$ that takes a one-copy hardness certificate
 $\Gamma$ to a $k$-copy candidate $\cE^{\uparrow}(\Gamma)$. To show that this candidate gives the
 desired lower bound, we use the adjoint map $\cE^{\downarrow}$ to relate its inner product with a
 $k$-copy target $\Lambda$ to a one-copy inner product:
\begin{align*}
\langle \cE^{\uparrow}(\Gamma),(J^{\otimes k}-\Lambda)\circ\kbr{\mu}^{\otimes k}\rangle
&=\langle\Gamma,\cE^{\downarrow}((J^{\otimes k}-\Lambda)\circ\kbr{\mu}^{\otimes k})\rangle
 =\langle\Gamma,(J-\Psi^{\downarrow})\circ\kbr{\mu}\rangle.
\end{align*}
The key is to show that $\cE^{\downarrow}$ takes valid $k$-copy targets to valid one-copy targets,
so that $\Psi^{\downarrow}\in\cG^\mu_\eps(f)$. Together with a bound on the query-filter norms of
$\cE^{\uparrow}(\Gamma)$, this allows us to obtain a $k$-copy hardness certificate from $\Gamma$.
\par
The upper bound follows a similar argument with the roles of $\Psi$ and $\Gamma$ reversed: we start
 with a one-copy algorithmic certificate $\Psi$ and lift it to the $k$-copy target $\Psi^{\otimes
 k}$. The lower and upper bounds are proved in \Cref{prop:direct-sum-lower-bound} and
 \Cref{prop:direct-sum-upper-bound}, respectively.
\par

\section{Relaxed and Shaltiel-free adversary method}\label{sec:relaxed-and-sf-adversary-method} In
this section, we introduce the measures needed for the composition theorem in
\Cref{sec:composition-adversary}. More specifically, we first relax the distributional adversary
method in \Cref{sec:relaxed-adversary} to obtain a more analytically flexible measure. We then use
the relaxed measure to define the Shaltiel-free adversary in \Cref{sec:shaltiel-free-adversary},
which will serve as the measure for the inner function in the composition lower bound.
\par
We note that, throughout this section, we will be working only on functional and decision problems
where disagreement and agreement matrices $\Delta_f$ and $A_f$ are well-defined.

\subsection{Relaxed adversary method}\label{sec:relaxed-adversary} We first motivate our definition
as follows. Let $\Psi$ be a given Gram matrix, and write $\Psi = \Psi \circ A_f + \Psi \circ
\Delta_f$. Then, intuitively, the disagreement part $\Psi \circ \Delta_f$ encodes the remaining
overlap between inputs $x,y$ with different outputs $f(x)\neq f(y)$, and therefore how
distinguishable the final states corresponding to $x$ and $y$ are. On the other hand, the agreement
part $\Psi \circ A_f$ records overlaps between inputs with the same answer $f(x)=f(y)$ and, in a
sense, reflects how the physical implementation of the algorithm retains information about the
inputs themselves.
\par
In this subsection, we consider a nonphysical relaxation based only on the disagreement part, using
a trace-norm bound. More specifically, we consider relaxed Hermitian targets $\Psi$ with the fixed
agreement part $\Psi \circ A_f = A_f$, so that $J-\Psi = (J-\Psi) \circ \Delta_f$ only records
disagreement entries. We do not require these targets to be PSD. We also measure the ``trace-norm
error'' $\delta$ in the disagreement part by considering $\|(\Psi - A_f) \circ \kbr{\mu}\|_{\tr} =
\|\Psi \circ \Delta_f \circ \kbr{\mu}\|_{\tr}\leq \delta$.
\par
As shown in \Cref{cor:radv-query-comparison}, this measure still bounds the quantum query
complexity, albeit with a quadratic error loss. Moreover, for decision problems,
\Cref{lem:radv-universal-hardness-cert} gives a very nice dual form that we will use to define the
composition measure from it.

We also note that our relaxation allows for Hermitian matrices $\Psi$. In fact, if $\Psi$ is a Gram
matrix with $\Psi \circ A_f = A_f$, all entries $\Psi(x,y)$ will depend only on the output values
$f(x)$ and $f(y)$, restricting the disagreement part $\Psi\circ\Delta_f$ as well.

\begin{definition}\label{def:radv} Let $\Sigma \subset \{0,1\}^{n}$, $f:\Sigma \to \Omega$ be a
    functional problem, and $\mu:\Sigma\to [0,1]$ be a distribution. Also, let
    $\delta=\delta(n)\in[0,1]$ be a relaxation parameter. We define the set of (relaxed) targets for
    $f$ as the following set of Hermitian matrices.
    \begin{align*}
     \cR_{\delta}^{\mu}(f):= \{\Psi\in \Herm(\Sigma): \Psi \circ A_f = A_f ; \|\Psi \circ \Delta_f \circ \kbr{\mu}\|_{\tr} \leq \delta\}.
    \end{align*}
    Then, the relaxed adversary measure for $f$ with respect to distribution $\mu$ and relaxation
    parameter $\delta$ is defined as
    \begin{align*}
         \mathrm{rAdv}^{\mu}_{\delta}(f) :=& \min_{\Psi \in \cR^{\mu}_{\delta}(f)} \bar{\gamma}_{2}^{\mu}(J-\Psi|\{\Delta_i\})
         = \progbox{\min_{\Psi \in \Herm(\Sigma)}}{\bar{\gamma}_2^{\mu}(J-\Psi|\{\Delta_i\})}{ \Psi \circ A_f = A_f   \\ \|\Psi \circ \Delta_f \circ \kbr{\mu}\|_{\tr} \leq \delta}
    \end{align*}
\end{definition}

The following lemma shows that we may restrict the feasible set to a compact subset, justifying
attainment of the minimum.
\begin{lemma}\label{lem:radv-compact-restriction} Let $\Sigma,f,\mu,\delta$ be as in
\Cref{def:radv}, and write $S:=\supp(\mu)$. Then,
\begin{align*}
    \mathrm{rAdv}^{\mu}_{\delta}(f)
    &= \min_{\Psi\in\cR}\bar{\gamma}_2^\mu(J-\Psi|\{\Delta_i\}),
\end{align*}
where $\cR:=\{\Psi\in\cR^\mu_\delta(f): \Psi\circ\Delta_f\circ J_S=\Psi\circ\Delta_f\}$ is a
nonempty, convex, and compact set. In particular, the value is finite and the minimum is attained.
\end{lemma}

\begin{proof}
Let $M:=A_f+\Delta_f\circ J_S$. Replacing $\Psi$ by $\Psi\circ M$ preserves both constraints, and by
    \Crefitem{item:avggamma2-support}, $\bar{\gamma}_2^\mu(J-\Psi\circ M|\{\Delta_i\})
    =\bar{\gamma}_2^\mu((J-\Psi)\circ J_S|\{\Delta_i\}) =\bar{\gamma}_2^\mu(J-\Psi|\{\Delta_i\})$.
    Therefore, \Cref{lem:adversary-support-replacement} applies with the mask $A_f+\Delta_f\circ
    J_S$ and function $F(\Psi) = -\bar{\gamma}_2^{\mu}(J-\Psi|\{\Delta_i\})$, so we may restrict to
    the set where $\Psi = \Psi\circ M = \Psi \circ (A_f + \Delta_f \circ J_S)$.

    This just means that $\Psi(x,y)=0$ whenever $f(x)\neq f(y)$ and $\mu(x)\mu(y)=0$. The set $\cR$
    is nonempty because $A_f\in\cR$, and convex because its defining constraints are affine or
    convex. It is also compact: it is closed, and its only unfixed entries satisfy
    $|\Psi(x,y)|\leq\delta/\sqrt{\mu(x)\mu(y)}$ for $x,y\in S$. Thus, by continuity of the objective
    in \Cref{lem:equiv-avg-gamma2-query}, the minimum is attained.
\end{proof}

\subsubsection{Equivalent primal and dual formulations}\label{sec:radv-dual} In this part, we first
provide a minimax characterization of $\mathrm{rAdv}^{\mu}_{\delta}(f)$ for functional problems. We
then obtain a simpler dual formulation for decision problems. Its pairing and spectral-norm penalty
are closely related to the distribution-distinguishing bound of \cite{BBH+18}.

\begin{proposition}\label{lem:radv-universal-hardness-cert} Let $\Sigma \subset \{0,1\}^{n}$,
$f:\Sigma\to\Omega$ be a functional problem, and $\mu:\Sigma\to[0,1]$ be a distribution. Then, for
any relaxation parameter $\delta=\delta(n)\in[0,1]$,
\begin{align*}
    \mathrm{rAdv}^{\mu}_{\delta}(f)
    = &\progat{6.5cm}
    {\max_{\Gamma=\Gamma^{\dagger}}}
    {\min_{\Psi\in\cR^{\mu}_{\delta}(f)}
    \langle\Gamma,(J-\Psi)\circ\kbr{\mu}\rangle}
    {&\Gamma\circ I=0; \\
     &\max_{i\in[n]}\|\Gamma\circ\Delta_i\|\leq 1.}
\end{align*}
Moreover, if $f$ is a decision problem, then
\begin{align*}
    \mathrm{rAdv}^{\mu}_{\delta}(f)
    = &\progat{6.5cm}
    {\max_{\Gamma=\Gamma^{\dagger}}}
    {\langle \Gamma,\Delta_f\circ \kbr{\mu}\rangle-\delta\|\Gamma\|}
    {&\Gamma\circ\Delta_f=\Gamma; \\
     &\max_{i\in[n]}\|\Gamma\circ\Delta_i\|\leq 1.}
\end{align*}
In particular, the value is finite and both maxima are attained.
\end{proposition}
\begin{proof}
Let $S:=\supp(\mu)$. We first apply Sion's minimax theorem to interchange the minimum and maximum.
For decision problems, we then use restrictions and a change of variables to evaluate the inner
optimization using duality between the spectral and trace norms.
\paragraph{Applying minimax.}
We begin by writing $\mathrm{rAdv}^{\mu}_{\delta}(f)$ in a form where we can apply Sion's minimax
theorem; that is,
\begin{align*}
    \mathrm{rAdv}^{\mu}_{\delta}(f)
    &= \min_{\Psi \in \cR'}\max_{\Gamma \in \cL}
    \langle\Gamma,(J-\Psi)\circ\kbr{\mu}\rangle,
\end{align*}
where we restricted the domains of $\Psi$ and $\Gamma$ to the sets
\begin{align*}
    \cR'
    &= \{\Psi \in \Herm(\Sigma): \Psi\circ A_f = A_f, \|\Psi \circ \Delta_f \circ \kbr{\mu}\|_{\tr}\leq \delta,
    \Psi\circ\Delta_f\circ J_S =\Psi\circ\Delta_f
    \}, \\
    \cL &= \{\Gamma \in \Herm(\Sigma): \Gamma\circ I=0,\ \max_{i\in[n]}\|\Gamma\circ\Delta_i\|\leq1\},
\end{align*}
by \Cref{lem:radv-compact-restriction} and \Cref{lem:equiv-avg-gamma2-query}, respectively. Both
$\cR'$ and $\cL$ are compact and convex sets and Sion's minimax \Cref{thm:sion-minimax} gives
\begin{align}
    \mathrm{rAdv}^{\mu}_{\delta}(f)
    =
   \max_{\Gamma \in \cL}  \min_{\Psi \in \cR'}
    {\langle\Gamma,(J-\Psi)\circ\kbr{\mu}\rangle}.
\end{align}
Enlarging $\cR'$ back to $\cR^{\mu}_{\delta}(f)$ does not change the inner minimum, proving the
general formulation.
We now continue with the case where $f$ is a decision problem.
\paragraph{Restricting $\cL$.}
Next, we claim that we can further restrict the outer set $\cL$ to
$\cL':=\{\Gamma\in\cL:\Gamma\circ\Delta_f\circ J_S=\Gamma\}$. Indeed, taking
$\Gamma'=\Gamma\circ\Delta_f\circ J_S$ does not change the objective, since
$(J-\Psi)\circ\kbr{\mu}=(J-\Psi)\circ\Delta_f\circ J_S\circ\kbr{\mu}$. Moreover, the constraints
remain satisfied, since
\begin{align*}
    \|\Gamma'\circ\Delta_i\|
    &\leq\gamma_2(\Delta_f)\gamma_2(J_S)\|\Gamma\circ\Delta_i\|\leq1,\\
    \Gamma'\circ I&=0,
\end{align*}
where $\gamma_2(\Delta_f)\leq1$ for decision problems and $\gamma_2(J_S)=1$. Therefore,
\Cref{lem:adversary-support-replacement} applies, giving
\begin{align*}
    \mathrm{rAdv}^{\mu}_{\delta}(f)
    &=\max_{\Gamma\in\cL'}\min_{\Psi\in\cR'}
    \langle\Gamma,(J-\Psi)\circ\kbr{\mu}\rangle.
\end{align*}
\paragraph{Change of variable from $\cR'$.}
Next, we apply a change of variable for the set $\cR'$ by considering the map $\Psi \to K := \Psi
\circ \Delta_f \circ \kbr{\mu}$, where the set $\cR'$ maps to the set
\begin{align*}
\cK := \{K \in \Herm(\Sigma): K \circ \Delta_f = K, \|K\|_{\tr} \leq \delta, K \circ J_S = K\}
\end{align*}
It is easy to verify that the map $\Psi \to K$ is a bijection between $\cR'$ and $\cK$ by taking the
inverse map $K \to \Psi := A_f + K \circ \ketbra{1/\mu}$ where $\ket{1/\mu} = \sum_{x\in S}
1/\sqrt{\mu(x)} \ket{x}$. Therefore, we can write
\begin{align}
    \mathrm{rAdv}^{\mu}_{\delta}(f)
    &=\max_{\Gamma \in \cL'}\min_{K \in \cK}
    \left(\langle\Gamma,\kbr{\mu}\rangle-\langle\Gamma,K\rangle\right) \nonumber\\
    &=\max_{\Gamma \in \cL'}\left(
    \langle\Gamma,\kbr{\mu}\rangle-\max_{K \in \cK}\langle\Gamma,K\rangle\right).
    \label{eq:radv-before-inner-max}
\end{align}

\paragraph{Recovering spectral norm.} We next replace the inner maximum using norm duality.

\begin{claim}
    If $\Gamma \circ \Delta_f \circ J_{S} = \Gamma \in \Herm(\Sigma)$, then $\max_{K\in \cK}\langle
    \Gamma,K\rangle  = \delta \cdot \|\Gamma\|$.
\end{claim}
\begin{proof}
Let $\cK':= \{K \in \Herm(\Sigma):\|K\|_{\tr}\leq \delta\}$. Then, for every $K\in\cK'$, $\langle
 \Gamma,K\rangle = \langle \Gamma \circ \Delta_f \circ J_S, K\rangle  = \langle \Gamma, K \circ
 \Delta_f \circ J_S \rangle$. Moreover, if $K \in \cK'$ so is $K\circ \Delta_f \circ J_S$ as
 $\|K\circ \Delta_f \circ J_S\|_{\tr}\leq \gamma_2(\Delta_f) \cdot \gamma_2(J_S) \cdot \|K\|_{\tr}
 \leq 1\cdot \delta$. Therefore, \Cref{lem:adversary-support-replacement} applies and we can remove
 the condition $K\circ \Delta_f \circ J_S = K$ to obtain
\begin{align*}
    \max_{K\in \cK}\langle \Gamma,K\rangle
    &=\max_{K \in \cK'}\langle \Gamma,K\rangle
    =\max_{\substack{K=K^\dagger\\\|K\|_{\tr}\leq\delta}}\langle \Gamma,K\rangle
    =\max_{\substack{K=K^\dagger\\\|K\|_{\tr}\leq1}}\delta\cdot\langle \Gamma,K\rangle
    =\delta\cdot\|\Gamma\|,
\end{align*}
where in the last step we used duality between the trace and spectral norms.
\end{proof}
Consequently, we can continue \Cref{eq:radv-before-inner-max} as
\begin{align*}
    \mathrm{rAdv}^{\mu}_{\delta}(f) = \max_{\Gamma \in \cL'}\left(\langle \Gamma,\kbr{\mu}\rangle-\delta\|\Gamma\|\right).
\end{align*}
\paragraph{Enlarging $\cL'$.}
Finally, to complete the proof, let $\cL'':=\{\Gamma\in\Herm(\Sigma):\Gamma\circ\Delta_f=\Gamma,\
\max_i\|\Gamma\circ\Delta_i\|\leq1\}$. Similar to the above, $\Gamma\circ J_S\in\cL'$ for every
$\Gamma\in\cL''$, and this replacement does not decrease the objective. Thus, by
\Cref{lem:adversary-support-replacement}, we may replace $\cL'$ with $\cL''$ in the maximum, giving
the stated form.
\end{proof}

\subsubsection{Bounding query complexity}\label{sec:radv-bounds} In this part, we show that, with a
loss in the error parameter, our relaxed measure can still upper and lower bound the distributional
adversary measure, and therefore, quantum query complexity.
\paragraph{Adversary lower bounds relaxed adversary.}

We show that $\mathrm{Adv}$ lower bounds $\mathrm{rAdv}$ by replacing the agreement part of a valid
Gram matrix $\Psi \in \cG^{\mu}_{\eps}(f)$ with $A_f$, giving a feasible relaxed target.
\begin{proposition}
    Let $\Sigma \subset \{0,1\}^{n}$, $f:\Sigma \to \Omega$ be a functional problem, and
    $\mu:\Sigma\to [0,1]$ be a distribution. Also, let $\eps=\eps(n)\in[0,1/2)$ be an error
    parameter. Then, for $\eps\leq1/16$, we have $\mathrm{Adv}_{\eps}^{\mu}(f) \geq
    \mathrm{rAdv}^{\mu}_{4\sqrt{\eps}}(f)/2.$
    \par
    Moreover, if $f:\Sigma \to \{0,1\}$ is a decision problem, then for every
    $\eps=\eps(n)\in[0,1/2)$, we have the better bound $\mathrm{Adv}^{\mu}_{\eps}(f) \geq
    \mathrm{rAdv}^{\mu}_{2\sqrt{\eps(1-\eps)}}(f).$
\end{proposition}
\begin{proof}
For the first bound, assume $\eps\leq1/16$. Let $\Psi\in\cG^{\mu}_{\eps}(f)$ and define
$\widehat{\Psi}:=A_f+\Psi\circ\Delta_f$. Then, $\widehat{\Psi}\circ A_f=A_f$ and
\begin{align*}
\|\widehat{\Psi}\circ\Delta_f \circ \kbr{\mu}\|_{\tr}
= \|\Psi\circ\Delta_f \circ \kbr{\mu}\|_{\tr}
\leq 4\sqrt{\eps},
\end{align*}
where we used \Cref{lem:trace-norm-upper-bound-for-boolean-functions} in the last step. Thus,
$\widehat{\Psi}\in\cR^{\mu}_{4\sqrt{\eps}}(f)$. Moreover, $J-\widehat{\Psi}=(J-\Psi)\circ\Delta_f$.
Therefore,
\begin{align*}
    \mathrm{rAdv}^{\mu}_{4\sqrt{\eps}}(f)
    &\leq \bar{\gamma}_2^{\mu}(J-\widehat{\Psi}|\{\Delta_i\})
    \leq \gamma_2(\Delta_f)\bar{\gamma}_2^{\mu}(J-\Psi|\{\Delta_i\})
    \leq 2\bar{\gamma}_2^{\mu}(J-\Psi|\{\Delta_i\}),
\end{align*}
where we applied \Crefitem{item:avggamma2-hadamard-bounds} and $\gamma_2(\Delta_f)\leq2$ in the last
two steps, respectively. Since $\Psi\in\cG^{\mu}_{\eps}(f)$ was arbitrary,
$\mathrm{rAdv}^{\mu}_{4\sqrt{\eps}}(f)
\leq\min_{\Psi\in\cG^{\mu}_{\eps}(f)}2\bar{\gamma}_2^{\mu}(J-\Psi|\{\Delta_i\})
=2\mathrm{Adv}^{\mu}_{\eps}(f)$.
\par
If $f$ is a decision problem, then \Cref{lem:trace-norm-upper-bound-for-boolean-functions} provides
a sharper bound, giving $\widehat{\Psi}\in\cR^{\mu}_{2\sqrt{\eps(1-\eps)}}(f)$ for every
$\eps\in[0,1/2)$. The claim similarly follows, since $\gamma_2(\Delta_f)\leq1$ for decision
problems.
\end{proof}

\paragraph{relaxed adversary upper bounds adversary up to a loss in error}
We next show that we can obtain a quantum algorithm from a $\mathrm{rAdv}$ certificate up to a loss
in error.
\begin{proposition}\label{prop:radv-upperbound} Let $\Sigma \subset \{0,1\}^{n}$, and let $f:\Sigma
    \to \Omega$ be a functional problem. Let $\delta=\delta(n)\in(0,1)$ be an error parameter,
    $\mu:\Sigma\to [0,1]$ be a distribution, and $\eta=\eta(n)\in(0,1]$ satisfy $(2+\eta)\delta<1$.
    Then,
    \begin{align*}
        \mathrm{Adv}^{\mu}_{(2+\eta)\delta}(f)\leq 2\mathsf{Q}^{\mu}_{(2+\eta)\delta}(f)
        =O\left(\frac{\mathrm{rAdv}^{\mu}_{\delta}(f)}{\eta^2\delta}\right).
    \end{align*}
\end{proposition}

\paragraph{Overview of proof.} The first inequality follows from \Cref{adv-lowerbound}. Therefore,
it remains to construct a quantum algorithm from a $\mathrm{rAdv}$ certificate. The key idea of the
proof is to ``repair'' the Hermitian matrix $\Psi=A_f+\Psi\circ\Delta_f\in\cR^\mu_\delta(f)$ into a
Gram matrix suitable for state conversion. We do this in two steps.
\begin{itemize}
    \item First, we repair $\Psi$ into a PSD matrix by compensating for the negative contribution of
    $\Psi\circ\Delta_f\circ\kbr{\mu}$.
    \item Then, we normalize the obtained PSD matrix into a Gram matrix $\widehat{\Psi}$.
\end{itemize}
To apply our $\bar{\gamma}_2^\mu$ state-conversion bound in \Cref{prop:distalg-from-avggamma2}, we
also apply the same procedure to $J$, obtaining $\widehat{J}$ such that
$\bar{\gamma}_2^\mu(\widehat{J}-\widehat{\Psi}|\{\Delta_i\})
\leq\bar{\gamma}_2^\mu(J-\Psi|\{\Delta_i\})$. Starting from $\widehat{J}$, this gives an algorithm
that produces a Gram matrix $\widetilde{\Psi}$ close to $\widehat{\Psi}$. However, the algorithm
must start from $J$. We therefore run the same algorithm from $J$, obtaining a final Gram matrix
$\doubletilde{\Psi}$, and bound the additional error using the fidelity bounds from the repair
process.

The following lemma formalizes the repair process.

\begin{lemma}[PSD repair lemma]\label{lem:radv-psd-repair} Let $\Sigma\subset \{0,1\}^{n}$,
    $f:\Sigma \to \Omega$ be a functional problem,  $\mu:\Sigma \to [0,1]$ be a distribution, and
    $\delta=\delta(n)\in[0,1]$. Let $\Psi\in \Herm(\Sigma)$ be a Hermitian matrix satisfying $\Psi
    \circ A_f = A_f$ and $\|\Psi \circ \Delta_f \circ \kbr{\mu}\|_{\tr}\leq \delta$.
    \par
    Then, there exists a pair of Gram matrices $\widehat{J},\widehat{\Psi}$ such that the following
    holds.
        \begin{enumerate}[i.]
            \item $\widehat{J},\widehat{\Psi}$ are both Gram matrices with $\widehat{J}\circ I = I$
            and $\widehat{\Psi}\circ I = I.$
            \item $\cF(\widehat{J}\circ \kbr{\mu}, J \circ \kbr{\mu}) \geq
            \frac{1}{\sqrt{1+\delta/2}}$ and $\cF(\widehat{\Psi}\circ \kbr{\mu}, A_f\circ \kbr{\mu})
            \geq \frac{1}{\sqrt{1+\delta/2}}.$
            \item
            $\bar{\gamma}_{2}^{\mu}(\widehat{J}-\widehat{\Psi}|\{\Delta_i\}) \leq
            \bar{\gamma}_{2}^{\mu}(J-\Psi|\{\Delta_i\})$.
        \end{enumerate}
    \end{lemma}

    \begin{proof}
        Let $S:=\supp(\mu)$. Without loss of generality, we assume that $\Psi-A_f$ is supported on
        $S\times S$. Otherwise, by \Crefitem{item:avggamma2-support}, we may replace $\Psi$ with
        $\Psi':=A_f+(\Psi\circ\Delta_f)\circ J_S$, so that $\Psi'\circ\kbr{\mu}=\Psi\circ\kbr{\mu}$,
        while
        \begin{align*}
            \bar{\gamma}_2^\mu(J-\Psi'|\{\Delta_i\})
            &=\bar{\gamma}_2^\mu((J-\Psi)\circ J_S|\{\Delta_i\})
            =\bar{\gamma}_2^\mu(J-\Psi|\{\Delta_i\}).
        \end{align*}
        \par
        We begin by factoring $\Psi \circ \Delta_f$ into a PSD difference as follows.
        \begin{claim}\label{clm:radv-psd-decomposition} There exist $H^{+},H^{-}\succeq0$ such that
            $\Psi\circ\Delta_f=H^{+}-H^{-}$ and $\Tr[H^{-}\circ\kbr{\mu}]\leq\delta/2$.
        \end{claim}
        \begin{proof}
        Indeed, let $H:=\Psi\circ\Delta_f$ and define $H_{\mu}:=H\circ\kbr{\mu}$. By the trace-norm
        characterization, there exist $H_{\mu}^{+},H_{\mu}^{-}\succeq0$ such that
        $H_{\mu}=H_{\mu}^{+}-H_{\mu}^{-}$ and
        $\Tr[H_{\mu}^{+}]+\Tr[H_{\mu}^{-}]=\|H_{\mu}\|_{\tr}\leq\delta$. Moreover, since $H_{\mu}$
        is zero on its diagonal entries, $H_{\mu}^{+}$ and $H_{\mu}^{-}$ have the same diagonal
        entries, so $\Tr[H_{\mu}^{+}]=\Tr[H_{\mu}^{-}]=\frac12\|H_{\mu}\|_{\tr}\leq\delta/2$. Next,
        we rescale back by $\kbr{1/\mu}:=\sum_{x,y\in S}\frac{\ketbra{x}{y}}{\sqrt{\mu(x)\mu(y)}}$
        to get $H^{\pm}:=H_{\mu}^{\pm}\circ\kbr{1/\mu}\succeq0$. Since we assumed that
        $\Psi\circ\Delta_f$ is supported on $S\times S$, we obtain $\Psi\circ\Delta_f=H=H^{+}-H^{-}$
        and $\Tr[H^{-}\circ\kbr{\mu}]=\Tr[H_{\mu}^{-}]\leq\delta/2$.
        \end{proof}
        \par
        As a result,
        $
            \Psi':=\Psi+H^{-}=A_f+\Psi\circ\Delta_f+H^{-}=A_f+H^{+}\succeq0
        $
        and $J':=J+H^{-}\succeq0$. Since $H^{+}$ and $H^{-}$ have the same diagonal, the matrices
        $\Psi'$ and $J'$ have the same diagonal $d(x):=\Psi'(x,x)=J'(x,x)=1+H^{-}(x,x)\geq1$. So, we
        can rescale our matrices using
        $\ket*{d^{-1}}:=\sum_{x\in\Sigma}\frac{1}{\sqrt{d(x)}}\ket*{x}$ to obtain Gram matrices
        $\widehat{\Psi}:=\Psi'\circ\kbr{d^{-1}}$ and $\widehat{J}:=J'\circ\kbr{d^{-1}}$. Since
        $\widehat{J}-\widehat{\Psi}=(J-\Psi)\circ\kbr{d^{-1}}$, we can also bound
        \begin{align*}
            \bar{\gamma}_2^\mu(\widehat{J}-\widehat{\Psi}|\{\Delta_i\})
            \leq\bar{\gamma}_2^\mu(J-\Psi|\{\Delta_i\})\gamma_2(\kbr{d^{-1}})
            \leq\bar{\gamma}_2^\mu(J-\Psi|\{\Delta_i\}),
        \end{align*}
        where we used \Crefitem{item:avggamma2-hadamard-bounds} and
        $\gamma_2(\kbr{d^{-1}})=\max_x1/d(x)\leq1$, respectively.
        \par
        It only remains to prove the fidelity lower bounds. To do so, let
        $A_f(x,y)=\braket{a_x}{a_y}$ and $H^{+}(x,y)=\braket{h_x}{h_y}$ be PSD decompositions. Then,
        the vectors
        \begin{align*}
            \ket*{a}&:=\sum_x\sqrt{\mu(x)}\ket*{x}\otimes(\ket*{a_x}\oplus0),\\
            \ket*{\widehat{\psi}}&:=\sum_x\sqrt{\frac{\mu(x)}{d(x)}}\ket*{x}
            \otimes(\ket*{a_x}\oplus\ket*{h_x})
        \end{align*}
        purify $(A_f\circ\kbr{\mu})^\top$ and $(\widehat{\Psi}\circ\kbr{\mu})^\top$, respectively.
        Thus, by Uhlmann's theorem,
        \begin{align*}
            \cF(\widehat{\Psi}\circ\kbr{\mu},A_f\circ\kbr{\mu})
            &\geq|\braket*{a}{\widehat{\psi}}|
            =\E_{x\sim\mu}\left[\frac{1}{\sqrt{d(x)}}\right]
            \geq\frac{1}{\sqrt{\E_{x\sim\mu}[d(x)]}}
            =\frac{1}{\sqrt{1+\Tr[H^{-}\circ\kbr{\mu}]}} \\
            &\geq\frac{1}{\sqrt{1+\delta/2}},
        \end{align*}
        where we used $\|a_x\|^2=A_f(x,x)=1$, Jensen's inequality, and
        \Cref{clm:radv-psd-decomposition}. Using the same argument with $J$ and $H^{-}$ gives
        $\cF(\widehat{J}\circ\kbr{\mu},J\circ\kbr{\mu})\geq1/\sqrt{1+\delta/2}$, completing the
        proof.
    \end{proof}

We can now present the proof of \Cref{prop:radv-upperbound} as follows.
\begin{proof}[Proof of \Cref{prop:radv-upperbound}]
    Let $\Psi\in\cR^{\mu}_{\delta}(f)$ be a Hermitian matrix attaining the minimum
    $\mathrm{rAdv}^{\mu}_{\delta}(f)= \bar{\gamma}_{2}^{\mu}(J-\Psi|\{\Delta_i\})$, and let
    $\widehat{J}$ and $\widehat{\Psi}$ be the Gram matrices given by the PSD repair lemma. Set
    $\alpha:=1-\eta^2\delta/32$. Applying \Cref{prop:distalg-from-avggamma2} to $\widehat{J}$ and
    $\widehat{\Psi}$ with fidelity parameter $\alpha$, we obtain a  Gram matrix $\widetilde{\Psi}$
    where $ \cF(\widetilde{\Psi}\circ\kbr{\mu}, \widehat{\Psi}\circ\kbr{\mu})\geq\alpha,
    $
     and
    \begin{align*}
        T := \mathsf{Q}(\widehat{J}\to \widetilde{\Psi})
        &\leq \frac{2}{1-\alpha}\bar{\gamma}_2^{\mu}(\widehat{J}-\widehat{\Psi}
        |\{\Delta_i\})
        \leq \frac{64\mathrm{rAdv}^{\mu}_{\delta}(f)}{\eta^2\delta}.
    \end{align*}
    We now run the same query algorithm starting from Gram matrix $J$, and let
$\doubletilde{\Psi}$ denote its final Gram matrix. By the data-processing inequality for fidelity
and the fidelity bound for $\hat{J}$ in \Cref{lem:radv-psd-repair},
    \begin{align*}
        \cF(\doubletilde{\Psi}\circ\kbr{\mu},\widetilde{\Psi}\circ\kbr{\mu})
        \geq \cF(J\circ\kbr{\mu},\widehat{J}\circ\kbr{\mu})
        \geq \frac{1}{\sqrt{1+\delta/2}},
    \end{align*}
    Finally, we apply the triangle inequality for the purified distance $\sqrt{1-\cF^2}$ to obtain
    \begin{align*}
        &\sqrt{1-\cF^2(\doubletilde{\Psi}\circ\kbr{\mu},A_f\circ\kbr{\mu})}\\
        \leq &\sqrt{1-\cF^2(\doubletilde{\Psi}\circ\kbr{\mu},\widetilde{\Psi}\circ\kbr{\mu})}
        +\sqrt{1-\cF^2(\widetilde{\Psi}\circ\kbr{\mu},\widehat{\Psi}\circ\kbr{\mu})}
        +\sqrt{1-\cF^2(\widehat{\Psi}\circ\kbr{\mu},A_f\circ\kbr{\mu})}\\
        \leq &2\sqrt{\frac{\delta}{2+\delta}}+\sqrt{1-\alpha^2}
        \leq\sqrt{2\delta}+\sqrt{1-\alpha^2}
        \leq\sqrt{2\delta}+\sqrt{2(1-\alpha)}
        =\sqrt{\delta}\left(\sqrt{2}+\frac{\eta}{4}\right)\\
        \leq &\sqrt{(2+\eta)\delta},
    \end{align*}
    where we used $\delta>0$, $\alpha\leq1$, and $\eta\leq1$, respectively. Consequently,
    $\cF(\doubletilde{\Psi}\circ\kbr{\mu},A_f\circ\kbr{\mu})\geq \sqrt{1-(2+\eta)\delta}$. Since
    $A_f$ is a valid witness for $f$, $\doubletilde{\Psi} \in \cG^{\mu}_{(2+\eta)\delta}(f)$, which
    completes the proof.
\end{proof}

As a corollary, we get the following.
\begin{corollary}\label{cor:radv-query-comparison} Let $\Sigma \subset \{0,1\}^{n}$, $f:\Sigma \to
    \Omega$ be a functional problem, and $\mu:\Sigma \to [0,1]$ be a distribution. Also, let
    $\eps=\eps(n)\in(0,1/2)$ be an error parameter. Then, for $\eps\leq1/16$, we have
    \begin{align*}
        \Omega\left(\mathrm{rAdv}^{\mu}_{4\sqrt{\eps}}(f)\right)
        \leq \mathsf{Q}^{\mu}_{\eps}(f)
        \leq O\left(\frac{\mathrm{rAdv}^{\mu}_{\eps/3}(f)}{\eps}\right).
    \end{align*}
    Moreover, if $f:\Sigma \to \{0,1\}$ is a decision problem, then for every
    $\eps=\eps(n)\in(0,1/2)$, we have
    \begin{align*}
        \Omega\left(\mathrm{rAdv}^{\mu}_{2\sqrt{\eps(1-\eps)}}(f)\right)
        \leq \mathsf{Q}^{\mu}_{\eps}(f)
        \leq O\left(\frac{\mathrm{rAdv}^{\mu}_{\eps/3}(f)}{\eps}\right).
    \end{align*}
\end{corollary}

\subsection{Shaltiel-free distributional adversary}\label{sec:shaltiel-free-adversary}

In this subsection, we define the Shaltiel-free distributional adversary $\mathrm{sfAdv}$ as a
measure for an inner function $\bg$ for lower bounding the composition complexity of $f\circ \bg$.
The terminology and motivation follow the classical Shaltiel-free measure of Ben-David and Blais
\cite{BB20a}, which accounts for Shaltiel's example \cite{Sha03}.

When computing a single value $g(x)$, the bias in the algorithm’s belief about $g(x)$ directly
determines its error. In composition, however, small biases about the inner outputs may suffice to
obtain a strong belief about the final answer. For instance, consider the extreme case of the
$\mathrm{TRIV}$ function, defined by $\mathrm{TRIV}(0^n)=0$ and $\mathrm{TRIV}(1^n)=1$. Then, small
biases in the algorithm’s beliefs about $g(x^1),\dots,g(x^n)$ across sufficiently many independent
inputs can accumulate into a strong belief about whether $\bg(\bx)=0^n$ or $1^n$.
\par
Consequently, when proving a composition lower bound, we need to be careful about the tradeoff
between the cost of an inner computation and the bias it provides. Intuitively, the following
measure reflects the best tradeoff that an algorithm can achieve between the cost and the bias it
provides.

\begin{definition}\label{def:sfadv} Let $\Sigma\subset \{0,1\}^{n}$ and $g:\Sigma\to \{0,1\}$ be a
    decision problem. Let $\nu_0,\nu_1:\Sigma\to [0,1]$ be a pair of distributions supported on the
    $0$-inputs $g^{-1}(0)$ and the $1$-inputs $g^{-1}(1)$, respectively, and let $\nu =
    \frac{1}{2}(\nu_0+\nu_1)$. We define the Shaltiel-free adversary measure for $g$ with respect to
    $(\nu_0,\nu_1)$ as
    \begin{align*}
        \mathrm{sfAdv}^{(\nu_0,\nu_1)}(g) :=\inf_{\delta\in [0,1)}
        \frac{\mathrm{rAdv}^{\nu}_{\delta}(g)}{1-\delta},
    \end{align*}
    where $\mathrm{rAdv}$ is the relaxed adversary measure in \Cref{def:radv}.
\end{definition}

\subsubsection{Monotonicity of the ratio}
In the following, we show that the ratio in \Cref{def:sfadv} is non-increasing as $\delta \to 1$.

\begin{lemma}\label{lemma:sfadv-monotonicity} Let $\Sigma \subset \{0,1\}^{m}$, $g:\Sigma\to
    \{0,1\}$ be a decision problem, and $\nu_0,\nu_1:\Sigma\to [0,1]$ be a pair of distributions
    supported on $g^{-1}(0)$ and $g^{-1}(1)$, respectively. Write $\nu=(\nu_0+\nu_1)/2$. Then, for
    every $0\leq \alpha < \beta <1$, we have
    \begin{align*}
        \frac{\mathrm{rAdv}^{\nu}_{\alpha}(g)}{1-\alpha} \geq \frac{\mathrm{rAdv}^{\nu}_{\beta}(g)}{1-\beta}.
    \end{align*}
\end{lemma}
\begin{proof}
    Let $\Psi\in\cR^{\nu}_{\alpha}(g)$ attain the minimum defining
    $\mathrm{rAdv}^{\nu}_{\alpha}(g)$, so that $\Psi \circ A_g = A_g$, $\|\Psi \circ \Delta_g \circ
    \kbr{\nu}\|_{\tr}\leq \alpha$, and $\mathrm{rAdv}^{\nu}_{\alpha}(g) =
    \bar{\gamma}_2^{\nu}(J-\Psi|\{\Delta_i\})$. We define $c:=\frac{1-\beta}{1-\alpha}\in(0,1)$ and
    $\Psi' := c\Psi + (1-c)J$. Then, $\Psi' \circ A_g = A_g$ and
    \begin{align*}
        \|\Psi' \circ \Delta_g \circ \kbr{\nu}\|_{\tr} &\leq c\|\Psi \circ \Delta_g \circ \kbr{\nu}\|_{\tr}
        + (1-c)\|\Delta_g\circ \kbr{\nu}\|_{\tr}  \\
        & \leq c \alpha + (1-c) \gamma_2(\Delta_g) \|\kbr{\nu}\|_{\tr} \\
        &\leq \alpha c + 1-c = \beta,
    \end{align*}
    where $\gamma_2(\Delta_g)\leq 1$ as $g$ is a decision problem. Therefore,
    \begin{align*}
        \mathrm{rAdv}^{\nu}_{\beta}(g) &\leq
        \bar{\gamma}_2^{\nu}(J-\Psi' |\{\Delta_i\}) = \bar{\gamma}_2^{\nu}(c (J - \Psi) |\{\Delta_i\})
        = \frac{1-\beta}{1-\alpha}  \mathrm{rAdv}^{\nu}_{\alpha}(g). \qedhere
    \end{align*}
\end{proof}

\subsubsection{Equivalent primal and dual formulations}
We also provide a dual formulation for $\mathrm{sfAdv}$ which yields a universal certificate of
hardness.
\begin{lemma}\label{lem:dual-sfadv-hardness} Let $\Sigma \subset \{0,1\}^{m}$ and $g:\Sigma\to
 \{0,1\}$ be a decision problem. Also, let $\nu_0,\nu_1:\Sigma\to [0,1]$ be a pair of distributions
 supported on $g^{-1}(0)$ and $g^{-1}(1)$, respectively. Then,
    \begin{align*}
        \mathrm{sfAdv}^{(\nu_0,\nu_1)}(g) = &\progline
        {\max_{\Gamma=\Gamma^\dagger}}
        {\frac{1}{\max_i\|\Gamma\circ\Delta_i\|}}
        {&\bra{\nu_0}\Gamma\ket{\nu_1}=\|\Gamma\|=1;\quad \Gamma\circ\Delta_g=\Gamma},
    \end{align*}
    where $\ket{\nu_b} = \sum_{x \in g^{-1}(b)} \sqrt{\nu_b(x)} \ket{x} \in \bR^{\Sigma}$ is the
    distribution vector supported on $g^{-1}(b)$, with $\|\nu_b\| = 1$. In particular, the value is
    finite and the maximum is attained by a real symmetric matrix $\Gamma$.
\end{lemma}
\begin{proof}
    Let $\ket{\nu}:=(\ket{\nu_0}+\ket{\nu_1})/\sqrt{2}$. Using the $\mathrm{rAdv}^{\nu}_{\delta}$
    characterization in \Cref{lem:radv-universal-hardness-cert}, we have
    \begin{align*}
        \begin{alignedat}{2}
        \mathrm{sfAdv}^{(\nu_0,\nu_1)}(g) = &\inf_{\delta\in[0,1)}\max_{\Gamma=\Gamma^\dagger}
        \frac{\langle\Gamma,\kbr{\nu}\rangle-\delta\|\Gamma\|}{1-\delta}
        &\quad\mathrm{s.t.}\quad &\Gamma\circ\Delta_g=\Gamma;\quad \max_i\|\Gamma\circ\Delta_i\|\leq1\\[0.6em]
        = &\inf_{\delta\in[0,1)}\max_{\Gamma=\Gamma^\dagger}
        \frac{\langle\Gamma,\kbr{\nu}\rangle-\|\Gamma\|}{1-\delta}+\|\Gamma\|
        &\quad\mathrm{s.t.}\quad &\Gamma\circ\Delta_g=\Gamma;\quad \max_i\|\Gamma\circ\Delta_i\|\leq1\\[0.6em]
        = &\inf_{r\in[1,\infty)}\max_{\Gamma=\Gamma^\dagger}
        r\bigl(\langle\Gamma,\kbr{\nu}\rangle-\|\Gamma\|\bigr)+\|\Gamma\|
        &\quad\mathrm{s.t.}\quad &\Gamma\circ\Delta_g=\Gamma;\quad \max_i\|\Gamma\circ\Delta_i\|\leq1,
        \end{alignedat}
    \end{align*}
    where we applied a change of variable $r:=1/(1-\delta)$. The feasible set of $\Gamma$ is compact
    and convex: it is closed, and every entry that can be nonzero is bounded by one of the filter
    constraints. Moreover, $[1,\infty)$ is convex, and the objective is continuous, concave in
    $\Gamma$ for $r\geq1$, and affine in $r$. Hence, by Sion's minimax theorem
    (\Cref{thm:sion-minimax}),
    \begin{align*}
        \mathrm{sfAdv}^{(\nu_0,\nu_1)}(g) = &\progline
        {\max_{\Gamma=\Gamma^\dagger}\inf_{r\in[1,\infty)}}
        {r\bigl(\langle\Gamma,\kbr{\nu}\rangle-\|\Gamma\|\bigr)+\|\Gamma\|}
        {&\Gamma\circ\Delta_g=\Gamma;\quad \max_i\|\Gamma\circ\Delta_i\|\leq1}
    \end{align*}
    Since $\|\nu\|=1$, $\langle \Gamma,\kbr{\nu}\rangle - \|\Gamma\| \leq 0$. Hence, if the
    inequality is strict, the infimum will be $-\infty$. So we may assume $\|\Gamma\| = \langle
    \Gamma,\kbr{\nu}\rangle$. Opening the inner product with respect to $\ket{\nu_0}$ and
    $\ket{\nu_1}$ we get
    \begin{align*}
        \langle \Gamma, \kbr{\nu}\rangle = \frac{1}{2}(\bra{\nu_0}+\bra{\nu_1})\Gamma (\ket{\nu_0}+\ket{\nu_1}) = \frac{1}{2} (\bra{\nu_0} \Gamma \ket{\nu_1} + \bra{\nu_1}\Gamma \ket{\nu_0}) = \Re \bra{\nu_0} \Gamma \ket{\nu_1},
    \end{align*}
    where we used $\bra{\nu_b}\Gamma\ket{\nu_b}=0$ for $b\in\{0,1\}$, as
    $\Gamma\circ\Delta_g=\Gamma$. Since $|\bra{\nu_0}\Gamma\ket{\nu_1}|\leq\|\Gamma\|$, the equality
    $\|\Gamma\| = \langle \Gamma,\kbr{\nu}\rangle$ holds if and only if
    $\bra{\nu_0}\Gamma\ket{\nu_1}=\|\Gamma\|$. Consequently, we can rewrite the program as
    \begin{align*}
        \mathrm{sfAdv}^{(\nu_0,\nu_1)}(g) = &\progline
        {\max_{\Gamma=\Gamma^\dagger}}
        {\|\Gamma\|}
        {&\Gamma\circ\Delta_g=\Gamma;\quad \max_i\|\Gamma\circ\Delta_i\|\leq1; \\
         &\bra{\nu_0}\Gamma\ket{\nu_1}=\|\Gamma\|}.
    \end{align*}
    This feasible set is closed and contained in the compact set above, so the maximum is finite and
    attained.
Next, we use homogeneity to obtain the quotient formulation. Note that $\Gamma_0 :=
\ketbra{\nu_0}{\nu_1}+\ketbra{\nu_1}{\nu_0}$ is feasible with norm $\|\Gamma_0\|=1$, so we may
restrict the problem to $\Gamma\neq 0$. Then $\langle\nu_0|\Gamma|\nu_1\rangle=\|\Gamma\|>0$ implies
that $\Gamma(x,y)\neq0$ for some $x,y$ with $g(x)\neq g(y)$, and hence
$\max_i\|\Gamma\circ\Delta_i\|>0$. Write $M:=\max_i\|\Gamma\circ\Delta_i\|$. Any nonzero $\Gamma$
feasible for the above formulation has $M\leq1$, so $\|\Gamma\|/M\geq\|\Gamma\|$. Conversely, any
$\Gamma$ feasible for the quotient formulation can be normalized to $\Gamma':=\Gamma/M$, which
satisfies $\max_i\|\Gamma'\circ\Delta_i\|=1$ and all the remaining constraints, with
$\|\Gamma'\|=\|\Gamma\|/M$. Therefore, the two optimal values agree, giving
\begin{align*}
        \mathrm{sfAdv}^{(\nu_0,\nu_1)}(g) = &\progline
        {\max_{\Gamma = \Gamma^{\dagger}\neq 0}}
        {\frac{\|\Gamma\|}{\max_{i}\|\Gamma \circ \Delta_i\|}}
        {&\bra{\nu_0}\Gamma\ket{\nu_1}=\|\Gamma\|;\quad \Gamma\circ\Delta_g=\Gamma}.
    \end{align*}
 This quotient is unchanged by scaling, so we may normalize
 $\langle\nu_0|\Gamma|\nu_1\rangle=\|\Gamma\|=1$, obtaining the stated form.
\par
Finally, to obtain a real symmetric maximizer, replace $\Gamma$ by $\Gamma' =
(\Gamma+\overline{\Gamma})/2$. This preserves $\Gamma'\circ\Delta_g=\Gamma'$ and
$\bra{\nu_0}\Gamma'\ket{\nu_1}=1$, since $\ket{\nu_0}$ and $\ket{\nu_1}$ are real. Moreover, the
operator norm and any filtered norms do not increase by the triangle inequality. The equality
$\bra{\nu_0}\Gamma'\ket{\nu_1}=1$ also forces the operator norm to be exactly $\|\Gamma'\| = 1$, so
the replacement is feasible and its objective value cannot decrease.
\end{proof}

\section{A composition theorem via the Shaltiel-free adversary}\label{sec:composition} In this
section, we present our composition theorems for computing $f\circ g^{\otimes n}$. More
specifically, in \Cref{sec:composition-adversary}, we prove composition theorems for our adversary
method and its relaxation, using the Shaltiel-free measure for the inner function. Then, in
\Cref{sec:composition-quantum}, we recover bounds on distributional quantum query complexity and
show that, in the limit, $\mathrm{sfAdv}$ is the optimal inner factor up to constant factors. We use
the notation in \Cref{sec:prelim-composition}.

\subsection{Composition for distributional adversary methods}\label{sec:composition-adversary}
\begin{theorem}\label{thm:composition-adversary} Let $\Sigma\subset \{0,1\}^{n}$,
    $f:\Sigma\to\Omega$ be a functional problem, and $\mu:\Sigma\to[0,1]$ be a distribution. Let
    $\Sigma_g\subset \{0,1\}^{m}$, $g:\Sigma_g\to\{0,1\}$ be a decision problem, and
    $\nu_0,\nu_1:\Sigma_g\to[0,1]$ be distributions supported on $g^{-1}(0)$ and $g^{-1}(1)$,
    respectively. Then, the following hold.
    \begin{itemize}
        \item For an error parameter $\eps=\eps(n)\in(0,1/2)$,
        \begin{align*}
            \mathrm{Adv}^{\mu\circ(\nu_0,\nu_1)}_{\eps}(f\circ\bg)
            &\geq\mathrm{Adv}^{\mu}_{\eps}(f)\cdot\mathrm{sfAdv}^{(\nu_0,\nu_1)}(g).
        \end{align*}
        \item For $\delta=\delta(n)\in(0,1)$,
        \begin{align*}
            \mathrm{rAdv}^{\mu\circ(\nu_0,\nu_1)}_{\delta}(f\circ\bg)
            &\geq\mathrm{rAdv}^{\mu}_{\delta}(f)\cdot\mathrm{sfAdv}^{(\nu_0,\nu_1)}(g).
        \end{align*}
        \item For every fixed $\delta\in(0,1)$,
        \begin{align*}
            \mathrm{sfAdv}^{(\nu_0,\nu_1)}(g)
            &=\lim_{k\to\infty}\frac{\mathrm{rAdv}^{\eta_k}_{\delta}
            (\mathrm{TRIV}_k\circ g^{\otimes k})}{1-\delta},
        \end{align*}
        where $\mathrm{TRIV}_k:\{0^k,1^k\}\to\{0,1\}$ is defined by $\mathrm{TRIV}_k(b^k)=b$ and
        $\eta_k=(\nu_0^{\otimes k}+\nu_1^{\otimes k})/2$.
    \end{itemize}
\end{theorem}

\paragraph{Proof overview.}
Our proof strategy is in essence a dual composition proof, with a lifting construction related to
those in \cite{HLS07,BL20}. We prove both inequalities of \Cref{thm:composition-adversary} through
the following steps.
\begin{itemize}
    \item \textbf{The lifting map and its adjoint.} We use a hardness certificate $\Gamma_g$ for
    $\mathrm{sfAdv}$ as in \Cref{lem:dual-sfadv-hardness} to define a lifting map
    $\cE^{\uparrow}_{\Gamma_g}$ and its adjoint $\cE^{\downarrow}_{\Gamma_g}$. The lifting map sends
    a hardness certificate $\Gamma_f$ of $f$ (as in \Cref{lem:adv-universal-hardness-cert} for
    $\mathrm{Adv}$ or \Cref{lem:radv-universal-hardness-cert} for $\mathrm{rAdv}$) to a candidate
    certificate $\cE^{\uparrow}_{\Gamma_g}(\Gamma_f)$ for the composed function $f\circ\bg$.

    \item \textbf{Preservation of valid targets.} We show that the adjoint preserves the weighted
    initial matrix and valid targets:
    $\cE^{\downarrow}_{\Gamma_g}(J\circ\kbr{\eta})=J\circ\kbr{\mu}$ and
    $\cE^{\downarrow}_{\Gamma_g}(\Lambda\circ\kbr{\eta})=\Phi\circ\kbr{\mu}$, where
    $\Phi\in\cG^{\mu}_{\eps}(f)$ for $\mathrm{Adv}$ and $\Phi\in\cR^{\mu}_{\delta}(f)$ for
    $\mathrm{rAdv}$.

    \item \textbf{Bounding filtered norms.} We bound the filtered norms
    $\|\cE^{\uparrow}_{\Gamma_g}(\Gamma_f)\circ\Delta_{ij}\|$ in terms of
    $\|\Gamma_f\circ\Delta_i\|$ and $\|\Gamma_g\circ\Delta_j\|$, so that the lifted certificate can
    be normalized for the composed problem.

    \item \textbf{Proof of the composition theorem.} Finally, we use adjointness to write
    \begin{align*}
        &\langle \cE^{\uparrow}_{\Gamma_g}(\Gamma_f),(J-\Lambda)\circ\kbr{\eta}\rangle =\langle\Gamma_f,\cE^{\downarrow}_{\Gamma_g}(J\circ\kbr{\eta})
        -\cE^{\downarrow}_{\Gamma_g}(\Lambda\circ\kbr{\eta})\rangle.
    \end{align*}
    The preservation properties and the filtered-norm bounds then give both composition
    inequalities.
\end{itemize}
These four steps are carried out in \Cref{sec:composition-lifting},
\Cref{sec:composition-preservation}, \Cref{sec:composition-filtered-norms}, and
\Cref{sec:composition-proof}, respectively. The tightness result for $\mathrm{TRIV}$ is proven in
\Cref{sec:composition-triv}.

\subsubsection{The lifting map and its adjoint}\label{sec:composition-lifting} In this part, we
define our lifting map and obtain its adjoint as well as a few useful properties. We begin by
encoding the action of the function $x\to g(x)$ as a matrix mapping $\ket{x} \to \ket{g(x)}$.
\begin{definition}\label{def:composition-G} Let $g:\Sigma_g\to\{0,1\}$ be a decision problem, where
    $\Sigma_g\subset\{0,1\}^{m}$. We define the matrix $G$ so that, for $x\in\Sigma_g$, it outputs
    $G\ket{x}=\ket{g(x)}$; i.e.,
    \begin{align*}
        G &:= \sum_{x\in\Sigma_g}\ketbra{g(x)}{x}\in\cL(\Sigma_g,\{0,1\}).
    \end{align*}
    Then, the map $\bG:=G^{\otimes n}\in\cL(\Sigma_g^n,\{0,1\}^n)$ maps
    $\bG\ket{\bx}=\ket{\bg(\bx)}$ for $\bx\in\Sigma_g^n$. Specifically, $\bG$ can be used to lift a
    matrix $B\in\cL(\{0,1\}^n)$ to $\bG^\dagger B\bG\in\cL(\Sigma_g^n)$.
\end{definition}
Expanding the above sum can be visually distracting.  So, to streamline future calculations, we
collect a few useful identities for $\bG$ in the following lemma. The proofs are straightforward
calculations and, for the sake of brevity, are postponed to \Cref{app:G-properties}.
\begin{lemma}\label{lem:G-properties} Let $g$ and $\bG$ be as in \Cref{def:composition-G},
    $X,X_1,X_2\in\cL(\{0,1\}^n)$, and $Y\in\cL(\Sigma_g^n)$. The following hold.
    \begin{enumerate}[(i)]
        \item\label{item:G-entries} For $\bx,\by\in\Sigma_g^n$, we have $(\bG^\dagger
        X\bG)(\bx,\by)=X(\bg(\bx),\bg(\by))$.
        \item\label{item:G-hadamard-lift} We have $\bG^\dagger(X_1\circ X_2)\bG =(\bG^\dagger
        X_1\bG)\circ(\bG^\dagger X_2\bG)$.
        \item\label{item:G-hadamard-pull-through} We have $\bG(Y\circ(\bG^\dagger X\bG))\bG^\dagger
        =X\circ(\bG Y\bG^\dagger)$.
    \end{enumerate}
\end{lemma}
We are now ready to define our lifting map with respect to an adversary matrix $\Gamma$ for the
inner decision problem $g$.
\begin{definition}[Lifting map]\label{def:composition-lifting-map} Let $g$ and $\bG$ be as in
    \Cref{def:composition-G}, and let $\Gamma\in\cL(\Sigma_g)$ be real symmetric with
    $\|\Gamma\|\leq1$ and $\Gamma\circ\Delta_g=\Gamma$. We define the lifting map
    $\cE^{\uparrow}_{\Gamma}:\cL(\{0,1\}^{n})\to\cL(\Sigma_g^n)$ by
    \[
        \cE^{\uparrow}_{\Gamma}(X):=\bG^\dagger X\bG\circ(\Gamma+I)^{\otimes n}.
    \]
\end{definition}
To get some intuition, we can imagine that the adversary $\Gamma=\Gamma\circ\Delta_g$ has a
predetermined query budget $\max_{j\in[m]}\|\Gamma\circ\Delta_j\|\leq1$, which it spends according
to how much it cares about the algorithm producing distinguishable states on each pair of inner
inputs. Given two composed inputs $\bx=(x^1,\dots,x^n)$ and $\by=(y^1,\dots,y^n)$, we use the same
weight $\Gamma(x^i,y^i)$ whenever $g(x^i)\neq g(y^i)$. When the outputs agree, however,
distinguishing $x^i$ from $y^i$ does not change the corresponding input to the outer function. But
we cannot assign weight zero, since the weights multiply across coordinates and this would remove
the entire comparison between $\bx$ and $\by$. Instead, we use the identity weight $I$, which lets
identical inner inputs contribute weight one without spending any query budget on distinguishing
them.
\par
The following lemma collects a few useful properties of our map and its adjoint, including their
interaction with Hadamard products and weighted matrices.
\begin{lemma}\label{lem:composition-channel-properties} Let $g$, $\Gamma$, and
    $\cE^{\uparrow}_{\Gamma}$ be as in \Cref{def:composition-lifting-map}. The following hold.
    \begin{enumerate}[(i)]
        \item\label{item:composition-diagonal-lift} For every $X\in\cL(\{0,1\}^n)$, we have
        $\cE^{\uparrow}_{\Gamma}(X\circ I)=\cE^{\uparrow}_{\Gamma}(X)\circ I$. In particular, for
        every $\tx\in\{0,1\}^n$, we have $\cE^{\uparrow}_{\Gamma}(\kbr{\tx})=I_{\bg^{-1}(\tx)}$ and
        $\cE^{\uparrow}_{\Gamma}(I)=I$.
        \item\label{item:composition-channel} The adjoint of $\cE^{\uparrow}_{\Gamma}$ is a quantum
        channel $\cE^{\downarrow}_{\Gamma}:\cL(\Sigma_g^n)\to\cL(\{0,1\}^n)$ given by
        $\cE^{\downarrow}_{\Gamma}(Y)=\bG(Y\circ(\Gamma+I)^{\otimes n})\bG^\dagger$.
        \item\label{item:composition-hadamard} For every $Y\in\cL(\Sigma_g^n)$ and
        $X\in\cL(\{0,1\}^n)$, we have $\cE^{\downarrow}_{\Gamma}(Y)\circ X
        =\cE^{\downarrow}_{\Gamma}(Y\circ(\bG^\dagger X\bG))$.
        \item\label{item:composition-weighted-input} Let $\nu=(\nu_0+\nu_1)/2$, where each $\nu_b$
        is a distribution supported on $g^{-1}(b)$ for $b\in\{0,1\}$, and let
        $\mu:\{0,1\}^n\to[0,1]$ be a distribution. Then, for every $Y\in\cL(\Sigma_g^n)$,
        \[
            \cE^{\downarrow}_{\Gamma}(Y \circ \kbr{\eta})=2^n\cE^{\downarrow}_{\Gamma}(Y\circ \kbr{\nu}^{\otimes n})\circ \kbr{\mu},
        \]
        where $\eta=\mu\circ(\nu_0,\nu_1)$ is the composed distribution.
    \end{enumerate}
\end{lemma}
\begin{proof}
    We prove each item as follows.
\begin{enumerate}[(i)]
    \item For $\bx,\by\in\Sigma_g^n$, by \Crefitem{item:G-entries}, we have
$(\bG^\dagger\bG)(\bx,\by) =\ind[\bg(\bx)=\bg(\by)]$. Thus, $\bG^\dagger\bG=A_g^{\otimes n}$. Since
$\Gamma\circ\Delta_g=\Gamma$, we have $\Gamma\circ A_g=0$ and $\Gamma\circ I=0$. Therefore,
\[
    \bG^\dagger\bG\circ(\Gamma+I)^{\otimes n}
    =(A_g\circ(\Gamma+I))^{\otimes n}=I.
\]
Now let $X\in\cL(\{0,1\}^n)$. Using \Crefitem{item:G-hadamard-lift}, we can write
\begin{align*}
    \cE^{\uparrow}_{\Gamma}(X\circ I)
    = &\bG^\dagger(X\circ I)\bG\circ(\Gamma+I)^{\otimes n}\\
    = &(\bG^\dagger X\bG)\circ(\bG^\dagger I\bG)
       \circ(\Gamma+I)^{\otimes n}\\
    = &(\bG^\dagger X\bG)\circ I\\
    = &\cE^{\uparrow}_{\Gamma}(X)\circ I.
\end{align*}
Taking $X=\kbr{\tx}$ for $\tx\in\{0,1\}^n$ gives
\[
    \cE^{\uparrow}_{\Gamma}(\kbr{\tx})
    =(\bG^\dagger\kbr{\tx}\bG)\circ I
    =I_{\bg^{-1}(\tx)},
\]
where we used $(\bG^\dagger\kbr{\tx}\bG)(\bx,\bx)=\ind[\bg(\bx)=\tx]$. Finally, taking $X=I$ gives
$\cE^{\uparrow}_{\Gamma}(I)=(\bG^\dagger\bG)\circ I=I$.
    \item The adjoint of the Hadamard-product map $X \to X \circ (I+\Gamma)^{\otimes n}$ is the map
$X \to X \circ \overline{(I+\Gamma)^{\otimes n}} = X\circ (I+\Gamma)^{\otimes n}$. Moreover, the
adjoint of the map $X\to \bG^{\dagger} X \bG$ is the map $Y\to \bG Y \bG^{\dagger}$. Composing the
two adjoints we get $\cE^{\downarrow}_{\Gamma}(Y) = \bG (Y \circ (\Gamma+I)^{\otimes n}
)\bG^{\dagger}.$
\par
We next prove that $\cE^{\downarrow}_{\Gamma}$ is a quantum channel. Hadamard product with PSD
 matrices and conjugation $\bG Y\bG^{\dagger}$ are both completely positive maps, and so is their
 composition. Moreover, since $\|\Gamma\|\leq 1$, we have that $I+\Gamma\succeq 0$. As a result,
 $\cE^{\downarrow}_{\Gamma}$ is completely positive. Additionally, the map is trace preserving. To
 see this, we can write
\begin{align*}
    \Tr[\cE^{\downarrow}_{\Gamma}(Y)]
    = &\langle I,\cE^{\downarrow}_{\Gamma}(Y)\rangle
    =\langle\cE^{\uparrow}_{\Gamma}(I),Y\rangle
    =\langle I,Y\rangle=\Tr[Y],
\end{align*}
where we used \Crefitem{item:composition-diagonal-lift}.

    \item We compute
$$\cE^{\downarrow}_{\Gamma}(Y\circ (\bG^{\dagger}X\bG)) = \bG \left(Y\circ (\bG^{\dagger}X\bG) \circ
(\Gamma + I)^{\otimes n}\right) \bG^{\dagger} = \bG \left(Y \circ (\Gamma + I)^{\otimes n}\right)
\bG^{\dagger} \circ X = \cE^{\downarrow}_{\Gamma}(Y)\circ X,$$ where the second equality follows
from \Crefitem{item:G-hadamard-pull-through}.

    \item For $\bx,\by\in \Sigma_g^{n}$, we have
\begin{align*}
    \kbr{\eta}(\bx,\by) &= \sqrt{\eta(\bx)\eta(\by)} = \sqrt{\mu(\bg(\bx))\mu(\bg(\by))}\cdot\prod_{i=1}^{n}\sqrt{\nu_{g(x^i)}(x^i)\nu_{g(y^i)}(y^i)}\\
    &= \sqrt{\mu(\bg(\bx))\mu(\bg(\by))}\cdot\prod_{i=1}^{n}2\sqrt{\nu(x^i)\nu(y^i)}\\
    &= (\bG^\dagger\kbr{\mu}\bG)(\bx,\by)\cdot(2\kbr{\nu})^{\otimes n}(\bx,\by),
\end{align*}
where the second and third lines follow from $\nu_{g(x)}(x)=2\nu(x)$ and \Crefitem{item:G-entries},
respectively. Consequently, $\kbr{\eta} = (\bG^\dagger\kbr{\mu}\bG) \circ(2\kbr{\nu})^{\otimes n}.$
By \Crefitem{item:composition-hadamard} and linearity, we immediately get
\begin{align*}
    \cE^{\downarrow}_{\Gamma}(Y\circ \kbr{\eta}) = &
    \cE^{\downarrow}_{\Gamma}(Y\circ (\bG^\dagger\kbr{\mu}\bG)\circ(2\kbr{\nu})^{\otimes n})
    = 2^n \cdot \cE^{\downarrow}_{\Gamma}(Y\circ \kbr{\nu}^{\otimes n})\circ \kbr{\mu}.\qedhere
\end{align*}
\end{enumerate}
\end{proof}

\subsubsection{Preservation of valid targets}\label{sec:composition-preservation} We now show that
the adjoint map takes valid targets for the composed problem back to valid targets for the outer
problem.
\begin{lemma}\label{lem:comp-adjoint-applications} Let $g$ and $\nu_0,\nu_1$ be as in
    \Cref{thm:composition-adversary}. Let $\Gamma\in\cL(\Sigma_g)$ be real symmetric with
    $\Gamma\circ\Delta_g=\Gamma$ and $\bra{\nu_0}\Gamma\ket{\nu_1}=\|\Gamma\|=1$. For
    $\cE^{\downarrow}_{\Gamma}$ as defined in \Cref{lem:composition-channel-properties}, the
    following hold.
    \begin{enumerate}[(i)]
        \item\label{item:comp-preserves-initial} Let $\mu:\Sigma\to[0,1]$ be a distribution with
        $\Sigma\subseteq\{0,1\}^n$, and set $\eta=\mu\circ(\nu_0,\nu_1)$. Then,
        $\cE^{\downarrow}_{\Gamma}(\kbr{\eta})=\kbr{\mu}$.
        \item\label{item:comp-preserves-gram} Let $f\subseteq\Sigma\times\Omega$ be a search problem
        and $\eps=\eps(n)\in[0,1)$. Then, for $\Psi\in\cG^{\eta}_{\eps}(f\circ\bg)$, we have
        \begin{align*}
            \cE^{\downarrow}_{\Gamma}(\Psi\circ\kbr{\eta}) = &\Psi^{\downarrow}\circ\kbr{\mu},
        \end{align*}
        for some $\Psi^{\downarrow}\in\cG^{\mu}_{\eps}(f)$.
        \item\label{item:comp-preserves-relaxed} Let $f:\Sigma\to\Omega$ be a functional problem and
        $\delta=\delta(n)\in[0,1]$. Then, for $\Psi\in\cR^{\eta}_{\delta}(f\circ\bg)$, we have
        \begin{align*}
            \cE^{\downarrow}_{\Gamma}(\Psi\circ\kbr{\eta}) = &\Psi^{\downarrow}\circ\kbr{\mu},
        \end{align*}
        for some $\Psi^{\downarrow}\in\cR^{\mu}_{\delta}(f)$.
    \end{enumerate}
    Here and in the proof, matrices indexed by $\Sigma$ and $\bg^{-1}(\Sigma)$ are implicitly
    extended by zero to $\{0,1\}^n$ and $\Sigma_g^n$, respectively, when applying the maps, and the
    resulting target matrices are restricted back to $\Sigma$.
\end{lemma}
\begin{proof}
    Let $\nu=(\nu_0+\nu_1)/2$. We prove each item separately as follows.
    \begin{enumerate}[(i)]
    \item \textbf{Channel preserves the initial state.}

    By \Crefitem{item:composition-weighted-input}, we can write
    \begin{align*}
        \cE^{\downarrow}_{\Gamma}(\kbr{\eta})
        = & 2^n\cE^{\downarrow}_{\Gamma}(\kbr{\nu}^{\otimes n})\circ\kbr{\mu}.
    \end{align*}
    We claim that
    \begin{equation}\label{eq:composition-balanced-input}
        2^n\cE^{\downarrow}_{\Gamma}(\kbr{\nu}^{\otimes n})=J_2^{\otimes n},
    \end{equation}
    which gives $\cE^{\downarrow}_{\Gamma}(\kbr{\eta}) =J_2^{\otimes n}\circ\kbr{\mu}=\kbr{\mu}$. By
    the definition of $\cE^{\downarrow}_{\Gamma}$ in \Crefitem{item:composition-channel},
    \begin{align*}
        2^n\cE^{\downarrow}_{\Gamma}(\kbr{\nu}^{\otimes n})
        = & 2^n\bG\bigl((\Gamma+I)^{\otimes n}\circ\kbr{\nu}^{\otimes n}\bigr)\bG^\dagger
        =\bigl(2G\bigl((\Gamma+I)\circ\kbr{\nu}\bigr)G^\dagger\bigr)^{\otimes n}.
    \end{align*}
    Therefore, it suffices to show that $2G\bigl((\Gamma+I)\circ\kbr{\nu}\bigr)G^\dagger=J_2$. To
    prove this, we partition the inputs as $\Sigma_g=g^{-1}(0)\sqcup g^{-1}(1)$. Writing $\one_b$
    for the all-ones vector on $g^{-1}(b)$, we have
\begin{align*}
    G &=
    \begin{bmatrix}
        \one_0^\top & 0 \\
        0 & \one_1^\top
    \end{bmatrix},
    &
    \Gamma &=
    \begin{bmatrix}
        0 & \Gamma_0 \\
        \Gamma_0^\dagger & 0
    \end{bmatrix},
\end{align*}
where the diagonal blocks of $\Gamma$ vanish since $\Gamma\circ\Delta_g=\Gamma$. Within these
blocks, we identify $\ket{\nu_b}$ with its restriction to $g^{-1}(b)$. Since
$\ket{\nu}=(\ket{\nu_0}+\ket{\nu_1})/\sqrt{2}$,
\begin{align*}
    (\Gamma+I)\circ(2\kbr{\nu})
    &=
    \begin{bmatrix}
        I\circ\kbr{\nu_0}
        & \Gamma_0\circ\ketbra{\nu_0}{\nu_1} \\
        \Gamma_0^\dagger\circ\ketbra{\nu_1}{\nu_0}
        & I\circ\kbr{\nu_1}
    \end{bmatrix}.
\end{align*}
Conjugating by $G$ sums the entries of each block. Thus,
\begin{align*}
    G\bigl((\Gamma+I)\circ(2\kbr{\nu})\bigr)G^\dagger
    &=
    \begin{bmatrix}
        \braket{\nu_0}{\nu_0}
        & \bra{\nu_0}\Gamma_0\ket{\nu_1} \\
        \bra{\nu_1}\Gamma_0^\dagger\ket{\nu_0}
        & \braket{\nu_1}{\nu_1}
    \end{bmatrix}
    =
    \begin{bmatrix}
        1 & 1 \\
        1 & 1
    \end{bmatrix}
    =J_2,
\end{align*}
where we used normalization of $\nu_0,\nu_1$ and $\bra{\nu_0}\Gamma\ket{\nu_1}=1$.
\item \textbf{Channel preserves valid Gram matrices
$\cG^{\eta}_{\eps}(f\circ\bg)\to\cG^{\mu}_{\eps}(f)$.} Let $\Psi\in\cG^{\eta}_{\eps}(f\circ\bg)$ be
given with corresponding witness $W\in\cW^{\eta}_{f\circ\bg}$ such that $\Tr[W\circ\kbr{\eta}]=1$
and $\cF(\Psi\circ\kbr{\eta},W\circ\kbr{\eta})\geq\sqrt{1-\eps}$. We define
$\Psi^{\downarrow}:=2^n\cE^{\downarrow}_{\Gamma}(\Psi\circ\kbr{\nu}^{\otimes n})$ and
$W^{\downarrow}:=2^n\cE^{\downarrow}_{\Gamma}(W\circ\kbr{\nu}^{\otimes n})$. Therefore, by
\Crefitem{item:composition-weighted-input}, we have
$\cE^{\downarrow}_{\Gamma}(\Psi\circ\kbr{\eta})=\Psi^{\downarrow}\circ\kbr{\mu}$ and
$\cE^{\downarrow}_{\Gamma}(W\circ\kbr{\eta})=W^{\downarrow}\circ\kbr{\mu}$. We claim that
$\Psi^{\downarrow}\in\cG^{\mu}_{\eps}(f)$.
    \begin{itemize}
        \item \textbf{Fidelity.} By monotonicity of fidelity under the quantum channel
        $\cE^{\downarrow}_{\Gamma}$,
        \begin{align*}
            \cF(\Psi^{\downarrow}\circ\kbr{\mu},W^{\downarrow}\circ\kbr{\mu})
            &=\cF\bigl(\cE^{\downarrow}_{\Gamma}(\Psi\circ\kbr{\eta}),
                      \cE^{\downarrow}_{\Gamma}(W\circ\kbr{\eta})\bigr)
            \geq\cF(\Psi\circ\kbr{\eta},W\circ\kbr{\eta})\\
            &\geq\sqrt{1-\eps}.
        \end{align*}
        \item \textbf{Gram matrixness.} Since $\cE^{\downarrow}_{\Gamma}$ is a quantum channel
        (\Crefitem{item:composition-channel}), it preserves positivity. Therefore,
        $\Psi^{\downarrow} = 2^n\cE^{\downarrow}_{\Gamma}(\Psi \circ \kbr{\nu}^{\otimes n}) \succeq
        0$. Let $\tx\in \Sigma$ be a given index. We prove that $\Psi^{\downarrow}(\tx,\tx) = 1$.
        Recall that, by definition of $\cG^{\eta}_{\eps}(f\circ\bg)$, $\Psi\circ I=I$ on the domain
        $\bg^{-1}(\Sigma)$ of $f\circ\bg$. Therefore, by \Crefitem{item:composition-diagonal-lift},
        $\cE^{\uparrow}_{\Gamma}(\kbr{\tx}) \circ \overline{\Psi} = I_{\bg^{-1}(\tx)}\circ
        \overline{\Psi} = \cE^{\uparrow}_{\Gamma}(\kbr{\tx})$. So,
        \begin{align*}
            \Psi^{\downarrow}(\tx,\tx)
            = & 2^n\langle \kbr{\tx}, \cE^{\downarrow}_{\Gamma}(\Psi\circ \kbr{\nu}^{\otimes n})\rangle\\
            = & 2^n\langle \cE^{\uparrow}_{\Gamma}(\kbr{\tx}), \Psi\circ \kbr{\nu}^{\otimes n}\rangle\\
            = & 2^n\langle \cE^{\uparrow}_{\Gamma}(\kbr{\tx}) \circ \overline{\Psi}, \kbr{\nu}^{\otimes n} \rangle\\
            = & 2^n\langle \cE^{\uparrow}_{\Gamma}(\kbr{\tx}),\kbr{\nu}^{\otimes n} \rangle\\
            = & 2^n\langle \kbr{\tx},\cE^{\downarrow}_{\Gamma}(\kbr{\nu}^{\otimes n})\rangle\\
            = & \bra{\tx}J_2^{\otimes n}\ket{\tx} = 1,
        \end{align*}
        where in the last line we used \Cref{eq:composition-balanced-input}.
        \item \textbf{Witness properties.} We have to show that $W^{\downarrow} \in \cW^{\mu}_f$.
        First of all, by trace preservation of $\cE^{\downarrow}_{\Gamma}$, $\Tr[W^{\downarrow}\circ
        \kbr{\mu}] = \Tr[\cE^{\downarrow}_{\Gamma}(W\circ \kbr{\eta})] =1$.
        \par
        Next, let $W = \sum_{s\in\Omega} W^{(s)}$ be the decomposition from \Cref{def:witness-set},
        where $W^{(s)}\succeq 0$ and $W^{(s)}(\bx,\bx) = 0$ whenever $s\notin (f\circ\bg)(\bx)$. We
        define $W^{(s)\downarrow} := 2^n\cE^{\downarrow}_{\Gamma}(W^{(s)}\circ\kbr{\nu}^{\otimes
        n})$. Consequently, by linearity of channels,
        \[
            W^{\downarrow} = 2^n\cE^{\downarrow}_{\Gamma}\Bigl(\Bigl(\sum_{s\in\Omega}W^{(s)}\Bigr)\circ\kbr{\nu}^{\otimes n}\Bigr)
            = \sum_{s\in\Omega} W^{(s)\downarrow},
        \]
        and each $W^{(s)\downarrow}\succeq0$ by positivity. Moreover, for $s\notin f(\tx)$, we have
        $I_{\bg^{-1}(\tx)}\circ W^{(s)}=0$, since $s\notin(f\circ\bg)(\bx)$ whenever $\bg(\bx)=\tx$.
        Therefore,
        \begin{align*}
            W^{(s)\downarrow}(\tx,\tx)
            = & 2^n\langle \kbr{\tx}, \cE^{\downarrow}_{\Gamma}(W^{(s)}\circ \kbr{\nu}^{\otimes n})\rangle\\
            = & 2^n\langle \cE^{\uparrow}_{\Gamma}(\kbr{\tx}), W^{(s)}\circ \kbr{\nu}^{\otimes n}\rangle\\
            = & 2^n\langle \cE^{\uparrow}_{\Gamma}(\kbr{\tx}) \circ \overline{W^{(s)}}, \kbr{\nu}^{\otimes n} \rangle\\
            = & 2^n\langle I_{\bg^{-1}(\tx)}\circ\overline{W^{(s)}}, \kbr{\nu}^{\otimes n} \rangle\\
            = & 2^n\langle 0, \kbr{\nu}^{\otimes n} \rangle = 0,
        \end{align*}
        where in the fourth line we used \Crefitem{item:composition-diagonal-lift}.
    \end{itemize}

    \item \textbf{Channel preserves relaxed matrices $\cR^{\eta}_{\delta}(f\circ\bg)\to
        \cR^{\mu}_{\delta}(f)$.} Let $\Psi\in\cR^{\eta}_{\delta}(f\circ\bg)$, so that $\Psi\circ
        A_{f\circ\bg}=A_{f\circ\bg}$ and
        $\|\Psi\circ\Delta_{f\circ\bg}\circ\kbr{\eta}\|_{\tr}\leq\delta$. Since $\Psi$ is Hermitian,
        $\Psi^{\downarrow} := 2^n\cE^{\downarrow}_{\Gamma}(\Psi\circ\kbr{\nu}^{\otimes n})$ is also
        Hermitian. By \Crefitem{item:composition-weighted-input}, we have
        $\cE^{\downarrow}_{\Gamma}(\Psi\circ\kbr{\eta})=\Psi^{\downarrow}\circ\kbr{\mu}$. By
        \Crefitem{item:G-entries},  $(\bG^{\dagger} A_f \bG)(\bx,\by) = A_f(\bg(\bx),\bg(\by)) =
        A_{f\circ\bg}(\bx,\by)$. Therefore, $\bG^{\dagger} A_f \bG = A_{f\circ\bg}$. So, we can
        compute
        \begin{align*}
            \Psi^{\downarrow} \circ A_f
            = & 2^n\cE^{\downarrow}_{\Gamma}(\Psi \circ \kbr{\nu}^{\otimes n}\circ(\bG^\dagger A_f\bG))\\
            = & 2^n\cE^{\downarrow}_{\Gamma}(\Psi \circ \kbr{\nu}^{\otimes n}\circ A_{f\circ\bg})\\
            = & 2^n\cE^{\downarrow}_{\Gamma}(\kbr{\nu}^{\otimes n}\circ A_{f\circ\bg})\\
            = & 2^n\cE^{\downarrow}_{\Gamma}(\kbr{\nu}^{\otimes n}) \circ A_f = A_f,
        \end{align*}
        where we used \Crefitem{item:composition-hadamard} in the first line and
        \Cref{eq:composition-balanced-input} in the last. Similarly, we can prove
        $\bG^{\dagger}\Delta_f \bG = \Delta_{f\circ\bg}$. Therefore,
        \begin{align*}
            \| \Psi^{\downarrow} \circ \Delta_f \circ \kbr{\mu}\|_{\tr}
            = & \| \cE^{\downarrow}_{\Gamma}(\Psi \circ \kbr{\eta}) \circ \Delta_f\|_{\tr}
                \quad\text{(By \Crefitem{item:composition-weighted-input})}\\
            = & \|\cE^{\downarrow}_{\Gamma}(\Psi \circ \kbr{\eta} \circ \Delta_{f\circ\bg})\|_{\tr}
                \quad\text{(By \Crefitem{item:composition-hadamard})}\\
            \leq & \|\Psi \circ \kbr{\eta} \circ \Delta_{f\circ\bg}\|_{\tr}
                \quad\text{(By \Cref{lem:channel-norm-contraction})}\\
            \leq & \delta. \qedhere
        \end{align*}
\end{enumerate}
\end{proof}


\subsubsection{Bounding filtered norms}\label{sec:composition-filtered-norms} Next, we bound the
filtered norms of the lifted certificate.
\begin{lemma}\label{lem:composition-filtered-norms} Let $g:\Sigma_g\to\{0,1\}$ be a decision problem
    on $\Sigma_g\subseteq\{0,1\}^m$, and let $\Gamma\in\cL(\Sigma_g)$ be real symmetric with
    $\|\Gamma\|\leq1$ and $\Gamma\circ\Delta_g=\Gamma$. Then, for every $X\in\Herm(\{0,1\}^n)$,
    $i\in[n]$, and $j\in[m]$, we have
    \[
        \|\cE^{\uparrow}_{\Gamma}(X)\circ\Delta_{ij}\|
        \leq \|X\circ\Delta_i\|\cdot\|\Gamma\circ\Delta_j\|,
    \]
    where $\cE^{\uparrow}_{\Gamma}$ is the lifting map in \Cref{def:composition-lifting-map}.
\end{lemma}
\begin{proof}
    By \Crefitem{item:G-entries}, $(\bG^{\dagger}\Delta_i\bG)(\bx,\by)
    =\Delta_i(\bg(\bx),\bg(\by))=\ind\{\bg(\bx)_i\neq\bg(\by)_i\}$. Therefore,
    $\bG^{\dagger}\Delta_i\bG=J^{\otimes(n-1)}\otimes_i\Delta_g$. Let
    $\lambda:=\|\Gamma\circ\Delta_j\|$. Then, $(\lambda I+\Gamma\circ\Delta_j)\circ\Delta_g
    =\lambda(I\circ\Delta_g)+(\Gamma\circ\Delta_g)\circ\Delta_j =\Gamma\circ\Delta_j$. Thus,
    \begin{align*}
        \cE^{\uparrow}_{\Gamma}(X)\circ\Delta_{ij}
        = &\bG^{\dagger}X\bG\circ(I+\Gamma)^{\otimes n}\circ\Delta_{ij}\\
        = &\bG^{\dagger}X\bG\circ(I+\Gamma)^{\otimes n}
        \circ\left(J^{\otimes(n-1)}\otimes_i\Delta_j\right)\\
        = &\bG^{\dagger}X\bG\circ\left[
        (I+\Gamma)^{\otimes(n-1)}
        \otimes_i\bigl((I+\Gamma)\circ\Delta_j\bigr)\right]\\
        = &\bG^{\dagger}X\bG\circ\left[
        (I+\Gamma)^{\otimes(n-1)}
        \otimes_i(\Gamma\circ\Delta_j)\right]\\
        = &\bG^{\dagger}X\bG\circ\left[
        (I+\Gamma)^{\otimes(n-1)}
        \otimes_i\bigl((\lambda I+\Gamma\circ\Delta_j)\circ\Delta_g\bigr)\right]\\
        = &\bG^{\dagger}X\bG\circ\left[
        \bigl((I+\Gamma)^{\otimes(n-1)}\circ J^{\otimes(n-1)}\bigr)
        \otimes_i\bigl((\lambda I+\Gamma\circ\Delta_j)\circ\Delta_g\bigr)\right]\\
        = &\bG^{\dagger}X\bG
        \circ\left[(I+\Gamma)^{\otimes(n-1)}
        \otimes_i(\lambda I+\Gamma\circ\Delta_j)\right]
        \circ\left(J^{\otimes(n-1)}\otimes_i\Delta_g\right)\\
        = &(\bG^{\dagger}X\bG)
        \circ\left[(I+\Gamma)^{\otimes(n-1)}
        \otimes_i(\lambda I+\Gamma\circ\Delta_j)\right]
        \circ(\bG^{\dagger}\Delta_i\bG)\\
        = &\bG^{\dagger}(X\circ\Delta_i)\bG
        \circ\left[(I+\Gamma)^{\otimes(n-1)}
        \otimes_i(\lambda I+\Gamma\circ\Delta_j)\right],
    \end{align*}
    where in the last line we used \Crefitem{item:G-hadamard-lift}.
    \par
    Next, we bound the spectral norm of the matrix on the right-hand side. Note that
    $\bG^{\dagger}(X\circ \Delta_i) \bG \preceq \|X\circ \Delta_i\| \bG^{\dagger}\bG$. Moreover, the
    final tensor on the right-hand side is PSD, since $I+\Gamma\succeq0$ and $\lambda
    I+\Gamma\circ\Delta_j\succeq0$. Therefore, by the Schur product theorem,
    \begin{align*}
        \cE^{\uparrow}_{\Gamma}(X)\circ \Delta_{ij}
        &=\bG^{\dagger}(X\circ\Delta_i)\bG
        \circ\left[(I+\Gamma)^{\otimes(n-1)}
        \otimes_i(\lambda I+\Gamma\circ\Delta_j)\right]\\
        &\preceq
        \|X\circ \Delta_i\| \bG^{\dagger}\bG \circ \left[(I+\Gamma)^{\otimes(n-1)}
        \otimes_i(\lambda I+\Gamma\circ\Delta_j)\right]\\
        & = \|X\circ \Delta_i\| A_g^{\otimes n}\circ \left[(I+\Gamma)^{\otimes(n-1)}
        \otimes_i(\lambda I+\Gamma\circ\Delta_j)\right]\\
        &= \|X\circ \Delta_i\|\cdot(\lambda I)\\
        &= \|X\circ \Delta_i\|\cdot \|\Gamma \circ \Delta_j\| \cdot I,
    \end{align*}
    where we used $A_g\circ\Gamma=0$ to resolve the Hadamard products. Similarly,
    $\bG^{\dagger}(X\circ \Delta_i) \bG \succeq -\|X\circ \Delta_i\| \bG^{\dagger}\bG$ and
    therefore,
    \begin{align*}
        \cE^{\uparrow}_{\Gamma}(X)\circ \Delta_{ij} \succeq -\|X\circ \Delta_i\|\cdot \|\Gamma \circ \Delta_j\| \cdot I.
    \end{align*}
    Having both sides, we conclude $\| \cE^{\uparrow}_{\Gamma}(X)\circ \Delta_{ij}\| \leq \|X\circ
    \Delta_i\|\cdot \|\Gamma \circ \Delta_j\|$, which completes the proof.
\end{proof}

\subsubsection{Proof of the composition theorem}\label{sec:composition-proof}
\begin{proposition}\label{prop:composition-adversary-bounds} Let $f,g,\mu,\nu_0,\nu_1$ be as in
    \Cref{thm:composition-adversary}, and set $\eta=\mu\circ(\nu_0,\nu_1)$. Then, for
    $\eps=\eps(n)\in(0,1/2)$ and $\delta=\delta(n)\in(0,1)$,
    \begin{align*}
        \mathrm{Adv}^{\eta}_{\eps}(f\circ\bg)
        &\geq \mathrm{Adv}^{\mu}_{\eps}(f)\cdot\mathrm{sfAdv}^{(\nu_0,\nu_1)}(g),\\
        \mathrm{rAdv}^{\eta}_{\delta}(f\circ\bg)
        &\geq \mathrm{rAdv}^{\mu}_{\delta}(f)\cdot\mathrm{sfAdv}^{(\nu_0,\nu_1)}(g).
    \end{align*}
\end{proposition}
\begin{proof}
    Let $\Gamma_g$ be a real symmetric maximizer in \Cref{lem:dual-sfadv-hardness}, so that
    \begin{align*}
        \Gamma_g\circ\Delta_g=\Gamma_g,
        &&\bra{\nu_0}\Gamma_g\ket{\nu_1}=\|\Gamma_g\|=1,
        &&\mathrm{sfAdv}^{(\nu_0,\nu_1)}(g)=\frac{1}{\max_j\|\Gamma_g\circ\Delta_j\|}.
    \end{align*}
    Let $\cE^{\uparrow}:=\cE^{\uparrow}_{\Gamma_g}$ be the lifting map in
    \Cref{def:composition-lifting-map}, and let $\cE^{\downarrow}$ be its adjoint. For either
    inequality, we may assume that the one-copy measure on the right-hand side is positive, since
    otherwise the claim is trivial.
    \par
    For the first inequality, let $\Gamma_f$ attain the maximum in
    \Cref{lem:adv-universal-hardness-cert}, satisfying $\Gamma_f=\Gamma_f^\dagger$, $\Gamma_f\circ
    I=0$, and $\max_i\|\Gamma_f\circ\Delta_i\|\leq1$. We lift $\Gamma_f$ with respect to $\Gamma_g$
    using $\cE^{\uparrow}(\Gamma_f)$. Here $\Gamma_f$ is extended by zero outside $\Sigma$, and its
    lift is restricted to $\bg^{-1}(\Sigma)$. By definition, the lift is Hermitian, and by
    \Crefitem{item:composition-diagonal-lift}, $\cE^{\uparrow}(\Gamma_f)\circ I =
    \cE^{\uparrow}(\Gamma_f \circ I) = 0$. Moreover, by \Cref{lem:composition-filtered-norms},
    \begin{align}\label{eq:composition-lift-filtered-norm}
        \max_{i,j}\|\cE^{\uparrow}(\Gamma_f)\circ\Delta_{ij}\|
        \leq &\max_i\|\Gamma_f\circ\Delta_i\|\cdot
        \max_j\|\Gamma_g\circ\Delta_j\|
        \leq\max_j\|\Gamma_g\circ\Delta_j\|.
    \end{align}
    Therefore,
    \begin{align*}
        \mathrm{Adv}^{\eta}_{\eps}(f\circ\bg)
        \geq &\min_{\Psi\in\cG^{\eta}_{\eps}(f\circ\bg)}
        \frac{\langle\cE^{\uparrow}(\Gamma_f),(J-\Psi)\circ\kbr{\eta}\rangle}
        {\max_{i,j}\|\cE^{\uparrow}(\Gamma_f)\circ\Delta_{ij}\|}
        &&\text{(By \Cref{lem:adv-universal-hardness-cert}).}\\
        \geq &\min_{\Psi\in\cG^{\eta}_{\eps}(f\circ\bg)}
        \frac{\langle\Gamma_f,\cE^{\downarrow}((J-\Psi)\circ\kbr{\eta})\rangle}
        {\max_j\|\Gamma_g\circ\Delta_j\|}
        &&\text{(By Eq.~\eqref{eq:composition-lift-filtered-norm}).}\\
        = &\min_{\Psi\in\cG^{\eta}_{\eps}(f\circ\bg)}
        \frac{\langle\Gamma_f,(J-\Psi^{\downarrow})\circ\kbr{\mu}\rangle}
        {\max_j\|\Gamma_g\circ\Delta_j\|}
        &&\text{($\Psi^{\downarrow}\in\cG^{\mu}_{\eps}(f)$ by \Cref{lem:comp-adjoint-applications})}\\
        \geq &\min_{\Phi\in\cG^{\mu}_{\eps}(f)}
        \frac{\langle\Gamma_f,(J-\Phi)\circ\kbr{\mu}\rangle}
        {\max_j\|\Gamma_g\circ\Delta_j\|}\\
        = &\mathrm{Adv}^{\mu}_{\eps}(f)\cdot\mathrm{sfAdv}^{(\nu_0,\nu_1)}(g),
        &&\text{(By the choice of $\Gamma_f$ and $\Gamma_g$).}
    \end{align*}
    \par
    For the second inequality, choose a maximizer $\Gamma_f$ in
    \Cref{lem:radv-universal-hardness-cert}, satisfying the same Hermitian, zero-diagonal, and
    filtered-norm conditions. Then,
    \begin{align*}
        \min_{\Phi\in\cR^{\mu}_{\delta}(f)}
        \langle\Gamma_f,(J-\Phi)\circ\kbr{\mu}\rangle
        = &\mathrm{rAdv}^{\mu}_{\delta}(f).
    \end{align*}
    The same lift and filtered-norm bound give
    \begin{align*}
        \mathrm{rAdv}^{\eta}_{\delta}(f\circ\bg)
        \geq &\min_{\Psi\in\cR^{\eta}_{\delta}(f\circ\bg)}
        \frac{\langle\cE^{\uparrow}(\Gamma_f),(J-\Psi)\circ\kbr{\eta}\rangle}
        {\max_{i,j}\|\cE^{\uparrow}(\Gamma_f)\circ\Delta_{ij}\|}\\
        = &\min_{\Psi\in\cR^{\eta}_{\delta}(f\circ\bg)}
        \frac{\langle\Gamma_f,\cE^{\downarrow}((J-\Psi)\circ\kbr{\eta})\rangle}
        {\max_{i,j}\|\cE^{\uparrow}(\Gamma_f)\circ\Delta_{ij}\|}\\
        \geq &\min_{\Psi\in\cR^{\eta}_{\delta}(f\circ\bg)}
        \frac{\langle\Gamma_f,(J-\Psi^{\downarrow})\circ\kbr{\mu}\rangle}
        {\max_j\|\Gamma_g\circ\Delta_j\|}
        &&\text{($\Psi^{\downarrow}\in\cR^{\mu}_{\delta}(f)$ by \Cref{lem:comp-adjoint-applications} )}\\
        \geq &\min_{\Phi\in\cR^{\mu}_{\delta}(f)}
        \frac{\langle\Gamma_f,(J-\Phi)\circ\kbr{\mu}\rangle}
        {\max_j\|\Gamma_g\circ\Delta_j\|}\\
        = &\mathrm{rAdv}^{\mu}_{\delta}(f)\cdot\mathrm{sfAdv}^{(\nu_0,\nu_1)}(g).
       &&  \qedhere
    \end{align*}
\end{proof}

\subsubsection{\texorpdfstring{TRIV and tightness for $\mathrm{rAdv}$}{TRIV and tightness for the
relaxed adversary}}\label{sec:composition-triv} Recall that $\mathrm{TRIV}:\{0^n,1^n\}\to\{0,1\}$ is
defined by $\mathrm{TRIV}(b^n)=b$.
\begin{lemma}\label{lem:triv-radv-certificate} For $\eps=\eps(n)\in[0,1]$, we have
    $\mathrm{rAdv}^{(1/2,1/2)}_{\eps}(\mathrm{TRIV})=1-\eps$.
\end{lemma}
\begin{proof}
    Let $\mu=(1/2,1/2)$ and $\Psi=\begin{bmatrix}1&\alpha\\\alpha^{\ast}&1\end{bmatrix}$. Then,
    $\Psi\circ\Delta_{\mathrm{TRIV}}\circ\kbr{\mu}
    =\begin{bmatrix}0&\alpha/2\\\alpha^{\ast}/2&0\end{bmatrix}$. Therefore,
    $\|\Psi\circ\Delta_{\mathrm{TRIV}}\circ\kbr{\mu}\|_{\tr}=|\alpha|$, so
    $\Psi\in\cR^{\mu}_{\eps}(\mathrm{TRIV})$ if and only if $|\alpha|\leq\eps$. Also, by
    \Cref{lem:equiv-avg-gamma2-query}, we may assume that the adversary matrix is
    $\Gamma=\begin{bmatrix}0&\beta\\\beta^{\ast}&0\end{bmatrix}$. Since
    $\Delta_i=\Delta_{\mathrm{TRIV}}$ for every $i\in[n]$, we have
    $\max_i\|\Gamma\circ\Delta_i\|=\|\Gamma\|=|\beta|$. Thus,
    \begin{align*}
        \bar{\gamma}_2^{\mu}(J-\Psi\mid\{\Delta_i\})
        &=\max_{\substack{\Gamma=\Gamma^{\dagger},\ \Gamma\circ I=0\\
            \max_i\|\Gamma\circ\Delta_i\|\leq1}}
            \langle\Gamma,(J-\Psi)\circ\kbr{\mu}\rangle\\
        &=\max_{|\beta|\leq1}
            \left\langle
            \begin{bmatrix}0&\beta\\\beta^{\ast}&0\end{bmatrix},
            \frac12\begin{bmatrix}0&1-\alpha\\1-\alpha^{\ast}&0\end{bmatrix}
            \right\rangle\\
        &=\max_{|\beta|\leq1}
            \frac{\beta^{\ast}(1-\alpha)+\beta(1-\alpha^{\ast})}{2}\\
        &=\max_{|\beta|\leq1}\Re\bigl(\beta^{\ast}(1-\alpha)\bigr)
        =|1-\alpha|.
    \end{align*}
    Therefore, by \Cref{def:radv},
    \begin{align*}
        \mathrm{rAdv}^{\mu}_{\eps}(\mathrm{TRIV})
        = &\min_{\Psi\in\cR^{\mu}_{\eps}(\mathrm{TRIV})}
        \bar{\gamma}_2^{\mu}(J-\Psi\mid\{\Delta_i\})
        =\min_{|\alpha|\leq\eps}|1-\alpha|=1-\eps.
    \end{align*}
    Taking $\alpha=\eps$ attains the minimum.
\end{proof}
We now present the proof of the third item in \Cref{thm:composition-adversary}. The distribution
$(\nu_0^{\otimes n}+\nu_1^{\otimes n})/2$ describes samples conditioned on a common output label, as
in the classical correlated-samples setting of Bassilakis et al. \cite{BDG+20}.
\begin{proposition}\label{prop:composition-triv-tightness} Let $\Sigma_g\subset\{0,1\}^{m}$,
    $g:\Sigma_g\to\{0,1\}$ be a decision problem, and $\nu_0,\nu_1:\Sigma_g\to[0,1]$ be
    distributions supported on $g^{-1}(0)$ and $g^{-1}(1)$, respectively. Then, for any fixed
    $\eps\in(0,1)$, we have
    \begin{align*}
        \mathrm{sfAdv}^{(\nu_0,\nu_1)}(g)
        = &\lim_{n\to\infty}
        \frac{\mathrm{rAdv}^{(\nu_0^{\otimes n}+\nu_1^{\otimes n})/2}_{\eps}(\mathrm{TRIV}\circ\bg)}{1-\eps}.
    \end{align*}
\end{proposition}


\begin{proof}
    Let $\nu=(\nu_0+\nu_1)/2$, and for $n\geq1$ let $\eta_n=(\nu_0^{\otimes n}+\nu_1^{\otimes n})/2$
    and $F_n=\mathrm{TRIV}\circ\bg$. By \Cref{prop:composition-adversary-bounds} and
    \Cref{lem:triv-radv-certificate}, for every $n\geq1$,
    \begin{align*}
        \mathrm{sfAdv}^{(\nu_0,\nu_1)}(g)
        &\leq \frac{\mathrm{rAdv}^{\eta_n}_{\eps}(F_n)}
            {\mathrm{rAdv}^{(1/2,1/2)}_{\eps}(\mathrm{TRIV})}
        =\frac{\mathrm{rAdv}^{\eta_n}_{\eps}(F_n)}{1-\eps}.
    \end{align*}
    For the other direction, let $\tau>0$ be arbitrary. By \Cref{def:sfadv}, there exists
    $\delta=\delta_{\tau}\in[0,1)$ such that
    \begin{align*}
        \frac{\mathrm{rAdv}^{\nu}_{\delta}(g)}{1-\delta}
        \leq \mathrm{sfAdv}^{(\nu_0,\nu_1)}(g)+\tau.
    \end{align*}
    Our goal is to show the following.
    \begin{claim}\label{claim:triv-amplification} Then, for every $n\geq1$, we have
        \begin{align*}
            \frac{\mathrm{rAdv}^{\eta_n}_{\delta^n}(F_n)}{1-\delta^n}
            \leq \frac{\mathrm{rAdv}^{\nu}_{\delta}(g)}{1-\delta}.
        \end{align*}
    \end{claim}
    Then, assuming the above, the proposition follows. Indeed,  for all sufficiently large $n\in\bN$
    we have $\delta^n\leq\eps$, and therefore, by \Cref{lemma:sfadv-monotonicity},
    \begin{align*}
        \frac{\mathrm{rAdv}^{\eta_n}_{\eps}(F_n)}{1-\eps}
        &\leq \frac{\mathrm{rAdv}^{\eta_n}_{\delta^n}(F_n)}{1-\delta^n}
        \leq \frac{\mathrm{rAdv}^{\nu}_{\delta}(g)}{1-\delta}
        \leq \mathrm{sfAdv}^{(\nu_0,\nu_1)}(g)+\tau.
    \end{align*}
    Since $\tau>0$ was arbitrary, taking the limit $n\to \infty$, we obtain a matching lower bound
    for $\mathrm{sfAdv}^{(\nu_0,\nu_1)}$, completing the proof.
    \par
    To prove the claim, we can intuitively think of $\mathrm{rAdv}^{\nu}_{\delta}(g)$ as the cost of
    an idealized algorithm that returns the correct answer with probability $1-\delta$ and otherwise
    reports failure. Given $n$ independent copies with the same answer, we could run this algorithm
    on successive copies until one succeeds. The failure probability would then be $\delta^n$, while
    the expected total cost would be
    \begin{align*}
        (1+\delta+\cdots+\delta^{n-1})\cdot\mathrm{rAdv}^{\nu}_{\delta}(g)
        =\frac{1-\delta^n}{1-\delta}\,\mathrm{rAdv}^{\nu}_{\delta}(g).
    \end{align*}
    Formally, we proceed as follows.
\paragraph{Proof of the claim.}
For ease of notation, we implicitly restrict all tensor-product matrices below to $D_n\times D_n$,
where $D_n:=g^{-1}(0)^n\cup g^{-1}(1)^n$ is the domain of $F_n$. Let $m$ be the input length of $g$.
By \Cref{def:radv}, we have
        \begin{align*}
            \mathrm{rAdv}^{\nu}_{\delta}(g)
            =\min_{\Psi\in\cR^{\nu}_{\delta}(g)}
            \bar{\gamma}_2^{\nu}(J-\Psi\mid\{\Delta_j\}).
        \end{align*}
        Let $\Psi\in\cR^{\nu}_{\delta}(g)$ attain the minimum.
        We first show that $\Psi^{\otimes n}\in\cR^{\eta_n}_{\delta^n}(F_n)$. The matrix
        $\Psi^{\otimes n}$ is Hermitian, since $\Psi$ is Hermitian. Also, for $\bx,\by\in D_n$ and
        $i\in [n]$, $g(x^i)=F_n(\bx)$ and $g(y^i)=F_n(\by)$. Hence,
        \begin{align*}
            (A_g^{\otimes n})(\bx,\by)
            &=\prod_{i=1}^{n}\ind\{g(x^i)=g(y^i)\}
            =\ind\{F_n(\bx)=F_n(\by)\}=A_{F_n}(\bx,\by),
        \end{align*}
        As a result, $\Psi^{\otimes n}\circ A_{F_n} =(\Psi\circ A_g)^{\otimes n} =A_g^{\otimes
        n}=A_{F_n}.$ Similarly,
        \begin{align*}
            (\Delta_g^{\otimes n})(\bx,\by)
            &=\prod_{i=1}^{n}\ind\{g(x^i)\neq g(y^i)\}
            =\ind\{F_n(\bx)\neq F_n(\by)\}=\Delta_{F_n}(\bx,\by).
        \end{align*}

        Therefore, $\Delta_g^{\otimes n} = \Delta_{F_n}$, and we can bound
        \begin{align*}
            \|\Psi^{\otimes n}\circ\Delta_{F_n}\circ\kbr{\eta_n}\|_{\tr}
            =&\frac12\Bigl\|\Psi^{\otimes n}\circ\Delta_g^{\otimes n}\circ\bigl[(\ketbra{\nu_0}{\nu_0})^{\otimes n}+(\ketbra{\nu_0}{\nu_1})^{\otimes n}+(\ketbra{\nu_1}{\nu_0})^{\otimes n}+(\ketbra{\nu_1}{\nu_1})^{\otimes n}\bigr]\Bigr\|_{\tr}\\
            =&\frac12\Bigl\|\Psi^{\otimes n}\circ\Delta_g^{\otimes n}\circ\bigl[(\ketbra{\nu_0}{\nu_1})^{\otimes n}+(\ketbra{\nu_1}{\nu_0})^{\otimes n}\bigr]\Bigr\|_{\tr}\\
            =&\|\Psi^{\otimes n}\circ\Delta_g^{\otimes n}\circ(\ketbra{\nu_0}{\nu_1})^{\otimes n}\|_{\tr}\\
            =&\|(\Psi\circ\Delta_g\circ\ketbra{\nu_0}{\nu_1})^{\otimes n}\|_{\tr}\\
            =&\|\Psi\circ\Delta_g\circ\ketbra{\nu_0}{\nu_1}\|_{\tr}^{n}\\
            =&\|\Psi\circ\Delta_g\circ\kbr{\nu}\|_{\tr}^{n}\\
            \leq &\delta^n,
        \end{align*}
        Here, we could remove and add the terms involving $\ketbra{\nu_b}{\nu_b}$ because $\nu_b$ is
        supported on $g^{-1}(b)$, so $\Delta_g\circ\kbr{\nu_b}=0$. This completes the proof of
        $\Psi^{\otimes n}\in\cR^{\eta_n}_{\delta^n}(F_n)$. Therefore,
        \begin{align*}
            \mathrm{rAdv}^{\eta_n}_{\delta^n}(F_n)
            \leq &\bar{\gamma}_2^{\eta_n}(J^{\otimes n}-\Psi^{\otimes n}\mid\{\Delta_{ij}\}).
        \end{align*}
\par
        We now directly upper bound the right-hand side by lifting a primal certificate for
        $\bar{\gamma}_2^{\nu}(J-\Psi\mid\{\Delta_j\})$ to one for
        $\bar{\gamma}_2^{\eta_n}(J^{\otimes n}-\Psi^{\otimes n}\mid\{\Delta_{ij}\})$. By the primal
        form in \Cref{prop:equiv-avg-gamma2}, let $X_j\in\Herm(\Sigma_g)$, $j\in[m]$, be such that
        \begin{align*}
            \mathrm{rAdv}^{\nu}_{\delta}(g)
            = &\sum_{j=1}^{m}\|X_j\circ\kbr{\nu}\|_{\tr},\\
            J-\Psi
            = &\sum_{j=1}^{m}X_j\circ\Delta_j
            &&\text{on $\supp(\nu)\times\supp(\nu)$.}
        \end{align*}
        Mirroring the definition of $\Delta_{ij}$, we define the lifted matrices
        $Z_{ij}\in\Herm(D_n)$, for $i\in[n]$ and $j\in[m]$, as
        \begin{align*}
            Z_{ij}:=\Delta_{F_n}\circ
            \bigl(J^{\otimes(i-1)}\otimes X_j\otimes\Psi^{\otimes(n-i)}\bigr).
        \end{align*}
        Then, on $\supp(\eta_n)\times\supp(\eta_n)$, we have
        \begin{align*}
            \sum_{i,j}Z_{ij}\circ\Delta_{ij}
            &=\sum_{i,j}\Delta_{F_n}\circ
            \bigl(J^{\otimes(i-1)}\otimes X_j\otimes\Psi^{\otimes(n-i)}\bigr)
            \circ\bigl(J^{\otimes(i-1)}\otimes\Delta_j\otimes J^{\otimes(n-i)}\bigr)\\
            &=\sum_{i,j}\Delta_{F_n}\circ
            \bigl(J^{\otimes(i-1)}\otimes(X_j\circ\Delta_j)\otimes\Psi^{\otimes(n-i)}\bigr)\\
            &=\sum_i\Delta_{F_n}\circ
            \bigl(J^{\otimes(i-1)}\otimes(J-\Psi)\otimes\Psi^{\otimes(n-i)}\bigr)\\
            &=\Delta_{F_n}\circ(J^{\otimes n}-\Psi^{\otimes n})\\
            &=J^{\otimes n}-\Psi^{\otimes n},
        \end{align*}
        where in the last line we used $\Psi^{\otimes n}\circ A_{F_n} = A_{F_n}$. Thus, $\{Z_{ij}\}$
        form a feasible primal certificate for $\bar{\gamma}_2^{\eta_n}(J^{\otimes n}-\Psi^{\otimes
        n}\mid\{\Delta_{ij}\})$ by \Cref{prop:equiv-avg-gamma2}. Hence,
        \begin{align*}
            \mathrm{rAdv}^{\eta_n}_{\delta^n}(F_n)
            \leq &\bar{\gamma}_2^{\eta_n}(J^{\otimes n}-\Psi^{\otimes n}\mid\{\Delta_{ij}\})\\
            \leq &\sum_{i,j}\|Z_{ij}\circ\kbr{\eta_n}\|_{\tr}\\
            = &\sum_{i,j}\Bigl\|\Delta_{F_n}\circ
            \bigl(J^{\otimes(i-1)}\otimes X_j\otimes\Psi^{\otimes(n-i)}\bigr)
            \circ\kbr{\eta_n}\Bigr\|_{\tr}\\
            = &\sum_{i,j}\Bigl\|\Delta_g^{\otimes n}\circ
            \bigl(J^{\otimes(i-1)}\otimes X_j\otimes\Psi^{\otimes(n-i)}\bigr)
            \circ(\ketbra{\nu_0}{\nu_1})^{\otimes n}\Bigr\|_{\tr}\\
            = &\sum_{i,j}\|\Delta_g\circ\ketbra{\nu_0}{\nu_1}\|_{\tr}^{i-1}
            \|X_j\circ\Delta_g\circ\ketbra{\nu_0}{\nu_1}\|_{\tr}
            \|\Psi\circ\Delta_g\circ\ketbra{\nu_0}{\nu_1}\|_{\tr}^{n-i}\\
            = &\sum_{i,j}\|X_j\circ\Delta_g\circ\kbr{\nu}\|_{\tr}
            \|\Psi\circ\Delta_g\circ\kbr{\nu}\|_{\tr}^{n-i}\\
            \leq &\sum_{i,j}\gamma_2(\Delta_g)\cdot\|X_j\circ\kbr{\nu}\|_{\tr}
            \cdot\delta^{n-i}\\
            \leq &\sum_{i=1}^{n}\delta^{n-i}
            \sum_{j=1}^{m}\|X_j\circ\kbr{\nu}\|_{\tr}\\
            = &\frac{1-\delta^n}{1-\delta}\,\mathrm{rAdv}^{\nu}_{\delta}(g),
        \end{align*}
        where we used $\gamma_2(\Delta_g)\leq1$ for decision problems. Dividing both sides by
        $1-\delta^n$ completes the proof of the claim.
    \end{proof}

\subsection{Composition for distributional quantum query complexity}\label{sec:composition-quantum}
In this section, we relate the adversary composition bounds back to the  distributional quantum
query complexity, giving the composition theorem and tightness of $\mathrm{sfAdv}$ in the limit.

\begin{theorem}\label{thm:composition-distributional-query} Let $\Sigma\subset\{0,1\}^{n}$,
    $f:\Sigma\to\Omega$ be a functional problem, and $\mu:\Sigma\to[0,1]$ be a distribution. Let
    $\Sigma_g\subset\{0,1\}^{m}$, $g:\Sigma_g\to\{0,1\}$ be a decision problem, and
    $\nu_0,\nu_1:\Sigma_g\to[0,1]$ be a pair of distributions supported on $g^{-1}(0)$ and
    $g^{-1}(1)$, respectively. For constants $0<\eps<\eps'<1/2$, we have
    \begin{align*}
        \mathsf{Q}^{\mu\circ(\nu_0,\nu_1)}_{\eps}(f\circ\bg)
        = &\Omega_{\eps',\eps}\!\left(
        \mathsf{Q}^{\mu}_{\eps'}(f)\cdot\mathrm{sfAdv}^{(\nu_0,\nu_1)}(g)\right).
    \end{align*}
    Moreover,
    \begin{align*}
        \mathrm{sfAdv}^{(\nu_0,\nu_1)}(g)
        = &\Theta_{\eps}\!\left(
        \lim_{n\to\infty}\mathsf{Q}^{(\nu_0^{\otimes n}+\nu_1^{\otimes n})/2}_{\eps}
        (\mathrm{TRIV}\circ\bg)\right).
    \end{align*}
\end{theorem}
\begin{proof}
    For the first bound, we have
    \begin{align*}
        \mathsf{Q}^{\mu\circ(\nu_0,\nu_1)}_{\eps}(f\circ\bg)
        &\geq \frac{1}{2}\mathrm{Adv}^{\mu\circ(\nu_0,\nu_1)}_{\eps}(f\circ\bg)
        &&\text{(By \Cref{adv-lowerbound}).}\\
        &\geq \frac{1}{2}\mathrm{Adv}^{\mu}_{\eps}(f)
        \cdot\mathrm{sfAdv}^{(\nu_0,\nu_1)}(g)
        &&\text{(By \Cref{prop:composition-adversary-bounds}).}
    \end{align*}
    Letting $\tau:=(\sqrt{\eps'}-\sqrt{\eps})^2$, we have $(\sqrt{\eps}+\sqrt{\tau})^2=\eps'$.
    Applying \Cref{prop:adv-upperbound} gives
    \begin{align*}
        \mathsf{Q}^{\mu}_{\eps'}(f)
        \leq \frac{4}{\tau}\mathrm{Adv}^{\mu}_{\eps}(f).
    \end{align*}
    Rearranging gives $\mathrm{Adv}^{\mu}_{\eps}(f)\geq\frac{\tau}{4}\mathsf{Q}^{\mu}_{\eps'}(f)$.
    Since $\tau>0$ depends only on $\eps',\eps$, combining with the above proves the first bound:
    \begin{align*}
        \mathsf{Q}^{\mu\circ(\nu_0,\nu_1)}_{\eps}(f\circ\bg)
        \geq &\frac{\tau}{8}\mathsf{Q}^{\mu}_{\eps'}(f)\cdot\mathrm{sfAdv}^{(\nu_0,\nu_1)}(g)
        =\Omega_{\eps',\eps}\!\left(
        \mathsf{Q}^{\mu}_{\eps'}(f)\cdot\mathrm{sfAdv}^{(\nu_0,\nu_1)}(g)\right).
    \end{align*}
    \par
    For the second statement, write $F_n:=\mathrm{TRIV}\circ\bg$ and $\mu_n=(\nu_0^{\otimes
    n}+\nu_1^{\otimes n})/2$. The sequence $\mathsf{Q}^{\mu_n}_{\eps}(F_n)$ is nonincreasing, since
    ignoring additional copies leaves the same marginal distribution, so its limit exists. Applying
    the first bound with outer problem $\mathrm{TRIV}$, outer distribution $(1/2,1/2)$, and
    $\eps''=(\eps+1/2)/2$, we get
    \begin{align*}
        \mathsf{Q}^{\mu_n}_{\eps}(F_n)
        = &\Omega_{\eps'',\eps}\!\left(
        \mathsf{Q}^{(1/2,1/2)}_{\eps''}(\mathrm{TRIV})
        \cdot\mathrm{sfAdv}^{(\nu_0,\nu_1)}(g)\right)
        =\Omega_{\eps'',\eps}\!\left(\mathrm{sfAdv}^{(\nu_0,\nu_1)}(g)\right),
    \end{align*}
    where we used $\mathsf{Q}^{(1/2,1/2)}_{\eps''}(\mathrm{TRIV})=1$. Since $\eps''$ depends only on
    the fixed $\eps$, we get $\lim_{n\to\infty}\mathsf{Q}^{\mu_n}_{\eps}(F_n)
    =\Omega_{\eps}\!\left(\mathrm{sfAdv}^{(\nu_0,\nu_1)}(g)\right)$.
    \par
    For the other direction, by \Cref{cor:radv-query-comparison}, we have
    \begin{align*}
        \mathsf{Q}^{\mu_n}_{\eps}(F_n)
        = &O_{\eps}\!\left(\mathrm{rAdv}^{\mu_n}_{\eps/3}(F_n)\right).
    \end{align*}
    Taking the limit on both sides and applying \Cref{prop:composition-triv-tightness}, we get
    \begin{align*}
        \lim_{n\to\infty}\mathsf{Q}^{\mu_n}_{\eps}(F_n)
        = &O_{\eps}\!\left(\lim_{n\to\infty}\mathrm{rAdv}^{\mu_n}_{\eps/3}(F_n)\right)
        =O_{\eps}\!\left((1-\eps/3)\mathrm{sfAdv}^{(\nu_0,\nu_1)}(g)\right)
        =O_{\eps}\!\left(\mathrm{sfAdv}^{(\nu_0,\nu_1)}(g)\right).
    \end{align*}
    Combining both directions completes the proof.
\end{proof}

\newpage
\section{\texorpdfstring{An XOR lemma via a multiplicative analogue of the $\gamma_2$ norm}{An XOR
lemma via a multiplicative analogue of the gamma 2 norm}}

In this section, we develop the tools needed for our new multiplicative adversary measure and the
direct-product theorem. In \Cref{sec:filtered-fidelity}, we introduce a filtered analogue of
fidelity, which we use in \Cref{sec:multiplicative-gamma2} to define our multiplicative analogue of
the $\gamma_2$ norm. In \Cref{sec:xor-lemma}, we then use this measure to prove an XOR lemma for
decision problems.

\subsection{Filtered fidelity}\label{sec:filtered-fidelity} We begin by restricting the minimization
in Alberti’s characterization of fidelity (\Cref{thm:alberti}) using a family of filters.
\begin{definition}
    For a finite set $\Sigma$, a constant $c>1$, and Gram matrices
    $\{P_i\}_{i\in[n]}\subset\Pos(\Sigma)$ with $P_i\circ I=I$, we define the filtered fidelity of a
    pair of PSD matrices $A,B\in\Pos(\Sigma)$ with cost $c$ and with respect to the filters
    $\{P_i\}$ as
    \begin{align*}
        \cF_c(A,B\mid\{P_i\})
        :=\progline{\inf_{\Gamma\succ0}}
        {\sqrt{\langle\Gamma,A\rangle\langle\Gamma^{-1},B\rangle}\quad}
        {\Gamma\circ P_i\preceq c\Gamma\quad\forall i\in[n]}.
    \end{align*}
\end{definition}
For our applications, we consider sign matrices $S_i(x,y)=(-1)^{x_i\oplus y_i}$ as our filters. We
can roughly think of $\cF_c(A,B\mid{S_i})$ as measuring the overlap between $A$ and $B$, using
adversary matrices whose change under a single query is bounded by a multiplicative factor $c$.
\par
In the following, we derive arithmetic and dual formulations of $\cF_c$ in
\Cref{prop:equiv-filtered-fidelity} and collect a few useful and natural properties in
\Cref{prop:filtered-fidelity-properties}.

\subsubsection{Equivalent primal and dual formulations}
\begin{proposition}[Minimax and dual formulation of $\cF_c$]\label{prop:equiv-filtered-fidelity} Let
$\Sigma$ be a finite set, $c>1$ be a constant, and $\{P_i\}_{i\in[n]}\subset\Pos(\Sigma)$ be Gram
matrices with $P_i\circ I=I$. Then, for $A,B\in\Pos(\Sigma)$, we have
\begin{equation}
\begin{progalign}
\progrow
{\cF_c(A,B\mid\{P_i\})}
{\inf_{\Gamma\succ0}}
{\frac12\bigl(\langle\Gamma,A\rangle+\langle\Gamma^{-1},B\rangle\bigr)}
{\Gamma\circ P_i\preceq c\Gamma & \forall i\in[n]}
\\[0.6em]
\progrow
{}
{\max_{C,X_i\succeq0}}
{\cF(C,B)}
{A-C=\sum_iX_i\circ(cJ-\overline{P_i}) &}
\end{progalign}
\end{equation}
In particular, the value is finite and the maximum is attained.
\end{proposition}
\begin{proof}
    By the AM-GM inequality, we immediately get that
    \begin{align*}
    \cF_{c}(A,B|\{P_i\}) \leq
    \progat{5cm}{\inf_{\Gamma \succ 0 }}{\frac{1}{2}(\langle \Gamma,A\rangle + \langle \Gamma^{-1},B\rangle )}{
        \Gamma \circ P_i \preceq c\Gamma && \forall i\in [n]
    }
    \end{align*}
    For the reverse inequality, let $\Gamma \succ 0$ be an arbitrary feasible matrix, where $\Gamma
    \circ P_i \preceq c\Gamma$ for all $i\in [n]$. For any $t>0$, $t\Gamma \succ 0$ is also
    feasible. If $A = 0$ or $B=0$, both sides become zero. So assume $A,B\neq 0$. Then, taking
    $t=\sqrt{\langle\Gamma^{-1},B\rangle/\langle\Gamma,A\rangle}$, we have
    \begin{align*}
        &\progat{5cm}{\inf_{\Gamma' \succ 0 }}{\frac{1}{2}(\langle \Gamma',A\rangle + \langle (\Gamma')^{-1},B\rangle )}{
        \Gamma' \circ P_i \preceq c\Gamma' && \forall i\in [n]}\\
        \leq&
        \progat{5cm}{\inf_{t>0}}{\frac{1}{2}(\langle t\Gamma,A\rangle + \langle (t\Gamma)^{-1},B\rangle )}{
        \Gamma \circ P_i \preceq c\Gamma && \forall i\in [n]}\\
        =&
        \sqrt{\langle \Gamma, A\rangle \langle \Gamma^{-1},B\rangle}.
    \end{align*}
    Since this holds for every feasible $\Gamma$, the two infima are equal, and hence,
    \begin{align*}
        \cF_{c}(A,B|\{P_i\}) &= \progat{5cm}{\inf_{\Gamma\succ 0}}{\frac{1}{2}(\langle \Gamma,A\rangle + \langle \Gamma^{-1},B\rangle )}{
            \Gamma \circ P_i \preceq c\Gamma && \forall i\in [n]}\\
            &=\progat{5cm}{\inf_{\substack{\Gamma\succ0\\H=H^\dagger}}}{\frac{1}{2}(\langle \Gamma,A\rangle + \langle H,B\rangle )}{
                 \Gamma \circ P_i \preceq c\Gamma && \forall i\in [n] \\ H\succeq \Gamma^{-1}}\\
            &=\progat{5cm}{\inf_{\substack{\Gamma=\Gamma^{\dagger}\\H=H^{\dagger}}}}{\frac{1}{2}(\langle \Gamma,A\rangle + \langle H,B\rangle )}{
                \Gamma \circ P_i \preceq c\Gamma && \forall i\in [n] \\
                \begin{bmatrix}
                    H & I \\
                    I & \Gamma
                \end{bmatrix} \succeq 0
            }.
    \end{align*}
    The second line holds as $B\succeq0$ implies $\langle
    H,B\rangle\geq\langle\Gamma^{-1},B\rangle$, and equality can be attained at $H=\Gamma^{-1}$. The
    third line holds by the Schur-complement equivalence
    $H\succeq\Gamma^{-1}\iff\left[\begin{smallmatrix}H&I\\I&\Gamma\end{smallmatrix}\right]\succeq0$.
    This is an SDP. The point $\Gamma = I$ and $H=2I$ is strictly feasible, since $cI-I\circ
    P_i=(c-1)I\succ0$ and $2I-I^{-1}=I\succ0$, and the objective is nonnegative and has a finite
    value at this point. Therefore, Slater's condition is satisfied, so strong duality holds and the
    dual maximum is attained (\Cref{thm:slater}). To get the maximization form, we introduce
    Lagrangian variables $Y\succeq0$ and $\{X_i\}_{i\in [n]}\succeq0$ for the second formulation
    above as follows.
    \begin{align*}
        \cF_{c}(A,B|\{P_i\}) &= \frac{1}{2}\inf_{\substack{\Gamma\succ0\\H=H^\dagger}}\sup_{Y,X_i\succeq0}
        \left(\langle \Gamma,A\rangle + \langle H,B\rangle
        - \sum_i\langle X_i,c\Gamma-\Gamma\circ P_i\rangle
        - \langle Y,H-\Gamma^{-1}\rangle\right)\\
        &= \frac{1}{2}\sup_{Y,X_i\succeq0}\inf_{\substack{\Gamma\succ0\\H=H^\dagger}}
        \left(\langle \Gamma,A\rangle + \langle H,B\rangle
        + \sum_i\langle X_i\circ\overline{P_i}-cX_i,\Gamma\rangle
        - \langle Y,H-\Gamma^{-1}\rangle\right)\\
        &= \frac{1}{2}\sup_{Y,X_i\succeq0}\inf_{\substack{\Gamma\succ0\\H=H^\dagger}}
        \left(\langle \Gamma,A+\sum_i(X_i\circ\overline{P_i}-cX_i)\rangle
        + \langle H,B-Y\rangle
        + \langle \Gamma^{-1},Y\rangle\right).
    \end{align*}
    If $B\neq Y$, then we may choose $H$ so that the infimum goes to $-\infty$. Thus, we may
    restrict to $Y=B$, in which case the term involving $H$ vanishes. In addition, if
    $C:=A+\sum_i(X_i\circ\overline{P_i}-cX_i)\not\succeq0$, choose a unit vector $\ket{v}$ such that
    $\bra{v}C\ket{v}<0$. Taking $\Gamma_t:=I+t\ketbra{v}$ for $t\geq0$ gives
    $0\prec\Gamma_t^{-1}\preceq I$, so $\langle\Gamma_t^{-1},B\rangle\leq\Tr B$, while
    $\langle\Gamma_t,C\rangle=\Tr C+t\bra{v}C\ket{v}\to-\infty$, making the infimum go to $-\infty$.
    So, we may additionally restrict to $C\succeq0$, which gives
    \begin{align*}
        \cF_{c}(A,B|\{P_i\})
        &= \frac{1}{2}\progat{5cm}{\sup_{X_i,C} \inf_{\Gamma \succ 0}}{\langle \Gamma,C\rangle + \langle \Gamma^{-1},B\rangle}{
            C,X_i \succeq 0 \quad \forall i \in [n] \\
            C = A+\sum_i(X_i\circ\overline{P_i}-cX_i)
        } \\
        \cF_{c}(A,B|\{P_i\})
        &= \progat{5cm}{\sup_{X_i,C}}{\cF(C,B)}{
            C,X_i \succeq 0 \quad \forall i \in [n] \\
            C = A+\sum_i(X_i\circ\overline{P_i}-cX_i)
        },
    \end{align*}
    where the last line is by Alberti's theorem (\Cref{thm:alberti}). The supremum is attained by
    the dual-attainment conclusion above, completing the proof.
\end{proof}

\subsubsection{Basic properties}

\begin{proposition}\label{prop:filtered-fidelity-properties} Let $\Sigma$ be a finite set, $c>1$ be
a constant, and $\{P_i\}_{i\in[n]}\subset\Pos(\Sigma)$ be Gram matrices with $P_i\circ I=I$. Then,
for $A,B\in\Pos(\Sigma)$, the following properties hold.
    \begin{enumerate}[i.]
        \item \textbf{Scalability.}\label{item:filtered-fidelity-scalability} For $a,b>0$,
        $\cF_c(aA,bB\mid\{P_i\})=\sqrt{ab}\,\cF_c(A,B\mid\{P_i\})$.
        \item
        \textbf{Tensorization.}\label{item:filtered-fidelity-tensorization} Let $\Sigma_1,\Sigma_2$
        be finite sets, $A_1,B_1\in\Pos(\Sigma_1)$  and $A_2,B_2\in\Pos(\Sigma_2)$. Also, let
        $\{P_i\}_{i=1}^{n}\subset\Pos(\Sigma_1)$ and $\{Q_j\}_{j=1}^{m}\subset\Pos(\Sigma_2)$ be
        families of Gram matrices with $P_i\circ I=I$ and $Q_j\circ I=I$ for every $i,j$. Then,
        \begin{align*}
            \cF_c\bigl(A_1\otimes A_2,B_1\otimes B_2
            \mid\{P_i\otimes J\}_i\cup\{J\otimes Q_j\}_j\bigr)
            =\cF_c(A_1,B_1\mid\{P_i\})\,
             \cF_c(A_2,B_2\mid\{Q_j\}).
        \end{align*}
        \item \textbf{Log-log convexity in $c$.}\label{item:filtered-fidelity-interpolation} For
        $c_0,c_1>1$ and $\lambda\in(0,1)$,
        \begin{align*}
            \cF_{c_0^{1-\lambda}c_1^\lambda}(A,B\mid\{P_i\})
            \le
            \cF_{c_0}(A,B\mid\{P_i\})^{1-\lambda}
            \cF_{c_1}(A,B\mid\{P_i\})^\lambda.
        \end{align*}
        \item \textbf{Joint concavity.}\label{item:filtered-fidelity-joint-concavity} For
        $A_1,A_2,B_1,B_2\in\Pos(\Sigma)$ and $\lambda\in[0,1]$,
        \begin{align*}
            &\cF_c\bigl(\lambda A_1+(1-\lambda)A_2,
            \lambda B_1+(1-\lambda)B_2\mid\{P_i\}\bigr)
            \ge
            \lambda\cF_c(A_1,B_1\mid\{P_i\})
            +(1-\lambda)\cF_c(A_2,B_2\mid\{P_i\}).
        \end{align*}
        \item \textbf{Data processing.} Let $\cE:\cL(\Sigma)\to\cL(\Sigma)$ be a quantum channel
        satisfying $\cE(X\circ\overline{P_i})=\cE(X)\circ\overline{P_i}$ for every $X\in\cL(\Sigma)$
        and every $i$. Then,
        $
            \cF_c\bigl(\cE(A),\cE(B)\mid\{P_i\}\bigr)
            \ge \cF_c(A,B\mid\{P_i\}).
        $
        \item \textbf{Continuity.}\label{item:filtered-fidelity-continuity} For fixed filters
        $\{P_i\}$, the map $(A,B,c)\mapsto\cF_c(A,B\mid\{P_i\})$ is jointly continuous on
        $\Pos(\Sigma)\times\Pos(\Sigma)\times(1,\infty)$.
        \item \textbf{Monotonicity in $c$.}\label{item:filtered-fidelity-cost-monotonicity} For
        $1<c\le d$,
        \begin{align*}
            0\le \cF(A,B)
            \le \cF_d(A,B\mid\{P_i\})
            \le \cF_c(A,B\mid\{P_i\})
            \le \sqrt{\Tr(A)\Tr(B)}.
        \end{align*}
        Moreover, $\lim_{c\to\infty}\cF_c(A,B\mid\{P_i\})=\cF(A,B)$.
        \item \textbf{Monotonicity in $A$ and $B$.}\label{item:filtered-fidelity-input-monotonicity}
        For $A\preceq A'$ and $B\preceq B'$
        \begin{align*}
            \cF_c(A,B\mid\{P_i\})
            \leq \cF_c(A',B'\mid\{P_i\}).
        \end{align*}
        \item \textbf{Lower bound on fidelity.}\label{item:filtered-fidelity-lowerbound}
        \begin{align*}
            \cF(A,B)
            \geq
            \frac{\sqrt{c}\,\cF_c(A,B\mid\{P_i\})-\sqrt{\Tr(A)\Tr(B)}}{\sqrt{c}-1}.
        \end{align*}
        \item \textbf{Hadamard restriction.}\label{item:filtered-fidelity-hadamard-restriction} Let
        $S\subseteq\Sigma$. Then,
        \[
            \cF_c(A\circ J_S,B\circ J_S\mid\{P_i\})
            =\cF_c(A|_S,B|_S\mid\{P_i|_S\}),
        \]
        where $X|_S$ denotes the principal submatrix of $X$ indexed by $S$.
    \end{enumerate}
    Moreover, suppose $\Sigma\subseteq\{0,1\}^m$, and let $\{S_i\}_{i\in[m]}$ be sign filters. Then,
    we have the following additional properties.
    \begin{enumerate}[i.]
        \setcounter{enumi}{10}
        \item \textbf{Symmetry.}\label{item:filtered-fidelity-symmetry}
        $\cF_c(A,B\mid\{S_i\})=\cF_c(B,A\mid\{S_i\})$.
        \item \textbf{Vanishing.}\label{item:filtered-fidelity-vanishing} $\cF_c(A,B\mid\{S_i\})=0$
        if and only if $A\circ B=0$.
    \end{enumerate}
\end{proposition}
For the sake of brevity, we postpone the proofs to \Cref{app:filtered-fidelity-properties}.

\subsection{\texorpdfstring{Multiplicative variant of $\gamma_2$ norm}{Multiplicative variant of the
gamma 2 norm}}\label{sec:multiplicative-gamma2} Next, we define a multiplicative variant of the
$\bar{\gamma}_2$ measure as follows.

\begin{definition}\label{def:filtered-delta2} Let $\Sigma\subseteq\{0,1\}^{n}$, and let
    $\{S_i\}_{i\in[n]}$ be sign filters on $\Sigma$. The filtered $\delta_2$ measure of a pair of
    PSD matrices $A,B\in\Pos(\Sigma)$ with respect to a distribution $\mu:\Sigma\to[0,1]$ is defined
    as\footnote{If the filtered fidelity vanishes for some $c>1$, we define the value of the
    function to be $+\infty$.}
    \begin{align*}
        \bar{\delta}_2^{\mu}(A,B\mid\{S_i\})
        :=\sup_{c>1}
        \frac{-2\log\cF_c\bigl(
            A\circ\kbr{\mu},B\circ\kbr{\mu}\mid\{S_i\}\bigr)}
        {\log c}.
    \end{align*}
\end{definition}
With the previous intuition for $\cF_c$, we can think of $ \bar{\delta}_2^{\mu}$ as measuring the
log-over-log ratio of distinguishability versus per-query cost. We first show how
$\bar{\gamma}_2^{\mu}$ can be recovered from $\bar{\delta}_2^{\mu}$ in
\Cref{prop:delta2-gamma2-connection}, and then establish a few useful properties of
$\bar{\delta}_2^{\mu}$ in \Cref{prop:avgdelta2-properties}. Finally, we show that it provides upper
and lower bounds on quantum query complexity in \Cref{prop:filtered-delta2-upperbound} and
\Cref{prop:filtered-delta2-lowerbound}, respectively. Finally, we obtain an asymptotic upper bound
in \Cref{prop:delta-asymptotic-state-conversion}.

\subsubsection{\texorpdfstring{Analogy with $\bar{\gamma}_2^{\mu}$}{Analogy with the distributional
gamma 2 measure}} We begin by mentioning that $\bar{\gamma}_2^\mu$ can be recovered from
$\bar{\delta}_2^\mu$ as follows.
\begin{proposition}\label{prop:delta2-gamma2-connection} Let $\Sigma\subseteq\{0,1\}^{n}$, let
    $\mu:\Sigma\to[0,1]$ be a distribution, and let $\{S_i\}_{i\in[n]}$ be sign filters on $\Sigma$.
    Then, for a pair of Gram matrices $\Phi$ and $\Psi$ with $\Phi \circ I = I$ and $\Psi \circ I =
    I$, we have the following characterization:
    \begin{align*}
        \frac12\bar{\gamma}_2^{\mu}(\Phi - \Psi|\{\Delta_i\}) =
        \sup_{c>1}
        \frac{1 - \cF_c^2(\Phi \circ \kbr{\mu}, \Psi \circ \kbr{\mu}|\{S_i\})}{c-1}
         = \bar{\delta}_2^{\mu}(\Phi, \Psi|\{S_i\}).
    \end{align*}
\end{proposition}
For the sake of continuity, we defer the proof to \Cref{app:delta2-gamma2-connection}.

\subsubsection{Basic properties}

\begin{proposition}\label{prop:avgdelta2-properties} Let $\Sigma\subseteq\{0,1\}^{n}$, let
    $A,B\in\Pos(\Sigma)$, and let $\{S_i\}_{i\in[n]}$ be sign filters on $\Sigma$. For any
    distribution $\mu:\Sigma\to[0,1]$, the following hold.
    \begin{enumerate}[i.]
        \item \textbf{Symmetry.} $\bar{\delta}_2^\mu(A,B\mid\{S_i\})
        =\bar{\delta}_2^\mu(B,A\mid\{S_i\})$.
        \item \textbf{Scalability.} For $a,b>0$, $\bar{\delta}_2^\mu(aA,bB\mid\{S_i\})
        =\bar{\delta}_2^\mu(A,abB\mid\{S_i\}) =\bar{\delta}_2^\mu(abA,B\mid\{S_i\})$.
        \item \textbf{Tensorization.} Let $k\in\bN$ with $k\geq1$. Then,
        \begin{align*}
            \bar{\delta}_2^{\mu^{\otimes k}}(A^{\otimes k},B^{\otimes k}\mid\{S_{ij}\}_{i\in[k],\,j\in[n]})
            =k\,\bar{\delta}_2^\mu(A,B\mid\{S_i\}),
        \end{align*}
        where $S_{ij}:=J^{\otimes(k-1)}\otimes_i S_j$.
        \item \textbf{Support restriction.}\label{item:avgdelta2-support} Let $S:=\supp(\mu)$. Then,
        \begin{align*}
            \bar{\delta}_2^\mu(A,B\mid\{S_i\})
            &=\bar{\delta}_2^\mu(A\circ J_S,B\circ J_S\mid\{S_i\})
            =\bar{\delta}_2^{\mu|_S}(A|_S,B|_S\mid\{S_i|_S\}),
        \end{align*}
        where $X|_S$ denotes the principal submatrix of $X$ indexed by $S$.
        \item \textbf{Nonnegativity and vanishing.}\label{item:avgdelta2-nonnegative}
        $\bar{\delta}_2^\mu(A,B\mid\{S_i\})\geq0$, with equality if and only if
        $\cF(A\circ\kbr{\mu},B\circ\kbr{\mu})\geq1$.
        \item \textbf{Lower semicontinuity.}\label{item:avgdelta2-lower-semicontinuity} For fixed
        $\mu$, the map $(A,B)\mapsto\bar{\delta}_2^\mu(A,B\mid\{S_i\})$ is lower semicontinuous.
    \end{enumerate}
\end{proposition}
\begin{proof}
    Properties (i)--(ii) follow directly from the corresponding properties of $\cF_c$ in
    \Crefitem{item:filtered-fidelity-symmetry,item:filtered-fidelity-scalability} by taking the
    supremum over $c>1$. We prove the remaining items as follows.
    \begin{itemize}
    \item[(iii)] Exact tensorization of $\cF_c$ (\Crefitem{item:filtered-fidelity-tensorization})
    gives
    \begin{align*}
        \bar{\delta}_2^{\mu^{\otimes k}}(A^{\otimes k},B^{\otimes k}\mid\{S_{ij}\})
        &=\sup_{c>1}\frac{-2\log\cF_c\bigl(A^{\otimes k}\circ\kbr{\mu^{\otimes k}},B^{\otimes k}\circ\kbr{\mu^{\otimes k}}\mid\{S_{ij}\}\bigr)}{\log c}\\
        &=\sup_{c>1}\frac{-2k\log\cF_c(A\circ\kbr{\mu},B\circ\kbr{\mu}\mid\{S_i\})}{\log c}\\
        &=k\sup_{c>1}\frac{-2\log\cF_c(A\circ\kbr{\mu},B\circ\kbr{\mu}\mid\{S_i\})}{\log c}
        =k\,\bar{\delta}_2^\mu(A,B\mid\{S_i\}).
    \end{align*}
    \item[(iv)] For $M\in\{A,B\}$, we have $(M\circ J_S)\circ\kbr{\mu}=M\circ\kbr{\mu}$, so
    $\bar{\delta}_2^\mu(A,B\mid\{S_i\})=\bar{\delta}_2^\mu(A\circ J_S,B\circ J_S\mid\{S_i\})$.
    Moreover, by \Crefitem{item:filtered-fidelity-hadamard-restriction}, for every $c>1$,
    \begin{align*}
        &\cF_c((A\circ J_S)\circ\kbr{\mu},(B\circ J_S)\circ\kbr{\mu}\mid\{S_i\})\\
        &\qquad=\cF_c(A|_S\circ\kbr{\mu|_S},B|_S\circ\kbr{\mu|_S}\mid\{S_i|_S\}).
    \end{align*}
    Taking the supremum over $c>1$ in the definition of $\bar{\delta}_2^\mu$ proves the second
    equality.
    \item[(v)] If the filtered fidelity vanishes for some $c>1$, then $\bar{\delta}_2^\mu=+\infty$
    by definition. Otherwise, let
    \begin{align*}
        M:=-\log\bigl(\Tr[A\circ\kbr{\mu}]\Tr[B\circ\kbr{\mu}]\bigr).
    \end{align*}
    By \Crefitem{item:filtered-fidelity-cost-monotonicity}, for every $c>1$ we have
    \begin{align*}
        -2\log\cF_c(A\circ\kbr{\mu},B\circ\kbr{\mu}\mid\{S_i\})\geq M.
    \end{align*}
    Therefore,
    \begin{align*}
        \bar{\delta}_2^\mu(A,B\mid\{S_i\})
        \geq\sup_{c>1}\frac{M}{\log c}
        =
        \begin{cases}
            +\infty,& M>0,\\
            0,& M\leq0.
        \end{cases}
    \end{align*}
    Moreover, $\bar{\delta}_2^\mu(A,B\mid\{S_i\})=0$ if and only if
    $\cF_c(A\circ\kbr{\mu},B\circ\kbr{\mu}\mid\{S_i\})\geq1$ for every $c>1$. By
    \Crefitem{item:filtered-fidelity-cost-monotonicity}, taking $c\to\infty$ implies
    $\cF(A\circ\kbr{\mu},B\circ\kbr{\mu})\geq1$. Conversely, if
    $\cF(A\circ\kbr{\mu},B\circ\kbr{\mu})\geq1$, then, for every $c>1$,
    \[
        \cF_c(A\circ\kbr{\mu},B\circ\kbr{\mu}\mid\{S_i\})
        \geq\cF(A\circ\kbr{\mu},B\circ\kbr{\mu})\geq1.
    \]
    So, $\bar{\delta}_2^\mu(A,B\mid\{S_i\})=0$.
    \item[(vi)] Let $A_k\to A$ and $B_k\to B$ be given. For every fixed $c>1$, continuity of $\cF_c$
    gives
    \begin{align*}
        \liminf_k\bar{\delta}_2^\mu(A_k,B_k\mid\{S_i\})
        &\geq \liminf_k
        \frac{-2\log\cF_c(A_k\circ\kbr{\mu},B_k\circ\kbr{\mu}\mid\{S_i\})}{\log c}\\
        &=\frac{-2\log\cF_c(A\circ\kbr{\mu},B\circ\kbr{\mu}\mid\{S_i\})}{\log c}.
    \end{align*}
    Taking the supremum over $c>1$ on the right gives
    \begin{align*}
        \liminf_k\bar{\delta}_2^\mu(A_k,B_k\mid\{S_i\})
        \geq\bar{\delta}_2^\mu(A,B\mid\{S_i\}),
    \end{align*}
    proving lower semicontinuity. \qedhere
    \end{itemize}
\end{proof}

\subsubsection{Lower bounding query complexity}
Similar to \cite{Spa08,LR12}, we track the multiplicative change of a progress function to obtain
the following lower bound.

\begin{proposition}\label{prop:filtered-delta2-lowerbound} Let $\Sigma\subset \{0,1\}^{n}$,
    $\mu:\Sigma\to[0,1]$ be a distribution, $W\in\Pos(\Sigma)$ be a target witness satisfying
    $\Tr[W\circ\kbr{\mu}]=1$, and $\alpha=\alpha(n)\in(0,1]$ be a success parameter. Then, any Gram
    matrix $\Psi\in\Pos(\Sigma)$ with $\Psi\circ I=I$ and
    $\cF(\Psi\circ\kbr{\mu},W\circ\kbr{\mu})\geq\sqrt{\alpha}$ has query complexity at least
   \begin{align*}
    \mathsf{Q}(\Psi) \geq \bar{\delta}_2^{\mu}(J,W/\alpha\mid\{S_i\}).
   \end{align*}
\end{proposition}
\begin{proof}
    Let $Q$ be a quantum algorithm attaining $T=\mathsf{Q}(\Psi)$. Let $\Psi^t$ be the Gram matrix
    of the algorithm after the $t$-th query, as in \Cref{lem:one-query-decomposition}, so $\Psi^0=J$
    and $\Psi^T=\Psi$. Then, for every $t\in[T]$, there exists a PSD decomposition
    $\{Z_i^t\}_{i=0}^n$ such that
    \begin{align*}
        \Psi^{t-1}=\sum_{i=0}^n Z_i^t,
        \qquad
        \Psi^{t-1}=\Psi^t+2\sum_{i=1}^n Z_i^t\circ\Delta_i.
    \end{align*}
    Letting $S_0=J$ and using $S_i=J-2\Delta_i$, we can therefore write $\Psi^t=\sum_{i=0}^n
    S_i\circ Z_i^t$.
\par
Fix $c>1$, and let $\Gamma\succ0$ be an arbitrary feasible matrix for $\cF_c$; that is, $\Gamma\circ
S_i\preceq c\Gamma$ for every $i\in[n]$. Defining $d_t:=\langle\Gamma,\Psi^t\circ\kbr{\mu}\rangle$
for $t\in[0:T]$, we have
\begin{align*}
    d_t &= \langle \Gamma, \Psi^{t} \circ \kbr{\mu}\rangle
    = \sum_{i=0}^n \langle \Gamma, S_i\circ Z_i^t\circ\kbr{\mu}\rangle
    = \sum_{i=0}^n \langle \Gamma\circ S_i,Z_i^t\circ\kbr{\mu}\rangle \\
    &\leq c\sum_{i=0}^n\langle\Gamma,Z_i^t\circ\kbr{\mu}\rangle
    = c d_{t-1},
\end{align*}
where in the last line we used $\Gamma\circ S_0=\Gamma\preceq c\Gamma$ and
$Z_i^t\circ\kbr{\mu}\succeq0$. By repeated iteration, we therefore get
\begin{align*}
    \langle\Gamma,\Psi\circ\kbr{\mu}\rangle
    =d_T\leq c^T d_0
    =c^T\langle\Gamma,\kbr{\mu}\rangle.
\end{align*}
\par
Next, by the fidelity assumption and Alberti's characterization (\Cref{thm:alberti}), we have
\begin{align*}
    \alpha &\leq \cF^2(\Psi\circ\kbr{\mu},W\circ\kbr{\mu})
    \leq \langle\Gamma,\Psi\circ\kbr{\mu}\rangle\langle\Gamma^{-1},W\circ\kbr{\mu}\rangle \leq c^T\langle\Gamma,\kbr{\mu}\rangle\langle\Gamma^{-1},W\circ\kbr{\mu}\rangle.
\end{align*}
Since $\Gamma$ was an arbitrary feasible matrix, we can take the infimum on the right-hand side to
get
$$\alpha\leq c^T\cF_c^2(\kbr{\mu},W\circ\kbr{\mu}\mid\{S_i\}).$$ By
\Crefitem{item:filtered-fidelity-scalability}, we can rescale both sides by $1/\alpha$ to get
\begin{align*}
    c^T\cF_c^2(\kbr{\mu},(W/\alpha)\circ\kbr{\mu}\mid\{S_i\})\geq1,
\end{align*}
which implies
\begin{align*}
    T\geq\frac{-2\log\cF_c(\kbr{\mu},(W/\alpha)\circ\kbr{\mu}\mid\{S_i\})}{\log c}.
\end{align*}
Finally, since $c>1$ was arbitrary, we can take the supremum over $c>1$ to get
\begin{align*}
    \mathsf{Q}(\Psi)=T
    &\geq \sup_{c>1}\frac{-2\log\cF_c(\kbr{\mu},(W/\alpha)\circ\kbr{\mu}\mid\{S_i\})}{\log c}\\
    &=\bar{\delta}_2^\mu(J,W/\alpha\mid\{S_i\}). \qedhere
\end{align*}
\end{proof}

\subsubsection{Upper bounding query complexity}
Using the extension of the catalyst state-conversion theorem in
(\Cref{thm:catalyst-state-conversion}), we can prove the following upper bound on quantum query
complexity.

\begin{proposition}\label{prop:filtered-delta2-upperbound} Let $\Sigma\subset\{0,1\}^{n}$,
    $\mu:\Sigma\to[0,1]$ be a distribution, and $W\in\Pos(\Sigma)$ be a target witness satisfying
    $\Tr[W\circ\kbr{\mu}]=1$. Let $\alpha=\alpha(n)$ and $\beta=\beta(n)$ be success parameters with
    $0<\beta<\alpha<1$, and assume that $\bar{\delta}_2^\mu(J,W/\alpha\mid\{S_i\})<+\infty$. Then,
    there exists a Gram matrix $\Psi\in\Pos(\Sigma)$ with $\Psi\circ I=I$ such that
    $\cF(\Psi\circ\kbr{\mu},W\circ\kbr{\mu})\geq\sqrt{\beta}$ and
    \begin{align*}
        \mathsf{Q}(\Psi)
        \leq
        \frac{10\alpha}{\alpha-\beta}
        \log\left(\frac{2(1-\beta)}{\alpha-\beta}\right)
        \bar{\delta}_2^\mu(J,W/\alpha\mid\{S_i\}).
    \end{align*}
\end{proposition}
\begin{proof}
    Without loss of generality, we may assume that $\Sigma=\supp(\mu)$. Indeed, replacing each
    matrix by its principal submatrix indexed by $S:=\supp(\mu)$ leaves $\bar{\delta}_2^\mu$
    unchanged (see \Crefitem{item:avgdelta2-support}). Also, an algorithm on $S$ extends to $\Sigma$
    without increasing its query cost or changing the fidelity under $\mu$.
    \par
    Let $d:= \bar{\delta}_{2}^{\mu}(J,W/\alpha|\{S_i\})< \infty$. If $d=0$, the definition of
    $\bar{\delta}_2^\mu$ and scalability give
    $\cF_c(\kbr{\mu},W\circ\kbr{\mu}\mid\{S_i\})\geq\sqrt{\alpha}$ for every $c>1$. Letting
    $c\to\infty$ gives $\cF(\kbr{\mu},W\circ\kbr{\mu})\geq\sqrt{\alpha}$, so taking $\Psi:=J$
    satisfies the fidelity condition with zero queries. Hence, we assume $d>0$.
    \par
    Define $m := (\alpha + \beta)/2$, and let $L:= (1-\sqrt{\beta})/(\sqrt{m}-\sqrt{\beta}) >1$.
    Then,
\begin{align}\label{eq:filtered-delta2-L-identity}
    L-1 &= \frac{1-\sqrt{m}}{\sqrt{m}-\sqrt{\beta}}
    = \frac{L\sqrt{m}-1}{\sqrt{\beta}}.
\end{align}

    Let $c:=(\alpha/m)^{1/d}>1$. By the definition of $d$, we have
    $(-2\log\cF_c(\kbr{\mu},W\circ\kbr{\mu}/\alpha\mid\{S_i\}))/(\log c)\leq d$. By rearrangement
    and scalability (\Crefitem{item:filtered-fidelity-scalability}), this gives
    \begin{align*}
        \cF_c(\kbr{\mu},W\circ\kbr{\mu}\mid\{S_i\})
        &=\sqrt{\alpha}\,\cF_c(\kbr{\mu},W\circ\kbr{\mu}/\alpha\mid\{S_i\})
        \geq\sqrt{\alpha}c^{-d/2}
        =\sqrt{m}.
    \end{align*}
    \par
    We claim that if $c \geq L^2$ ($d\leq\log(\alpha/m)/(2\log L)$), then taking $\Psi:=J$ satisfies
    the fidelity condition, and we are done with zero queries. Indeed, in that case $c=
    (\alpha/m)^{1/d}\geq L^2$, and therefore, $\sqrt{c}\geq L>1$. By
    \Crefitem{item:filtered-fidelity-lowerbound}, and since both matrices have trace one, we get
    \begin{align*}
        \cF(\kbr{\mu},W\circ\kbr{\mu})
        &\geq\frac{\sqrt{c}\cF_c(\kbr{\mu},W\circ\kbr{\mu}\mid\{S_i\})-1}{\sqrt{c}-1}
        \geq\frac{\sqrt{c}\sqrt{m}-1}{\sqrt{c}-1}
        \geq\frac{L\sqrt{m}-1}{L-1}
        =\sqrt{\beta}.
    \end{align*}
    The last equality follows from \eqref{eq:filtered-delta2-L-identity}. Hence, we assume $c<L^2$
    throughout the rest of the proof.
    \par
    The rest of the proof is organized as follows.
    \begin{itemize}
        \item \textbf{Constructing $\Psi$.} First, we use the maximization form of $\cF_c$ to
        construct a target Gram matrix $\Psi$ with large overlap
        $\cF(\Psi\circ\kbr{\mu},W\circ\kbr{\mu})\geq\sqrt{m}$.
        \item \textbf{State conversion.} We then use the decomposition of $J-\Psi$ to apply the
        state-conversion theorem \Cref{thm:catalyst-state-conversion} with query cost $T$. This
        gives us an approximate target $\widetilde{\Psi}$ satisfying
        $\cF(\widetilde{\Psi}\circ\kbr{\mu},W\circ\kbr{\mu})\geq
        \eta(1-\eps_T)-\sqrt{\eps_T(1-\eta)}$.
        \item \textbf{Choosing $T$.} Finally, we choose $T$ to obtain fidelity at least
        $\sqrt{\beta}$ and the required query bound.
    \end{itemize}

    \paragraph{Constructing $\Psi$.}
    By the equivalent maximization formulation in \Cref{prop:equiv-filtered-fidelity}, there exist
    $C,X_i\succeq 0$ such that
    \begin{align*}
        \cF(C,W\circ\kbr{\mu})
        &=\cF_c(\kbr{\mu},W\circ\kbr{\mu}\mid\{S_i\})\geq\sqrt{m}
    \end{align*}
    and
    \begin{align}\label{eq:filtered-delta2-primal-decomposition}
        J\circ\kbr{\mu}-C
        &=\sum_iX_i\circ(cJ-S_i)
        =(c-1)\sum_iX_i+2\sum_iX_i\circ\Delta_i.
    \end{align}
    We first normalize both sides to get rid of the extra $\ketbra{\mu}$. More specifically, we
    define $\hat{C}\in\Pos(\Sigma)$ and $\hat{X}_i\in\Pos(\Sigma)$ so that $\hat{C}\circ\kbr{\mu}=C$
    and $\hat{X}_i\circ\kbr{\mu}=X_i$. These matrices are well-defined since $\mu(x)>0$ for every
    $x\in\Sigma$, and positive semidefinite by the Schur product theorem. Therefore,
    \begin{align*}
        J-\hat{C}=(c-1)\sum_i\hat{X}_i+2\sum_i\hat{X}_i\circ\Delta_i.
    \end{align*}
    Next, we claim that $\Psi:=\hat{C}+(c-1)\sum_i\hat{X}_i$ is a Gram matrix with $\Psi\circ I=I$.
    Indeed, by the above, $\Psi(x,x)=J(x,x)-2\sum_i(\hat{X}_i\circ\Delta_i)(x,x)=1$, as
    $\Delta_i(x,x)=0$. Also, since $\Psi \succeq \hat{C}$,
    \begin{align*}
        \eta:=\cF(\Psi \circ \kbr{\mu}, W\circ \kbr{\mu}) \geq \cF(\hat{C} \circ \kbr{\mu},W\circ \kbr{\mu}) \geq \sqrt{m}.
    \end{align*}
    Additionally,
    \begin{align*}
        J-\Psi=2\sum_i\hat{X}_i\circ\Delta_i,
        \qquad
        \Psi\succeq(c-1)\sum_i\hat{X}_i,
    \end{align*}
    where the first equality follows from the equation above.

    \paragraph{Applying state conversion.}
    For every integer $T\geq1$, \Cref{thm:catalyst-state-conversion}, applied with $\Phi=J$, $r=c$,
    and witnesses $2\hat{X}_i$, gives a Gram matrix $\widetilde{\Psi}$ with $\widetilde{\Psi}\circ
    I=I$, $\mathsf{Q}(\widetilde{\Psi})\leq T$, and
    \begin{align*}
        \cF(\widetilde{\Psi}\circ\kbr{\mu},W\circ\kbr{\mu})
        \geq \eta(1-\eps_T)-\sqrt{\eps_T(1-\eta)},
    \end{align*}
    where
    \begin{align}\label{eq:filtered-delta2-epsilon-bound}
        \eps_T &:=\frac{2(c-1)}{c^T-1}\sum_i\Tr[X_i]=\frac{2(c-1)}{c^T-1}\cdot\frac{1-\Tr[C]}{c-1}\leq\frac{2(1-m)}{c^T-1}\leq\frac{4(1-\sqrt{m})}{c^T-1}.
    \end{align}
    Here, we used the primal decomposition \eqref{eq:filtered-delta2-primal-decomposition} and
    $\Tr[C]=\Tr[C]\Tr[W\circ\kbr{\mu}]\geq\cF(C,W\circ\kbr{\mu})^2\geq m$. Consequently, we can
    continue the lower bound as
    \begin{align*}
        \cF(\widetilde{\Psi}\circ\kbr{\mu},W\circ\kbr{\mu})
        &\geq\sqrt{m}(1-\eps_T)-\sqrt{\eps_T(1-\sqrt{m})}
        &&\left(\text{By }\sqrt{m}\leq\eta\right)\\
        &\geq\sqrt{m}-\eps_T-\sqrt{\eps_T(1-\sqrt{m})}
        &&\left(\text{By }\sqrt{m}\leq1\right)\\
        &\geq\sqrt{m}-(1-\sqrt{m})
        \left(\frac{4}{c^T-1}+\frac{2}{\sqrt{c^T-1}}\right)
        &&\left(\text{By }\eqref{eq:filtered-delta2-epsilon-bound}\right)\\
        &\geq\sqrt{m}-(1-\sqrt{m})\cdot\frac{2}{\sqrt{c^T-1}-2}
        &&\text{(By geometric series)}\\
        &=\sqrt{m}-(\sqrt{m}-\sqrt{\beta})
        \frac{2(L-1)}{\sqrt{c^T-1}-2}
        &&\left(\text{By }\eqref{eq:filtered-delta2-L-identity}\right).
    \end{align*}
    In the fourth line, we assumed $c^T-1>4$.
    \paragraph{Choosing the parameter $T$}
    To obtain fidelity at least $\sqrt{\beta}$, it suffices for the fraction in the last line to be
    at most $1$, i.e., $c^T-1\geq4L^2$. Therefore, we choose
    \begin{align*}
        T:=\left\lceil\frac{\log(1+4L^2)}{\log c}\right\rceil.
    \end{align*}
    Since $L>1$, this choice also satisfies the condition $c^T-1>4$, so the fidelity bound follows.
    To obtain the query upper bound, let $R:=2(1-\beta)/(\alpha-\beta)\geq2$. Then,
    \begin{align}\label{eq:filtered-delta2-L-upperbound}
        L
        &=\frac{1-\sqrt{\beta}}{\sqrt{m}-\sqrt{\beta}}
        =\frac{(1-\sqrt{\beta})(\sqrt{m}+\sqrt{\beta})}{m-\beta}
        \leq R,
    \end{align}
    where we used $\sqrt{m}\leq1$ and $m-\beta=(\alpha-\beta)/2$. Therefore, we can bound the query
    complexity as follows.
    \begin{align*}
        \mathsf{Q}(\widetilde{\Psi})
        &\leq T=\left\lceil\frac{\log(1+4L^2)}{\log c}\right\rceil\\
        &\leq\frac{\log(1+4L^2)}{\log c}+1\\
        &\leq\frac{3+2\log L}{\log c}+\frac{2\log L}{\log c}
        &&\left(\text{By }c<L^2\right)\\
        &\leq\frac{7\log R}{\log c}
        &&\left(\text{By }L\leq R,\ 2\leq R\right).
    \end{align*}
    Next, we bound $\log c$. By the definitions of $c$ and $m$,
    \begin{align}\label{eq:filtered-delta2-log-c-bound}
        \log c
        &=\frac{1}{d}\log(\alpha/m)
        =\frac{1}{d}\log\frac{2\alpha}{\alpha+\beta}
        =-\frac{1}{d}\log\left(1-\frac{\alpha-\beta}{2\alpha}\right)
        \geq\frac{\alpha-\beta}{2\alpha d\ln2},
    \end{align}
    where we used $-\ln(1-x)\geq x$ for $0<x<1$. Consequently,
    \begin{align*}
        \mathsf{Q}(\widetilde{\Psi})
        &\leq\frac{7}{\log c}\log R
        \leq\frac{14\alpha d\ln2}{\alpha-\beta}
        \log\frac{2(1-\beta)}{\alpha-\beta}
        \leq\frac{10\alpha d}{\alpha-\beta}
        \log\frac{2(1-\beta)}{\alpha-\beta}.\qedhere
    \end{align*}

\end{proof}

\subsubsection{Asymptotic upper bound on query complexity}

\begin{proposition}[Asymptotic state conversion]\label{prop:delta-asymptotic-state-conversion} Let
    $\Sigma\subseteq\{0,1\}^n$, let $\mu:\Sigma\to[0,1]$ be a distribution, and
    $\Phi,W\in\Pos(\Sigma)$ satisfy $\Phi\circ I=I$ and $\Tr(W\circ\kbr{\mu})=1$.  Let $0<\alpha
    \leq 1$ be a success parameter and suppose $\bar{\delta}_2^{\mu}(J,W/\alpha|\{S_i\}) <\infty $
    is finite. The, there exists a sequence of Gram matrices $\Psi_k$ on $\Sigma^{k}$ such that for
    large enough $k$, we have
    \begin{align*}
        &Q(\Psi_k) \leq 2k\bar{\delta}_2^{\mu}(J,W/\alpha|\{S_i\}) + o(k),
        &&
        \cF(\Psi_k\circ (\kbr{\mu})^{\otimes k}, W^{\otimes k}\circ (\kbr{\mu})^{\otimes k}) \geq \sqrt{\alpha^{k}2^{-o(k)}}.
    \end{align*}
\end{proposition}

\begin{proof}
    As in the proof of \Cref{prop:filtered-delta2-upperbound}, we may assume that
    $\Sigma=\supp(\mu)$ and $\delta:=\bar{\delta}_2^\mu(J,W/\alpha\mid\{S_i\})>0$. For simplicity,
    for a Gram matrix $\Psi \in \Pos(\Sigma^{k})$ we denote $F(\Psi) := \cF(\Psi \circ
    (\kbr{\mu})^{\otimes k}, W^{\otimes k}\circ (\kbr{\mu})^{\otimes k})$. We claim that it suffices
    to show that for sufficiently large $k$, there exists a Gram matrix $\Psi_k\in \Pos(\Sigma^{k})$
    such that
\begin{align}\label{eq:asymptotic-state-conversion-target}
        &\frac{Q(\Psi_k)}{k} \leq 2\delta + \frac{2}{k}
        &&
        \frac{\log F(\Psi_k)}{k} \geq \frac{\log \alpha}{2} - \frac{4}{k}.
\end{align}
Indeed, repeating the same algorithm on $k$ many blocks of $k$ inputs (that is, $k^2$ inputs) yields
an algorithm with query complexity $Q(\Psi^{\otimes k}_k) \leq 2k^2 \delta + O(k)$ and has fidelity
\begin{align*}
    F(\Psi_k^{\otimes k}) \geq \alpha^{k^2/2}2^{-4k}  = \sqrt{\alpha^{k^2}2^{-8k}}.
\end{align*}
To find such an algorithm, we proceed in two steps: we first prove the following claim, which gives
a trade-off between query complexity and fidelity for every value of $r>1$. Then, use this claim to
btain two algorithms: one faster but less accurate and the other slower but more acurrate. Then, we
balance these two algorithms on multiple blocks to achieve both of the desired  fidelity and query
complexities.
    \begin{claim}
       Let $k\geq1$ be a fixed integer. Then, for every $r>1$, there exist a Gram matrix
       $\Psi\in\Pos(\Sigma^k)$ and $t\geq0$ such that
        \begin{align*}
            Q(\Psi)&\leq kt,
            &\frac{\log F(\Psi)}{k}
            &\geq\frac{\log\alpha}{2}
            +\frac{t-2\delta}{4}\log r-\frac{2}{k}.
        \end{align*}
    \end{claim}
\begin{proof}
    By the definition of $\delta$ and the dual characterization in
    \Cref{prop:equiv-filtered-fidelity}, for every $r>1$ there exist $A,X_1,\ldots,X_n\succeq0$ such
    that
    \begin{align*}
        J-A&=\sum_{i=1}^n(rJ-S_i)\circ X_i,
        &F(A)&\geq\sqrt{\alpha}r^{-\delta/2},
    \end{align*}
    where $F(A):=\cF(A\circ\kbr{\mu},W\circ\kbr{\mu})$. Batching this factorization gives
    \begin{align*}
        J^{\otimes k}-A^{\otimes k}
        &=\sum_{j=1}^k J^{\otimes(j-1)}\otimes(J-A)\otimes A^{\otimes(k-j)}\\
        &=\sum_{j=1}^k\sum_{i=1}^n
          (rJ^{\otimes k}-S_{ij})\circ
          \left(J^{\otimes(j-1)}\otimes X_i\otimes A^{\otimes(k-j)}\right) \\
        &= \sum_{j=1}^k\sum_{i=1}^n
        (rJ^{\otimes k}-S_{ij})\circ \hat{X}_{ij}.
    \end{align*}
    Here, $S_{ij}:=J^{\otimes(j-1)}\otimes S_i\otimes J^{\otimes(k-j)}$ is the $ij$-th sign filter
    and $\hat{X}_{ij} = J^{\otimes(j-1)}\otimes X_i\otimes A^{\otimes(k-j)}$.  Moreover,
    multiplicativity of fidelity gives $F(A^{\otimes k})=F(A)^k$. Applying
    \Cref{prop:multiplicative-state-conversion} to this factorization, we obtain a Gram matrix
    $\Psi$ and an integer $T\geq0$ such that  $Q(\Psi)\leq T$ and
    $F(\Psi)\geq\frac{r^{T/4}}{4}F(A)^k\geq\frac{\alpha^{k/2}}{4}r^{(T-2k\delta)/4}$. Defining
    $t:=T/k$ and taking logarithms proves the claim.
\end{proof}
We now find the two algorithms mentioned in the above. Letting $r_1:=2^{1/(k\delta)}>1$, the above
claim gives $t_1,\Psi_1$ such that
\begin{align*}
    Q(\Psi_1)&\leq kt_1,
    &\frac{\log F(\Psi_1)}{k}
    &\geq\frac{\log\alpha}{2}
      +\frac{t_1}{2k\delta}-\frac{3}{k}.
\end{align*}
If $t_1\leq2\delta$, then $Q(\Psi_1)\leq2k\delta$ and $\log F(\Psi_1)/k\geq (\log\alpha)/2-3/k$.
Thus, taking $\Psi:=\Psi_1$ would yield the required bounds in
\eqref{eq:asymptotic-state-conversion-target}. Otherwise, we have $t_1>2\delta$. Let us define
$r_{\ell} := (r_1)^{\ell} = 2^{\ell/(k\delta)}$ so that $r_{\ell+1}/r_{\ell} = r_1$. The above claim
the gives $t_{\ell}, \Psi_{\ell}$ satisfying $Q(\Psi_\ell) \leq k t_{\ell} $ and
\begin{align*}
    \frac{\log F(\Psi_\ell)}{k} \geq\frac{\log\alpha}{2}
      +\frac{t_\ell - 2\delta}{4}\cdot \ell \log r_1 -\frac{2}{k}.
\end{align*}
Since $F(\Psi)\leq1$, there exists a smallest $\ell$ for which $t_{\ell} \leq 2 \delta$. If equality
holds, again $t_{\ell}$ and $\Psi_{\ell}$ satisfy the condition of
\eqref{eq:asymptotic-state-conversion-target}. Otherwise, we have
\begin{align*}
        & r_{\ell}/r_{\ell-1} = 2^{1/(k\delta)} && t_{\ell} < 2\delta < t_{\ell-1}
\end{align*}
We next balance the faster algorithm $\Psi_\ell$ with the more accurate algorithm $\Psi_{\ell-1}$ to
obtain a query rate of $2\delta$. More specifically, let $\lambda \in (0,1)$ be such that $\lambda
t_{\ell-1} + (1-\lambda)t_{\ell} = 2\delta$. Then, the algorithm $\hat{\Psi} :=
\Psi_{\ell-1}^{\otimes \lambda k} \otimes \Psi_{\ell}^{\otimes (1-\lambda)k}$ compues $k^2$ many
instances of $W$ and has query complexity $Q(\hat{\Psi}) \leq \lambda k(kt_{\ell-1}) + (1-\lambda)k
(kt_{\ell}) = 2 k^2 \delta$. Moreover, we can lower bound the total fidelity with respect to
$W^{\otimes k^2}\circ (\kbr{\mu})^{\otimes k^2}$ as
\begin{align*}
    \frac{\log F(\hat{\Psi})}{k^2}
    &=\frac{\lambda k\log F(\Psi_{\ell-1})+(1-\lambda)k\log F(\Psi_\ell)}{k^2}\\
    &\geq\frac{\log\alpha}{2}
      +\frac{\lambda(t_{\ell-1}-2\delta)(\ell-1)+(1-\lambda)(t_\ell-2\delta)\ell}{4}
       \cdot\log r_1-\frac{2}{k}\\
    &=\frac{\log\alpha}{2}
      +\frac{(1-\lambda)(t_\ell-2\delta)}{4}\cdot\log r_1-\frac{2}{k}\\
    &\geq\frac{\log\alpha}{2}
      -\frac{2(1-\lambda)\delta}{2k\delta}-\frac{2}{k}\\
    &\geq\frac{\log\alpha}{2}-\frac{3}{k},
\end{align*}
achiviing the desired bounds in \eqref{eq:asymptotic-state-conversion-target}.
\end{proof}


\subsection{XOR lemma}\label{sec:xor-lemma} In this section, we prove the XOR lemma for decision
problems. For the worst-case setting, see \cite{LR12}.

\begin{theorem}\label{thm:distributional-xor} Let $\Sigma \subset \{0,1\}^{n}$, $f:\Sigma \to
    \{0,1\}$ be a decision problem, and $\mu:\Sigma \to [0,1]$ be a distribution. Then, for
    $\eps=\eps(n)\in(0,1/2)$ and $k\in\bN$ with $k\geq1$, we have
    \begin{align*}
        \mathsf{Q}^{\mu^{\otimes k}}_{\frac12-\frac12(2\eps)^{k/2}}(f^{\oplus k})
        \geq
        \frac{k(1-2\eps)}{40\log(1/\eps)}
        \mathsf{Q}^{\mu}_{\frac12-\eps}(f),
    \end{align*}
    where $f^{\oplus k}(x_1,\dots,x_k):=f(x_1)\oplus\cdots\oplus f(x_k)$.
\end{theorem}

We use the following lemma which states query complexity of decision problems directly in terms of
$\bar{\delta}_2^{\mu}$.

\begin{lemma}\label{lem:decision-delta2-bounds} Let $\Sigma\subseteq\{0,1\}^{n}$,
    $g:\Sigma\to\{0,1\}$ be a decision problem, and $\mu:\Sigma\to[0,1]$ be a distribution. Then,
    for every bias $\beta=\beta(n)\in(0,1)$, we have
    \begin{align*}
        \frac12\,
        \bar{\delta}_2^\mu\left(J,\frac{S_g}{\beta^2}\mid\{S_i\}\right)
        \leq \mathsf{Q}^{\mu}_{(1-\beta)/2}(g)
        \leq
        \frac{20\log(2/\beta)}{1-\beta}\,
        \bar{\delta}_2^\mu\left(J,\frac{S_g}{\beta}\mid\{S_i\}\right).
    \end{align*}
\end{lemma}
\begin{proof}
    We prove each direction as follows.
    \paragraph{Lower bound.}
    Let $T:=\mathsf{Q}^{\mu}_{(1-\beta)/2}(g)$. By \Cref{lem:decision-sign-characterization}, there
    is a Gram matrix $\Phi$ with $\mathsf{Q}(\Phi)\leq2T$ and
    $\cF(\Phi\circ\kbr{\mu},S_g\circ\kbr{\mu})\geq\beta$. Applying
    \Cref{prop:filtered-delta2-lowerbound} with $W=S_g$ and success parameter $\beta^2$ gives
    $\bar{\delta}_2^\mu(J,S_g/\beta^2\mid\{S_i\})\leq\mathsf{Q}(\Phi)\leq2T$.

    \paragraph{Upper bound.}
    If $\bar{\delta}_2^\mu(J,S_g/\beta\mid\{S_i\})=+\infty$, the bound is immediate. Otherwise,
    since $S_g\succeq0$ and $\Tr[S_g\circ\kbr{\mu}]=1$, \Cref{prop:filtered-delta2-upperbound}, with
    success parameters $\beta$ and $\beta^2$, gives a Gram matrix $\Psi$ with $\Psi\circ I=I$ such
    that $\cF(\Psi\circ\kbr{\mu},S_g\circ\kbr{\mu})\geq\beta$ and
    \begin{align*}
        \mathsf{Q}(\Psi)
        &\leq\frac{10}{1-\beta}
        \log\left(\frac{2(1+\beta)}{\beta}\right)
        \bar{\delta}_2^\mu\left(J,\frac{S_g}{\beta}\mid\{S_i\}\right)\\
        &\leq\frac{20\log(2/\beta)}{1-\beta}
        \bar{\delta}_2^\mu\left(J,\frac{S_g}{\beta}\mid\{S_i\}\right),
    \end{align*}
    where the last inequality uses $1+\beta\leq2/\beta$. By
    \Cref{lem:decision-sign-characterization}, we have
    $\mathsf{Q}^{\mu}_{(1-\beta)/2}(g)\leq\mathsf{Q}(\Psi)$, proving the upper bound.
\end{proof}

\begin{proof}[Proof of \Cref{thm:distributional-xor}]
    Since $S_{f^{\oplus k}}=S_f^{\otimes k}$, filtered-fidelity tensorization
    (\Crefitem{item:filtered-fidelity-tensorization}) and \Cref{def:filtered-delta2} give
    \begin{align*}
        \bar{\delta}_2^{\mu^{\otimes k}}\left(
        J^{\otimes k},\frac{S_{f^{\oplus k}}}{(2\eps)^k}\mid\{S_{ij}\}_{i\in[k],j\in[n]}\right)
        =k\,\bar{\delta}_2^\mu\left(J,\frac{S_f}{2\eps}\mid\{S_j\}\right),
    \end{align*}
    where $S_{ij}:=J^{\otimes(i-1)}\otimes S_j\otimes J^{\otimes(k-i)}$. Applying the lower bound in
    \Cref{lem:decision-delta2-bounds} with bias $(2\eps)^{k/2}$ and its upper bound with bias
    $2\eps$ gives
    \begin{align*}
        \mathsf{Q}^{\mu^{\otimes k}}_{\frac12-\frac12(2\eps)^{k/2}}(f^{\oplus k})
        &\geq\frac{k}{2}\,
        \bar{\delta}_2^\mu\left(J,\frac{S_f}{2\eps}\mid\{S_j\}\right)
        \geq\frac{k(1-2\eps)}{40\log(1/\eps)}
        \mathsf{Q}^{\mu}_{\frac12-\eps}(f). \qedhere
    \end{align*}
\end{proof}
\section{A direct-product theorem via the Alberti adversary method}
In this section, we finally turn to the direct-product theorems. In
\Cref{sec:alberti-adversary-method}, we formally define the new Alberti adversary measure and study
its dual formulation and relation to query complexity. Then, in \Cref{sec:direct-product}, we use
these characterizations to prove direct-product theorems for the Alberti adversary measure and
distributional quantum query complexity. \subsection{Alberti adversary method}
\label{sec:alberti-adversary-method} In this subsection, we define the Alberti adversary measure and
relate it to distributional quantum query complexity. To obtain a measure for a problem $f$, we
define the Alberti adversary measure as the minimum $\bar{\delta}_2^\mu$ cost of transitioning from
$J$ to $W/\alpha$, over all valid $\mu$-witnesses $W$ for $f$. We emphasize that the success
parameter $\alpha$ enters by scaling the second argument, which allows us to use the tensorization
property of $\bar{\delta}_2^\mu$.

\begin{definition}\label{def:aadv} Let $\Sigma\subseteq\{0,1\}^{n}$, and let
    $f\subseteq\Sigma\times\Omega$ be a search problem. Let $\alpha=\alpha(n)\in(0,1]$ be a success
    parameter and $\mu:\Sigma\to[0,1]$ be a distribution. Let $\{S_i\}_{i\in[n]}$ be sign filters.
    Let $\cW_f^\mu$ be the set of $\mu$-witnesses for $f$ as in \Cref{def:witness-set}. Then, the
    Alberti adversary measure for $f$ with respect to distribution $\mu$ and success $\alpha$ is
    defined as
    \begin{align*}
        \mathrm{AAdv}_{\alpha}^{\mu}(f)
        :=\min_{W\in\cW_f^\mu}
        \bar{\delta}_2^\mu\bigl(J,W/\alpha\mid\{S_i\}\bigr).
    \end{align*}
    In particular, the value $0\leq\mathrm{AAdv}_{\alpha}^{\mu}(f)\leq n$ is finite, and the minimum
    is attained.\footnote{The lower bound follows from $\bar{\delta}_2^\mu(A,B\mid\{S_i\})\geq0$
    (see \Crefitem{item:avgdelta2-nonnegative}). For the upper bound, $I\in\cW_f^\mu$, and
    \Cref{prop:filtered-delta2-lowerbound} gives
    $\bar{\delta}_2^\mu(J,I/\alpha\mid\{S_i\})\leq\mathsf{Q}(I)\leq n$. For attainment, without loss
    of generality, we may replace $\cW_f^\mu$ with the compact set
    $\widehat{\cW}_f^\mu:=\{W\in\cW_f^\mu:W\circ J_S=W\}$, where $S=\supp(\mu)$ and $J_S$ is the
    indicator matrix of $S\times S$. Moreover, $W\mapsto\bar{\delta}_2^\mu(J,W/\alpha\mid\{S_i\})$
    is lower semicontinuous by \Crefitem{item:avgdelta2-lower-semicontinuity}. Therefore, it attains
    its minimum on $\widehat{\cW}_f^\mu$.}
\end{definition}

\subsubsection{Equivalent primal and dual formulations}
We now give a characterization of the Alberti adversary measure by a single positive definite
matrix.

\begin{lemma}\label{lem:aadv-universal-hardness-cert} Let $\Sigma\subseteq\{0,1\}^{n}$, let
    $f\subseteq\Sigma\times\Omega$ be a search problem, and let $\mu:\Sigma\to[0,1]$ be a
    distribution. Also, for $s\in \Omega$ let $I_{f^{-1}(s)}:=\sum_{\substack{x\in\Sigma\\s\in
    f(x)}}\kbr{x}$ be the projector onto the subspace of inputs for which $s$ is a valid answer.
    Then, for any success parameter $\alpha=\alpha(n)\in(0,1]$, we have
    \begin{align*}
        \mathrm{AAdv}_{\alpha}^{\mu}(f)
        &=\progat{6cm}{\sup_{\Gamma\succ0}}
        {\frac{\log\alpha-\log\langle\Gamma,\kbr{\mu}\rangle}
        {\max_{i\in[n]}D_{\max}(\Gamma\circ S_i\Vert\Gamma)}}
        {\Gamma\succeq I_{f^{-1}(s)}\quad\forall s\in\Omega}.
    \end{align*}

    The quotient is defined to be zero when its denominator vanishes.
\end{lemma}
\begin{proof}
    We begin with the primal formulation
    \begin{align*}
        \mathrm{AAdv}^{\mu}_{\alpha}(f)
        &=\min_{W\in\cW^{\mu}_{f}}
        \bar{\delta}_2^{\mu}(J,W/\alpha\mid\{S_i\})\\
        &=\min_{W\in\cW^{\mu}_{f}}
        \sup_{c>1}
        \frac{-2\log\cF_c(\kbr{\mu},W\circ\kbr{\mu}/\alpha)}{\log c}\\
        &=\min_{W\in\cW^{\mu}_{f}}
        \sup_{r>0}
        -2r\log\cF_{2^{r^{-1}}}(\kbr{\mu},W\circ\kbr{\mu}/\alpha).
    \end{align*}
    Here we substituted $r:=1/\log c$, with $c>1$ corresponding to $r>0$.
    \paragraph{Swapping $\min_{W\in \cW^{\mu}_f} \sup_{r>0}$.} To apply Sion's minimax, we first
    restrict the set of minimizers. Let $S:=\supp(\mu)$, let $J_S$ be the indicator matrix of
    $S\times S$, and define
    \begin{align*}
        \widehat{\cW}^{\mu}_{f}
        :=\{W\in\cW^{\mu}_{f}:W\circ J_S=W\}.
    \end{align*}
    Since the value depends on $W$ only through $W\circ\kbr{\mu}$, by
    \Cref{lem:adversary-support-replacement} we may replace $\cW^{\mu}_{f}$ by
    $\widehat{\cW}^{\mu}_{f}$ to get

    \begin{align*}
        \mathrm{AAdv}^{\mu}_{\alpha}(f)
        =\min_{W\in\widehat{\cW}^{\mu}_{f}}\sup_{r>0}
        -2r\log\cF_{2^{r^{-1}}}(\kbr{\mu},W\circ\kbr{\mu}/\alpha).
    \end{align*}
    The set $\widehat{\cW}^{\mu}_{f}$ is convex and compact (see \Cref{lem:compactness-gmueps}), and
    the set $(0,\infty)$ is convex. The following claim establishes the remaining continuity,
    convexity, and concavity needed to apply Sion's minimax.
    \begin{claim}
        The function $F(W,r) := -2r\log \cF_{2^{r^{-1}}}(\kbr{\mu}, W\circ \kbr{\mu}/\alpha)$ is
        jointly continuous, convex in $W$, and concave in $r$.
    \end{claim}
    \begin{proof}
        Convexity in $W$ is a direct consequence of
        \Crefitem{item:filtered-fidelity-joint-concavity}. For  concavity in $r$, let
        $h(x):=-2\log\cF_{2^x}(\kbr{\mu},W\circ\kbr{\mu}/\alpha)$ for $x>0$. This function is
        concave. Indeed, for $\lambda \in [0,1]$, we have
    \begin{align*}
        h(\lambda x_1 + (1-\lambda)x_2) &=
        -2\log \cF_{2^{\lambda x_1 + (1-\lambda)x_2 }}(\kbr{\mu},W\circ \kbr{\mu}/\alpha)  \\
        &\geq
        -2\log \left(\cF_{2^{x_1}}(\kbr{\mu},W\circ \kbr{\mu}/\alpha)^{\lambda} \cF_{2^{x_2}}(\kbr{\mu},W\circ \kbr{\mu}/\alpha)^{(1-\lambda)}\right)\\
        &= \lambda h(x_1) + (1-\lambda) h(x_2).
    \end{align*}
    Here we used \Crefitem{item:filtered-fidelity-interpolation}. Moreover, since $h$ is concave,
    the function $r \cdot h(1/r)$ is also concave. More specifically, for $r_1,r_2>0$, let
    $r=\lambda r_1+(1-\lambda)r_2$. Then, we can decompose
    \begin{align*}
        \frac{1}{r}
        =\frac{\lambda r_1}{r}\cdot\frac{1}{r_1}
        +\frac{(1-\lambda)r_2}{r}\cdot\frac{1}{r_2}.
    \end{align*}
    Therefore, by Jensen's inequality
    \begin{align*}
        r h\left(\frac{1}{r}\right)
        &\geq r\left(
        \frac{\lambda r_1}{r}h\left(\frac{1}{r_1}\right)
        +\frac{(1-\lambda)r_2}{r}h\left(\frac{1}{r_2}\right)\right)
        =\lambda r_1h\left(\frac{1}{r_1}\right)
        +(1-\lambda)r_2h\left(\frac{1}{r_2}\right).
    \end{align*}
    We now verify continuity. By \Crefitem{item:filtered-fidelity-continuity}, the map $(W,r) \to
    \cF_{2^{1/r}}(\kbr{\mu},W\circ \kbr{\mu}/\alpha)$ is jointly continuous. Moreover, by
    \Crefitem{item:filtered-fidelity-vanishing}, $\cF_{2^{1/r}}(\kbr{\mu},W\circ
    \kbr{\mu}/\alpha)>0$ as $\Tr[W\circ \kbr{\mu}] = 1$ and $\kbr{\mu}\circ W\circ
    \kbr{\mu}/\alpha\neq 0$. Since $\log$ is continuous on $(0,\infty)$, $F(W,r)$ is jointly
    continuous.
    \end{proof}
    Applying Sion's minimax, we get
    \begin{align*}
        &\begin{progalign}
        \progexprrow
        {\mathrm{AAdv}^{\mu}_{\alpha}(f)}
        {\sup_{r>0}\min_{W\in\widehat{\cW}^{\mu}_{f}}
        \left[-2r\log\cF_{2^{1/r}}
        \bigl(\kbr{\mu},W\circ\kbr{\mu}/\alpha\bigr)\right]}\\[0.6em]
        \progrow
        {}
        {\sup_{r>0}}
        {-r\log\left[
        \max_{W\in\widehat{\cW}^{\mu}_{f}}\inf_{\Gamma}
        \langle\Gamma,\kbr{\mu}\rangle
        \langle\Gamma^{-1},W\circ\kbr{\mu}/\alpha\rangle
        \right]}
        {\Gamma\succ0;\\
        \Gamma\circ S_i\preceq2^{1/r}\Gamma\quad\forall i\in[n]}\\[0.6em]
        \progrow
        {}
        {\sup_{r>0}}
        {-r\log\left[
        \max_{W\in\widehat{\cW}^{\mu}_{f}}\inf_{\Gamma}
        \langle\Gamma^{-1},W\circ\kbr{\mu}/\alpha\rangle
        \right]}
        {\Gamma\succ0,\quad\langle\Gamma,\kbr{\mu}\rangle=1;\\
        \Gamma\circ S_i\preceq2^{1/r}\Gamma\quad\forall i\in[n]}.
        \end{progalign}
    \end{align*}
    \paragraph{Swapping $\max_{W\in \widehat{\cW}^{\mu}_f}\inf_{\Gamma}$.} For each fixed $r>0$, let
    \begin{align*}
        \cL:=\left\{\Gamma\succ0:
        \langle\Gamma,\kbr{\mu}\rangle=1,\quad
        \Gamma\circ S_i\preceq2^{1/r}\Gamma\quad\forall i\in[n]\right\}.
    \end{align*}
    This set is convex. The function
    $(\Gamma,W)\mapsto\langle\Gamma^{-1},W\circ\kbr{\mu}/\alpha\rangle$ is convex in $\Gamma$ and
    linear in $W$, since inversion is operator convex and $W\circ\kbr{\mu}/\alpha\succeq0$.
    Therefore, we can apply Sion's minimax once again and obtain
    \begin{align*}
        &\begin{progalign}
        \progexprrow
        {\mathrm{AAdv}^{\mu}_{\alpha}(f)}
        {\sup_{r>0}-r\log\left[
        \inf_{\Gamma\in\cL}\max_{W\in\widehat{\cW}^{\mu}_{f}}
        \langle\Gamma^{-1},W\circ\kbr{\mu}/\alpha\rangle
        \right]}\\[0.6em]
        \progrow
        {}
        {\sup_{\Gamma}\sup_{c>1}}
        {\frac{-\log\left[
        \max_{W\in\widehat{\cW}^{\mu}_{f}}
        \langle\Gamma^{-1},W\circ\kbr{\mu}/\alpha\rangle
        \right]}{\log c}}
        {\Gamma\succ0,\quad\langle\Gamma,\kbr{\mu}\rangle=1;\\
        \Gamma\circ S_i\preceq c\Gamma\quad\forall i\in[n]}.
        \end{progalign}
    \end{align*}
    \paragraph{Replacing $\sup_{c>1}$ with $D_{\max}$.} For this step, let us abbreviate
    $d_{\Gamma}:=\max_{i\in[n]}D_{\max}(\Gamma\circ S_i\Vert\Gamma)\geq0$ and
    $M_{\Gamma}:=\max_{W\in\widehat{\cW}^{\mu}_{f}}
    \langle\Gamma^{-1},W\circ\kbr{\mu}/\alpha\rangle>0$. We consider the following cases.
    \begin{enumerate}
        \item If $d_{\Gamma}=0$, then $\Gamma$ is diagonal. Since
        $\langle\Gamma,\kbr{\mu}\rangle=1$, there exists $x\in S$ such that $\Gamma(x,x)\leq1$.
        Taking $W=\kbr{x}/\mu(x)$ gives $M_{\Gamma}\geq1/(\alpha\Gamma(x,x))\geq1$.
        $$
        \sup_{c>1}\frac{-\log M_{\Gamma}}{\log c}= 0
        $$
        \item
        If $d_{\Gamma}>0$ and $M_{\Gamma}\geq1$, we have
        \begin{align*}
            \sup_{c>1, \log c>d_{\Gamma}}\frac{-\log M_{\Gamma}}{\log c}
            = \sup_{t>d_{\Gamma}}\frac{-\log M_{\Gamma}}{t}
            = 0.
        \end{align*}
        \item Finally, if $d_{\Gamma} >0$ and $M_{\Gamma}<1$, we get
            \begin{align*}
                  \sup_{c>1, \log c>d_{\Gamma}}\frac{-\log M_{\Gamma}}{\log c}
                = \sup_{t>d_{\Gamma}}\frac{-\log M_{\Gamma}}{t}
                = \frac{-\log M_{\Gamma}}{d_{\Gamma}}.
            \end{align*}
    \end{enumerate}
    Therefore, defining the quotient to be zero whenever its denominator vanishes, we can rewrite
    the program as
    \begin{align*}
        &\begin{progalign}
        \progrow
        {\mathrm{AAdv}^{\mu}_{\alpha}(f)}
        {\sup_{\Gamma}}
        {\frac{-\log\left[
        \max_{W\in\widehat{\cW}^{\mu}_{f}}
        \langle\Gamma^{-1},W\circ\kbr{\mu}/\alpha\rangle
        \right]}{\max_{i\in[n]}D_{\max}(\Gamma\circ S_i\Vert\Gamma)}}
        {\Gamma\succ0,\quad\langle\Gamma,\kbr{\mu}\rangle=1}.
        \end{progalign}
    \end{align*}
    \paragraph{Resolving $\max_{W\in \widehat{\cW}^{\mu}_f}$} Next, we remove the maximization over
    $W$ as follows.
    \begin{claim}
        For every $\Gamma\succ0$,
        $
            \max_{W\in\widehat{\cW}^{\mu}_{f}}
            \langle\Gamma^{-1},W\circ\kbr{\mu}\rangle
            =\max_{s\in\Omega}
            \|\Gamma^{-1}\circ J_S\circ J_{f^{-1}(s)}\|,
        $
        where $J_{f^{-1}(s)}:=\sum_{\substack{x,y\in\Sigma\\s\in f(x)\cap f(y)}}\ket{x}\bra{y}$ is
        the indicator matrix of $f^{-1}(s)\times f^{-1}(s)$.
    \end{claim}
    \begin{proof}
    Expanding the definition of $\widehat{\cW}^{\mu}_{f}$, we have
    \begin{align*}
        \widehat{\cW}^{\mu}_{f}
        &=\left\{\sum_s W^s: W^s\succeq0,\ W^s\circ J_S=W^s,\ W^s(x,x)=0\ (s\notin f(x)),\ \sum_s\Tr[W^s\circ\kbr{\mu}]=1\right\}\\
        &=\left\{\sum_s W^s: W^s\succeq0,\ W^s\circ J_S\circ J_{f^{-1}(s)}=W^s,\ \sum_s\Tr[W^s\circ\kbr{\mu}]=1\right\},
    \end{align*}
    where we used the equivalence between the diagonal-zero condition and $W^s\circ
    J_{f^{-1}(s)}=W^s$ for PSD blocks. Write $\Omega_S:=\{s\in\Omega:J_S\circ J_{f^{-1}(s)}\neq0\}$.
    Since $W\circ J_S=W$, the map $W\mapsto X:=W\circ\kbr{\mu}$ can be uniquely inverted and gives a
    bijection from $\widehat{\cW}^{\mu}_{f}$ to
    \begin{align*}
        \cX_f
        &=\left\{\sum_{s\in\Omega_S}X^s:X^s\succeq0,\ X^s\circ J_S\circ J_{f^{-1}(s)}=X^s,\ \sum_s\Tr[X^s]=1\right\}\\
        &=\left\{\E_{s\sim p}X^s:X^s\succeq0,\ X^s\circ J_S\circ J_{f^{-1}(s)}=X^s,\ \Tr[X^s]=1,\ p\text{ a distribution on }\Omega_S\right\}.
    \end{align*}
    The second equality follows by taking $p(s):=\Tr[X^s]$ and normalizing $(X')^s:=X^s/p(s)$ when
    $p(s)>0$.\footnote{The blocks $(X')^s$ with $p(s)=0$ may be chosen arbitrarily among trace-one
    PSD matrices in their answer subspaces.} Let $p$ range over distributions on $\Omega_S$. Then,
    \begin{align*}
        &\begin{progalign}
        \progexprrow[]
        {}
        {\max_{W\in\widehat{\cW}^{\mu}_{f}}
        \langle\Gamma^{-1},W\circ\kbr{\mu}\rangle}\\[0.6em]
        \progrow
        {}
        {\max_p\max_{\{X^s\succeq0\}}}
        {\E_{s\sim p}\langle\Gamma^{-1},X^s\rangle}
        {\Tr[X^s]=1,\quad X^s\circ J_S\circ J_{f^{-1}(s)}=X^s}\\[0.6em]
        \progrow
        {}
        {\max_p\E_{s\sim p}\max_{X^s\succeq0}}
        {\langle\Gamma^{-1},X^s\circ J_S\circ J_{f^{-1}(s)}\rangle}
        {\Tr[X^s]=1,\quad X^s\circ J_S\circ J_{f^{-1}(s)}=X^s}\\[0.6em]
        \progrow
        {}
        {\max_p\E_{s\sim p}\max_{X^s\succeq0}}
        {\langle\Gamma^{-1}\circ J_S\circ J_{f^{-1}(s)},X^s\rangle}
        {\Tr[X^s]=1,\quad X^s\circ J_S\circ J_{f^{-1}(s)}=X^s}\\[0.6em]
        \progrow
        {}
        {\max_p\E_{s\sim p}\max_{X^s\succeq0}}
        {\langle\Gamma^{-1}\circ J_S\circ J_{f^{-1}(s)},X^s\rangle}
        {\Tr[X^s]\leq1,\quad X^s\circ J_S\circ J_{f^{-1}(s)}=X^s}\\[0.6em]
        \progrow
        {}
        {\max_p\E_{s\sim p}\max_{X^s\succeq0}}
        {\langle\Gamma^{-1}\circ J_S\circ J_{f^{-1}(s)},X^s\rangle}
        {\Tr[X^s]\leq1}\\[0.6em]
        \progexprrow
        {}
        {\max_p\E_{s\sim p}
        \|\Gamma^{-1}\circ J_S\circ J_{f^{-1}(s)}\|}\\[0.6em]
        \progexprrow
        {}
        {\max_{s\in\Omega}
        \|\Gamma^{-1}\circ J_S\circ J_{f^{-1}(s)}\|}.
        \end{progalign}
    \end{align*}
    In the third-to-last equality, we remove the support condition using
    \Cref{lem:adversary-support-replacement}, since $X^s\mapsto X^s\circ J_S\circ J_{f^{-1}(s)}$
    preserves positivity, the trace bound, and the objective. The penultimate equality is
    trace--spectral norm duality. The final equality follows by taking $p$ to be a point mass at a
    maximizing answer; answers outside $\Omega_S$ contribute zero.
    \end{proof}
    As a result, we can write
    \begin{align*}
    \mathrm{AAdv}^{\mu}_{\alpha}(f) &=
    \progline{
    {\sup_{\Gamma}}}
    {\frac{-\log\left[
    \max_{W\in\widehat{\cW}^{\mu}_{f}}
    \langle\Gamma^{-1},W\circ\kbr{\mu}/\alpha\rangle
    \right]}{\max_{i\in[n]}D_{\max}(\Gamma\circ S_i\Vert\Gamma)}}
    {\Gamma\succ0,\quad\langle\Gamma,\kbr{\mu}\rangle=1} \\
    &=
    \progline{
    {\sup_{\Gamma}}}
    {\frac{\log \alpha -\log \max_{s} \|\Gamma^{-1}\circ J_S\circ J_{f^{-1}(s)}\|}{\max_{i\in[n]}D_{\max}(\Gamma\circ S_i\Vert\Gamma)}}
    {\Gamma\succ0,\quad\langle\Gamma,\kbr{\mu}\rangle=1}
    \end{align*}
    Finally, set $m:=\max_{s\in\Omega}\|\Gamma^{-1}\circ J_S\circ J_{f^{-1}(s)}\|$ and
    $\Gamma':=m\Gamma$. Then $\langle\Gamma',\kbr{\mu}\rangle=m$ and
    $\max_{s\in\Omega}\|(\Gamma')^{-1}\circ J_S\circ J_{f^{-1}(s)}\|=1$, while the denominator is
    unchanged. Renaming $\Gamma'$ as $\Gamma$, we get
    \begin{align*}
        &\begin{progalign}
        \progrow
        {\mathrm{AAdv}^{\mu}_{\alpha}(f)}
        {\sup_{\Gamma\succ0}}
        {\frac{\log\alpha-\log\langle\Gamma,\kbr{\mu}\rangle}
        {\max_{i\in[n]}D_{\max}(\Gamma\circ S_i\Vert\Gamma)}}
        {\max_{s\in\Omega}\|\Gamma^{-1}\circ J_S\circ J_{f^{-1}(s)}\|=1}\\[0.6em]
        \progrow
        {}
        {\sup_{\Gamma\succ0}}
        {\frac{\log\alpha-\log\langle\Gamma,\kbr{\mu}\rangle}
        {\max_{i\in[n]}D_{\max}(\Gamma\circ S_i\Vert\Gamma)}}
        {\max_{s\in\Omega}\|\Gamma^{-1}\circ J_S\circ J_{f^{-1}(s)}\|\leq1}\\[0.6em]
        \progrow
        {}
        {\sup_{\Gamma\succ0}}
        {\frac{\log\alpha-\log\langle\Gamma,\kbr{\mu}\rangle}
        {\max_{i\in[n]}D_{\max}(\Gamma\circ S_i\Vert\Gamma)}}
        {\Gamma\succeq I_S\circ I_{f^{-1}(s)}\quad\forall s\in\Omega}.
        \end{progalign}
    \end{align*}
    For the penultimate equality, if the maximum is $m<1$, replacing $\Gamma$ by $m\Gamma$ makes it
    equal to one without decreasing the objective.
    \paragraph{Removing support filter.}
    Finally, we remove the $I_S$ filter in the constraint. Let $\Gamma\succ0$ be a feasible
    adversary matrix with $\Gamma\succeq I_S\circ I_{f^{-1}(s)}$ for all $s\in\Omega$. Define
    $\Gamma':=\Gamma\circ J_S+I_{\Sigma\setminus S}$. Then $\Gamma'\succ0$ and
    \begin{align*}
        \Gamma'
        &=\Gamma\circ J_S+I_{\Sigma\setminus S}\\
        &\succeq I_S\circ I_{f^{-1}(s)}+I_{\Sigma\setminus S}
        \succeq I_{f^{-1}(s)}
        \qquad\forall s\in\Omega.
    \end{align*}
    Moreover, $\langle\Gamma',\kbr{\mu}\rangle=\langle\Gamma,\kbr{\mu}\rangle$. For every $i\in[n]$
    and $c\geq1$ satisfying $\Gamma\circ S_i\preceq c\Gamma$, we have
    \begin{align*}
        \Gamma'\circ S_i-c\Gamma'
        &=(\Gamma\circ J_S+I_{\Sigma\setminus S})\circ S_i
        -c(\Gamma\circ J_S+I_{\Sigma\setminus S})\\
        &=(\Gamma\circ S_i-c\Gamma)\circ J_S
        +(1-c)I_{\Sigma\setminus S}\\
        &\preceq0,
    \end{align*}
    where we used the Schur product theorem and $c\geq1$. Therefore,
    \begin{align*}
        D_{\max}(\Gamma'\circ S_i\Vert\Gamma')
        &=\inf\{\log c:c\geq1,\ \Gamma'\circ S_i\preceq c\Gamma'\}\\
        &\leq\inf\{\log c:c\geq1,\ \Gamma\circ S_i\preceq c\Gamma\}\\
        &=D_{\max}(\Gamma\circ S_i\Vert\Gamma).
    \end{align*}
    Therefore, the value of the optimization problem can only increase. This yields
    \begin{align*}
        &\begin{progalign}
        \progrow
        {\mathrm{AAdv}^{\mu}_{\alpha}(f)}
        {\sup_{\Gamma\succ0}}
        {\frac{\log\alpha-\log\langle\Gamma,\kbr{\mu}\rangle}
        {\max_{i\in[n]}D_{\max}(\Gamma\circ S_i\Vert\Gamma)}}
        {\Gamma\succeq I_{f^{-1}(s)}\quad\forall s\in\Omega.}
        \end{progalign} \qedhere
    \end{align*}
\end{proof}

\subsubsection{Bounding query complexity}
We relate the Alberti adversary measure to distributional quantum query complexity.

\begin{proposition}\label{prop:aadv-lowerbound}\label{prop:aadv-upperbound} Let
    $\Sigma\subseteq\{0,1\}^{n}$, let $f\subseteq\Sigma\times\Omega$ be a search problem, and let
    $\mu:\Sigma\to[0,1]$ be a distribution. Then, for success parameters $\alpha=\alpha(n)$ and
    $\beta=\beta(n)$ with $0<\beta<\alpha<1$,
    \begin{align*}
        \frac{\alpha-\beta}
        {10\alpha\log\left(\frac{2(1-\beta)}{\alpha-\beta}\right)}
        \mathsf{Q}_{1-\beta}^{\mu}(f)
        \leq \mathrm{AAdv}_{\alpha}^{\mu}(f)
        \leq \mathsf{Q}_{1-\alpha}^{\mu}(f).
    \end{align*}
    The right-hand inequality also holds for $\alpha=1$.
\end{proposition}

\begin{proof}
    We first prove the lower bound as follows. By \Cref{thm:fidelity-charecterization}, there exists
    $\Psi\in\cG^\mu_{1-\alpha}(f)$ such that $\mathsf{Q}(\Psi)=\mathsf{Q}^\mu_{1-\alpha}(f)$. By
    \Cref{def:valid-grams}, there exists $W\in\cW^\mu_f$ such that
    $\cF(\Psi\circ\kbr{\mu},W\circ\kbr{\mu})\geq\sqrt{\alpha}$. From
    \Cref{prop:filtered-delta2-lowerbound}, we immediately get
    \begin{align*}
        \mathsf{Q}^\mu_{1-\alpha}(f)=\mathsf{Q}(\Psi)
        \geq\bar{\delta}_2^\mu(J,W/\alpha\mid\{S_i\})
        \geq\mathrm{AAdv}^\mu_\alpha(f),
    \end{align*}
    where the last inequality follows from \Cref{def:aadv}.
    \par
    Let $W\in\cW^\mu_f$ attain the minimum in \Cref{def:aadv}, so that
    $\bar{\delta}_2^\mu(J,W/\alpha\mid\{S_i\})=\mathrm{AAdv}^\mu_\alpha(f)$. Applying
    \Cref{prop:filtered-delta2-upperbound} gives a Gram matrix $\Psi$ with $\Psi\circ I=I$ and
    $\cF(\Psi\circ\kbr{\mu},W\circ\kbr{\mu})\geq\sqrt{\beta}$ such that
    \begin{align*}
        \mathsf{Q}(\Psi)
        \leq
        \frac{10\alpha}{\alpha-\beta}
        \log\left(\frac{2(1-\beta)}{\alpha-\beta}\right)
        \mathrm{AAdv}^\mu_\alpha(f).
    \end{align*}
    Hence $\Psi\in\cG^\mu_{1-\beta}(f)$ by \Cref{def:valid-grams}, and
    \Cref{thm:fidelity-charecterization} gives $\mathsf{Q}^\mu_{1-\beta}(f)\leq\mathsf{Q}(\Psi)$,
    completing the proof.
\end{proof}

\subsection{Direct-product theorem} \label{sec:direct-product} In this subsection, we prove
direct-product theorems for the Alberti adversary measure and distributional quantum query
complexity.

\subsubsection{\texorpdfstring{Exact direct product for $\mathrm{AAdv}$}{Exact direct product for
the Alberti adversary}} We first show that the Alberti adversary measure satisfies an exact
direct-product identity.
\begin{theorem}\label{prop:aadv-direct-product} Let $\Sigma\subseteq\{0,1\}^{n}$, and let
    $f\subseteq\Sigma\times\Omega$ be a search problem. Let $\mu:\Sigma\to[0,1]$ be a distribution,
    $\alpha=\alpha(n)\in(0,1]$, and $k\in\bN$ with $k\geq1$. Then,
    \begin{align*}
        \mathrm{AAdv}_{\alpha^k}^{\mu^{\otimes k}}(f^{\otimes k})
        =k\cdot\mathrm{AAdv}_{\alpha}^{\mu}(f).
    \end{align*}
\end{theorem}
\begin{proof}
    We prove each direction as follows.
    \paragraph{Proof of $\leq$.}
    Let $W\in\cW_f^\mu$ be such that
    $\mathrm{AAdv}_\alpha^\mu(f)=\bar{\delta}_2^\mu\left(J,W/\alpha\mid\{S_j\}\right)$. We claim
    that $W^{\otimes k} \in \cW^{\mu^{\otimes k}}_{f^{\otimes k}}$, which immediately gives
    \begin{align*}
        \mathrm{AAdv}^{\mu^{\otimes k}}_{\alpha^k}(f^{\otimes k})
        &\leq\bar{\delta}_2^{\mu^{\otimes k}}\left(
        J^{\otimes k},\frac{W^{\otimes k}}{\alpha^k}
        \mid\{S_{ij}\}_{i\in[k],j\in[n]}\right)\\
        &=\bar{\delta}_2^{\mu^{\otimes k}}\left(
        J^{\otimes k},\frac{W^{\otimes k}}{\alpha^k}
        \mid\bigcup_{i\in[k]}
        \left\{J^{\otimes(i-1)}\otimes S_j\otimes J^{\otimes(k-i)}\right\}_{j\in[n]}
        \right)\\
        &=k\,\bar{\delta}_2^\mu\left(J,\frac{W}{\alpha}\mid\{S_j\}_{j\in[n]}\right)\\
        &=k\,\mathrm{AAdv}_\alpha^\mu(f),
    \end{align*}
    where the tensorization equality follows from \Crefitem[iii]{prop:avgdelta2-properties}.
    \par
    We now prove the claim. Let $W=\sum_{s\in\Omega}W^s$ be a decomposition of $W$ where
    $W^s\succeq0$ and $W^s(x,x)=0$ whenever $s\notin f(x)$. Then, we have
    \begin{align*}
        W^{\otimes k}
        =\left(\sum_{s\in\Omega}W^s\right)^{\otimes k}
        =\sum_{s_1,\dots,s_k\in\Omega}W^{s_1}\otimes\cdots\otimes W^{s_k}.
    \end{align*}
    First of all, each summand on the right-hand side is PSD since each tensor factor is. Next, let
    $\mathbf{s}=(s_1,\dots,s_k)\in\Omega^k$ and $\bx=(x_1,\dots,x_k)\in\Sigma^k$ be such that
    $\mathbf{s}\notin f^{\otimes k}(\bx)$. Therefore, there exists an index $j$ such that $s_j\notin
    f(x_j)$. Consequently,
    \begin{align*}
        \left(W^{s_1}\otimes\cdots\otimes W^{s_k}\right)(\bx,\bx)
        =\prod_{\ell=1}^{k}W^{s_\ell}(x_\ell,x_\ell)=0.
    \end{align*}
    Finally, $\Tr[W^{\otimes k}\circ\kbr{\mu}^{\otimes k}] =\bigl(\Tr[W\circ\kbr{\mu}]\bigr)^k=1$,
    which completes the proof of the claim.

    \paragraph{Proof of $\geq$.}
    For the lower bound, we use the dual formulation in \Cref{lem:aadv-universal-hardness-cert}:
    \begin{align*}
        \mathrm{AAdv}_\alpha^\mu(f)
        &=\progat{6cm}{\sup_{\Gamma\succ0}}
        {\frac{\log\alpha-\log\langle\Gamma,\kbr{\mu}\rangle}
        {\max_{j\in[n]}D_{\max}(\Gamma\circ S_j\Vert\Gamma)}}
        {\Gamma\succeq I_{f^{-1}(s)}\quad\forall s\in\Omega}.
    \end{align*}
    We use the following claim to lift each adversary matrix for $\mathrm{AAdv}_\alpha^\mu(f)$ to an
    adversary matrix for $\mathrm{AAdv}_{\alpha^k}^{\mu^{\otimes k}}(f^{\otimes k})$.
    \begin{claim}
        Let $\Gamma\succ0$ satisfy $\Gamma\succeq I_{f^{-1}(s)}$ for every $s\in\Omega$. Then,
        \begin{enumerate}[(i)]
            \item For every $\mathbf{s}\in\Omega^k$, $\Gamma^{\otimes k}\succeq I_{(f^{\otimes
            k})^{-1}(\mathbf{s})}$.
            \item $\displaystyle \max_{i\in[k],j\in[n]}D_{\max}(\Gamma^{\otimes k}\circ
            S_{ij}\Vert\Gamma^{\otimes k}) =\max_{j\in[n]}D_{\max}(\Gamma\circ S_j\Vert\Gamma)$.
        \end{enumerate}
    \end{claim}
    Assuming the claim, for any feasible $\Gamma$ we have
    \begin{align*}
        \mathrm{AAdv}_{\alpha^k}^{\mu^{\otimes k}}(f^{\otimes k})
        &\geq\frac{\log(\alpha^k)-\log\langle\Gamma^{\otimes k},\kbr{\mu}^{\otimes k}\rangle}
        {\max_{i\in[k],j\in[n]}D_{\max}(\Gamma^{\otimes k}\circ S_{ij}\Vert\Gamma^{\otimes k})}\\
        &=\frac{k\log\alpha-\log\left(\langle\Gamma,\kbr{\mu}\rangle^k\right)}
        {\max_{j\in[n]}D_{\max}(\Gamma\circ S_j\Vert\Gamma)}\\
        &=k\,\frac{\log\alpha-\log\langle\Gamma,\kbr{\mu}\rangle}
        {\max_{j\in[n]}D_{\max}(\Gamma\circ S_j\Vert\Gamma)}.
    \end{align*}
    Taking the supremum over feasible $\Gamma$ gives
    \begin{align*}
        \mathrm{AAdv}_{\alpha^k}^{\mu^{\otimes k}}(f^{\otimes k})
        &\geq\progat{6cm}{\sup_{\Gamma\succ0}}
        {k\,\frac{\log\alpha-\log\langle\Gamma,\kbr{\mu}\rangle}
        {\max_{j\in[n]}D_{\max}(\Gamma\circ S_j\Vert\Gamma)}}
        {\Gamma\succeq I_{f^{-1}(s)}\quad\forall s\in\Omega}\\
        &=k\,\mathrm{AAdv}_\alpha^\mu(f).
    \end{align*}
    We now prove the claim. For every $\mathbf{s}=(s_1,\dots,s_k)\in\Omega^k$, we have
    \begin{align*}
        I_{(f^{\otimes k})^{-1}(\mathbf{s})}
        &=\sum_{\substack{x_1,\dots,x_k\in\Sigma\\s_i\in f(x_i)\ \forall i\in[k]}}
        \kbr{x_1\cdots x_k}\\
        &=\bigotimes_{i=1}^{k}
        \left(\sum_{\substack{x_i\in\Sigma\\s_i\in f(x_i)}}\kbr{x_i}\right)
        =\bigotimes_{i=1}^{k}I_{f^{-1}(s_i)}.
    \end{align*}
    Item (i) then follows from $\Gamma\succeq I_{f^{-1}(s)}$ for every $s\in\Omega$ and monotonicity
    of tensor products of positive semidefinite matrices.
    \par
    For the second part, for every $i\in[k]$ and $j\in[n]$, we have
    \begin{align*}
        D_{\max}(\Gamma^{\otimes k}\circ S_{ij}\Vert\Gamma^{\otimes k})
        &=D_{\max}\left(
        \Gamma^{\otimes(i-1)}\otimes(\Gamma\circ S_j)\otimes\Gamma^{\otimes(k-i)}
        \Vert\Gamma^{\otimes k}\right)\\
        &=(k-1)D_{\max}(\Gamma\Vert\Gamma)+D_{\max}(\Gamma\circ S_j\Vert\Gamma)\\
        &=D_{\max}(\Gamma\circ S_j\Vert\Gamma),
    \end{align*}
    where we used $D_{\max}(A\otimes B\Vert C\otimes D)=D_{\max}(A\Vert C)+D_{\max}(B\Vert D)$ and
    $D_{\max}(\Gamma\Vert\Gamma)=0$. Taking the maximum over $i,j$ proves item (ii) and completes
    the proof.
\end{proof}
\subsubsection{Direct product for query complexity}
Combining the direct-product identity with the query-complexity bounds gives the following theorem.

\begin{theorem}\label{thm:direct-product-distributional-query} Let $\Sigma\subseteq\{0,1\}^{n}$, and
     let $f\subseteq\Sigma\times\Omega$ be a search problem. Let $\mu:\Sigma\to[0,1]$ be a
     distribution, let $\alpha=\alpha(n)$ and $\beta=\beta(n)$ be success parameters with
     $0<\beta<\alpha<1$, and let $k\in\bN$ with $k\geq1$. Then,
    \begin{align*}
        \mathsf{Q}_{1-\alpha^k}^{\mu^{\otimes k}}(f^{\otimes k})
        \geq
        \frac{k(\alpha-\beta)}
        {10\alpha\log\left(\frac{2(1-\beta)}{\alpha-\beta}\right)}
        \mathsf{Q}_{1-\beta}^{\mu}(f).
    \end{align*}
\end{theorem}
\begin{proof}
    By the lower bound in \Cref{prop:aadv-lowerbound} and \Cref{prop:aadv-direct-product}, we have
    \begin{align}\label{eq:direct-product-aadv-lower-bound}
        \mathsf{Q}_{1-\alpha^k}^{\mu^{\otimes k}}(f^{\otimes k})
        &\geq\mathrm{AAdv}_{\alpha^k}^{\mu^{\otimes k}}(f^{\otimes k}) =k\cdot\mathrm{AAdv}_{\alpha}^{\mu}(f).
    \end{align}
    Applying the upper bound in \Cref{prop:aadv-upperbound} completes the proof.

\end{proof}

The preceding bounds also give the following asymptotic characterization.
\begin{corollary}[Asymptotic characterization]\label{cor:aadv-asymptotic-characterization} Let
    $\Sigma\subseteq\{0,1\}^{n}$, let $f\subseteq\Sigma\times\Omega$ be a search problem, and let
    $\mu:\Sigma\to[0,1]$ be a distribution. For a success parameter $\alpha=\alpha(n)\in(0,1)$, we
    have
    \begin{align*}
        \mathrm{AAdv}_{\alpha}^{\mu}(f)
        &=\Theta\!\left(
            \lim_{k\to\infty}
            \frac{\mathsf{Q}_{1-\alpha^k}^{\mu^{\otimes k}}(f^{\otimes k})}{k}
        \right).
    \end{align*}
\end{corollary}
\begin{proof}
    Let $\alpha < \alpha' <1$ where $\alpha'$ is arbitrarily close to $\alpha$. Then, by
    \Cref{def:aadv}, there exists a matrix $W\in\cW_f^\mu$ such that
    $\mathrm{AAdv}_{\alpha'}^{\mu}(f) = \bar{\delta}_2^{\mu}(J,W/\alpha'|\{S_i\})$. By
    \Cref{prop:delta-asymptotic-state-conversion}, for large enough $k$, there exists a Gram matrix
    $\Psi_k$ with
    \begin{align*}
        Q(\Psi_{k}) \leq 2k \bar{\delta}_2^{\mu}(J,W/\alpha'|\{S_i\}) + o(k) = 2k\mathrm{AAdv}^{\mu}_{\alpha'}(f) + o(k)
    \end{align*}
    and
    \begin{align*}
        \cF(\Psi_k\circ (\kbr{\mu})^{\otimes k},W^{\otimes k}\circ (\kbr{\mu})^{\otimes k})\geq \sqrt{(\alpha')^{k}2^{-o(k)}} \geq \sqrt{\alpha^{k}}.
    \end{align*}
    As shown in the proof of \Cref{prop:aadv-direct-product}, $W^{\otimes k}\in\cW_{f^{\otimes
    k}}^{\mu^{\otimes k}}$. Therefore, by \Cref{thm:fidelity-charecterization},
    \begin{align*}
       & Q^{\mu^{\otimes k}}_{1-\alpha^{k}}(f^{\otimes k}) \leq Q(\Psi_k) \leq 2k\mathrm{AAdv}^{\mu}_{\alpha'}(f) + o(k).
    \end{align*}
    Dividing by $k$ and taking limit $k\to \infty$ we obtain
    \begin{align*}
       \lim_{k\to \infty} \frac{Q^{\mu^{\otimes k}}_{1-\alpha^{k}}(f^{\otimes k})}{k} \leq 2\mathrm{AAdv}^{\mu}_{\alpha'}(f).
    \end{align*}
    Finally, by the dual formulation in \Cref{lem:aadv-universal-hardness-cert},
    $\mathrm{AAdv}^{\mu}_{\alpha'}(f)$ is continuous on $\alpha'$. Since the above inequality holds
    for any $\alpha' > \alpha$, we get
    \begin{align*}
        \lim_{k\to \infty}  \frac{Q^{\mu^{\otimes k}}_{1-\alpha^{k}}(f^{\otimes k})}{k} \leq \lim_{\alpha'\to \alpha} 2\mathrm{AAdv}^{\mu}_{\alpha'}(f) = 2 \mathrm{AAdv}^{\mu}_{\alpha}(f).
    \end{align*}
    Combined with \eqref{eq:direct-product-aadv-lower-bound}, the statement follows.
\end{proof}

\subsubsection{Direct sum for query complexity}
As a consequence of \Cref{thm:direct-product-distributional-query}, we obtain the following
direct-sum theorem.
\begin{corollary}\label{cor:direct-sum-distributional-query} Let $\Sigma\subseteq\{0,1\}^{n}$, let
    $f\subseteq\Sigma\times\Omega$ be a search problem, and let $\mu:\Sigma\to[0,1]$ be a
    distribution. For $\eps=\eps(n)\in(0,1/2)$ and every integer $k\geq1$, we have
    \begin{align*}
        \Omega\!\left(\eps k\cdot\mathsf{Q}^{\mu}_{\eps}(f)\right)
        \leq
        \mathsf{Q}^{\mu^{\otimes k}}_{\eps}(f^{\otimes k})
        \leq
        k\cdot\mathsf{Q}^{\mu}_{1-(1-\eps)^{1/k}}(f),
    \end{align*}
\end{corollary}
\begin{proof}
    For the upper bound, we can simply run an optimal one-copy algorithm with error $1-(1-\eps)^{1/k}$
    indepndenty on each of the $k$ copies which gives an algorithm for the $k$ copies with error
    $\eps$.
    \par
    If $k\leq 2$, the lower-bound is trivial. So, assume $k \geq 2$. Then, $\eps \leq  1-
    (1-\eps/2)^k= 1-\alpha^{k}$, and we can apply \Cref{thm:direct-product-distributional-query}
    with $\alpha = 1-\eps/2$ and $\beta = 1-\eps$ to get
    \begin{align*}
        \mathsf{Q}^{\mu^{\otimes k}}_{\eps}(f^{\otimes k}) &\geq
        \mathsf{Q}^{\mu^{\otimes k}}_{1-(1-\eps/2)^k}(f^{\otimes k})
        \geq
        \frac{k(\eps/2)}
        {10(1-\eps/2)\log\left(\frac{2\eps}{\eps/2}\right)}
        \mathsf{Q}^{\mu}_{\eps}(f)
        =\Omega\left(\eps k\cdot\mathsf{Q}^{\mu}_{\eps}(f)\right). \qedhere
    \end{align*}
\end{proof}

\subsubsection{Worst-case direct product theorem}
As another corollary, we obtain the following direct-product theorem for the worst-case query
complexity.
\begin{corollary}[Worst-case direct product]\label{cor:direct-product-worst-case-query} Let
    $\Sigma\subseteq\{0,1\}^{n}$, and let $f:\Sigma\to\Omega$ be a function. Let
    $\delta=\delta(n)\in(0,1)$ be a bias parameter, and let $k\in\bN$ with $k\geq1$. Then,
    \begin{align*}
        \mathsf{Q}_{1-(1/2+\delta/2)^k}(f^{\otimes k})
        =
        \Omega\left(
            \frac{k\delta}{-\log(\delta(1-\delta))}
            \mathsf{Q}_{1/2-\delta/2}(f)
        \right).
    \end{align*}
\end{corollary}
\begin{proof}
    Let $\eps:=1/2-\delta/4\in(1/4,1/2)$, and let $\mu$ be a hard
    distribution for $Q_{\eps}(f)$; i.e., $Q_{\eps}(f)=Q^{\mu}_{\eps}(f)$. Such a distribution exists
    by the minimax principle for quantum algorithms with a fixed query budget.
    Then, applying the direct-product theorem
    \Cref{thm:direct-product-distributional-query} with $\alpha=1/2+\delta/2$ and
    $\beta=1/2+\delta/4=1-\eps$, we obtain
    \begin{align*}
        Q_{1-(1/2+\delta/2)^k}(f^{\otimes k})
        &\geq Q^{\mu^{\otimes k}}_{1-(1/2+\delta/2)^k}(f^{\otimes k})\\
        &\geq \frac{k\delta}
        {20(1+\delta)\log\left(\frac{2(2-\delta)}{\delta}\right)}
        Q^{\mu}_{\eps}(f)\\
        &=\Omega\left(\frac{k\delta}{\log(2/\delta)}Q_{\eps}(f)\right).
    \end{align*}
    It only remains to amplify the error from $\eps_0 := 1/2-\delta/4$ to $\eps_1 = 1/2-\delta/2$. If $\delta>0.1$,
    then using Chernoff bound, we can achieve error $\eps_1$ with
    \begin{align*}
        O\left(\frac{\log(\eps_1^{-1})}{(\eps_0-1/2)^2}\right)
        &= O\left(\frac{-\log(1-\delta)}{\delta^2}\right) = O\left(-\log(1-\delta)\right).
    \end{align*}
    For the case where $\delta \leq 0.1$, majority vote on three independent runs
    decreases error from $\eps_0$ to at most
    \begin{align*}
        3\eps_0^2(1-\eps_0)+\eps_0^3
        &= \frac{1}{2}-\frac{3\delta}{8}+\frac{\delta^3}{32}
        = \frac{1}{2}-\frac{\delta(3-\delta^2/4)}{8}
        \leq \frac{1}{2}-\frac{\sqrt{2}\delta}{4}.
    \end{align*}
    So, applying this procedure twice (9 repetitions in total), we obtain error at most
    $1/2-\delta/2$. Consequently, combining both cases,
    \begin{align*}
        Q_{1/2-\delta/2}(f)
        = O\left(\log\left(\frac{2}{1-\delta}\right)Q_{1/2-\delta/4}(f)\right).
    \end{align*}
     Combining this with the preceding lower bound gives
    \begin{align*}
        Q_{1-(1/2+\delta/2)^k}(f^{\otimes k})
        &= \Omega\left(
            \frac{k\delta}{\log\left(\frac{2}{\delta}\right)
            \log\left(\frac{2}{1-\delta}\right)}
            Q_{1/2-\delta/2}(f)
        \right)\\
        &= \Omega\left(
            \frac{k\delta}{-\log(\delta(1-\delta))}
            Q_{1/2-\delta/2}(f)
        \right),
    \end{align*}
    where we used
    Since one of $\log(2/\delta)\log(2/(1-\delta)) = O(-\log(\delta(1-\delta)))$ for $\delta \in (0,1)$.
\end{proof}

\section*{Acknowledgments}
This research is supported in part by the Natural Sciences and Engineering Research Council of
Canada (NSERC), RGPIN-2019-04804.\footnote{Cette recherche a été financée par le Conseil de
recherches en sciences naturelles et en génie du Canada (CRSNG), RGPIN-2019-04804.}

\section*{AI methodology}
ChatGPT was used extensively to brainstorm different adversary measures.
In particular, the Shaltiel-free, filtered-fidelity,
multiplicative-\(\gamma_2\), and Alberti-adversary formulations were
developed through many rounds of discussion, in which
we tried out, compared, and refined several variants until
we settled on what we thought were the ``right'' definitions.
To this end, ChatGPT was particularly helpful in deriving primal and
dual programs, finding relaxations and equivalent formulations, and
working through minimax arguments, which resulted in the
characterizations in
\Cref{lem:dual-sfadv-hardness,lem:aadv-universal-hardness-cert}
and \Cref{prop:equiv-filtered-fidelity}. ChatGPT also assisted with proofs of the basic properties in
\Cref{prop:filtered-fidelity-properties}, the extension of the
catalyst-reuse lemma (\Cref{lem:catalyst-reuse}), and the
asymptotic state-conversion result in
\Cref{prop:multiplicative-state-conversion,prop:delta-asymptotic-state-conversion}.
It also supplied the short arguments deriving
\Cref{lem:decision-delta2-bounds,thm:distributional-xor}
from the preceding results. The models used were mainly gpt-5.6-sol, gpt-6-astra, and gpt-6.1-sol.

The tools were also used to assist the authors with proofreading,
phrasing, clarity, and \LaTeX{} presentation. All AI outputs used in this work were checked
by the authors. The document was written
by the human authors.

\printbibliography

\appendix
\addtocontents{toc}{\protect\setcounter{tocdepth}{1}}
\section{Proofs of individual results}\label{app:individual-proofs}

\subsection{\texorpdfstring{Proof of \Cref{lem:compactness-gmueps}: compactness and convexity}{Proof
of the compactness and convexity lemma}}\label{app:compactness}

\begin{proof}[Proof of \Cref{lem:compactness-gmueps}]
    Let $\lambda\in [0,1]$, $W_1 , W_2 \in \cW^{\mu}_{f}$ with corresponding decompositions $W_1 =
    \sum_{s\in \Omega} W_1^{s}\in \cW^{\mu}_{f}$ and $W_{2} = \sum_{s\in \Omega} W_2^{s}$ as in
    \Cref{def:witness-set}. Then, define $ W^{s} := \lambda W_1^{s} + (1-\lambda) W_2^{s}$; it
    satisfies both PSDness and $W^{s}(x,x) = 0$ whenever $s\notin f(x)$. Consequently, $W := \lambda
    W_1 + (1-\lambda) W_2 = \sum_{s\in\Omega}  W^{s}$ also is in $\cW^{\mu}_{f}$ as $\Tr[W\circ
    \kbr{\mu}] = 1$.
    \par
    For closedness of $\cW^{\mu}_f$, consider a converging family $W_{k} \to W$ with decompositions
    $W_{k} = \sum_{s} W^{s}_k$. Since $W_k\to W$, there exists $r \geq 0$ such that $\|W_k\|\leq r$.
    Let $\mathcal{B} := \{X: \|X\| \leq r\}$. Since $0\preceq W_k^s\preceq W_k$, we have
    $\|W_k^s\|\leq\|W_k\|\leq r$. Hence, for all $s\in \Omega$ and $k\in \bN$, $W^{s}_k \in
    \mathcal{B}$. Moreover, $\mathcal{B}$ is compact, and hence any sequence of elements $W^{s}_k$
    has a convergent subsequence. Since $\Omega$ is finite, we can successively extract convergent
    subsequences to obtain a common subsequence $\{k_j:j\in\bN\}$ along which $W^s_{k_j}\to W^s$ for
    every $s\in\Omega$. Consequently,
    $$W = \lim_{j\to \infty} W_{k_j} = \lim_{j\to \infty} \sum_{s\in \Omega} W^{s}_{k_j} =
    \sum_{s\in \Omega}\lim_{j\to \infty} W^{s}_{k_j} = \sum_{s\in \Omega} W^{s}.$$ Since PSDness and
    the correctness constraints are preserved under limits, and $\Tr[W\circ\kbr{\mu}]=1$, this is a
    feasible decomposition of $W$. It follows that $W \in \cW^{\mu}_{f}$, and $\cW^{\mu}_f$ is
    closed.
    \par
    Next, let $S:=\supp(\mu)$. The constrained set $\widehat{\cW}^{\mu}_{f}$ is closed and convex,
    since it is the intersection of $\cW^{\mu}_{f}$ with the closed linear subspace $\{W:W\circ
    J_S=W\}$. Moreover, for every $W\in\widehat{\cW}^{\mu}_{f}$, we have
    \begin{align*}
        \|W\|\leq\Tr[W]
        =\sum_{x\in S}W(x,x)
        \leq\frac{1}{\min_{x\in S}\mu(x)}
        \sum_{x\in S}\mu(x)W(x,x)
        =\frac{1}{\min_{x\in S}\mu(x)}.
    \end{align*}
    Thus, $\widehat{\cW}^{\mu}_{f}$ is bounded, and hence compact.
    \par
    We now turn to the set $\cG^{\mu}_{\eps}(f)$. For $W=\sum_sW^s\in\cW^{\mu}_{f}$, replacing each
    $W^s$ by $W^s\circ J_S$ preserves PSDness and the correctness constraints. Also, $(W\circ
    J_S)\circ\kbr{\mu}=W\circ\kbr{\mu}$, so normalization under $\mu$ is preserved and $W\circ
    J_S\in\widehat{\cW}^{\mu}_{f}$. The fidelity objective is unchanged as well. Therefore, by
    \Cref{lem:adversary-support-replacement}, for every Gram matrix $\Psi$,
    \begin{align*}
        \sup_{W\in\cW^{\mu}_{f}}\cF(W\circ\kbr{\mu},\Psi\circ\kbr{\mu})
        =\max_{W\in\widehat{\cW}^{\mu}_{f}}\cF(W\circ\kbr{\mu},\Psi\circ\kbr{\mu}),
    \end{align*}
    where the maximum on the right is attained by compactness and continuity of fidelity.
    Consequently, the maximum in \Cref{def:valid-grams} is attained, and we may choose its witnesses
    in $\widehat{\cW}^{\mu}_{f}$.
    \par
    For convexity, let $\lambda\in [0,1]$, $\Psi_1 , \Psi_2 \in \cG^{\mu}_{\eps}(f)$ with
    corresponding witnesses $W_1, W_2 \in \widehat{\cW}_f^{\mu}$ where $\cF(\Psi_i \circ \kbr{\mu},
    W_i\circ \kbr{\mu}) \geq \sqrt{1-\eps}$. First of all, $\lambda \Psi_1 + (1-\lambda)\Psi_2$ is
    also a Gram matrix. Moreover, by convexity of $\widehat{\cW}^{\mu}_f$, $\lambda W_1 +
    (1-\lambda) W_2 \in \widehat{\cW}^{\mu}_f$. By joint concavity of fidelity,
    \begin{align*}
        &\cF( (\lambda \Psi_1 + (1-\lambda)\Psi_2)\circ \kbr{\mu},(\lambda W_1 + (1-\lambda)W_2)\circ \kbr{\mu} ) &\\
        \geq& \lambda \cF(\Psi_1\circ \kbr{\mu},W_1\circ \kbr{\mu})  +
         (1-\lambda)\cF(\Psi_2\circ \kbr{\mu},W_2\circ \kbr{\mu}) \\ \geq& \sqrt{1-\eps}.
    \end{align*}
    For closedness, let $\Psi_k\in\cG^{\mu}_{\eps}(f)$ converge to $\Psi$, with corresponding
    witnesses $W_k\in\widehat{\cW}^{\mu}_{f}$. Then, $\Psi$ is a Gram matrix. By compactness of
    $\widehat{\cW}^{\mu}_{f}$, there exists a subsequence $W_{k_j}$ converging to
    $W\in\widehat{\cW}^{\mu}_{f}$. Therefore, by continuity of fidelity,
    \begin{align*}
        \cF(\Psi\circ\kbr{\mu},W\circ\kbr{\mu})
        =\lim_{j\to\infty}\cF(\Psi_{k_j}\circ\kbr{\mu},W_{k_j}\circ\kbr{\mu})
        \geq\sqrt{1-\eps}.
    \end{align*}
    Thus, $\Psi\in\cG^{\mu}_{\eps}(f)$, and $\cG^{\mu}_{\eps}(f)$ is closed.
    \par
    Finally, the set $\cG^{\mu}_{\eps}(f)$ is bounded as $\Psi \in \cG^{\mu}_{\eps}(f)$ is a Gram
    matrix with bounded norm $\|\Psi\| \leq \Tr[\Psi] = 1\cdot |\Sigma|$. Consequently,
    $\cG^{\mu}_{\eps}(f)$ is a compact subset of $\cL(\Sigma)$.
\end{proof}

\subsection{\texorpdfstring{Proof of \Cref{lem:catalyst-reuse}: catalyst-reuse lemma}{Proof of the
catalyst-reuse lemma}}\label{app:catalyst-reuse}


\begin{proof}
    Write $s:=\sum_{t=0}^{T-1}r^t$. Let $\mathcal{T}:=\bC^T$ be a clock register with computational
    basis $\{\ket{t}:t\in[T]\}$. Our algorithm will work in the space
    $(\mathcal{T}\otimes\cV)\oplus\cK$; that is, we attach the clock register to the $\cV$-space and
    add the catalyst space $\cK$ to the end.
    \par
    Let $U$ be the input-independent unitary in the lemma. For $t\in[T]$, we define $U_t$ to be the
    extension of $U$ to the space $(\mathcal{T}\otimes\cV)\oplus\cK$, where it applies $U$ jointly
    to the $t$-th $\cV$-component and the catalyst space $\cK$, and leaves the other clock
    components unchanged. More specifically, for all $\{\ket{v_j}\}_{j=1}^{T}\subset\cV$ and
    $\ket{\kappa}\in\cK$,
    \begin{align*}
        U_t\left(
            \left(
                \sum_{j=1}^{T}
                \ket{j}_{\mathcal T}\ket{v_j}_{\cV}
            \right)
            \oplus\ket{\kappa}_{\cK}
        \right)
        =
        \left(
            \sum_{j\neq t}
            \ket{j}_{\mathcal T}\ket{v_j}_{\cV}
            +
            \ket{t}_{\mathcal T}\ket{v_t'}_{\cV}
        \right)
        \oplus\ket{\kappa_t'}_{\cK},
    \end{align*}
    where $ \ket{v_t'}_{\cV}\oplus\ket{\kappa_t'}_{\cK} =
        U\left(\ket{v_t}_{\cV}\oplus\ket{\kappa}_{\cK}\right).$
    \par
    Given the initial vector $\ket{a_x}\oplus \ket{\kappa_x}/\sqrt{s}$, our algorithm proceeds as
    follows.
    \begin{itemize}
        \item First, we split $\ket{a_x}$ into a geometric superposition of $T$ branches by applying
        a unitary $H$ satisfying $H\ket{1}_{\mathcal{T}} =
        \sum_{j}r^{(j-1)/2}\ket{j}_{\mathcal{T}}/\sqrt{s}$ to the clock register; i.e., we prepare
        the vector
        \begin{align*}
            \ket{\sigma_x^0}
            :=\frac{1}{\sqrt{s}}\left(\sum_{j=1}^{T}r^{(j-1)/2} \ket{j}_{\mathcal{T}} \ket{a_x}_{\cV}\right)\oplus
            \frac{\ket{\kappa_x}_{\cK}}{\sqrt{s}}
        \end{align*}
        \item Then, we inductively apply a series of $T$ many oracle queries $\cO^{x}$ followed by
        the unitary $U_t$ to flip each coordinate $\ket{j}_{\mathcal{T}}\ket{a_x}_{\cV}$ into
        $\ket{j}_{\mathcal{T}} \ket{b_x}_{\cV}$ by reusing the catalyst $\ket{\kappa_x}$. More
        specifically, let $t\in[T]$, and suppose that after $t-1$ steps we have computed the vector
        \begin{align*}
            \ket{\sigma_x^{t-1}}
            &=\frac{1}{\sqrt{s}}\left(\sum_{j=1}^{t-1}r^{(j-1)/2}\ket{j}_{\mathcal T}
                \ket{b_x}_{\cV} + \sum_{j=t}^{T}r^{(j-1)/2}\ket{j}_{\mathcal T}
                \ket{a_x}_{\cV}\right) \oplus \frac{r^{(t-1)/2}\ket{\kappa_{x}}_{\cK}}{\sqrt{s}}
        \end{align*}
        As a result, at step $t$ we obtain the next vector as
        \begin{align*}
            \ket{\sigma^{t}_x} :=& U_t(I \oplus \cO^{x}) \ket{\sigma^{t-1}_x} \\
            =&U_t(I \oplus \cO^{x}) \left(\frac{1}{\sqrt{s}}\left(\sum_{j=1}^{t-1}r^{(j-1)/2}\ket{j}_{\mathcal T}
                \ket{b_x}_{\cV} + \sum_{j=t}^{T}r^{(j-1)/2}\ket{j}_{\mathcal T}
                \ket{a_x}_{\cV}\right) \oplus \frac{r^{(t-1)/2}\ket{\kappa_{x}}_{\cK}}{\sqrt{s}}\right) \\
            =& U_t \left(\frac{1}{\sqrt{s}}\left(\sum_{j=1}^{t-1}r^{(j-1)/2}\ket{j}_{\mathcal T}
            \ket{b_x}_{\cV} + \sum_{j=t}^{T}r^{(j-1)/2}\ket{j}_{\mathcal T}
            \ket{a_x}_{\cV}\right) \oplus \frac{r^{(t-1)/2}\cO^{x}\ket{\kappa_{x}}_{\cK}}{\sqrt{s}}\right) \\
            =&
            \frac{1}{\sqrt{s}}\left(\sum_{j=1}^{t-1}r^{(j-1)/2}\ket{j}_{\mathcal T}
            \ket{b_x}_{\cV}\right) \oplus 0_{\cK}
            + \left(\ket{t}_{\mathcal{T}}\frac{r^{(t-1)/2}\ket{b_x}}{\sqrt{s}}\right)\oplus \frac{r^{t/2}\ket{\kappa_x}_{\cK}}{\sqrt{s}}
            \\ &\quad{}+\frac{1}{\sqrt{s}}\left(\sum_{j=t+1}^{T}r^{(j-1)/2}\ket{j}_{\mathcal T}
            \ket{a_x}_{\cV} \right)\oplus 0_{\cK} \\
            =&\frac{1}{\sqrt{s}}\left(\sum_{j=1}^{t}r^{(j-1)/2}\ket{j}_{\mathcal T}
            \ket{b_x}_{\cV} + \sum_{j=t+1}^{T}r^{(j-1)/2}\ket{j}_{\mathcal T}
            \ket{a_x}_{\cV}\right) \oplus \frac{r^{t/2}\ket{\kappa_{x}}_{\cK}}{\sqrt{s}}
        \end{align*}
        where we used the fact that $U(\ket{a_x}\oplus \cO^{x}\ket{\kappa_x}) = \ket{b_x}\oplus
        \sqrt{r}\,\ket{\kappa_x}$, and the unitary $(I\oplus \cO^{x})$ can be implemented using a
        single query to $\cO^{x}$.
        \item Consequently, after $T$ queries, we obtain
        \begin{align*}
            \ket{{\sigma}_x^T}
            =\frac{1}{\sqrt{s}}\left(\sum_{j=1}^{T}r^{(j-1)/2} \ket{j}_{\mathcal{T}} \ket{b_x}_{\cV}\right)\oplus
            \frac{r^{T/2}\ket{\kappa_x}_{\cK}}{\sqrt{s}}
        \end{align*}
       \item Finally, we unprepare the clock register by applying $H^{\dagger}\ket{\sigma_{x}^{T}}$
        and obtain the desired final vector $\ket{b_x}\oplus r^{T/2}\ket{\kappa_x}/\sqrt{s}$.
    \end{itemize}
    Each step makes exactly one use of $\cO^x$, and therefore, the algorithm makes a total of $T$
    queries.
\end{proof}

\subsection{\texorpdfstring{Proof of \Cref{lem:geometric-convergence}: geometric convergence}{Proof
of the geometric-convergence lemma}}
\label{app:geometric-convergence}

\begin{proof}
    We first decompose
    \begin{align*}
        V(0 \oplus \ket{\kappa}) = V_{\cV}\ket{\kappa} \oplus V_{\cK} \ket{\kappa},
    \end{align*}
    where $\|V_{\cV}\| \leq 1$ and $\|V_{\cK}\|\leq 1$.  Therefore, for $t\geq 0$, we have
    \begin{align*}
        \ket{h_{t+1}} &= V_{\cK} \ket{h_{t}} &&
        \ket{b_{t+1}} = V_{\cV} \ket{h_t}.
    \end{align*}
    In addition, by the given unitary equality,
    \begin{align*}
        \ket{b_0} \oplus \ket{h_0}
        &= V(\ket{a} \oplus 0)\\
        &= V(\ket{a} \oplus \ket{\kappa}) - V(0 \oplus \ket{\kappa})\\
        &= \ket{b} \oplus \sqrt{r}\ket{\kappa}
        -(V_{\cV}\ket{\kappa}\oplus V_{\cK}\ket{\kappa}).
    \end{align*}
    Therefore, from the second component, for $r>1$ we can write
    \begin{align*}
        \ket{\kappa}
        &= \frac{\ket{h_0}}{\sqrt{r}}+\frac{V_{\cK}\ket{\kappa}}{\sqrt{r}}
        = \sum_{t=0}^{\infty}\frac{V_{\cK}^{t}\ket{h_0}}{\sqrt{r^{t+1}}}
        = \sum_{t=0}^{\infty}\frac{\ket{h_t}}{\sqrt{r^{t+1}}}.
    \end{align*}
    Here, the series converges in norm because $\|V_{\cK}/\sqrt{r}\|\leq1/\sqrt{r}<1$. On the other
    hand, from the first component, we obtain that
    \begin{align*}
        \ket{b}
        &= \ket{b_0}+V_{\cV}\ket{\kappa}
        = \ket{b_0}+\sum_{t=0}^{\infty}\frac{V_{\cV}\ket{h_t}}{\sqrt{r^{t+1}}}\\
        &= \ket{b_0}+\sum_{t=0}^{\infty}\frac{\ket{b_{t+1}}}{\sqrt{r^{t+1}}}
        = \sum_{t=0}^{\infty}\frac{\ket{b_t}}{\sqrt{r^t}}.
    \end{align*}
    For $r>1$, the norm distance is the tail of the series, and
    \begin{align*}
        \left\|\sum_{t=T+1}^{\infty}\frac{\ket{b_t}}{\sqrt{r^t}}\right\|
        &= \left\|\sum_{t=T}^{\infty}
            \frac{V_{\cV}V_{\cK}^{t}\ket{h_0}}{\sqrt{r^{t+1}}}\right\|\\
        &=\frac{1}{\sqrt{r^T}}\left\|V_{\cV}V_{\cK}^{T}\ket{\kappa}\right\|
        \leq \frac{1}{\sqrt{r^T}}\|\ket{\kappa}\|. \qedhere
    \end{align*}
\end{proof}

\subsection{\texorpdfstring{Proof of \Cref{prop:direct-sum-distributional-adversary}: direct sum for
the distributional adversary}{Proof of the direct-sum proposition for the distributional adversary}}
\label{app:adv-direct-sum}

We prove the lower and upper bounds separately.

\subsubsection{Lower bound}
We begin with the harder direction, the lower bound.

\begin{proposition}\label{prop:direct-sum-lower-bound} Let $\Sigma \subset \{0,1\}^{n}$, $f\subseteq
    \Sigma \times \Omega$ be a search problem, and $\mu:\Sigma \to [0,1]$ be a distribution. Then,
    for $\eps=\eps(n)\in(0,1)$ and any integer $k\geq1$, $k\cdot \mathrm{Adv}^{\mu}_{\eps}(f)\leq
    \mathrm{Adv}^{\mu^{\otimes k}}_{\eps}(f^{\otimes k})$.
 \end{proposition}
\begin{proof}
    Let $\cH=\bC^\Sigma$, using the tensor notation from \Cref{sec:prelim-hilbert-spaces}. By
    \Cref{lem:adv-universal-hardness-cert}, we can write
    \begin{align*}
        \mathrm{Adv}^{\mu}_{\eps}(f)
        =\max_{\substack{\Gamma=\Gamma^\dagger\\
        \max_i\|\Gamma\circ\Delta_i\|\leq1}}
        \min_{\Psi \in \cG^{\mu}_{\eps}(f)}
        \langle \Gamma, (J-\Psi)\circ \kbr{\mu}\rangle.
    \end{align*}
    Let $\Gamma$ be the adversary matrix achieving the maximum. The main idea is to lift $\Gamma$
    into a good adversary matrix for $f^{\otimes k}$. To do so, we define a lifting map
    $\cE^{\uparrow}:\cL(\cH)\to \cL(\cH^{\otimes k})$ as
    \begin{align*}
        \cE^{\uparrow}(X)
        :=\frac{1}{k}\sum_{\ell=1}^{k}I^{\otimes(k-1)}\otimes_\ell X
        =\frac{1}{k}\sum_{\ell=1}^{k}I^{\otimes(\ell-1)}\otimes X\otimes I^{\otimes(k-\ell)}.
    \end{align*}
    That is, $\cE^{\uparrow}(\Gamma)$ consists of a uniform combination of a single copy of $\Gamma$
    on a block $\ell \in [k]$ with identity on the remaining blocks $\bar{\ell}$. To show that
    $\cE^{\uparrow}(\Gamma)$ is a good adversary candidate, we need to show that the lifting map
    behaves well with respect to the query filters $\Delta_{ij}$ and inner-products with target Gram
    matrices in $\cG^{\mu^{\otimes k}}_{\eps}(f^{\otimes k})$.

    \paragraph{Query filters.}
    \begin{claim}\label{claim:direct-sum-lower-filters} For $i\in[k]$ and $j\in[n]$, the query
        filter corresponding to coordinate $j$ of block $i$ is
        $\Delta_{ij}:=J^{\otimes(k-1)}\otimes_i\Delta_j$. Moreover, for every $X\in\cL(\cH)$,
        $
            \cE^{\uparrow}(X)\circ\Delta_{ij}
            =\frac{1}{k}I^{\otimes(k-1)}\otimes_i(X\circ\Delta_j).
        $
        Consequently,
        \begin{align*}
            \max_{i,j}\|\cE^{\uparrow}(\Gamma)\circ\Delta_{ij}\|
            =\frac{1}{k}\max_j\|\Gamma\circ\Delta_j\|.
        \end{align*}
    \end{claim}
    \begin{proof}
        For $\mathbf{x},\mathbf{y}\in\Sigma^k$, we have
        $
            (J^{\otimes(k-1)}\otimes_i\Delta_j)(\mathbf{x},\mathbf{y})
            =\Delta_j(x^i,y^i)=\ind[x_j^i\neq y_j^i],
        $
        which proves the stated form of $\Delta_{ij}$. Moreover, for every $X\in\cL(\cH)$,
        \begin{align*}
            \cE^{\uparrow}(X)\circ\Delta_{ij}
            &=\frac{1}{k}\sum_{\ell=1}^{k}
            (I^{\otimes(k-1)}\otimes_\ell X)
            \circ(J^{\otimes(k-1)}\otimes_i\Delta_j)
        =\frac{1}{k}I^{\otimes(k-1)}\otimes_i(X\circ\Delta_j),
        \end{align*}
        where every term with $\ell\neq i$ vanishes because its $i$-th tensor factor is
        $I\circ\Delta_j=0$. Taking $X=\Gamma$ and using multiplicativity of the operator norm gives
        $\|\cE^{\uparrow}(\Gamma)\circ\Delta_{ij}\| =\|\Gamma\circ\Delta_j\|/k$, which proves the
        claim.
    \end{proof}

    \paragraph{Inner product.}
    \begin{claim}\label{claim:direct-sum-lower-inner-product} For every
        $\Lambda\in\cG_{\eps}^{\mu^{\otimes k}}(f^{\otimes k})$, there exists
        $\Psi_\Lambda\in\cG_{\eps}^{\mu}(f)$ such that
        \begin{align*}
            \langle\Gamma,(J-\Psi_\Lambda)\circ\kbr{\mu}\rangle
            =\langle\cE^{\uparrow}(\Gamma),
            (J^{\otimes k}-\Lambda)\circ\kbr{\mu}^{\otimes k}\rangle.
        \end{align*}
    \end{claim}
    \begin{proof}
        We first compute the adjoint of $\cE^{\uparrow}$ as
        \begin{align*}
            \cE^{\downarrow}(Y)
            :=(\cE^{\uparrow})^{\dagger}(Y)
            =\frac{1}{k}\sum_{\ell=1}^{k}\Tr_{\bar\ell}[Y].
        \end{align*}
        The correctness of the adjoint follows from linearity and the fact that the adjoint of
        $X\mapsto I^{\otimes(k-1)}\otimes_\ell X$ is the partial trace $Y\mapsto\Tr_{\bar\ell}[Y]$.
        Also, $\cE^{\downarrow}$ is a quantum channel since it is an average of partial-trace
        channels. We note that, for every $X\in\cL(\cH^{\otimes k})$,
        \begin{align}\label{eq:direct-sum-weighted-reduction}
            \cE^{\downarrow}(X\circ\kbr{\mu}^{\otimes k})
            =\kbr{\mu}\circ
            \E_{\substack{\ell\sim[k]\\z\sim\mu^{\otimes(k-1)}}}
            \left[(\bra{z}\otimes_\ell I)X(\ket{z}\otimes_\ell I)\right].
        \end{align}
        Indeed, for every $x,y\in\Sigma$, the $(x,y)$-entry of the left-hand side is
        \begin{align*}
            &\frac{1}{k}\sum_{\ell=1}^k\sum_{z\in\Sigma^{k-1}}
            (\bra{z}\otimes_\ell\bra{x})(X\circ\kbr{\mu}^{\otimes k})
            (\ket{z}\otimes_\ell\ket{y})\\
            &\qquad=\sqrt{\mu(x)\mu(y)}\frac{1}{k}
            \sum_{\ell=1}^k\sum_{z\in\Sigma^{k-1}}
            \mu^{\otimes(k-1)}(z)
            (\bra{z}\otimes_\ell\bra{x})X(\ket{z}\otimes_\ell\ket{y}),
        \end{align*}
        which is the $(x,y)$-entry of the right-hand side. Define $\Psi_\Lambda :=
        \E_{\ell,z}[(\bra{z}\otimes_{\ell} I)\Lambda (\ket{z}\otimes_{\ell} I)]$. For $X =
        J^{\otimes k}$, we have $(\bra{z}\otimes_{\ell} I)J^{\otimes k}(\ket{z}\otimes_{\ell} I) =
        J^{\otimes(k-1)}(z,z)J = J$. Therefore, applying \eqref{eq:direct-sum-weighted-reduction} to
        $J^{\otimes k}$ and $\Lambda$, respectively, gives
        \begin{align*}
            \cE^{\downarrow}(J^{\otimes k}\circ\kbr{\mu}^{\otimes k})
            &=J\circ\kbr{\mu},\\
            \cE^{\downarrow}(\Lambda \circ \kbr{\mu}^{\otimes k}) &=
            \kbr{\mu}\circ \Psi_\Lambda.
        \end{align*}
        As a result, by definition of the adjoint map,
        \begin{align*}
            \langle\cE^{\uparrow}(\Gamma),
            (J^{\otimes k}-\Lambda)\circ\kbr{\mu}^{\otimes k}\rangle
            &=\langle\Gamma,
            \cE^{\downarrow}((J^{\otimes k}-\Lambda)\circ\kbr{\mu}^{\otimes k})\rangle\\
            &=\langle\Gamma,(J-\Psi_\Lambda)\circ\kbr{\mu}\rangle.
        \end{align*}
        \par
        It remains to prove that $\Psi_\Lambda\in\cG_{\eps}^{\mu}(f)$. Each compression of $\Lambda$
        in the definition of $\Psi_\Lambda$ is positive semidefinite. Moreover, for every
        $x\in\Sigma$,
        \begin{align*}
            \Psi_\Lambda(x,x)=\E_{\ell,z}
            (\bra{z}\otimes_\ell\bra{x})\Lambda
            (\ket{z}\otimes_\ell\ket{x})=1.
        \end{align*}
        Hence, $\Psi_\Lambda$ is a Gram matrix. Let $W =
        \sum_{\mathbf{s}\in\Omega^k}W^{\mathbf{s}}\in \cW_{f^{\otimes k}}^{\mu^{\otimes k}}$ be a
        witness such that $\cF(\Lambda\circ\kbr{\mu}^{\otimes k}, W\circ\kbr{\mu}^{\otimes k}) \geq
        \sqrt{1-\eps}$, and define $U := \E_{\ell,z}[(\bra{z}\otimes_{\ell} I)W
        (\ket{z}\otimes_{\ell} I)]$. By \eqref{eq:direct-sum-weighted-reduction},
        $\cE^{\downarrow}(W\circ\kbr{\mu}^{\otimes k}) = \kbr{\mu}\circ U$. For every $a \in
        \Omega$, define $U^a := \E_{\ell,z}\sum_{\mathbf{s}:s_{\ell}=a}(\bra{z}\otimes_{\ell} I)
        W^{\mathbf{s}}(\ket{z}\otimes_{\ell} I)$. Each $U^a$ is positive semidefinite and $\sum_a
        U^a = U$. Also, if $a \notin f(x)$, then for every $\ell,z$ and $\mathbf{s}$ with $s_{\ell}
        = a$, the output $\mathbf{s}$ is invalid for the product input corresponding to
        $\ket{z}\otimes_\ell\ket{x}$. The correctness of $W^{\mathbf{s}}$ therefore gives
        $U^a(x,x)=0$. Finally,
        \begin{align*}
            \Tr[U\circ\kbr{\mu}]
            =\Tr[\cE^{\downarrow}(W\circ\kbr{\mu}^{\otimes k})]
            =\Tr[W\circ\kbr{\mu}^{\otimes k}]=1.
        \end{align*}
        Therefore, $U\in\cW_f^\mu$. The matrices $\Lambda\circ\kbr{\mu}^{\otimes k}$ and
        $W\circ\kbr{\mu}^{\otimes k}$ are density matrices. Since $\cE^{\downarrow}$ is a quantum
        channel, monotonicity of fidelity gives
        \begin{align*}
            \cF(\Psi_\Lambda\circ\kbr{\mu},U\circ\kbr{\mu})
            &\geq\cF(\Lambda\circ\kbr{\mu}^{\otimes k},
            W\circ\kbr{\mu}^{\otimes k})
            \geq\sqrt{1-\eps}.
        \end{align*}
        Hence, $\Psi_\Lambda\in\cG_{\eps}^{\mu}(f)$.
    \end{proof}

    \paragraph{Putting everything together.}
    By \Cref{claim:direct-sum-lower-filters},
    $\max_{i,j}\|k\cE^{\uparrow}(\Gamma)\circ\Delta_{ij}\|\leq1$. Thus, by
    \Cref{lem:adv-universal-hardness-cert}, we have
    \begin{align*}
        \mathrm{Adv}_{\eps}^{\mu^{\otimes k}}(f^{\otimes k})
        &=
        \progline{\max_{\Theta=\Theta^\dagger}
        \min_{\Lambda \in \cG_{\eps}^{\mu^{\otimes k}}(f^{\otimes k})}}
        {\langle \Theta,(J^{\otimes k}-\Lambda)\circ\kbr{\mu}^{\otimes k}\rangle}
        {\max_{i,j}\|\Theta\circ\Delta_{ij}\|\leq1}\\
        &\geq
        \min_{\Lambda \in \cG_{\eps}^{\mu^{\otimes k}}(f^{\otimes k})}
        \langle k\cE^{\uparrow}(\Gamma),
        (J^{\otimes k}-\Lambda)\circ\kbr{\mu}^{\otimes k}\rangle\\
        &=k\min_{\Lambda \in \cG_{\eps}^{\mu^{\otimes k}}(f^{\otimes k})}
        \langle \Gamma,(J-\Psi_\Lambda)\circ\kbr{\mu}\rangle\\
        &\geq k\min_{\Psi \in \cG_{\eps}^{\mu}(f)}
        \langle \Gamma,(J-\Psi)\circ\kbr{\mu}\rangle\\
        &=k\cdot\mathrm{Adv}_{\eps}^{\mu}(f).
    \end{align*}
    Here, the second line follows by taking $\Theta=k\cE^{\uparrow}(\Gamma)$, the third line follows
    from \Cref{claim:direct-sum-lower-inner-product}, the fourth line uses
    $\Psi_\Lambda\in\cG_{\eps}^{\mu}(f)$, and the last equality follows from our choice of $\Gamma$.
\end{proof}
\subsubsection{Upper bound}
Next, we prove the upper bound in \Cref{prop:direct-sum-distributional-adversary}.
\begin{proposition}\label{prop:direct-sum-upper-bound} Let $\Sigma \subset \{0,1\}^{n}$, $f\subseteq
    \Sigma \times \Omega$ be a search problem, $\mu:\Sigma \to [0,1]$ be a distribution,
    $\eps=\eps(n)\in(0,1)$, and $k\in \bN$ with $k\geq 1$. Then $\mathrm{Adv}^{\mu^{\otimes
    k}}_{\eps}(f^{\otimes k}) \leq k\cdot \mathrm{Adv}^{\mu}_{1-(1-\eps)^{1/k}}(f)$.
\end{proposition}
\begin{proof}
    Let $\delta := 1- (1-\eps)^{1/k}$. We begin with the min-max formulation in
    \Cref{def:adv,lem:equiv-avg-gamma2-query}:
    \begin{align*}
        \mathrm{Adv}^{\mu}_{\delta}(f)
        =\min_{\Psi \in \cG^{\mu}_{\delta}(f)}
        \max_{\substack{\Gamma=\Gamma^\dagger\\
        \max_i\|\Gamma\circ\Delta_i\|\leq1}}
        \langle \Gamma,(J-\Psi)\circ\kbr{\mu}\rangle,
    \end{align*}
    and choose $\Psi$ attaining the minimum above. Our goal is to show that the lifted Gram matrix
    $\Psi^{\otimes k}$ belongs to $\cG^{\mu^{\otimes k}}_{\eps}(f^{\otimes k})$ and use it as a
    certificate to upper bound $\mathrm{Adv}^{\mu^{\otimes k}}_{\eps}(f^{\otimes k})$.
    \paragraph{The lifted witness}
    \begin{claim}\label{claim:direct-sum-upper-target} $ \Psi^{\otimes k} \in \cG^{\mu^{\otimes
       k}}_{\eps}(f^{\otimes k}).$
    \end{claim}
\begin{proof}
    Since $\Psi$ is a Gram matrix, $\Psi^{\otimes k}$ is also a Gram matrix. Let $W\in \cW^{\mu}_f$
    be the corresponding witness of $\Psi$, where $\cF(\Psi \circ \kbr{\mu}, W\circ \kbr{\mu}) \geq
    \sqrt{1-\delta}$. Also, let $W = \sum_{s\in \Omega} W^{s}$ be the decomposition of $W$
    satisfying $W^{s}\succeq 0$ and $W^{s}(x,x) = 0$ for all $s\in \Omega$ and $s\notin f(x)$.
    \par
    Let $\hat{W} := \bigotimes_{j\in [k]} W$. Then,
    \begin{align*}
        \Tr[\hat{W}\circ \kbr{\mu^{\otimes k}}] =
        \Tr[\left(
        \bigotimes_{j\in [k]} W \right)\circ \left(\bigotimes_{j\in [k]} \kbr{\mu} \right)]
        = \prod_{j\in [k]} \Tr[W\circ \kbr{\mu}] = 1.
    \end{align*}
    Moreover, by the decomposition $W = \sum_{s} W^{s}$, we can write
    \begin{align*}
        \hat{W} = \bigotimes_{j\in [k]} \left(\sum_{s\in \Omega} W^s\right) = \sum_{(s^1,\dots,s^{k})\in\Omega^k}
        W^{s^1}\otimes \dots \otimes W^{s^{k}}.
    \end{align*}
    By PSDness of each $W^{s^{j}}$, each summand is also PSD. Moreover, if $(s^1,\dots,s^{k}) \notin
    f^{\otimes k}(x^1,\dots,x^k)$, there exists some index $j$ such that $s^{j} \notin f(x^{j})$.
    Consequently,
    $$\left(W^{s^1}\otimes \dots\otimes W^{s^{k}}\right) ((x^1,\dots,x^{k}),(x^1,\dots,x^{k})) =
    W^{s^{1}}(x^{1},x^1)\cdots W^{s^{j}}(x^{j},x^j)\cdots W^{s^{k}}(x^{k},x^k) = 0.$$ Thus, we
    conclude that $\hat{W} \in \cW^{\mu^{\otimes k}}_{f^{\otimes k}}$.
    \par
    We finish the proof by computing the fidelity as follows.
    \begin{align*}
        \cF(\Psi^{\otimes k} \circ \kbr{\mu}^{\otimes k}, \hat{W}\circ \kbr{\mu^{\otimes k}}) &=
        \cF(\Psi^{\otimes k} \circ \kbr{\mu}^{\otimes k}, W^{\otimes k}\circ \kbr{\mu}^{\otimes k})\\
        &=  \cF(\Psi\circ \kbr{\mu}, W\circ \kbr{\mu})^{k} \\
        &\geq (\sqrt{1-\delta})^{k} = \sqrt{1-\eps}.
        \qedhere
    \end{align*}
\end{proof}
\paragraph{Inner product}
\begin{claim}\label{claim:direct-sum-upper-inner-product} There exists a map $\cE^{\downarrow} :
    \Herm(\Sigma^{k}) \to \Herm(\Sigma)$ such that for all $\Gamma \in \Herm(\Sigma^{k})$,
    \begin{align}
        \langle \Gamma, (J^{\otimes k}-\Psi^{\otimes k})\circ \kbr{\mu}^{\otimes k}\rangle =
        \langle \cE^{\downarrow}(\Gamma), (J - \Psi)\circ \kbr{\mu}\rangle,
        \label{eq:direct-sum-upper-inner-product}
    \end{align}
    and $\max_j\|\cE^{\downarrow}(\Gamma)\circ \Delta_j\| \leq k\max_{\ell,j}\|\Gamma \circ
    \Delta_{\ell j}\|$.
\end{claim}
\begin{proof}
    Let $\rho_\ell := (\Psi\circ \kbr{\mu})^{\otimes (\ell-1)}\otimes (J\circ
    \kbr{\mu})^{\otimes(k-\ell)}$, and define the map $\cE^{\uparrow}(X) := \sum_{\ell \in
    [k]}\rho_\ell\otimes_{\ell} X$. Both $\Psi\circ\kbr{\mu}$ and $J\circ\kbr{\mu}$ are density
    operators, and hence, so is $\rho_\ell$. Then, by the telescoping identity,
    \begin{align*}
        \cE^{\uparrow}((J-\Psi)\circ \kbr{\mu})
        &=\sum_{\ell=1}^{k}(\Psi\circ \kbr{\mu})^{\otimes (\ell-1)}\otimes
        ((J-\Psi)\circ \kbr{\mu})\otimes
        (J\circ \kbr{\mu})^{\otimes(k-\ell)} \\
        &= (J\circ \kbr{\mu})^{\otimes k} - (\Psi \circ \kbr{\mu})^{\otimes k} \\
        &= (J^{\otimes k}- \Psi^{\otimes k})\circ \kbr{\mu}^{\otimes k}.
    \end{align*}
    Using the fact that adjoint and sum commute, and the adjoint of tensor product is a partial
    trace on the corresponding legs, the adjoint of $\cE^{\uparrow}$ can be computed as
    \begin{align*}
        \cE^{\downarrow}(Y) := (\cE^{\uparrow})^{\dagger}(Y)
        = \sum_{\ell=1}^{k} \Tr_{\bar{\ell}}[(\rho_\ell\otimes_{\ell} I_{\ell})Y],
    \end{align*}
    which maps Hermitian matrices to Hermitian matrices. As a result, we can compute
    \begin{align*}
        \langle \Gamma, (J^{\otimes k}-\Psi^{\otimes k})\circ \kbr{\mu}^{\otimes k}\rangle
        &=\langle \Gamma, \cE^{\uparrow}((J-\Psi)\circ \kbr{\mu})\rangle\\
        &=\langle \cE^{\downarrow}(\Gamma), (J-\Psi)\circ \kbr{\mu}\rangle.
    \end{align*}
    It only remains to show the bound on the filtered norms. For every $\ell\in[k]$, we have
    \begin{align*}
        \Tr_{\bar{\ell}}[(\rho_\ell\otimes_{\ell} I_{\ell})\Gamma]\circ\Delta_j
        = &\Tr_{\bar{\ell}}[((\rho_\ell\otimes_{\ell} I_{\ell})\Gamma)
        \circ(J^{\otimes(k-1)}\otimes_{\ell}\Delta_j)]\\
        = &\Tr_{\bar{\ell}}[(\rho_\ell\otimes_{\ell} I_{\ell})
        (\Gamma\circ(J^{\otimes(k-1)}\otimes_{\ell}\Delta_j))]\\
        = &\Tr_{\bar{\ell}}[(\rho_\ell\otimes_{\ell} I_{\ell})
        (\Gamma\circ\Delta_{\ell j})].
    \end{align*}
    The map $X\mapsto\rho_\ell\otimes_{\ell}X$ is a quantum channel, whose adjoint is
    $Y\mapsto\Tr_{\bar{\ell}}[(\rho_\ell\otimes_{\ell}I_{\ell})Y]$. Therefore, by
    \Cref{lem:channel-norm-contraction},
    \begin{align*}
        \|\cE^{\downarrow}(\Gamma) \circ \Delta_{j}\|
        &\leq\sum_{\ell=1}^{k}\|\Tr_{\bar{\ell}}[(\rho_\ell\otimes_{\ell}I_{\ell})
        (\Gamma\circ\Delta_{\ell j})]\|\\
        &\leq\sum_{\ell=1}^{k}\|\Gamma\circ\Delta_{\ell j}\|
        \leq k\max_{\ell}\|\Gamma\circ\Delta_{\ell j}\|.
    \end{align*}
    Taking the maximum over $j$ proves the claim.
\end{proof}
\paragraph{Concluding the upper bound.}
By \Cref{claim:direct-sum-upper-inner-product}, $\max_j\|\cE^{\downarrow}(\Gamma)\circ\Delta_j\|\leq
k$ whenever $\max_{\ell,j}\|\Gamma\circ\Delta_{\ell j}\|\leq1$. Thus, by
\Cref{lem:adv-universal-hardness-cert}, we have
\begin{align*}
    \mathrm{Adv}^{\mu^{\otimes k}}_{\eps}(f^{\otimes k})
    &=\progline{\max_{\Gamma=\Gamma^\dagger}
    \min_{\Lambda\in\cG_{\eps}^{\mu^{\otimes k}}(f^{\otimes k})}}
    {\langle \Gamma,(J^{\otimes k}-\Lambda)\circ\kbr{\mu}^{\otimes k}\rangle}
    {\max_{\ell,j}\|\Gamma\circ\Delta_{\ell j}\|\leq1}\\
    &\leq\progline{\max_{\Gamma=\Gamma^\dagger}}
    {\langle \Gamma,(J^{\otimes k}-\Psi^{\otimes k})\circ\kbr{\mu}^{\otimes k}\rangle}
    {\max_{\ell,j}\|\Gamma\circ\Delta_{\ell j}\|\leq1}\\
    &=\progline{\max_{\Gamma=\Gamma^\dagger}}
    {\langle \cE^{\downarrow}(\Gamma),(J-\Psi)\circ\kbr{\mu}\rangle}
    {\max_{\ell,j}\|\Gamma\circ\Delta_{\ell j}\|\leq1}\\
    &\leq\progline{\max_{\Theta=\Theta^\dagger}}
    {\langle \Theta,(J-\Psi)\circ\kbr{\mu}\rangle}
    {\max_j\|\Theta\circ\Delta_j\|\leq k}\\
    &=k\progline{\max_{\Theta=\Theta^\dagger}}
    {\langle \Theta,(J-\Psi)\circ\kbr{\mu}\rangle}
    {\max_j\|\Theta\circ\Delta_j\|\leq1}\\
    &=k\cdot\mathrm{Adv}^{\mu}_{\delta}(f).
\end{align*}
Here, the second line follows from \Cref{claim:direct-sum-upper-target}, the third line follows from
\Cref{claim:direct-sum-upper-inner-product}, and the fourth line uses the bound above. The fifth
line follows by replacing $\Theta$ with $k\Theta$, and the last equality follows from our choice of
$\Psi$.
\end{proof}

\section{Proofs of basic properties}\label{app:basic-properties}

\subsection{\texorpdfstring{Proof of \Cref{prop:avggamma2-properties}: properties of
$\bar{\gamma}_2$}{Proof of the distributional gamma 2 properties}}\label{app:avggamma2-properties}

\begin{proof}
We prove each property as follows.
        \begin{enumerate}[i.]
            \item For the triangle inequality, we have
            \begin{align*}
                \bar{\gamma}_2^{\mu}(A+B|\{ Z_i\})
                &= \max_{\substack{\Gamma\\\max_i\|\Gamma\circ Z_i\|\leq1}}
                |\langle \Gamma, (A+B)\circ\kbr{\mu}\rangle| \\
                &\leq \max_{\substack{\Gamma\\\max_i\|\Gamma\circ Z_i\|\leq1}}
                \bigl(|\langle \Gamma,A\circ\kbr{\mu}\rangle|
                +|\langle \Gamma,B\circ\kbr{\mu}\rangle|\bigr) \\
                &\leq \bar{\gamma}_2^{\mu}(A|\{ Z_i\}) + \bar{\gamma}_2^{\mu}(B|\{ Z_i\}).
            \end{align*}
            \item For $\lambda > 0$, we have
            \[
            \bar{\gamma}_2^{\mu}(\lambda A|\{ Z_i\})
            =\max_{\substack{\Gamma\\\max_i\|\Gamma\circ Z_i\|\leq1}}
            |\langle \Gamma,\lambda A\circ\kbr{\mu}\rangle|
            =|\lambda|\bar{\gamma}_2^{\mu}(A|\{ Z_i\}).
            \]
            \item Since $\gamma_2(\overline{B})=\gamma_2(B)$, we obtain
            \begin{alignat*}{2}
            \bar{\gamma}_2^{\mu}(A\circ B|\{Z_i\}) &=\max_{\Gamma}\;
            |\langle\Gamma\circ\overline{B},A\circ\kbr{\mu}\rangle| &&\quad\mathrm{s.t.}\quad
            \max_i\|\Gamma\circ Z_i\|\leq1 \\
            &\leq\max_{\Gamma}\; |\langle\Gamma,A\circ\kbr{\mu}\rangle| &&\quad\mathrm{s.t.}\quad
            \max_i\|\Gamma\circ Z_i\|\leq\gamma_2(B) \\
            &=\gamma_2(B)\bar{\gamma}_2^{\mu}(A|\{Z_i\}). \intertext{Here, the inequality follows
            from $\|\Gamma\circ\overline{B}\circ Z_i\| \leq\gamma_2(B)\|\Gamma\circ Z_i\|$, and the
            last equality follows by homogeneity. Similarly,}
            \gamma_2(B)\bar{\gamma}_2^{\mu}(A|\{Z_i\circ B\}) &=\max_{\Gamma}\;
            |\langle\Gamma,A\circ\kbr{\mu}\rangle| &&\quad\mathrm{s.t.}\quad \max_i\|\Gamma\circ
            Z_i\circ B\|\leq\gamma_2(B) \\
            &\geq\max_{\Gamma}\; |\langle\Gamma,A\circ\kbr{\mu}\rangle| &&\quad\mathrm{s.t.}\quad
            \max_i\|\Gamma\circ Z_i\|\leq1 \\
            &=\bar{\gamma}_2^{\mu}(A|\{Z_i\}).
            \intertext{If $B$ is real, then}
            \bar{\gamma}_2^{\mu}(A\circ B|\{Z_i\circ B\}) &=\max_{\Gamma}\; |\langle\Gamma\circ
            B,A\circ\kbr{\mu}\rangle| &&\quad\mathrm{s.t.}\quad \max_i\|\Gamma\circ Z_i\circ
            B\|\leq1 \\
            &\leq\bar{\gamma}_2^{\mu}(A|\{Z_i\}),
            \end{alignat*}
            since $\Gamma\circ B$ is feasible for the maximization on the right.
            \item Since $\|\Gamma\circ Z_i\|\leq\|\Gamma\|\gamma_2(Z_i)$ for every $i$, we obtain
            \begin{alignat*}{2}
            \bar{\gamma}_2^{\mu}(A|\{Z_i\}) &=\max_{\Gamma}\; |\langle\Gamma,A\circ\kbr{\mu}\rangle|
            &&\quad\mathrm{s.t.}\quad \max_i\|\Gamma\circ Z_i\|\leq1 \\
            &\geq\max_{\Gamma}\; |\langle\Gamma,A\circ\kbr{\mu}\rangle| &&\quad\mathrm{s.t.}\quad
            \bigl(\max_i\gamma_2(Z_i)\bigr)\|\Gamma\|\leq1 \\
            &=\frac{1}{\max_i\gamma_2(Z_i)} \max_{\Gamma}\; |\langle\Gamma,A\circ\kbr{\mu}\rangle|
            &&\quad\mathrm{s.t.}\quad \|\Gamma\|\leq1 \\
            &=\frac{\|A\circ\kbr{\mu}\|_{\tr}}{\max_i\gamma_2(Z_i)}.
            \end{alignat*}
            The last equality is the duality between the operator norm and the trace norm.
            \item By H\"older's inequality, we have
            \begin{alignat*}{2}
            \bar{\gamma}_2^{\mu}(A\circ B|\{Z_i\}) &=\max_{\Gamma}\;
            |\langle\Gamma\circ\overline{B},A\circ\kbr{\mu}\rangle| &&\quad\mathrm{s.t.}\quad
            \max_i\|\Gamma\circ Z_i\|\leq1 \\
            &\leq\max_{\Gamma}\; \|\Gamma\circ\overline{B}\|\,\|A\circ\kbr{\mu}\|_{\tr}
            &&\quad\mathrm{s.t.}\quad \max_i\|\Gamma\circ Z_i\|\leq1 \\
            &=\|A\circ\kbr{\mu}\|_{\tr} \max_{\Gamma}\; \|\Gamma\circ\overline{B}\|
            &&\quad\mathrm{s.t.}\quad \max_i\|\Gamma\circ Z_i\|\leq1 \\
            &=\|A\circ\kbr{\mu}\|_{\tr}\gamma_2(B|\{Z_i\}),
            \end{alignat*}
            The last equality follows from \Cref{eq:filtered-gamma2-dual} and
            $\gamma_2(\overline{B}|\{Z_i\})=\gamma_2(B|\{Z_i\})$, since the filters are real.
            \item Since $\supp(\mu)\subseteq S$, we have $(A\circ
            J_S)\circ\kbr{\mu}=A\circ\kbr{\mu}$. Thus, the objective is unchanged, proving the first
            equality. For the second equality, we prove each direction separately. Let $\Gamma$ be
            feasible on $\Sigma$. Then, $\|(\Gamma|_S)\circ(Z_i|_S)\|\leq\|\Gamma\circ Z_i\|\leq1$
            for every $i$, so $\Gamma|_S$ is feasible on $S$. The objective is unchanged, as
            \[
                \langle\Gamma,A\circ J_S\circ\kbr{\mu}\rangle
                =\langle\Gamma|_S,(A|_S)\circ\kbr{\mu|_S}\rangle.
            \]
            Taking suprema gives
            \[
                \bar{\gamma}_2^\mu(A\circ J_S\mid\{Z_i\})
                \leq\bar{\gamma}_2^{\mu|_S}(A|_S\mid\{Z_i|_S\}).
            \]
            Conversely, extending any feasible matrix on $S$ by zero preserves the objective and the
            filtered operator norms. Taking suprema gives
            \[
                \bar{\gamma}_2^{\mu|_S}(A|_S\mid\{Z_i|_S\})
                \leq\bar{\gamma}_2^\mu(A\circ J_S\mid\{Z_i\}),
            \]
           which completes the proof. \qedhere
        \end{enumerate}
\end{proof}

\subsection{\texorpdfstring{Proof of \Cref{lem:G-properties}: identities for the composition
matrix}{Proof of identities for the composition
matrix}}\label{app:G-properties}\label{app:composition-auxiliary}

\begin{proof}
    For (\ref{item:G-entries}), let $\bx,\by\in\Sigma_g^n$. We have
    \begin{align*}
        (\bG^\dagger X\bG)(\bx,\by)
        &= \bra{\bx}\bG^\dagger X\bG\ket{\by}= \bra{\bg(\bx)}X\ket{\bg(\by)}
        = X(\bg(\bx),\bg(\by)).
    \end{align*}
    For (\ref{item:G-hadamard-lift}), applying this entrywise formula gives
    \begin{align*}
        (\bG^\dagger(X_1\circ X_2)\bG)(\bx,\by)
        &= (X_1\circ X_2)(\bg(\bx),\bg(\by))
        = X_1(\bg(\bx),\bg(\by))X_2(\bg(\bx),\bg(\by))\\
        &= (\bG^\dagger X_1\bG)(\bx,\by)\cdot
           (\bG^\dagger X_2\bG)(\bx,\by).
    \end{align*}

    For (\ref{item:G-hadamard-pull-through}), fix $\tx,\ty\in\{0,1\}^n$. Expanding the matrix
    product, we obtain
    \begin{align*}
        \bigl(\bG(Y\circ(\bG^\dagger X\bG))\bG^\dagger\bigr)(\tx,\ty)
        &=
        \bra{\tx}
        \left(\sum_{\bx\in\Sigma_g^n}\ket{\bg(\bx)}\bra{\bx}\right)
        (Y\circ(\bG^\dagger X\bG))
        \left(\sum_{\by\in\Sigma_g^n}\ket{\by}\bra{\bg(\by)}\right)
        \ket{\ty}\\
        &=
        \sum_{\bx,\by\in\Sigma_g^n}
        \braket{\tx}{\bg(\bx)}
        \bra{\bx}Y\ket{\by}
        \bra{\bx}\bG^\dagger X\bG\ket{\by}
        \braket{\bg(\by)}{\ty}\\
        &=
        \sum_{\bx,\by\in\Sigma_g^n}
        \braket{\tx}{\bg(\bx)}
        \bra{\bx}Y\ket{\by}
        \bra{\bg(\bx)}X\ket{\bg(\by)}
        \braket{\bg(\by)}{\ty}\\
        &=
        \sum_{\substack{\bx,\by\in\Sigma_g^n\\
                        \bg(\bx)=\tx,\ \bg(\by)=\ty}}
        \bra{\bx}Y\ket{\by}
        \bra{\bg(\bx)}X\ket{\bg(\by)}\\
        &=
        \bra{\tx}X\ket{\ty}\cdot
        \sum_{\substack{\bx,\by\in\Sigma_g^n\\
                        \bg(\bx)=\tx,\ \bg(\by)=\ty}}
        \bra{\bx}Y\ket{\by}\\
        &=
        \bra{\tx}X\ket{\ty}\cdot
        \sum_{\bx,\by\in\Sigma_g^n}
        \braket{\tx}{\bg(\bx)}
        \bra{\bx}Y\ket{\by}
        \braket{\bg(\by)}{\ty}\\
        &=
        \bra{\tx}X\ket{\ty}\cdot
        \bra{\tx}\bG Y\bG^\dagger\ket{\ty}\\
        &=
        X(\tx,\ty)\cdot(\bG Y\bG^\dagger)(\tx,\ty). \qedhere
        \end{align*}
\end{proof}

\subsection{\texorpdfstring{Proof of \Cref{prop:filtered-fidelity-properties}: properties of
filtered fidelity}{Proof of the filtered-fidelity properties
proposition}}\label{app:filtered-fidelity-properties}

\begin{proof}
    We prove each property as follows.
    \begin{enumerate}[i.]
        \item Let $a,b>0$. Then,
        \begin{align*}
            \cF_c(aA,bB|\{P_i\})
            &=\progat{6cm}{\inf_{\Gamma\succ0}}
            {\sqrt{\langle\Gamma,aA\rangle\langle\Gamma^{-1},bB\rangle}}
            {\Gamma\circ P_i\preceq c\Gamma\quad\forall i}\\
            &=\progat{6cm}{\inf_{\Gamma\succ0}}
            {\sqrt{ab}\sqrt{\langle\Gamma,A\rangle\langle\Gamma^{-1},B\rangle}}
            {\Gamma\circ P_i\preceq c\Gamma\quad\forall i}\\
            &=\sqrt{ab}\cF_c(A,B|\{P_i\}).
        \end{align*}
        \item We prove each direction separately. First,
        \begin{align*}
            &\ \cF_c(A_1,B_1\mid\{P_i\})
            \cF_c(A_2,B_2\mid\{Q_j\})\\
            &=\progat{8.5cm}{\inf_{\Gamma_1,\Gamma_2\succ0}}
            {\sqrt{
                \langle\Gamma_1,A_1\rangle
                \langle\Gamma_1^{-1},B_1\rangle
                \langle\Gamma_2,A_2\rangle
                \langle\Gamma_2^{-1},B_2\rangle
            }}
            {\Gamma_1\circ P_i\preceq c\Gamma_1\quad\forall i\\
             \Gamma_2\circ Q_j\preceq c\Gamma_2\quad\forall j}\\
            &=\progat{8.5cm}{\inf_{\Gamma_1,\Gamma_2\succ0}}
            {\sqrt{
                \langle\Gamma_1\otimes\Gamma_2,A_1\otimes A_2\rangle
                \langle(\Gamma_1\otimes\Gamma_2)^{-1},B_1\otimes B_2\rangle
            }}
            {(\Gamma_1\otimes\Gamma_2)\circ(P_i\otimes J)
             \preceq c(\Gamma_1\otimes\Gamma_2)\quad\forall i\\
             (\Gamma_1\otimes\Gamma_2)\circ(J\otimes Q_j)
             \preceq c(\Gamma_1\otimes\Gamma_2)\quad\forall j}\\
            &\geq\cF_c\!\left(
                A_1\otimes A_2,B_1\otimes B_2
                \mid\{P_i\otimes J\}_i\cup\{J\otimes Q_j\}_j
            \right).
        \end{align*}
        The last inequality follows because $\Gamma_1\otimes\Gamma_2$ is feasible for the combined
        filter family. For the other direction, we use the fidelity formulation in
        \Cref{prop:equiv-filtered-fidelity}. Let $X_i,C_1\succeq0$  be such that
        $\cF_c(A_1,B_1|\{P_i\})=\cF(C_1,B_1)$ and $A_1-C_1=\sum_iX_i\circ(cJ-\overline{P_i})$.
        Similarly, we get $Y_j,C_2\succeq0$ such that $\cF_c(A_2,B_2|\{Q_j\})=\cF(C_2,B_2)$ and
        $A_2-C_2=\sum_jY_j\circ(cJ-\overline{Q_j})$. Then, define $C:=C_1\otimes C_2$,
        $X'_i:=X_i\otimes A_2$, and $Y'_j:=C_1\otimes Y_j$, which are all positive semidefinite. We
        have
        \begin{align*}
            A_1\otimes A_2-C
            &=(A_1-C_1)\otimes A_2+C_1\otimes(A_2-C_2)\\
            &=\sum_i\bigl(X_i\circ(cJ-\overline{P_i})\bigr)\otimes A_2
              +C_1\otimes\sum_j\bigl(Y_j\circ(cJ-\overline{Q_j})\bigr)\\
            &=\sum_iX'_i\circ(cJ\otimes J-\overline{P_i}\otimes J)
              +\sum_jY'_j\circ(cJ\otimes J-J\otimes\overline{Q_j}).
        \end{align*}
        Consequently,
        \begin{align*}
            &\ \cF_c(A_1\otimes A_2,B_1\otimes B_2|\{P_i\otimes J\}_i\cup\{J\otimes Q_j\}_j)\\
            &\geq\cF(C_1\otimes C_2,B_1\otimes B_2)\\
            &=\cF(C_1,B_1)\cF(C_2,B_2)\\
            &=\cF_c(A_1,B_1|\{P_i\})\cF_c(A_2,B_2|\{Q_j\}),
        \end{align*}
        which completes the proof.
 \item
 We claim that it suffices to show $\cF_{\sqrt{c_0c_1}}(A,B|\{P_i\}) \leq
 \sqrt{\cF_{c_0}(A,B|\{P_i\})\cF_{c_1}(A,B|\{P_i\})}$. First, we can extend the mid-point inequality
 to $\cF_{c_0^{1-\lambda}c_1^\lambda}(A,B|\{P_i\}) \leq \cF_{c_0}(A,B|\{P_i\})^{1-\lambda}
 \cF_{c_1}(A,B|\{P_i\})^\lambda$ for any dyadic number $\lambda=k/2^n$ by induction on $n$: if $k$
 is even, we are done as $\lambda=(k/2)/2^{n-1}$. Otherwise, we can write
 \begin{align*}
     \lambda=\frac12\left(\frac{(k-1)/2}{2^{n-1}}+\frac{(k+1)/2}{2^{n-1}}\right).
 \end{align*}
 Then, by the midpoint inequality and the induction hypothesis, the claim extends to
 $\lambda=k/2^n$.
Next, we can extend the result to any $\lambda \in [0,1]$ as dyadic numbers are dense and $\cF_{c}$
is monotone with respect to $c$.
\par
It only remains to prove the claim. Let $\Gamma_0$ and $\Gamma_1$ be arbitrary feasible solutions
for $\cF_{c_0}(A,B|\{P_i\})$ and $\cF_{c_1}(A,B|\{P_i\})$, respectively.
            \par
             The idea is to construct a geometric average of $\Gamma_0$ and $\Gamma_1$. Let
            $\Gamma_1^{1/2} \Gamma_0^{-1/2} = VDW^{\dagger}$ be the singular-value decomposition,
            where $D\succ 0$ is diagonal and $V,W$ are unitary. Defining $U =WV^{\dagger}$, we get
            \begin{align*}
                \Gamma:= \Gamma_0^{1/2}U\Gamma_1^{1/2} = \Gamma_0^{1/2} WV^{\dagger} VDW^{\dagger} \Gamma_0^{1/2}
                 = \Gamma_0^{1/2} WDW^{\dagger} \Gamma_0^{1/2} \succ 0.
            \end{align*}
            We first prove that $\Gamma$ is a feasible matrix for $\cF_{\sqrt{c_0c_1}}$. Since
            $\Gamma= \Gamma^{\dagger}$, we have
            \begin{align*}
                \Gamma\Gamma_0^{-1}\Gamma
                &=\Gamma_1^{1/2}U^\dagger
                \Gamma_0^{1/2}\Gamma_0^{-1}\Gamma_0^{1/2}
                U\Gamma_1^{1/2}
                =\Gamma_1^{1/2}U^\dagger U\Gamma_1^{1/2}
                =\Gamma_1.
            \end{align*}
        As a result,
        \begin{align*}
            &\phantom{\iff}~\Gamma_1\succeq\Gamma\Gamma_0^{-1}\Gamma\\
            &\iff
            \begin{bmatrix}\Gamma_0&\Gamma\\\Gamma&\Gamma_1\end{bmatrix}\succeq0
            &&\text{(By the Schur complement)}\\
            &\Rightarrow
            \begin{bmatrix}\Gamma_0&\Gamma\\\Gamma&\Gamma_1\end{bmatrix}
            \circ
            \begin{bmatrix}P_i&P_i\\P_i&P_i\end{bmatrix}
            =
            \begin{bmatrix}\Gamma_0\circ P_i&\Gamma\circ P_i\\\Gamma\circ P_i&\Gamma_1\circ P_i\end{bmatrix}
            \succeq0
            &&\text{(By the Schur product theorem)}\\
            &\Rightarrow
            \begin{bmatrix}c_0\Gamma_0&\Gamma\circ P_i\\\Gamma\circ P_i&c_1\Gamma_1\end{bmatrix}\succeq0
            &&\text{(By choice of }\Gamma_0,\Gamma_1\text{)}\\
            &\iff c_1\Gamma_1\succeq
            \frac1{c_0}(\Gamma\circ P_i)\Gamma_0^{-1}(\Gamma\circ P_i)
            &&\text{(By the Schur complement)}\\
            &\iff c_0c_1\Gamma_1\succeq
            (\Gamma\circ P_i)\Gamma_0^{-1}(\Gamma\circ P_i)\\
            &\iff c_0c_1\Gamma\Gamma_0^{-1}\Gamma\succeq
            (\Gamma\circ P_i)\Gamma_0^{-1}(\Gamma\circ P_i)\\
            &\iff c_0c_1\Gamma_0^{-1/2}\Gamma\Gamma_0^{-1}\Gamma\Gamma_0^{-1/2}
            \succeq
            \Gamma_0^{-1/2}(\Gamma\circ P_i)\Gamma_0^{-1}(\Gamma\circ P_i)\Gamma_0^{-1/2}\\
            &\iff c_0c_1\bigl(\Gamma_0^{-1/2}\Gamma\Gamma_0^{-1/2}\bigr)^2
            \succeq
            \bigl(\Gamma_0^{-1/2}(\Gamma\circ P_i)\Gamma_0^{-1/2}\bigr)^2\\
            &\Rightarrow
            \sqrt{c_0c_1}\,\Gamma_0^{-1/2}\Gamma\Gamma_0^{-1/2}
            \succeq
            \Gamma_0^{-1/2}(\Gamma\circ P_i)\Gamma_0^{-1/2}
            &&\text{(By monotonicity of the square root)}\\
            &\iff \sqrt{c_0c_1}\,\Gamma\succeq\Gamma\circ P_i.
        \end{align*}
        Therefore, we can bound the filtered fidelity as
        \begin{align*}
            \cF_{\sqrt{c_0c_1}}(A,B|\{P_i\})
            \leq \sqrt{\langle\Gamma,A\rangle\langle\Gamma^{-1},B\rangle}.
        \end{align*}
        We bound each factor on the right hand side as follows. Let $X_b:=\Gamma_b^{1/2}A^{1/2}$ for
        $b\in\{0,1\}$. Then,
            \begin{align*}
                \langle \Gamma,A\rangle &= \Tr[\Gamma_0^{1/2}U\Gamma_1^{1/2}A] = \Tr[A^{1/2}\Gamma_0^{1/2}U\Gamma_1^{1/2}A^{1/2}] = \langle X_0,UX_1\rangle \\
                & \leq \|X_0\|_2\|UX_1\|_2 = \|X_0\|_2\|X_1\|_2 = \sqrt{\Tr[X_0^\dagger X_0]\Tr[X_1^\dagger X_1]} \\
                &=\sqrt{\Tr[\Gamma_0^{1/2}A\Gamma_0^{1/2}]\Tr[\Gamma_1^{1/2}A\Gamma_1^{1/2}]} \\
                & = \sqrt{\langle \Gamma_0,A\rangle\langle \Gamma_1,A\rangle}.
            \end{align*}
            Similarly, we can prove
            \begin{align*}
                \langle \Gamma^{-1},B\rangle &\leq \sqrt{\langle \Gamma_0^{-1},B\rangle\langle \Gamma_1^{-1},B\rangle}.
            \end{align*}
            Combining these estimates gives
            \begin{align*}
                \cF_{\sqrt{c_0c_1}}(A,B|\{P_i\})
                &\leq \sqrt{\langle\Gamma,A\rangle\langle\Gamma^{-1},B\rangle}\\
                &\leq \sqrt{
                    \sqrt{\langle\Gamma_0,A\rangle\langle\Gamma_0^{-1},B\rangle}
                    \sqrt{\langle\Gamma_1,A\rangle\langle\Gamma_1^{-1},B\rangle}
                }.
            \end{align*}
            Since $\Gamma_0$ and $\Gamma_1$ were arbitrary feasible solutions, taking infima over
            them completes the proof of the claim.
        \item Let $\lambda\in[0,1]$, and let $A_1,A_2,B_1,B_2\in\Pos(\Sigma)$ be given. By the
        additive formulation of $\cF_c$ in \Cref{prop:equiv-filtered-fidelity}, we have
        \begin{align*}
            &\ \cF_c(\lambda A_1+(1-\lambda)A_2,\lambda B_1+(1-\lambda)B_2|\{P_i\})\\
            &=\inf_{\substack{\Gamma\succ0\\\Gamma\circ P_i\preceq c\Gamma\ \forall i}}
            \frac{\langle\Gamma,\lambda A_1+(1-\lambda)A_2\rangle
            +\langle\Gamma^{-1},\lambda B_1+(1-\lambda)B_2\rangle}{2}\\
            &\geq\lambda\inf_{\substack{\Gamma\succ0\\\Gamma\circ P_i\preceq c\Gamma\ \forall i}}
            \frac{\langle\Gamma,A_1\rangle+\langle\Gamma^{-1},B_1\rangle}{2}
            +(1-\lambda)\inf_{\substack{\Gamma\succ0\\\Gamma\circ P_i\preceq c\Gamma\ \forall i}}
            \frac{\langle\Gamma,A_2\rangle+\langle\Gamma^{-1},B_2\rangle}{2}\\
            &=\lambda\cF_c(A_1,B_1|\{P_i\})+(1-\lambda)\cF_c(A_2,B_2|\{P_i\}).
        \end{align*}
        \item Let $\cE$ be the given channel. By the fidelity formulation in
        \Cref{prop:equiv-filtered-fidelity}, there exist $X_i,C\succeq0$ such that
        $\cF_c(A,B|\{P_i\})=\cF(C,B)$ and $A-C=\sum_iX_i\circ(cJ-\overline{P_i})$. By data
        processing, we have $\cF(\cE(C),\cE(B))\geq\cF(C,B)$. Moreover,
        \begin{align*}
            \cE(A)-\cE(C)
            &=\sum_i\bigl(c\cE(X_i)-\cE(X_i\circ\overline{P_i})\bigr)\\
            &=\sum_i\cE(X_i)\circ(cJ-\overline{P_i}).
        \end{align*}
        Here we used linearity of $\cE$ as well as $\cE(X\circ\overline{P_i}) =
        \cE(X)\circ\overline{P_i}$. Thus, $\cE(C)$ and $\cE(X_i)$ form a feasible solution.
        Consequently,
        \begin{align*}
            \cF_c(\cE(A),\cE(B)|\{P_i\})
            &\geq\cF(\cE(C),\cE(B))
            \geq\cF(C,B)=\cF_c(A,B|\{P_i\}).
        \end{align*}

        \item We prove that the map $(A,B,c) \to \cF_{c}(A,B|\{P_i\})$ is both jointly upper
        semicontinuous and lower semicontinuous, which gives the joint continuity. So, let $A_k\to
        A$, $B_k\to B$, and $c_k \to c$ be given.
        \paragraph{Upper semicontinuity.}
        We want to show $ \limsup_{k} \cF_{c_k}(A_k,B_k|\{P_i\}) \leq \cF_{c}(A,B|\{P_i\})$. For
        simplicity, we denote $F(A,B,\Gamma) :=
        \sqrt{\langle\Gamma,A\rangle\langle\Gamma^{-1},B\rangle}$, which is jointly continuous in
        $A,B,\Gamma$. Let $\Gamma \succ 0$ be a feasible solution with $\Gamma \circ P_i \preceq c
        \Gamma$. For $\eps >0$, we define $\Gamma_{\eps} := \Gamma + \eps I$. Then, for every $i$,
        we have
        \begin{align*}
            \Gamma_{\eps} \circ P_i \preceq \Gamma \circ P_i + \eps I\circ P_i
            \preceq c \Gamma + \eps I
            = c(\Gamma + \eps I) - (c-1) \eps I =
            c\Gamma_{\eps} - (c-1)\eps I.
        \end{align*}
        Since $c_k\to c$, for large enough $k$, this yields $ \Gamma_{\eps} \circ P_i \preceq
            c_k\Gamma_{\eps}$. As a result, $\Gamma_{\eps}$ is a feasible solution for such $k$ and
            we get
        \begin{align*}
            \limsup_{k} \cF_{c_k}(A_k,B_k|\{P_i\}) \leq F(A,B,\Gamma_{\eps}).
        \end{align*}
        Since $\eps>0$ was arbitrary, taking $\eps \to 0$ and using continuity of $F$ gives $
        \limsup_{k} \cF_{c_k}(A_k,B_k|\{P_i\})  \leq F(A,B,\Gamma)$. Since $\Gamma$ was also an
        arbitrary solution, the inequality follows by taking the infimum of the right hand side over
        $\Gamma$.
        \paragraph{Lower semicontinuity.}
        We use the following claim.
        \begin{claim}
            Let $c>1$, $X,Y,D\succeq0$.
            \[
                \cF_c(X+D,Y)\leq \cF_c(X,Y)
                + \sqrt{\Tr[D]\Tr[Y]}.
            \]
        \end{claim}
        \begin{proof}
            Let $\Gamma$ be a feasible solution for $\cF_{c}$. We claim that $\Gamma_{\eps} :=
            (\Gamma^{-1} + \eps I)^{-1}$ is also feasible. Indeed, for any matrix $M\succ 0$, the
            Schur complement and the Schur product theorem give
            \begin{align*}
                \begin{bmatrix}
                    M & I \\
                    I & M^{-1}
                \end{bmatrix} \succeq 0
                \Rightarrow
                \begin{bmatrix}
                    M\circ P_i & I \\
                    I & M^{-1}\circ P_i
                \end{bmatrix} \succeq 0
                \Rightarrow
                M^{-1} \circ P_i \succeq (M \circ P_i)^{-1}.
            \end{align*}
            Then, we can write
            \begin{align*}
                \Gamma_\eps\circ P_i
                &= (\Gamma^{-1}+\eps I)^{-1}\circ P_i\\
                &= \frac{1}{\eps}\bigl(I-(I+\eps\Gamma)^{-1}\bigr)\circ P_i
                &&\left(\text{By~}\frac{1}{x^{-1}+\eps}
                   =\frac{1}{\eps}\left(1-\frac{1}{1+\eps x}\right)\right)\\
                &\preceq \frac{1}{\eps}\bigl(I-(I+\eps(\Gamma\circ P_i))^{-1}\bigr)
                &&\text{(By the above)}\\
                &\preceq \frac{1}{\eps}\bigl(I-(I+\eps c\Gamma)^{-1}\bigr)
                &&\left(\text{By~}\Gamma\circ P_i\preceq c\Gamma\right)\\
                &= c\Gamma(I+\eps c\Gamma)^{-1}
                &&\left(\text{By~}
                   1-\frac{1}{1+x}
                   =\frac{x}{1+x}\right)\\
                &\preceq c\Gamma(I+\eps\Gamma)^{-1}
                &&\left(\text{By~}c > 1\right)\\
                &= c\Gamma_\eps
                &&\left(\text{By~}\frac{x}{1+\eps x}
                   =\frac{1}{x^{-1}+\eps}\right).
            \end{align*}
            Therefore, we can compute
            \begin{align*}
                \cF_c(X+D,Y)
                &\leq \frac{1}{2}\bigl(
                    \langle\Gamma_{\eps},X+D\rangle
                    + \langle\Gamma_{\eps}^{-1},Y\rangle
                \bigr)\\
                &= \frac{1}{2}\bigl(
                    \langle\Gamma_{\eps},X\rangle
                    + \langle\Gamma_{\eps},D\rangle
                    + \langle\Gamma^{-1}+\eps I,Y\rangle
                \bigr)\\
                &\leq \frac{1}{2}\bigl(
                    \langle\Gamma,X\rangle
                    + \langle\frac{1}{\eps}I,D\rangle
                    + \langle\Gamma^{-1}+\eps I,Y\rangle
                \bigr)
                &&\left(\text{By }
                    \Gamma_{\eps}\preceq\Gamma,\
                    \Gamma_{\eps}\preceq\frac{1}{\eps}I
                \right)\\
                &= \frac{1}{2}\bigl(
                    \langle\Gamma,X\rangle
                    + \langle\Gamma^{-1},Y\rangle
                \bigr)
                + \frac{\Tr[D]}{2\eps}
                + \frac{\eps\Tr[Y]}{2}.
            \end{align*}
            Optimizing $\eps$ in the above estimate, we get that
            \begin{align*}
            \cF_{c}(X+D,Y) \leq \frac{1}{2}\left( \langle\Gamma,X\rangle
            + \langle\Gamma^{-1},Y\rangle\right)
            + \sqrt{\Tr[D]\Tr[Y]}.
            \end{align*}
            Taking the infimum over feasible $\Gamma$ gives the upper bound.
        \end{proof}
        We now prove that $\liminf_{k} \cF_{c_k}(A_k,B_k|\{P_i\}) \geq \cF_{c}(A,B|\{P_i\})$. This
         time, we start with feasible solutions $C,X_i\succeq 0$  such that
        \begin{align*}
            &\cF_{c}(A,B|\{P_i\}) = \cF(C,B), && A - C = \sum_{i} X_i\circ(cJ - \bar{P_i}).
        \end{align*}
        Define $C_k := C + (A_k-A) - (c_k-c)\sum_i X_i$, so $C_k\to C$. If $C_k \succeq 0$, we would
        be done. However, that does not need to be the case. Instead, we take the Jordan
        decomposition $C_k = C_k^{+} - C_k^{-}$, where $C_k^{-}\to 0$, and write $(A_k+C_k^-)-C_k^+
        = \sum_i X_i\circ(c_kJ-\overline{P_i})$. Consequently, by the above claim,

\begin{align*}
    \cF(C_k^{+},B_k) \leq \cF_{c_k}(A_k + C_k^{-}, B_k ) \leq \cF_{c_k}(A_k,B_k) + \sqrt{\Tr[C_k^{-}]\Tr[B_k]}
\end{align*}
    Therefore, we have
    \begin{align*}
        \liminf_{k} \cF_{c_k}(A_k,B_k) &\geq \liminf_{k}\left(\cF(C_k^{+},B_k) - \sqrt{\Tr[C_k^{-}]\Tr[B_k]}\right) \\
                                        &= \cF(C,B)\\
                                        &= \cF_{c}(A,B),
    \end{align*}
    where we used continuity of fidelity as well as $C_{k}^{-} \to 0$.

        \item Let $1<c\leq d$. Then, $\Gamma\circ P_i\preceq c\Gamma$ implies that $\Gamma\circ
        P_i\preceq d\Gamma$. Therefore, $\cF_d(A,B|\{P_i\})\leq\cF_c(A,B|\{P_i\})$. Moreover, for
        any $c>1$, we have $I\circ P_i=I\preceq cI$. Thus,
        \begin{align*}
            \cF_c(A,B|\{P_i\})
            \leq\sqrt{\langle I,A\rangle\langle I,B\rangle}
            =\sqrt{\Tr[A]\Tr[B]}.
        \end{align*}
        Also, $0\leq\cF(A,B)\leq\cF_c(A,B|\{P_i\})$ by Alberti's characterization
        (\Cref{thm:alberti}) and restriction of the infimum. Finally, every fixed $\Gamma\succ0$ is
        feasible for all sufficiently large $c$. Therefore,
        \begin{align*}
            \lim_{c\to\infty}\cF_c(A,B|\{P_i\})
            &=\lim_{c\to\infty}
            \progline{\inf_{\Gamma\succ0}}
            {\sqrt{\langle\Gamma,A\rangle\langle\Gamma^{-1},B\rangle}}
            {\Gamma\circ P_i\preceq c\Gamma\quad\forall i}\\
            &=\inf_{\Gamma\succ0}\sqrt{\langle\Gamma,A\rangle\langle\Gamma^{-1},B\rangle}\\
            &=\cF(A,B).
        \end{align*}
        \item Let $A\preceq A'$ and $B\preceq B'$. For every feasible $\Gamma\succ0$, we have
        $\langle\Gamma,A\rangle\leq\langle\Gamma,A'\rangle$ and
        $\langle\Gamma^{-1},B\rangle\leq\langle\Gamma^{-1},B'\rangle$. Taking infima over the same
        feasible set proves monotonicity.
        \item
        If $A =0$ or $B=0$, the inequality is trivial. So assume $A,B\neq 0$. Moreover, we may
        assume $\Tr[A] = \Tr[B] = 1$ as we can scale both sides of the equation by a factor of
        $\sqrt{ab}$.
        Let $\Gamma \succ 0$ be an arbitrary matrix. We clip the eigenvalues of $\Gamma$ so that it
        becomes feasible for $\cF_{c}$. More specifically, we define
        \begin{align*}
            h(x) := \min\left\{\frac{1}{\sqrt{c}} + \frac{(\sqrt{c}-1)x}{\sqrt{c}}, \sqrt{c}
            \right\}.
\end{align*}
        Then, since $\sqrt{c} > 1/\sqrt{c}$,
        \begin{align*}
            \frac{1}{\sqrt{c}}\leq h(x) \leq \sqrt{c}.
        \end{align*}
        Therefore,
        \begin{align*}
            h(\Gamma) \circ P_i \preceq \sqrt{c} \cdot I \circ P_i = \sqrt{c} \cdot I = c \cdot (\frac{1}{\sqrt{c}} \cdot I) \preceq c \cdot h(\Gamma)
        \end{align*}
        So $h(\Gamma)$ is feasible for $\cF_{c}$. Additionally, for $x>0$, we observe that
        \begin{align*}
            h(x) \leq \frac{1}{\sqrt{c}} + \frac{(\sqrt{c}-1)x}{\sqrt{c}}
        \end{align*}
        by definition, and
        \begin{align*}
            \frac{1}{h(x)} &\leq \max\{ (\frac{1}{\sqrt{c}} + \frac{(\sqrt{c}-1)x}{\sqrt{c}})^{-1}, (\sqrt{c})^{-1}\}\\
                            &\leq \max\{ \frac{1}{\sqrt{c}} + \frac{(\sqrt{c}-1)}{\sqrt{c}}\frac{1}{x} ,\frac{1}{\sqrt{c}}\} \\
                & = \frac{1}{\sqrt{c}} + \frac{(\sqrt{c}-1)}{\sqrt{c}}\frac{1}{x},
        \end{align*}
        where we used $(a+bx)^{-1} \leq a+ b/x$ for $a,b,x> 0$ with $a+b \geq 1$.
        Therefore,
        \begin{align*}
            \langle h(\Gamma),A\rangle
            + \langle h(\Gamma)^{-1},B\rangle
            &\leq
            \left\langle
                \frac{1}{\sqrt{c}}I
                + \frac{\sqrt{c}-1}{\sqrt{c}}\Gamma, A
            \right\rangle
            + \left\langle
                \frac{1}{\sqrt{c}}I
                + \frac{\sqrt{c}-1}{\sqrt{c}}\Gamma^{-1}, B
            \right\rangle\\
            &= \frac{1}{\sqrt{c}}\bigl(\Tr[A]+\Tr[B]\bigr)
            + \frac{\sqrt{c}-1}{\sqrt{c}}
            \bigl(\langle\Gamma,A\rangle
            + \langle\Gamma^{-1},B\rangle\bigr)\\
            &= \frac{2}{\sqrt{c}}
            + \frac{\sqrt{c}-1}{\sqrt{c}}
            \bigl(\langle\Gamma,A\rangle
            + \langle\Gamma^{-1},B\rangle\bigr).
        \end{align*}
        Since $\Gamma$ was arbitrary, taking the infimum over $\Gamma$ and using the additive
        formulations of fidelity gives
        \begin{align*}
            \cF_c(A,B|\{P_i\})
            \leq \frac{1+(\sqrt{c}-1)\cF(A,B)}{\sqrt{c}}.
        \end{align*}
        \item We prove each direction as follows. First, let $\Gamma\succ0$ be feasible on $\Sigma$.
        Taking principal submatrices gives $(\Gamma|_S)\circ(P_i|_S)\preceq c(\Gamma|_S)$ for every
        $i$. Moreover, the Schur complement gives $(\Gamma|_S)^{-1}\preceq(\Gamma^{-1})|_S$. Since
        $A\circ J_S$ and $B\circ J_S$ are supported on $S$, we have
        \begin{align*}
            \langle\Gamma|_S,A|_S\rangle&=\langle\Gamma,A\circ J_S\rangle,\\
            \langle(\Gamma|_S)^{-1},B|_S\rangle
            &\leq\langle(\Gamma^{-1})|_S,B|_S\rangle
            =\langle\Gamma^{-1},B\circ J_S\rangle.
        \end{align*}
        Thus, taking the infimum over feasible $\Gamma$ gives
        \[
            \cF_c(A|_S,B|_S\mid\{P_i|_S\})
            \leq\cF_c(A\circ J_S,B\circ J_S\mid\{P_i\}).
        \]
        Conversely, let $H\succ0$ be feasible on $S$, and define $\Gamma:=H\oplus I_{\Sigma\setminus
        S}$. Then,
        \[
            \Gamma\circ P_i
            =(H\circ(P_i|_S))\oplus I_{\Sigma\setminus S}
            \preceq c\Gamma.
        \]
        Since $A\circ J_S=A|_S\oplus0$ and $B\circ J_S=B|_S\oplus0$, this extension preserves the
        objective. Taking the infimum over feasible $H$ gives
        \[
            \cF_c(A\circ J_S,B\circ J_S\mid\{P_i\})
            \leq\cF_c(A|_S,B|_S\mid\{P_i|_S\}),
        \]
        proving the equality.

    \end{enumerate}
    For the sign filters, we prove the remaining properties as follows.
    \begin{enumerate}[i.]
        \setcounter{enumi}{10}
        \item Let $D_i := \sum_{x} (-1)^{x_i} \ketbra{x}$ be a diagonal matrix with $D_i^2 = I$.
        Then, for $\Gamma \succ 0$ we have $\Gamma \circ S_i = D_i \Gamma D_i$. Consequently,
        \begin{align*}
        D_i\Gamma^{-1}D_i \preceq c\Gamma^{-1} \iff
        \Gamma^{-1} \preceq c \cdot D_i \Gamma^{-1} D_i \iff
        \Gamma \succeq \frac{1}{c} D_i \Gamma D_i \iff
        c\Gamma \succeq D_i \Gamma D_i.
        \end{align*}
        Therefore, $\Gamma \succ 0$ is feasible for $\cF_{c}$ if and only if $\Gamma^{-1}\succ 0$
        is. Consequently,
        \begin{align*}
        \cF_{c}(A,B|\{S_i\}) &= \progline{\inf_{\Gamma \succ 0}}{\sqrt{\langle \Gamma, A\rangle \langle \Gamma^{-1},B \rangle}}{\Gamma \circ S_i \preceq c \Gamma & (\forall i)} \\
        &=
         \progline{\inf_{\Gamma \succ 0}}{\sqrt{\langle \Gamma^{-1}, A\rangle \langle \Gamma,B \rangle}}{\Gamma\circ S_i \preceq c \Gamma & (\forall i)} \\
        &= \cF_{c}(B,A|\{S_i\}).
        \end{align*}
        \item
        Suppose $A\circ B = 0$, and let $I_{A} = \sum_{x:A(x,x)>0} \ketbra{x}$ be the identity
        matrix restricted to non-zero diagonals of $A$. We define $\Gamma_{t} := tI_{A} + t^{-1}(I-
        I_{A})$. Since $\Gamma_{t}$ is diagonal, $\Gamma_t \circ S_i = \Gamma_t\leq c\Gamma_t$ is a
        feasible solution for any $c>1$. Consequently,
        \begin{align*}
        \cF_{c}(A,B|\{S_i\}) \leq \inf_{t>0} \sqrt{\langle \Gamma_t, A\rangle \langle \Gamma_t^{-1},B\rangle}
        = \inf_{t>0} \sqrt{t\Tr[A]\cdot t\Tr[B]} = 0.
        \end{align*}
         Next, suppose there exist $x,y$ such that $A(x,y)B(x,y) \neq 0$. Since $A,B\succeq 0$, this
         implies that $A(x,x)B(x,x) \neq 0$. Naturally, we use a primal form to complete the proof.
         The idea is to make the matrix $A$ into a diagonal matrix $C$ through the sign matrices
         $S_i$ while maintaining a lower bound on the fidelity. We define $C_0 := A$ and for $i\in
         [m]$ we define
         \begin{align*}
         C_{i} := \frac{C_{i-1} \circ (J+S_i)}{c+1} = \frac{C_{i-1} + C_{i-1}\circ S_i}{c+1} = C_{i-1} - \frac{C_{i-1}}{c+1} \circ (cJ-S_i)
         \end{align*}
         The first equality shows that $C_{i}$ remains PSD for all $i \in [m]$. Eventually, at the
         end we get
         \begin{align*}
            C_{m} = \frac{C_0 \circ (J+ S_1) \circ \dots \circ (J+ S_m)}{(c+1)^{m}}
                = \frac{2^m}{(c+1)^m}(A\circ I) \succeq 0.
         \end{align*}
         In addition, we can write
         \begin{align*}
         A - C_m = \sum_{i\in [m]}\bigl(C_{i-1} - C_i\bigr)
                =  \sum_{i\in [m]} \frac{C_{i-1}}{c+1} \circ (cJ - S_i)
         \end{align*}

         Consequently,
         \begin{align*}
            \cF_{c}(A,B) &\geq \cF(C_m, B) \geq \sqrt{\Tr[C_m B]}
            = (\frac{2}{c+1})^{m/2}\sqrt{\Tr[(A\circ I )B]} \\
            &\geq (\frac{2}{c+1})^{m/2} \sqrt{A(x,x)B(x,x)} > 0. \qedhere
         \end{align*}
    \end{enumerate}
\end{proof}

\section{Analogies between measures}\label{app:measure-connections}

\subsection{\texorpdfstring{Proof of \Cref{prop:average-to-worst-gamma2}: analogy between
distributional and worst-case $\gamma_2$ norms}{Proof of the proposition: Analogy between
distributional and worst-case gamma 2 norms}}\label{app:average-to-worst-gamma2}

For the proof, we use the following more general lemma that relates inner products with probability
distributions to the spectral norm.
\begin{lemma}\label{lem:max-over-dist} Let $\Sigma$ be a finite set and $\mathcal{K}\subset
\Herm(\Sigma)$ be invariant under negation and conjugation with diagonal phase matrices. That is,
$\Gamma\in\mathcal{K}$ if and only if $-\Gamma\in\mathcal{K}$, and for any diagonal phase matrix
$D\in \bC^{\Sigma\times\Sigma}$, where $|D(x,x)|=1$ for all $x\in \Sigma$, and any
$\Gamma\in\mathcal{K}$, we have $D^{\dagger}\Gamma D\in\mathcal{K}$. Then,
    $$\max_{\mu} \sup_{\Gamma\in \mathcal{K} } \langle \Gamma, \kbr{\mu}\rangle = \sup_{\Gamma\in
    \mathcal{K}} \|\Gamma\|$$ where the maximum is over all distributions $\mu:\Sigma \to [0,1]$ and
    is either attained or is unbounded.
\end{lemma}
\begin{proof}
    For every $\Gamma\in\mathcal{K}$ and every distribution $\mu$, we have
    $\langle\Gamma,\kbr{\mu}\rangle\leq\|\Gamma\|$. For the other direction, let
    $\Gamma\in\mathcal{K}$ and let $\ket{u}\in\bC^{\Sigma}$ be a unit eigenvector satisfying
    $|\bra{u}\Gamma\ket{u}|=\|\Gamma\|$. Replacing $\Gamma$ by $-\Gamma$ if necessary, we assume
    that $\bra{u}\Gamma\ket{u}=\|\Gamma\|$. Also, let $\mu:\Sigma\to[0,1]$ and
    $\theta:\Sigma\to[0,2\pi)$ be such that $(\ket{u})_x=\sqrt{\mu(x)}e^{i\theta(x)}$. Then, $\mu$
    is a distribution since $\|u\|=1$, and for the diagonal matrix $D(x,x)=e^{i\theta(x)}$ we have
    $D\ket{\mu}=\ket{u}$, where as usual $\braket{x}{\mu}=\sqrt{\mu(x)}$. Defining
    $\Gamma':=D^{\dagger}\Gamma D$, we have
    $$\|\Gamma'\|=\|\Gamma\|=\bra{u}\Gamma\ket{u} =\bra{\mu}D^{\dagger}\Gamma
    D\ket{\mu}=\bra{\mu}\Gamma'\ket{\mu}.$$ Therefore,
    \begin{align*}
        \|\Gamma\|
        &=\bra{\mu}\Gamma'\ket{\mu}
        \leq\sup_{K\in\mathcal{K}}\bra{\mu}K\ket{\mu}
        \leq\sup_{\mu}\sup_{K\in\mathcal{K}}\bra{\mu}K\ket{\mu}.
    \end{align*}
    Taking the supremum over $\Gamma\in\mathcal{K}$ gives
    \begin{align*}
        \sup_{\Gamma \in \cK} \|\Gamma\| = \sup_{\mu}\sup_{K\in\mathcal{K}}\bra{\mu}K\ket{\mu}.
    \end{align*}
    It only remains to show that the maximum is attained or unbounded.
    \par
    Indeed, if $\mathcal{K}$ is bounded, $\sup_{K\in \cK} \bra{\mu}K\ket{\mu}$ is continuous in
    $\mu$. Since the set $\{\mu:\Sigma \to [0,1]: \sum_{x}\mu(x) = 1\}$ is compact, the maximum is
    attained. Otherwise, choose $K_j\in\cK$ with $\|K_j\|\to\infty$. By the construction above,
    there exist $K'_j\in\cK$ and distributions $\mu_j$ such that
    $\bra{\mu_j}K'_j\ket{\mu_j}=\|K_j\|\to\infty$.
\end{proof}
\begin{proof}
Let $\cK:=\{\Gamma\circ\overline{A}:\Gamma=\Gamma^{\dagger}, \ \|\Gamma\circ Z_i\|\leq 1\text{ for
all }i\in[n]\}.$ By \Cref{prop:equiv-avg-gamma2},
$$\max_{\mu}\bar{\gamma}_2^{\mu}(A|\{Z_i\}) =\max_{\mu}\max_{K\in\cK}\langle K,\kbr{\mu}\rangle.$$
Since $A$ is Hermitian, $\cK\subset\Herm(\Sigma)$. We prove that $\cK$ satisfies the conditions of
\Cref{lem:max-over-dist}. First, $K\in\cK$ if and only if $-K\in\cK$. Let
$K=\Gamma\circ\overline{A}\in\cK$ and $\theta:\Sigma\to\bC$ be such that $|\theta(x)|=1$ for all
$x\in\Sigma$. Let $D:=\operatorname{Diag}(\theta)$ and define $\Gamma':=D^{\dagger}\Gamma D$. We
want to prove that $K':=D^{\dagger}KD\in\cK$ as well. First, since $D$ is diagonal, we can write
$$K'=D^{\dagger}(\Gamma\circ\overline{A})D =(D^{\dagger}\Gamma
D)\circ\overline{A}=\Gamma'\circ\overline{A}.$$ Moreover, $\Gamma'=\Gamma'{}^{\dagger}$, and for all
$i\in[n]$, we have
\begin{align*}
\|\Gamma'\circ Z_i\|
&=\|D^{\dagger}(\Gamma\circ Z_i)D\|
=\|\Gamma\circ Z_i\|\leq1.
\end{align*}
Therefore, $K'\in\cK$. It immediately follows that
\begin{align*}
\max_{\mu}\bar{\gamma}_2^{\mu}(A|\{Z_i\})
&=\max_{\mu}\max_{K\in\cK}\langle K,\kbr{\mu}\rangle
=\max_{K\in\cK}\|K\|
=\gamma_2(A|\{Z_i\}),
\end{align*}
where we applied \Cref{lem:max-over-dist} and the dual formulation \eqref{eq:filtered-gamma2-dual},
respectively.
\end{proof}

\subsection{\texorpdfstring{Proof of \Cref{prop:recovering-worst-case-adversary}: analogy between
distributional and worst-case adversary bounds}{Proof of the proposition: Analogy between
distributional and worst-case adversary bounds}}\label{app:recovering-worst-case-adversary}

\begin{proof}
    We first prove the upper bound. For every distribution $\mu$, we have $A_f\in\cW_f^\mu$ since
    $A_f\succeq0$, $A_f\circ\Delta_f=0$, and $\Tr[A_f\circ\kbr{\mu}]=1$. Moreover, $A_f\circ I=I$,
    and choosing $W=A_f$ gives fidelity $1$. Thus, $A_f\in\cG_\eps^\mu(f)$. Consequently,
    \begin{align*}
        \max_{\mu}\mathrm{Adv}^{\mu}_{\eps}(f)
        &\leq \max_{\mu}\bar{\gamma}^{\mu}_2(J-A_f|\{\Delta_i\})
        && \text{(By \Cref{def:adv})} \\
        &=\max_{\mu}\bar{\gamma}^{\mu}_2(\Delta_f|\{\Delta_i\}) \\
        &=\gamma_2(\Delta_f|\{\Delta_i\})
        && \text{(By \Cref{prop:average-to-worst-gamma2})} \\
        &=O\bigl(\mathrm{Adv}^{\pm}(f)\bigr)
        && \text{(By the adversary equivalence \cite{LMR+11}).}
    \end{align*}
    When $f$ is a decision problem, $\gamma_2(\Delta_f\mid\{\Delta_i\})=\mathrm{Adv}^{\pm}(f)$,
    which gives the sharper upper bound.
    \par
    For the other direction, fix a distribution $\mu$ and let $\Psi\in\cG^{\mu}_{\eps}(f)$. Write
    $\Delta_i^f:=\Delta_i\circ\Delta_f$. Since $\Delta_f\circ\Delta_f=\Delta_f$, by
    \Cref{prop:avggamma2-properties},
    \begin{align*}
        \bar{\gamma}_2^{\mu}(J-\Psi|\{\Delta_i\})
        &\geq \bar{\gamma}_2^{\mu}((J-\Psi)\circ\Delta_f|\{\Delta_i^f\})
        && \text{(By \Crefitem{item:avggamma2-hadamard-bounds})} \\
        &\geq \bar{\gamma}_2^{\mu}(\Delta_f|\{\Delta_i^f\})
        -\bar{\gamma}_2^{\mu}(\Psi\circ\Delta_f|\{\Delta_i^f\})
        && \text{(By (i))} \\
        &\geq \bar{\gamma}_2^{\mu}(\Delta_f|\{\Delta_i^f\})
        -\|\Psi\circ\Delta_f\circ\kbr{\mu}\|_{\tr}
        \gamma_2(\Delta_f|\{\Delta_i^f\})
        && \text{(By \Crefitem{item:avggamma2-tr-upper})} \\
        &\geq \bar{\gamma}_2^{\mu}(\Delta_f|\{\Delta_i^f\})
        -4\sqrt{\eps}\,\mathrm{Adv}^{\pm}(f)
        && \text{(By \Cref{lem:trace-norm-upper-bound-for-boolean-functions}).}
    \end{align*}
    Taking the minimum over $\Psi$ and then the maximum over $\mu$, we obtain
    \begin{align*}
        \max_{\mu}\mathrm{Adv}^{\mu}_{\eps}(f)
        &\geq \max_{\mu}\bar{\gamma}_2^{\mu}
        (\Delta_f|\{\Delta_i^f\})-4\sqrt{\eps}\,\mathrm{Adv}^{\pm}(f) \\
        &=\gamma_2(\Delta_f|\{\Delta_i^f\})
        -4\sqrt{\eps}\,\mathrm{Adv}^{\pm}(f)
        && \text{(By \Cref{prop:average-to-worst-gamma2})} \\
        &=(1-4\sqrt{\eps})\mathrm{Adv}^{\pm}(f).
    \end{align*}
    Finally, if $f$ is a decision problem, \Cref{lem:trace-norm-upper-bound-for-boolean-functions}
    gives $\|\Psi\circ\Delta_f\circ\kbr{\mu}\|_{\tr}\leq 2\sqrt{\eps(1-\eps)}$. Repeating the same
    argument gives the sharper lower bound.
\end{proof}

\subsection{\texorpdfstring{Proof of \Cref{prop:delta2-gamma2-connection}: analogy between
multiplicative and additive $\gamma_2$ measures}{Proof of the proposition: Analogy between
multiplicative and additive gamma 2 measures}}\label{app:delta2-gamma2-connection}

\begin{proof}
For the sake of notation, we abbreviate
$F_c:=\cF_c(\Phi\circ\kbr{\mu},\Psi\circ\kbr{\mu}\mid\{S_i\})$ for $c>1$. Also, let us denote
\begin{align*}
    R &:= \sup_{c>1}\frac{1-F_c^2}{c-1},
    && \gamma:=\bar{\gamma}_2^\mu(\Phi-\Psi\mid\{\Delta_i\}),
    && \delta:=\bar{\delta}_2^\mu(\Phi,\Psi\mid\{S_i\}).
\end{align*}
We prove the equalities by showing $\gamma/2\leq R\leq\delta\leq\gamma/2$.

\paragraph{$\gamma/2\leq R$.}
Let $\Gamma = \Gamma^{\dagger}$ with $\max_{i}\|\Gamma \circ \Delta_i\|\leq 1$ attain $\gamma$, so
that
\begin{align*}
    \gamma = \langle \Gamma,(\Phi - \Psi)\circ\kbr{\mu}\rangle.
\end{align*}
Let $t>0$ be arbitrarily small so that $t\|\Gamma\|<1$, and define $\Gamma_{t} := I - t\Gamma \succ
0$. Therefore,
\begin{align*}
    \Gamma_{t} \circ S_i = I \circ S_i - t \Gamma \circ (J-2\Delta_i) =   \Gamma_{t} + 2t(\Gamma \circ \Delta_i)
    \preceq \Gamma_t + 2tI \preceq \left(1 + \frac{2t}{1-t\|\Gamma\|}\right)\Gamma_t,
\end{align*}
where we used $\Gamma \circ \Delta_i \preceq I$ and $(1-t\|\Gamma\|) I \preceq \Gamma_{t}$,
respectively. As a result, for $c_t = 1 + \frac{2t}{1-t\|\Gamma\|}$ we have $\Gamma_{t} \circ S_i
\preceq c_t\Gamma_{t}$. In addition, using the Taylor expansion $(1-tx)^{-1} = 1 + tx + O(t^2)$
around $t=0$, we can bound the inverse as
\begin{align*}
    \Gamma_{t}^{-1} \preceq I + t\Gamma + O(t^2) I
\end{align*}
Consequently, we can upper bound $F_{c_t}^2$ by
\begin{align*}
    F_{c_t}^2
    &\quad\leq
    \langle\Gamma_t,\Phi\circ\kbr{\mu}\rangle
    \langle\Gamma_t^{-1},\Psi\circ\kbr{\mu}\rangle\\
    &\quad\leq
    \langle I-t\Gamma,\Phi\circ\kbr{\mu}\rangle
    \langle I+t\Gamma+O(t^2)I,\Psi\circ\kbr{\mu}\rangle\\
    &\quad=
    \bigl(1-t\langle\Gamma,\Phi\circ\kbr{\mu}\rangle\bigr)
    \bigl(1+t\langle\Gamma,\Psi\circ\kbr{\mu}\rangle+O(t^2)\bigr)\\
    &\quad=1-t\langle\Gamma,(\Phi-\Psi)\circ\kbr{\mu}\rangle+O(t^2)\\
    &\quad=1-t\gamma+O(t^2).
\end{align*}
Therefore, taking the limit as $t\to 0 $, we get
\begin{align*}
    \lim_{t\to 0}
    \frac{1-F_{c_t}^2}{c_t-1}
    &\geq \lim_{t\to 0}
    \frac{1-t\|\Gamma\|}{2t}\bigl(t\gamma-O(t^2)\bigr)\\
    &=\lim_{t\to 0}
    (1-t\|\Gamma\|)\left(\frac{\gamma}{2}-O(t)\right)\\
    &=\frac12\bar{\gamma}_2^\mu(\Phi-\Psi\mid\{\Delta_i\}),
\end{align*}
completing the first inequality.
\paragraph{$R\leq\delta$.}
    By \Crefitem{item:filtered-fidelity-cost-monotonicity}, for every $c>1$ we have
    \begin{align*}
        0 \leq F_c^2
        \leq \Tr[\Phi\circ\kbr{\mu}]\Tr[\Psi\circ\kbr{\mu}] = 1.
    \end{align*}
    The function $x\to \frac{\log x}{x-1}$ is monotone non-increasing for $x>0$, with its continuous
    extension at $x=1$. Therefore, for $F_c^2 <1 <c$, we get that
    \begin{align*}
        \frac{\log F^2_c}{\cF^2_{c}-1} \geq \frac{\log c}{c-1} \Rightarrow
        \frac{1-F_c^2}{c-1}\leq \frac{-2\log F_c}{\log c}.
    \end{align*}
    The above inequality is also trivial when $\cF_c = 1$. Therefore, taking suprema over $c>1$ on
    both sides proves $R\leq\delta$.

\paragraph{$\delta\leq\gamma/2$.} Let $c>1$ be arbitrary and let $\Gamma \succ 0$ with $\Gamma \circ
    S_i \preceq c \Gamma$. Then, we also get the direction $\Gamma = (\Gamma \circ S_i)\circ S_i
    \preceq c \Gamma \circ S_i$. Therefore,
    \begin{align*}
            \frac{1}{c} \Gamma \preceq \Gamma \circ S_i \preceq c \Gamma.
    \end{align*}
    We claim that $\log (\Gamma \circ S_i) = \log(\Gamma) \circ S_i$. To see this, we define the
    unitary matrix $D_i := \mathrm{Diag}((-1)^{x_i})$ so that $X \circ S_i = D_i X D_i$. Taking the
    PSD decomposition $\Gamma = \sum_{j} \lambda_{j} \ketbra{v_j}$ we have that
    \begin{align*}
        \log (\Gamma \circ S_i)
        &= \log\left(\sum_{j} \lambda_{j} \ketbra{D_iv_j}\right)
        = \sum_{j} \log(\lambda_j) \ketbra{D_iv_j}
        = \log (\Gamma) \circ S_i.
    \end{align*}
    Therefore, applying the logarithm to the above inequalities, we get
    \begin{align*}
        \log(\Gamma) - (\log c)I \preceq \log(\Gamma)\circ S_i
        \preceq \log(\Gamma) + (\log c)I.
    \end{align*}
    Consequently, for $i\in [n]$ we can bound the spectral norm as follows.
    \begin{align*}
        \|\log(\Gamma)\circ\Delta_i\|
        &= \frac{\|\log(\Gamma)\circ(J-S_i)\|}{2}
        = \frac{\|\log(\Gamma)-\log(\Gamma)\circ S_i\|}{2}
        \leq \frac{\log c}{2}.
    \end{align*}
    Moreover, for a matrix $A\succeq0$ with $\Tr[A] =1$ we have
    \begin{align*}
        \langle \log \Gamma, A\rangle \leq \log \langle \Gamma, A\rangle
    \end{align*}
    Indeed, $ \langle \log \Gamma, A\rangle = \sum_{j} \log(\lambda_j) \bra{v_j} A \ket{v_j} \leq
        \log (\sum_{j}  \lambda_j \bra{v_j} A \ket{v_j}) = \log \langle \Gamma, A\rangle.
    $
    Taking $A:=\Phi\circ\kbr{\mu}$ and $B:=\Psi\circ\kbr{\mu}$, we get
    \begin{align*}
        \langle \log \Gamma, A - B\rangle =
        \langle \log \Gamma, A\rangle + \langle \log \Gamma^{-1}, B\rangle
        \leq \log (\langle \Gamma, A\rangle) + \log (\langle  \Gamma^{-1}, B\rangle )
        = \log (\langle \Gamma, A\rangle \langle \Gamma^{-1},B\rangle).
    \end{align*}
    Since $\Gamma$ was arbitrary among the feasible matrices for $\cF_c$, we have
    \begin{align*}
        \frac{-2\log F_c}{\log c}
        &\leq \progline
        {\sup_{\Gamma\succ0}}
        {\frac{-\log\bigl(\langle\Gamma,A\rangle
        \langle\Gamma^{-1},B\rangle\bigr)}{\log c}}
        {\Gamma\circ S_i\preceq c\Gamma\quad\forall i\in[n]}\\
        &\leq \frac12\sup_{\Gamma\succ0}
        \frac{|\langle\log\Gamma,A-B\rangle|}
        {\max_i\|(\log\Gamma)\circ\Delta_i\|}\\
        &=\frac12\bar{\gamma}_2^\mu(\Phi-\Psi\mid\{\Delta_i\}).
    \end{align*}
    Here the norm quotient is interpreted as zero when its denominator vanishes; then $\log\Gamma$
    is diagonal and $\langle\log\Gamma,A-B\rangle=0$. The last equality follows from the definition
    of $\bar{\gamma}_2^\mu$, since $\log\Gamma$ ranges over all Hermitian matrices. Since $c$ was
    also arbitrary, taking the supremum over $c>1$ proves the desired bound.
\end{proof}

\end{document}